\PassOptionsToPackage{svgnames}{xcolor}

\documentclass[a4paper,UKenglish,cleveref,autoref,thm-restat]{lmcs}
\pdfoutput=1

\usepackage{lastpage}
\lmcsdoi{21}{2}{5}
\lmcsheading{}{\pageref{LastPage}}{}{}%
{Nov.~30,~2023}{Apr.~18,~2025}{}

\usepackage[utf8]{inputenc}

\keywords{Session Types, Failure Handling, Concurrency, Session Calculus}

\usepackage{booktabs}   %
\usepackage{subcaption} %
\usepackage{amsthm}
\usepackage{amssymb}
\usepackage{mathrsfs}
\usepackage{titletoc}
\usepackage{etoolbox}%

\usepackage[nomargin,inline,index,%
  status=draft %
]{fixme} %
\fxusetheme{colorsig}
\FXRegisterAuthor{fxAS}{anfxAS}{AS}%
\FXRegisterAuthor{fxAB}{anfxAB}{AB}%
\FXRegisterAuthor{fxNY}{anfxNY}{NY}%
\FXRegisterAuthor{fxFZ}{anfxFZ}{FZ}
\FXRegisterAuthor{fxPH}{anfxPH}{PH}

\makeatletter
\DeclareFontFamily{OMX}{MnSymbolE}{}
\DeclareSymbolFont{MnLargeSymbols}{OMX}{MnSymbolE}{m}{n}
\SetSymbolFont{MnLargeSymbols}{bold}{OMX}{MnSymbolE}{b}{n}
\DeclareFontShape{OMX}{MnSymbolE}{m}{n}{
    <-6>  MnSymbolE5
   <6-7>  MnSymbolE6
   <7-8>  MnSymbolE7
   <8-9>  MnSymbolE8
   <9-10> MnSymbolE9
  <10-12> MnSymbolE10
  <12->   MnSymbolE12
}{}
\DeclareFontShape{OMX}{MnSymbolE}{b}{n}{
    <-6>  MnSymbolE-Bold5
   <6-7>  MnSymbolE-Bold6
   <7-8>  MnSymbolE-Bold7
   <8-9>  MnSymbolE-Bold8
   <9-10> MnSymbolE-Bold9
  <10-12> MnSymbolE-Bold10
  <12->   MnSymbolE-Bold12
}{}

\let\llangle\@undefined
\let\rrangle\@undefined
\DeclareMathDelimiter{\llangle}{\mathopen}%
                     {MnLargeSymbols}{'164}{MnLargeSymbols}{'164}
\DeclareMathDelimiter{\rrangle}{\mathclose}%
                     {MnLargeSymbols}{'171}{MnLargeSymbols}{'171}
\makeatother

\usepackage[inline]{enumitem}%

\usepackage{nicefrac}%

\usepackage{proof}%

\usepackage[most]{tcolorbox}%
\newtcolorbox{myframe}[1][]{
  enhanced,
  arc=0pt,
  outer arc=0pt,
  colback=white,
  boxrule=0.5pt,
  boxsep=0mm,
  left=1mm,
  right=1mm,
  top=0.5mm,
  bottom=0.5mm,
  #1
}

\lstdefinelanguage{Scribble}{%
  basicstyle=\ttfamily,
  stringstyle=\color{Blue},
  showstringspaces=false,
  keywords={nested,new,calls,and,as,at,by,catches,choice,continue,do,from,global,import,instantiates,interruptible,local,module,or,par,protocol,rec,role,sig,throws,to,type,with,int,aux,reliable,crash},
  morestring=[b]",
  morestring=[b]',
  morecomment=[l][\color{greencomments}]{//},
}

\lstdefinelanguage{nuScr}{%
  basicstyle=\footnotesize\ttfamily,
  stringstyle=\color{Blue},
  showstringspaces=false,
  keywords={
    nested,new,calls,and,as,at,by,catches,choice,continue,do,from,global,import,instantiates,interruptible,local,module,or,par,protocol,rec,role,sig,throws,to,type,with,int,aux,
    safe
  },
  morestring=[b]",
  morestring=[b]',
  morecomment=[l][\color{greencomments}]{//},
  morecomment=[s][\color{magenta}]{(*}{*)},
}

\lstdefinelanguage{effpi}{
  keywords=[1]{
    case,class,sealed,abstract,object,extends,type,def,val,if,else,new,var,match
  },
  keywords=[2]{
    InChan,OutChan,RecVar,Rec,Out,In,InErr,Loop,
  },
  keywords=[3]{
    rec,send,receive,receiveErr,eval,par,Channel
  },
  keywordstyle=[1]{\color{blue}},
  keywordstyle=[2]{\color{ImperialIris}}, %
  keywordstyle=[3]{\color{OliveGreen}},
  otherkeywords={=>,.type,<:,>>:},
  morecomment=[l][\color{darkgray}]{//},
}

\usepackage{stmaryrd}
\usepackage{thmtools}%
\usepackage{arydshln}

\usepackage{xifthen}%
\usepackage{xspace}%

\usepackage{mathtools}%

\usepackage{balance}%

\usepackage{hyperref}
\hypersetup{hidelinks}
\usepackage[capitalise]{cleveref}

\definecolor{ImperialBlue}{HTML}{003E74}
\definecolor{ImperialDarkGreen}{HTML}{02893B}
\definecolor{ImperialTangerine}{HTML}{EC7300}
\definecolor{ImperialIris}{HTML}{751E66}

\definecolor{RYB1}{RGB}{141, 211, 199}
\definecolor{RYB2}{RGB}{255, 255, 179}
\definecolor{RYB3}{RGB}{190, 186, 218}
\definecolor{RYB4}{RGB}{251, 128, 114}
\definecolor{RYB5}{RGB}{128, 177, 211}
\definecolor{RYB6}{RGB}{253, 180, 98}
\definecolor{RYB7}{RGB}{179, 222, 105}

\usetikzlibrary{automata, positioning, arrows, calc, fit}
\tikzset{
  >=stealth,
  node distance=2cm,
  every state/.style={thick, fill=gray!10},
  initial text=$ $,
}
\usepackage{pgfplots}
\pgfplotscreateplotcyclelist{stackexchange}{
  {RYB1!50!black,fill=RYB1},
  {RYB4!50!black,fill=RYB4},
  {RYB3!50!black,fill=RYB3},
  {RYB2!50!black,fill=RYB2},
  {RYB5!50!black,fill=RYB5},
  {RYB6!50!black,fill=RYB6},
  {RYB7!50!black,fill=RYB7},
}
\pgfplotscreateplotcyclelist{stackexchangePlusOne}{
  {RYB4!50!black,fill=RYB4},
  {RYB3!50!black,fill=RYB3},
  {RYB2!50!black,fill=RYB2},
  {RYB5!50!black,fill=RYB5},
  {RYB6!50!black,fill=RYB6},
  {RYB7!50!black,fill=RYB7},
}
\pgfplotsset{
  compat=1.8,
  /pgfplots/bar cycle list/.style={/pgfplots/cycle list={%
    {brown!60!black,fill=brown!30!white,mark=none},
    {red,fill=red!30!white,mark=none},
    {blue,fill=blue!30!white,mark=none},
    {black,fill=gray,mark=none},
    }
  },
}
\usepackage{tabularx}
\newcolumntype{L}{>{$}l<{$}}
\newcolumntype{C}{>{$}c<{$}}
\newcolumntype{P}[1]{>{\centering\arraybackslash$}p{#1}<{$}}
\usepackage{multirow}

\usepackage[export]{adjustbox}

\Crefname{section}{\S\!}{\S\!}%
\Crefname{subsection}{\S\!}{\S\!}%
\Crefname{subsubsection}{\S\!}{\S\!}%
\Crefname{appendix}{Appendix \S\!}{Appendix \S\!}
\Crefname{definition}{Def.\@}{Defs.\@}%
\Crefname{figure}{Fig.\@}{Figs.\@}%
\Crefname{example}{Ex.\@}{Exs.\@}%
\Crefname{corollary}{Cor.\@}{Cors.\@}%
\Crefname{theorem}{Thm.\@}{Thms.\@}%
\Crefname{proposition}{Prop.\@}{Props.\@}%
\Crefname{lemma}{Lem.\@}{Lems.\@}
\Crefname{equation}{Eq.\@}{Eqs.\@}

\crefname{section}{\S\!}{\S\!}%
\crefname{subsection}{\S\!}{\S\!}%
\crefname{subsubsection}{\S\!}{\S\!}%
\crefname{appendix}{Appendix \S\!}{Appendix \S\!}
\crefname{definition}{Def.\@}{Defs.\@}%
\crefname{figure}{Fig.\@}{Figs.\@}%
\crefname{example}{Ex.\@}{Exs.\@}%
\crefname{corollary}{Cor.\@}{Cors.\@}%
\crefname{theorem}{Thm.\@}{Thms.\@}%
\crefname{proposition}{Prop.\@}{Props.\@}%
\crefname{lemma}{Lem.\@}{Lems.\@}
\crefname{equation}{Eq.\@}{Eqs.\@}

\usepackage{cite} %

\usepackage{magicvariables}

\newif\ifdraft%
  \draftfalse%
  \drafttrue%

\newcommand{\ifempty}[3]{%
  \ifthenelse{\isempty{#1}}{#2}{#3}%
}%

\newtoggle{techreport}%

\newtoggle{cruft}%

\newcommand{\makeFunc}[2]{%
  \expandafter\newcommand\csname #1\endcsname[1]{%
    \operatorname{#2}\!\left({##1}\right)%
  }%
}%
\newcommand{\dom}[1]{{\color{black}\operatorname{dom}\!\left({#1}\right)}}%
\makeFunc{fsucc}{suc}
\makeFunc{finv}{inv}
\newcommand{\fv}[1]{\operatorname{fv}\!\left({#1}\right)}%
\newcommand{\unfoldOne}[1]{%
  {\color{black}\operatorname{unf}\!\left({#1}\right)}}%
\newcommand{\notImpliedBy}{\mathrel{{\kern 1em}{\not{\kern -1em}\impliedby}}}%

\newcommand{\coloncolonequals}{\Coloneqq}%
\newcommand{\bnfdef}{\coloncolonequals}%
\newcommand{\bnfsep}{\mathbin{\;\big|\;}}%

\newcommand{\Erlang}[0]{\textsc{Erlang}\xspace}

\def\aka{a.k.a.\@\xspace}%
\def\cf{cf.\@\xspace}%
\def\eg{e.g.\@\xspace}%
\def\ie{i.e.\@\xspace}%
\def\wrt{w.r.t.\@\xspace}%
\def\etal{\emph{et al.}\@\xspace}%

\definecolor{ruleColor}{rgb}{0.1, 0.3, 0.1}%
\newcommand{\inferrule}[1]{{\color{ruleColor}\text{\upshape\textsc{\scriptsize[#1]}}}}%
\newcommand{\inference}[3][]{\infer[\ifempty{#1}{}{\inferrule{#1}}]{#3}{#2}}%
\newcommand{\cinference}[3][]{\infer=[\ifempty{#1}{}{\inferrule{#1}}]{#3}{#2}}%

\newcommand{\setenum}[1]{\mathord{{\color{black}\left\{#1\right\}}}}%
\newcommand{\setcomp}[2]{\mathord{%
  {\color{black}\left\{{#1} \,\middle|\, {#2}\right\}}}}%

\newcommand{\predP}[1][]{\ifempty{#1}{\varphi}{\varphi_{#1}}}%
\newcommand{\predPApp}[2][]{\ifempty{#1}{\predP}{\predP[{#1}]}\!\left({#2}\right)}%

\newcommand{\bind}[2]{\nicefrac{#2}{#1}}%
\newcommand{\substenum}[1]{\mathord{\left\{{#1}\right\}}}%
\newcommand{\subst}[2]{\substenum{\bind{#1}{#2}}}%

\definecolor{hlColor}{rgb}{0.65, 1.0, 0.65}%

\newcommand{\highlight}[2][\highlightColour]{\mathchoice%
  {\setlength{\fboxsep}{0pt}\colorbox{#1}{$\displaystyle#2$}}%
  {\setlength{\fboxsep}{0pt}\colorbox{#1}{$\textstyle#2$}}%
  {\setlength{\fboxsep}{0pt}\colorbox{#1}{$\scriptstyle#2$}}%
  {\setlength{\fboxsep}{0pt}\colorbox{#1}{$\scriptscriptstyle#2$}}}%

\newcommand{\runtime}[2][\runtimeColour]{\mathchoice%
  {\setlength{\fboxsep}{0pt}\colorbox{#1}{$\displaystyle#2$}}%
  {\setlength{\fboxsep}{0pt}\colorbox{#1}{$\textstyle#2$}}%
  {\setlength{\fboxsep}{0pt}\colorbox{#1}{$\scriptstyle#2$}}%
  {\setlength{\fboxsep}{0pt}\colorbox{#1}{$\scriptscriptstyle#2$}}}%

\newcommand{\lbbar}{\{\kern-0.2em|}
\newcommand{\rbbar}{|\kern-0.2em\}}

\newcommand{\tyGroundSet}{\stFmt{\mathcal{B}}}
\newcommand{\tyGround}[1][]{\stFmt{\ifempty{#1}{B}{B_{#1}}}}%
\newcommand{\tyGroundi}[1][]{\stFmt{\ifempty{#1}{B'}{B'_{#1}}}}%
\newcommand{\tyGroundii}[1][]{\stFmt{\ifempty{#1}{B''}{B''_{#1}}}}%
\newcommand{\tyGroundiii}[1][]{\stFmt{\ifempty{#1}{B'''}{B'''_{#1}}}}%
\newcommand{\tyBool}{\stFmtC{bool}}%
\newcommand{\tyUnit}{\stFmtC{unit}}%
\newcommand{\tyInt}{\stFmtC{int}}%
\newcommand{\tyString}{\stFmtC{str}}%
\newcommand{\tyNat}{\stFmtC{nat}}%
{\centerline{\bf --- Begin Copied From Previous Paper ---} \hrule}%
{\hrule \centerline{\bf --- End Copied From Previous Paper ---}}%

{\centerline{\bf --- Begin Discussion\ifempty{#1}{}{: {#1}} ---}
  \hrule\vspace{1mm}}%
{\hrule\vspace{1mm}\centerline{\bf --- End Discussion ---}}%

\definecolor{roleColor}{rgb}{0.5, 0.0, 0.0}%
\newcommand{\roleCol}[1]{{\color{roleColor}#1}}%
\newcommand{\roleSet}{\roleCol{\mathcal{R}}}%
\newcommand{\roleFmt}[1]{\ensuremath{{\boldsymbol{\roleCol{\mathtt{#1}}}}}\xspace}%
\newcommand{\roleCrashedSym}{\roleCol{\lightning}}
\newcommand{\roleMaybeCrashedSym}{\roleCol{\dagger}}

\newcommand{\roleP}[1][]{%
  \ifempty{#1}{{\color{roleColor}\roleFmt{p}}}{{\color{roleColor}\roleFmt{p}_{#1}}}%
}%
\newcommand{\rolePCrashed}[1][]{%
  \ifempty{#1}{{\color{roleColor}\roleFmt{p^{\roleCrashedSym}}}}{{\color{roleColor}\roleFmt{p^{\roleCrashedSym}}_{#1}}}%
}%
\newcommand{\rolePMaybeCrashed}[1][]{%
  \ifempty{#1}{{\color{roleColor}\roleFmt{p^{\roleMaybeCrashedSym}}}}{{\color{roleColor}\roleFmt{p^{\roleMaybeCrashedSym}}_{#1}}}%
}%
\newcommand{\roleQ}[1][]{%
  \ifempty{#1}{{\color{roleColor}\roleFmt{q}}}{{\color{roleColor}\roleFmt{q}_{#1}}}%
}%
\newcommand{\roleQCrashed}[1][]{%
  \ifempty{#1}{{\color{roleColor}\roleFmt{q^{\roleCrashedSym}}}}{{\color{roleColor}\roleFmt{q^{\roleCrashedSym}}_{#1}}}%
}%
\newcommand{\roleQMaybeCrashed}[1][]{%
  \ifempty{#1}{{\color{roleColor}\roleFmt{q^{\roleMaybeCrashedSym}}}}{{\color{roleColor}\roleFmt{q^{\roleMaybeCrashedSym}}_{#1}}}%
}%
\newcommand{\roleR}[1][]{%
  \ifempty{#1}{{\color{roleColor}\roleFmt{r}}}{{\color{roleColor}\roleFmt{r}_{\!#1}}}%
}%
\newcommand{\roleRMaybeCrashed}[1][]{%
  \ifempty{#1}{{\color{roleColor}\roleFmt{r^{\roleMaybeCrashedSym}}}}{{\color{roleColor}\roleFmt{r^{\roleMaybeCrashedSym}}_{#1}}}%
}%
\newcommand{\roleS}[1][]{%
  \ifempty{#1}{{\color{roleColor}\roleFmt{s}}}{{\color{roleColor}\roleFmt{s}_{\!#1}}}%
}%
\newcommand{\roleSMaybeCrashed}[1][]{%
  \ifempty{#1}{{\color{roleColor}\roleFmt{s^{\roleMaybeCrashedSym}}}}{{\color{roleColor}\roleFmt{s^{\roleMaybeCrashedSym}}_{#1}}}%
}%
\newcommand{\roleT}[1][]{%
  \ifempty{#1}{{\color{roleColor}\roleFmt{t}}}{{\color{roleColor}\roleFmt{t}_{\!#1}}}%
}%
\newcommand{\roleTCrashed}[1][]{%
  \ifempty{#1}{{\color{roleColor}\roleFmt{t^{\roleCrashedSym}}}}{{\color{roleColor}\roleFmt{t^{\roleCrashedSym}}_{#1}}}%
}%
\newcommand{\roleTMaybeCrashed}[1][]{%
  \ifempty{#1}{{\color{roleColor}\roleFmt{t^{\roleMaybeCrashedSym}}}}{{\color{roleColor}\roleFmt{t^{\roleMaybeCrashedSym}}_{#1}}}%
}%
\newcommand{\roleU}[1][]{%
  \ifempty{#1}{{\color{roleColor}\roleFmt{u}}}{{\color{roleColor}\roleFmt{u}_{\!#1}}}%
}%
\newcommand{\roleUMaybeCrashed}[1][]{%
  \ifempty{#1}{{\color{roleColor}\roleFmt{u^{\roleMaybeCrashedSym}}}}{{\color{roleColor}\roleFmt{u^{\roleMaybeCrashedSym}}_{#1}}}%
}%

\makeVars{roleA}{roleFmt}{a}
\makeVars{roleC}{roleFmt}{c}
\makeVars{roleL}{roleFmt}{l}
\makeVars{roleMhd}{roleFmt}{m_1}
\makeVars{roleMtl}{roleFmt}{m_2}

\makeVars{roles}{roleFmt}{\mathscr P}
\makeVars{rolesR}{roleFmt}{\mathscr R}
\makeVars{rolesC}{roleFmt}{\mathscr C}
\makeVars{rolesS}{roleFmt}{\mathfrak R}
\newcommand{\rolesEmpty}[0]{\gtFmt{\emptyset}}

\definecolor{gtColor}{rgb}{0.43, 0.21, 0.1}%
\newcommand{\gtFmt}[1]{\ensuremath{{\color{gtColor}#1}}\xspace}%
\newcommand{\gtMsgFmt}[1]{\gtFmt{\labFmt{#1}}}%
\newcommand{\gtLabFmt}[1]{\gtFmt{\labFmt{#1}}}%
\newcommand{\gtLab}[1][]{%
  \ifempty{#1}{\gtMsgFmt{m}}{{\color{gtColor}\gtMsgFmt{m}_{#1}}}%
}%
\newcommand{\gtLabi}[1][]{%
  \ifempty{#1}{\gtMsgFmt{m'}}{{\color{gtColor}\gtMsgFmt{m'}_{#1}}}%
}%

\makeVars{gtLabL}{gtLabFmt}{l}
\newcommand{\gtLabProp}[0]{\gtLabFmt{req}}

\newcommand{\gtLabCommit}[0]{\gtLabFmt{accept}}
\newcommand{\gtLabVeto}[0]{\gtLabFmt{reject}}
\newcommand{\gtLabAbort}[0]{\gtLabFmt{abort}}
\newcommand{\gtLabSucc}[0]{\gtLabFmt{enact}}
\newcommand{\gtLabPromote}[0]{\gtLabFmt{promote}}
\newcommand{\gtLabNext}[0]{\gtLabFmt{retry}}

\makeVarsEx{gtG}{gtFmt}{G}

\newcommand{\gtSeq}{\mathbin{\gtFmt{.}}}%

\newcommand{\gtCommRaw}[3]{%
  \gtFmt{%
    {#1} {\to} {#2}{:}%
    \left\{%
      {#3}%
    \right\}%
  }%
}%

\newcommand{\gtCommRawNB}[3]{%
  \gtFmt{%
    {#1} {\to} {#2}{:}%
      {#3}%
  }%
}%

\newcommand{\gtComm}[6]{%
  \gtFmt{%
    \gtCommRaw{#1}{#2}{%
      \gtCommChoice{#4}{#5}{#6}%
    }_{#3}%
  }%
}%
\newcommand{\gtCommSmall}[6]{%
  \gtFmt{%
    \gtCommRaw{#1}{#2}{%
      \gtCommChoiceSmall{#4}{#5}{#6}%
    }_{#3}%
  }%
}%
\newcommand{\gtCommSingle}[5]{%
  \gtFmt{%
    {#1} {\to} {#2}{:}%
    \gtCommChoice{#3}{#4}{#5}%
  }%
}%
\newcommand{\gtCommTransitSingle}[5]{%
  \gtFmt{%
    {#1} {\rightsquigarrow} {#2}{:}%
    \gtCommChoice{#3}{#4}{#5}%
  }%
}%
\newcommand{\gtCommTransitRaw}[4]{%
  \gtFmt{%
    {#1} {\rightsquigarrow} {#2}{:}{#4}%
    \left\{%
      {#3}%
    \right\}%
  }%
}%

\newcommand{\gtCommTransit}[7]{%
  \gtFmt{%
    \gtCommTransitRaw{#1}{#2}{%
      \gtCommChoice{#4}{#5}{#6}%
    }{#7}_{#3}%
  }%
}%

\newcommand{\gtCommSingleErr}[6]{%
  \gtFmt{%
    {#1} {\to} {#2}{:}%
    \{\gtCommChoice{#3}{#4}{#5},\quad \gtErrKFmt{#6}\}%
  }%
}%
\newcommand{\gtCommChoice}[3]{%
  \gtFmt{%
    \gtMsgFmt{#1}\ifempty{#2}{}{({#2})}%
    \ifempty{#3}{}{\vphantom{x} \!\gtSeq\! {#3}}%
  }%
}%
\newcommand{\gtCommChoiceSmall}[3]{%
  \gtFmt{%
    \gtMsgFmt{#1}\ifempty{#2}{}{({#2})}%
    \ifempty{#3}{}{\vphantom{x} \!\gtSeq\! {#3}}%
  }%
}%

\newcommand{\gtEnd}{\gtFmt{\mathsf{end}}}%

\newcommand{\gtRec}[2]{\gtFmt{\mu{#1}.{#2}}}%
\newcommand{\gtRecVarBase}{\gtFmt{\mathbf{t}}}%
\makeVars{gtRecVar}{gtFmt}{\gtRecVarBase}

\newcommand{\gtRoles}[1]{{\color{roleColor} \operatorname{roles}(\gtFmt{#1})}}%
\newcommand{\gtRolesCrashed}[1]{{\color{roleColor}\operatorname{roles}^{\roleCrashedSym}(\gtFmt{#1})}}%

\newcommand{\gtProj}[3][]{%
  {\color{stColor}\gtFmt{#2} \ifempty{#1}{\upharpoonright}{\upharpoonright_{#1}} \roleFmt{#3}}%
}%

\newcommand{\gtStopLab}[0]{\gtMsgFmt{\mathsf{crash}}}
\newcommand{\gtCrashLab}[0]{\gtMsgFmt{\mathsf{crash}}}

\newcommand{\gtWithCrashedRoles}[2]{\langle{#1};{#2}\rangle}

\newcommand{\gtCrashRole}[2]{\gtFmt{#1}\gtFmt{\lightning}\roleFmt{#2}}

\newcommand{\gtMove}[2][\phantom{\stEnvAnnotGenericSym}]{\xrightarrow{#1}_{#2}} %
\newcommand{\gtMoveStar}[1][]{\ifempty{#1}{\gtMove[]{}^{\!\!\!*}}{\gtMove[]{#1}^{\!*}}} %

\newcommand{\iruleGtMove}[1]{GR-{#1}}
\newcommand{\iruleGtMoveCrash}[0]{\iruleGtMove{$\lightning$}}

\newcommand{\iruleGtMoveOut}[0]{\iruleGtMove{$\oplus$}}
\newcommand{\iruleGtMoveIn}[0]{\iruleGtMove{$\&$}}
\newcommand{\iruleGtMoveRec}[0]{\iruleGtMove{$\mu$}}
\newcommand{\iruleGtMoveCrDe}[0]{\iruleGtMove{$\odot$}}
\newcommand{\iruleGtMoveOrph}[0]{\iruleGtMove{$\lightning\gtLab$}}
\newcommand{\iruleGtMoveCtx}[0]{\iruleGtMove{Ctx-i}}
\newcommand{\iruleGtMoveCtxi}[0]{\iruleGtMove{Ctx-ii}}

\newcommand{\labFmt}[2][]{\ensuremath{\ifempty{#1}{\mathtt{#2}}{\mathtt{#2}\textsubscript{#1}}}\xspace}%
\newcommand{\labSet}{\gtFmt{\mathcal{M}}}

\newcommand{\gtErrKFmt}[1]{\gtStopLab.{#1}}

\definecolor{stColor}{rgb}{0, 0, 0.9}%
\newcommand{\stFmt}[1]{\ensuremath{{\color{stColor}#1}}\xspace}%
\newcommand{\stFmtC}[1]{\stFmt{\operatorname{#1}}}
\newcommand{\stIn}[4]{\ifempty{#1}{}{\roleFmt{#1}}\stFmt{\&\left\{{#2}\ifempty{#3}{}{({#3})} \ifempty{#4}{}{\stSeq #4}\right\}}}%
\newcommand{\stInNB}[4]{\ifempty{#1}{}{\roleFmt{#1}}\stFmt{\&{#2}\ifempty{#3}{}{({#3})} \ifempty{#4}{}{\stSeq #4}}}%
\newcommand{\stOut}[3]{\ifempty{#1}{}{\roleFmt{#1}}\stFmt{\oplus{#2}\ifempty{#3}{}{({#3})}}}%

\newcommand{\stChoice}[2]{\stLabFmt{#1}\ifempty{#2}{}{\stFmt{({#2})}}}%

\newcommand{\stSeq}{\mathbin{\!\stFmt{.}\!}}%
\newcommand{\stIntC}{\mathbin{\stFmt{\oplus}}}%
\newcommand{\stIntSum}[3]{\roleFmt{#1}\stFmt{\oplus\!\left\{#3\right\}_{#2}}}%
\newcommand{\stIntSumNB}[3]{\roleFmt{#1}\stFmt{\oplus\!#3_{#2}}}%
\newcommand{\stExtC}{\mathbin{\stFmt{\&}}}%
\newcommand{\stExtSum}[3]{\roleFmt{#1}\stFmt{\&\!\left\{#3\right\}_{#2}}}%
\newcommand{\stExtSumNB}[3]{\roleFmt{#1}\stFmt{\&\!#3_{#2}}}%
\newcommand{\stRec}[2]{\stFmt{\mu{#1}.{#2}}}%
\newcommand{\stEnd}{\stFmt{\mathsf{end}}}%
\newcommand{\stStop}{\stFmt{\mathsf{stop}}}%

\newcommand{\stLabFmt}[1]{\stFmt{\labFmt{#1}}}%
\newcommand{\stLab}[1][]{%
  \ifempty{#1}{\stLabFmt{m}}{\stLabFmt{m}_{{\color{stColor}#1}}}
}%
\newcommand{\stLabi}[1][]{%
  \ifempty{#1}{\stLabFmt{m'}}{\stLabFmt{m}'_{{\color{stColor}#1}}}
}%
\newcommand{\stLabii}[1][]{%
  \ifempty{#1}{\stLabFmt{m''}}{\stLabFmt{m}''_{{\color{stColor}#1}}}
}%

\newcommand{\stCrashLab}[0]{\stLabFmt{\mathsf{crash}}}
\newcommand{\stLabOK}[0]{\stLabFmt{ok}}
\newcommand{\stLabKO}[0]{\stLabFmt{ko}}

\newcommand{\stLabProp}[0]{\stLabFmt{req}}

\newcommand{\stLabCommit}[0]{\stLabFmt{accept}}
\newcommand{\stLabVeto}[0]{\stLabFmt{reject}}
\newcommand{\stLabAbort}[0]{\stLabFmt{abort}}
\newcommand{\stLabSucc}[0]{\stLabFmt{enact}}
\newcommand{\stLabPromote}[0]{\stLabFmt{promote}}
\newcommand{\stLabNext}[0]{\stLabFmt{retry}}

\makeVars{stS}{stFmt}{S}
\makeVars{stT}{stFmt}{T}
\makeVars{stU}{stFmt}{U}

\newcommand{\stRecVarBase}{\stFmt{\mathbf{t}}}%
\makeVars{stRecVar}{stFmt}{\stRecVarBase}

\newcommand{\stMerge}[2]{\stFmt{\bigsqcap_{#1}{#2}}}%
\newcommand{\stBinMerge}{\mathbin{\stFmt{\sqcap}}}%

\newcommand{\stSub}{\mathrel{\stFmt{\leqslant}}}%

\newcommand{\ltsCrDe}[3]{{#2}\stFmt{\mathord{\odot}}\roleFmt{#3}}
\newcommand{\ltsCrash}[2]{{#2}\stFmt{\lightning}}
\newcommand{\ltsCrashSmall}[2]{{#2}\stFmt{\lightning}}

\newcommand{\ltsSubject}[1]{{\color{roleColor} \operatorname{subj}({#1})}}%

\definecolor{mpColor}{rgb}{0, 0, 0}%
\newcommand{\mpFmt}[1]{{\color{mpColor}#1}}%

\newcommand{\mpLab}[1][]{%
  \mpFmt{\ifempty{#1}{\labFmt{m}}{{\labFmt{m}}_{\mathnormal #1}}}%
}%
\newcommand{\mpLabi}[1][]{%
  \mpFmt{\ifempty{#1}{\labFmt{m}'}{\labFmt{m}'_{\mathnormal #1}}}%
}%

\newcommand{\mpCrashLab}[0]{\mpFmt{\labFmt{\mathsf{crash}}}}
\newcommand{\mpLabFmt}[1]{\mpFmt{\labFmt{#1}}}%

\newcommand{\mpLabProp}[0]{\mpLabFmt{req}}

\newcommand{\mpLabCommit}[0]{\mpLabFmt{accept}}
\newcommand{\mpLabVeto}[0]{\mpLabFmt{reject}}
\newcommand{\mpLabAbort}[0]{\mpLabFmt{abort}}
\newcommand{\mpLabSucc}[0]{\mpLabFmt{enact}}
\newcommand{\mpLabPromote}[0]{\mpLabFmt{promote}}
\newcommand{\mpLabNext}[0]{\mpLabFmt{retry}}

\newcommand{\mpTrue}{\mpFmt{\text{\text{\texttt{true}}}}}%
\newcommand{\mpFalse}{\mpFmt{\text{\text{\texttt{false}}}}}%
\newcommand{\mpString}[1]{\text{\text{\texttt{"{#1}"}}}}%
\newcommand{\mpInt}{\mpFmt{\text{\texttt{i}}}}
\newcommand{\mpNat}{\mpFmt{\text{\texttt{n}}}}%
\newcommand{\mpUnit}[1]{\text{\texttt{({#1})}}}%
\newcommand{\mpSucc}[1]{\text{\texttt{succ}}({#1})}%
\newcommand{\mpNeg}[1]{\text{\texttt{neg}}({#1})}%

\newcommand{\mpChanRole}[2]{\mpFmt{{#1}[{#2}]}}%

\newcommand{\mpNil}{\mpFmt{\mathbf{0}}}%
\newcommand{\mpIf}[3]{%
  \mpFmt{\mathsf{if}\,{#1}\,\mathsf{then}}\,{#2}\,\mathsf{else}\,{#3}%
}%
\newcommand{\mpPar}{\mathbin{\mpFmt{\mid}}}%
\newcommand{\mpRec}[2]{\mpFmt{\mu{#1}.{#2}}}%

\newcommand{\mpPart}[2]{#1 \triangleleft  #2}
\newcommand{\mpCrash}[0]{\mpFmt{\ensuremath{\lightning}}}

\makeVars{mpECtx}{mpFmt}{\mathbb{E}}

\newcommand{\mpV}[1][]{\mpFmt{\ifempty{#1}{v}{v_{#1}}}}
\newcommand{\mpVi}[1][]{\mpFmt{\ifempty{#1}{v'}{v'_{#1}}}}

\newcommand{\mpW}[1][]{\mpFmt{\ifempty{#1}{w}{w_{#1}}}}

\makeVars{mpE}{mpFmt}{e}
\makeVars{mpM}{mpFmt}{\mathbb{M}}
\makeVars{mpH}{mpFmt}{h}
\makeVars{mpx}{mpFmt}{x}
\newcommand{\mpQEmpty}{\mpFmt{\epsilon}}%
\newcommand{\mpQUnavail}{\mpFmt{\oslash}}

\newcommand{\mpS}[1][]{\mpFmt{\ifempty{#1}{s}{s_{#1}}}}%

\newcommand{\mpX}[1][]{\mpFmt{\ifempty{#1}{X}{X_{#1}}}}%

\newcommand{\mpP}[1][]{\mpFmt{\ifempty{#1}{P}{P_{#1}}}}%
\newcommand{\mpPi}[1][]{\mpFmt{\ifempty{#1}{P'}{P'_{#1}}}}%
\newcommand{\mpQ}[1][]{\mpFmt{\ifempty{#1}{Q}{Q_{#1}}}}%
\newcommand{\mpQi}[1][]{\mpFmt{\ifempty{#1}{Q'}{Q'_{#1}}}}%
\newcommand{\mpR}[1][]{\mpFmt{\ifempty{#1}{R}{R_{#1}}}}%

\newcommand{\mpMove}[1][]{\ifempty{#1}{\to}{\to_{#1}}}%

\newcommand{\mpMoveStar}[1][]{\mathrel{\mpMove[#1]^{*}}}%
\newcommand{\mpNotMove}[1][]{\mathrel{\not{\!\!\mpMove[#1]}}}

\newcommand{\iruleSafeComm}{S-${\stIntC}{\stExtC}$}%
\newcommand{\iruleSafeCrash}{S-${\stFmt{\lightning}}{\stExtC}$}%
\newcommand{\iruleSafeRec}{S-$\stFmt{\mu}$}%
\newcommand{\iruleSafeMove}{S-$\stEnvMoveMaybeCrash$}%

\newcommand{\iruleStSubEnd}{Sub-$\stEnd$}
\newcommand{\iruleStSubStop}{Sub-$\stStop$}
\newcommand{\iruleStSubRecL}{Sub-$\stFmt{\mu}$L}
\newcommand{\iruleStSubRecR}{Sub-$\stFmt{\mu}$R}
\newcommand{\iruleStSubOut}{Sub-$\stFmt{\oplus}$}
\newcommand{\iruleStSubIn}{Sub-$\stFmt{\&}$}

\newcommand{\stStopSym}{\stFmt{\ensuremath{\lightning}}}

\newcommand{\iruleTCtxOut}{$\stEnv$-$\stFmt{\oplus}$}%
\newcommand{\iruleTCtxIn}{$\stEnv$-$\stFmt{\&}$}%
\newcommand{\iruleTCtxRec}{$\stEnv$-$\mu$}%
\newcommand{\iruleTCtxCrash}{$\stEnv$-$\lightning$}
\newcommand{\iruleTCtxCrashDetect}{$\stEnv$-$\odot$}

\makeVars{stEnv}{stFmt}{\Gamma}
\newcommand{\stEnvEmpty}{\stFmt{\emptyset}}%
\newcommand{\stEnvMap}[2]{\stFmt{\mpFmt{#1}\mathbin{\!\triangleright\!}{#2}}}%
\newcommand{\stEnvComp}{\mathpunct{\stFmt{,}}}%
\newcommand{\stEnvApp}[2]{\stFmt{#1\!\left(\mpFmt{#2}\right)}}%
\newcommand{\stEnvUpd}[3]{\stFmt{#1 {} [#2 \mapsto #3]}}

\newcommand{\stEnvAssoc}[3]{\stFmt{{#2} \mathrel{\stFmt{\sqsubseteq}_{#3}} {#1}}}

\newcommand{\stEnvMove}{\mathrel{\stFmt{\to}}}%
\newcommand{\stEnvMoveMaybeCrash}[1][]{%
  \ifempty{#1}{%
    \mathrel{\stFmt{\to_{\!\lightning}}}%
  }{%
    \mathrel{\stFmt{\to_{\!{#1}}}}%
  }
}%
\newcommand{\stEnvAnnotOutSym}{\stFmt{\oplus}}%
\newcommand{\stEnvAnnotInSym}{\stFmt{\&}}%
\newcommand{\stEnvAnnotGenericSym}[1][]{\stFmt{\ifempty{#1}{\alpha}{\alpha_{#1}}}}%
\newcommand{\stEnvMoveAnnot}[1]{\mathrel{\stFmt{\xrightarrow{#1}}}}
\newcommand{\stEnvMoveGenAnnot}{\stEnvMoveAnnot{\stEnvAnnotGenericSym}}%
\newcommand{\stEnvMoveInAnnot}[3]{%
  \stEnvMoveAnnot{\stEnvInAnnot{#1}{#2}{#3}}%
}%
\newcommand{\stEnvMoveOutAnnot}[3]{%
  \stEnvMoveAnnot{\stEnvOutAnnot{#1}{#2}{#3}}%
}%

\newcommand{\stEnvInAnnot}[3]{{#1}{\stEnvAnnotInSym}{#2}:{#3}}%
\newcommand{\stEnvOutAnnot}[3]{{#1}{\stEnvAnnotOutSym}{#2}:{#3}}%
\newcommand{\stEnvInAnnotSmall}[3]{{#1}{\stEnvAnnotInSym}{#2}:{#3}}%
\newcommand{\stEnvInAnnotSmallLab}[3]{{#1}{\stEnvAnnotInSym}{#2}\!:\!{#3}}%
\newcommand{\stEnvOutAnnotSmallLab}[3]{{#1}{\stEnvAnnotOutSym}{#2}\!:\!{#3}}%
\newcommand{\stEnvOutAnnotSmall}[3]{{#1}{\stEnvAnnotOutSym}{#2}:{#3}}%
\newcommand{\stEnvMoveP}[1]{{#1}\!\stEnvMove}%
\newcommand{\stEnvMoveMaybeCrashP}[2][]{{#2}\!\!\stEnvMoveMaybeCrash[#1]}%
\newcommand{\stEnvNotMoveP}[1]{{#1}\!\not\stEnvMove}%
\newcommand{\stEnvNotMoveMaybeCrashP}[2][]{{#2}\!\not\stEnvMoveMaybeCrash[#1]}%
\newcommand{\stEnvMoveStar}{\mathrel{\stFmt{\stEnvMove{}^{\!\!\!*}}}}%
\newcommand{\stEnvMoveMaybeCrashStar}[1][]{%
  \ifempty{#1}{%
    \mathrel{\stFmt{\to^*_{\!\lightning}}}%
  }{%
    \mathrel{\stFmt{\to^*_{\!{#1}}}}%
  }
}%
\newcommand{\stEnvMoveAnnotP}[2]{{#1}\!\stEnvMoveAnnot{#2}}%
\newcommand{\stEnvMoveGenAnnotP}[1]{{#1}\!\stEnvMoveGenAnnot}%

\newcommand{\stIdxRemoveCrash}[2]{{#1}^{{#2} \setminus \stCrashLab}}

\makeVars{qEnv}{stFmt}{\Delta}
\newcommand{\qApp}[3]{{#1}({#2}, {#3})}
\makeVars{stQ}{stFmt}{\tau}
\makeVars{stM}{stFmt}{M}

\makeVars{qEnvPartial}{stFmt}{\delta}

\newcommand{\stQEmpty}{\stFmt{\epsilon}}%
\newcommand{\stQMsg}[2]{\stFmt{{#1}\ifempty{#2}{}{({#2})}}}%
\newcommand{\stQCons}[2]{\stFmt{{#1}\mathbin{\!\cdot\!}{#2}}}%
\newcommand{\stQUnavail}{\stFmt{\oslash}}

\definecolor{tyColorCustom}{rgb}{0.0, 0.0, 0.85}

\definecolor{purpleish}{rgb}{0.41, 0.16, 0.38}

\makeVars{ldVarX}{ldVarFmt}{x}
\makeVars{ldVarY}{ldVarFmt}{y}
\makeVars{ldVarV}{ldVarFmt}{v}
\makeVars{ldTermM}{ldTermFmt}{M}
\makeVars{ldTermN}{ldTermFmt}{N}
\makeVars{ldValX}{ldTermFmt}{X}
\makeVars{ldValY}{ldTermFmt}{Y}
\makeVars{ldValV}{ldTermFmt}{V}
\makeVars{ldValW}{ldTermFmt}{W}

\makeVars{val}{}{v}
\makeVars{valr}{}{i}
\makeVars{valn}{}{n}

\newcommand{\procin}[3]{#1 ? #2.#3}%

\newcommand{\procout}[4]{#1 ! #2\langle #3 \rangle.#4}%

\newcommand{\procoutNoVal}[3]{#1 ! #2.#3}%
\newcommand{\eval}[2]{#1 \downarrow #2}

\newcommand{\redLabel}[1]{\mpMove}
\newcommand{\redSend}[3]{\mpMove}
\newcommand{\redRecv}[3]{\mpMove}
\newcommand{\redCrash}[2]{\mpMove[#2]}
\newcommand{\redIf}[1]{\mpMove}
\newcommand{\msg}[3]{(#1,#2(#3))}

\newtoggle{full}
\settoggle{full}{false}
\iftoggle{full}
{
	
}{
	
}

\makeatletter
\newcommand{\iftoggleverb}[1]{%
  \ifcsdef{etb@tgl@#1}
    {\csname etb@tgl@#1\endcsname\iftrue\iffalse}
    {\etb@noglobal\etb@err@notoggle{#1}\iffalse}%
}
\makeatother

\begin{document}
\title{Crash-Stop Failures in Asynchronous Multiparty Session Types}
\thanks{}

\author[A.D.~Barwell]{Adam D. Barwell\lmcsorcid{0000-0003-1236-7160}}[a]
\author[P.~Hou]{Ping Hou\lmcsorcid{0000-0001-6899-9971}}[b]
\author[N.~Yoshida]{Nobuko Yoshida\lmcsorcid{0000-0002-3925-8557}}[b]
\author[F.~Zhou]{Fangyi Zhou\lmcsorcid{0000-0002-8973-0821}}[c]

\address{University of St Andrews, UK}
\email{adb23@st-andrews.ac.uk}

\address{University of Oxford, UK}
\email{ping.hou@cs.ox.ac.uk, nobuko.yoshida@cs.ox.ac.uk}

\address{No affiliation}
\email{me@fangyi.io}

\begin{abstract}
  Session types provide a typing discipline for message-passing systems.
  However,  their theory often assumes an ideal world: one in which
  everything is reliable and without failures. Yet this is in stark contrast with
  distributed systems in the real world. To address this limitation, 
  we introduce a new asynchronous 
  \emph{multiparty session types}~(MPST) theory with \emph{crash-stop failures}, 
  where processes may crash arbitrarily and cease to interact after crashing. %
  We augment asynchronous MPST and processes with \emph{crash handling} branches, and integrate 
crash-stop failure semantics into types and processes. Our approach
requires no user-level syntax extensions for global types, 
and features a formalisation of global semantics, which captures
complex behaviours induced by crashed/crash handling processes.

Our new theory covers the entire \emph{spectrum} of unreliability, ranging from the ideal world of total reliability 
to entirely unreliable scenarios where any process may crash, using \emph{optional reliability assumptions}. 
Under these assumptions, we demonstrate 
the sound and complete correspondence between global and local type semantics, 
which guarantee deadlock-freedom, protocol conformance, and liveness of 
well-typed processes \emph{by construction}, even in the presence of crashes.

\end{abstract}

\maketitle

\section{Introduction}
\label{sec:introduction}
\paragraph*{Background}
As distributed programming grows increasingly prevalent,
significant research effort has been devoted to improve the reliability of
distributed systems.
A key aspect of this research focuses on studying
\emph{un}reliability (or, more specifically, failures).
Modelling unreliability and failures
enables a distributed system to be designed to be more tolerant of failures,
and thus more resilient.

In pursuit of methods to achieve fundamental
reliability -- \emph{safety}
in distributed communication systems --
\emph{session types}~\cite{ESOP98Session} provide a lightweight,
type system--based approach to message-passing concurrency.
In particular,
\emph{Multiparty Session Types} (MPST)~\cite{HYC08} facilitate the
specification and verification of communication between message-passing
processes in concurrent and distributed systems.
The typing discipline
prevents common communication-based errors,
\eg deadlocks and communication mismatches~\cite{HYC16,POPL19LessIsMore}.

Nevertheless, the challenge to account for unreliability and failures persists
for session types:
most session type systems assume that both participants and message transmissions
are \emph{reliable}, \ie without failures.
In a real-world setting, however, participants may crash, communication %
channels may fail, and messages may be lost.
The lack of failure modelling in session type theories prevents
their application to large-scale distributed systems.

 Recent works~\cite{OOPSLA21FaultTolerantMPST, FORTE22FaultTolerant,
ECOOP22AffineMPST, ESOP23MAGPi, CONCUR22MPSTCrash}
address the gap in failure modelling within session
types with various techniques.
Viering \etal~\cite{OOPSLA21FaultTolerantMPST} introduce \emph{failure
suspicion},
where a participant may suspect their communication partner has failed,
and act accordingly.
Peters \etal~\cite{FORTE22FaultTolerant} introduce \emph{reliability
annotations} at type level, and fall back to a given \emph{default} value in
case of failures.
Lagaillardie \etal~\cite{ECOOP22AffineMPST} propose a framework of \emph{affine}
multiparty session types, where a session can terminate prematurely, \eg
in case of failures.
Barwell \etal~\cite{CONCUR22MPSTCrash} incorporate \emph{crash-stop failures}
in session types, where a generalised type system validates safety and liveness
properties. %
Le Brun and Dardha~\cite{ESOP23MAGPi}
adopt a similar approach but extend it to include additional failure types, \eg message losses, reordering, and delays.

\paragraph*{This Paper}
Unlike the aforementioned approaches, we advance failure modelling in session types
by introducing a new asynchronous multiparty session type theory
that incorporates \emph{crash-stop} failures~\cite[\S 2.2]{DBLP:books/daglib/0025983}, where
 a process may crash and cease to interact with others.
 This model is standard in distributed systems, and is used in
related work on session types with error-handling
capabilities~\cite{ESOP18CrashHandling,OOPSLA21FaultTolerantMPST}.
However, unlike previous work,
we allow \emph{any} process to crash arbitrarily,
and
support optional assumptions on non-crashing processes.
Our extended asynchronous MPST theory
effectively models failures with crash-stop
semantics,
and demonstrates that the usual session type guarantees remain
valid, \ie~communication safety, deadlock-freedom, and liveness.

 In our new theory, %
 we add \emph{crashing} and \emph{crash handling}
semantics to processes and types.
With \emph{minimal} changes to the standard surface syntax compared with the
original MPST theory,
we
model a variety of subtle, complex behaviours arising from
unreliable communicating
processes.
An active process $\mpP$ may crash arbitrarily, and a process $\mpQ$
interacting with $\mpP$ might need to be prepared to handle possible crashes.
Messages sent from $\mpQ$ to a crashed $\mpP$ are lost
-- but if $\mpQ$ tries to receive from $\mpP$, then $\mpQ$ can detect that
$\mpP$ has crashed, and take a crash handling branch.
Meanwhile, another process $\mpR$ may (or may not) have detected $\mpP$'s crash,
and may be handling it -- and in either case, any interaction between $\mpQ$
and $\mpR$ should remain correct.

A key design feature of our framework is
to support \emph{optional reliability assumptions}: a
programmer may declare that some peers will \emph{never} crash for the duration
of the protocol.
Such optional assumptions allow for simplifying protocols and programs: if a
participant is declared reliable, its peers can interact with it
without implementing crash detection and handling.
Moreover, by
making such assumptions explicit and customisable,
our theory supports a
\emph{spectrum} of assumptions ranging from only having sessions
with reliable participants (as in classic MPST works~\cite{HYC16,POPL19LessIsMore}),
to having no reliable participants at all.

Another fundamental aspect of our theory is its adherence to a \emph{top-down} methodology,
which stems from the original MPST theory \cite{HYC08}.
This approach starts the design of multiparty session types
with a given \emph{global} type
(top),
and processes rely on \emph{local} types (bottom) obtained from the
global type.
The global and local types reflect the global and local communication
behaviours respectively.
Well-typed processes that conform to a global type are
guaranteed to be \emph{correct by construction},
enjoying full guarantees (safety, deadlock-freedom, and liveness) from the theory.

The use of global types in our design for handling failures in multiparty session types
presents three distinct advantages:
\begin{enumerate*}[label=\emph{(\arabic*)}]
  \item there is no user-level syntax extension of global types compared with
    the original MPST global types, apart from a special label to signify crash
    handling branches;
  \item global types provide a simple, high-level means to both specify a protocol
    abstractly and automatically derive local types; and,
  \item under (optional) reliability assumptions, desirable behavioural properties such as communication safety,
    deadlock-freedom, and liveness are guaranteed by construction. %
  \end{enumerate*}

In contrast to the \emph{synchronous} semantics used by Barwell
\etal~\cite{CONCUR22MPSTCrash}, we model \emph{asynchronous} semantics,
where messages can be buffered whilst in transit.
We focus on asynchronous systems
since most communication in the real distributed world is asynchronous.
Although Le Brun and Dardha~\cite{ESOP23MAGPi} %
develop a generic typing system incorporating asynchronous semantics,
their approach results in  type-level properties
becoming undecidable~\cite[\S 4.4]{ESOP23MAGPi}.
With global types,
we restore decidability %
at a minor cost to
expressivity.

\paragraph*{Outline}
This article represents an extended rendition of the work on
 asynchronous multiparty session type theory with
 crash-stop failures as initially introduced in~\cite{DBLP:conf/ecoop/BarwellHY023}.
 Our enhanced version provides a more comprehensive and refined presentation, offering detailed definitions,
 an in-depth exploration of related work, and additional examples.
 Moreover, we incorporate a new section~(Section~\ref{sec:casestudy}) that illustrates our approach
 through an extensive case study of the Non-Blocking Atomic Commits abstraction,
 described in~\cite{DBLP:books/daglib/0025983}.
 In contrast to~\cite{DBLP:conf/ecoop/BarwellHY023},
 we have omitted the implementation of our theory, as our primary focus is on the theoretical aspects.

Note that, for clarity and simplicity, our framework does not support delegation (session passing), as our primary focus is on developing a top-down methodology for handling failures in multiparty session types. However, all the results presented in the paper are applicable to systems that do incorporate delegation.
For readers interested in the generalised MPST theory that includes both
crash-stop failures and delegation,
please refer to~\cite{CONCUR22MPSTCrash}.

\begin{itemize}[leftmargin=0pt]

\item[] {\bf Section~\ref{sec:overview}} We begin with an overview of our methodology.
  
\item[] %
{\bf Section~\ref{sec:process}} Addressing the challenge of defining
the semantics of crashing and crash handling
within the context of the crash-stop model,
we introduce an asynchronous multiparty session
calculus with minimal syntax changes.

\item[]
    {\bf Section~\ref{sec:gtype}}  We introduce an extended theory of asynchronous multiparty session
    types with semantic modelling of crash-stop failures.
    The challenges here involve capturing the non-trivial crash-stop semantics
     of global and local types under \emph{optional} reliability assumptions,
     as well as demonstrating that the sound and
     complete operational correspondence between the two semantics
     (Theorems~\ref{thm:gtype:proj-comp} and~\ref{thm:gtype:proj-sound}) can ensure desirable
     global type properties such as communication safety, deadlock-freedom, and
     liveness in local types through projection,
     even in the presence of crashes (\autoref{cor:allproperties}).

\item[] 
   {\bf Section~\ref{sec:typing_system}} We present a typing system for our asynchronous multiparty session calculus.
  The challenges comprise proving behavioural properties for typed sessions: subject reduction, session fidelity, deadlock-freedom, and liveness, in corresponding theorems
  (Theorems \ref{lem:sr}, \ref{lem:sf}, \ref{lem:session_deadlock_free}, and
  \ref{lem:session_live}).

\item[]
{\bf Section~\ref{sec:casestudy}}  We provide a case study on the Non-Blocking Atomic Commits abstraction,
which facilitates the reliable execution of transactions in distributed database systems.
\end{itemize}
We discuss related work in {\bf Section~\ref{sec:related}} and conclude in
{\bf Section~\ref{sec:conclusion}}.

\section{Overview}
\label{sec:overview}

In this work, 
we follow a standard top-down design approach enabling \emph{correctness by construction}, but enrich asynchronous MPST with crash-stop semantics.  As depicted in~\cref{fig:overview-of-topdown}, we formalise (asynchronous) multiparty protocols with crash-stop failures
as  global types with \emph{crash handling branches} ($\gtCrashLab$). These are projected onto local types, which may similarly contain crash handling branches ($\stCrashLab$).
The projected local types %
are, in turn, used to type-check
processes (also with crash handling branches ($\mpCrashLab$)) that are written in a session calculus.
As an example, we consider a simple
\emph{distributed logging} scenario, which is inspired by
the logging-management protocol~\cite{ECOOP22AffineMPST}, but extended with a third
participant. \iftoggle{full}{The full distributed logging protocol can be found in~\Cref{sec:appendix:eval}.}{}
\begin{figure}[t]
 \centering
  \begin{tikzpicture}
  \node (Gtext) {A Global Type $\gtG$ $\highlight{\text{with }\gtCrashLab}$};
  \node[below= 0.5mm of Gtext, xshift=-39mm, align=center] (proj) {
    \small{\bf{projection}}\,($\upharpoonright$)};
  \node[below=6mm of Gtext, xshift=-41mm] (LA) {\footnotesize Local Type for $\roleFmt{L}$
  \small \boxed{\stT_{\roleFmt{L}}}};
  \node[below=6mm of Gtext] (LB) {\footnotesize Local Type for $\roleFmt{I}$ $\highlight{\text{with }\stCrashLab}$ \small \boxed{\stT_{\roleFmt{I}}}};
  \node[below=6mm of Gtext, xshift=41mm] (LC) {\footnotesize Local Type for $\roleFmt{C}$
   \small  \boxed{\stT_{\roleFmt{C}}}};
  \draw[->] (Gtext) -- (LA);
  \draw[->] (Gtext) -- (LB);
  \draw[->] (Gtext) -- (LC);
   \node[below= 1mm of LA, xshift=-9.0mm, align=center] (typ) {
    \small {\bf{typing}}\,($\vdash$)};
   \node[below=6mm of LA, xshift=0mm] (PA) {\, \, \, \footnotesize Process for $\roleFmt{L}$ \small
    \boxed{P_{\roleFmt{L}}}};
   \node[below=6mm of LB, xshift=0mm] (PB) {\, \, \, \footnotesize Process for $\roleFmt{I}$ $\highlight{\text{with }\mpCrashLab}$ \small
    \boxed{P_{\roleFmt{I}}}};
     \node[below=6mm of LC, xshift=0mm] (PC) {\, \, \, \footnotesize Process for $\roleFmt{C}$ \small
    \boxed{P_{\roleFmt{C}}}};
   \draw[->] (LA) -- (PA);
   \draw[->] (LB) -- (PB);
   \draw[->] (LC) -- (PC);
\end{tikzpicture}
\caption{Top-down view of MPST with crash.}
\label{fig:overview-of-topdown}
\end{figure}

The Simpler Logging protocol consists of a \emph{logger} (\roleFmt{L}) that controls the
logging services, an \emph{interface} ($\roleFmt{I}$) that provides communications between
logger and client, and a \emph{client} ($\roleFmt{C}$) that requires logging services via interface.
Initially, the logger process $\roleFmt{L}$ sends a heartbeat message $\gtMsgFmt{trigger}$ to the
interface process $\roleFmt{I}$. Then the client $\roleFmt{C}$ sends a command to
the interface
to read the logs (\gtMsgFmt{read}).
When a \gtMsgFmt{read} request is
sent, it is forwarded to the logger, and the logger responds
with a \gtMsgFmt{report}, which is then forwarded onto the client.
When all participants (\ie logger, interface, and client) are assumed reliable, \ie without failures or crashes,
this logging behaviour can be represented by the \emph{global type}
$\gtG[0]$:

\smallskip
\centerline{\(
\gtG[0] \;=\; 
    \gtCommSingle{\roleFmt{L}}{\roleFmt{I}}{
  \gtMsgFmt{trigger}}{}{
    \gtCommSingle{\roleFmt{C}}{\roleFmt{I}}{
   \gtMsgFmt{read}}{}{
    \gtCommSingle{\roleFmt{I}}{\roleFmt{L}}
    {\gtMsgFmt{read}}{}{
    \gtCommSingle{\roleFmt{L}}{\roleFmt{I}}
    {\gtMsgFmt{report}}{\stFmtC{log}}{
    \gtCommSingle{\roleFmt{I}}{\roleFmt{C}}{\gtMsgFmt{report}}{\stFmtC{log}}{
    \gtEnd
    }}}}
       }
\)}

\smallskip
\noindent
Here, $\gtG[0]$ is a specification of the Simpler Logging protocol between multiple roles from
a global perspective.

In the real distributed world, all participants in the Simpler Logging system
may fail. Ergo, we may need to model protocols with failures or
crashes and handling behaviour thereof, \eg should the client fail after the logging has started, the interface
will inform the logger to stop and exit.
We follow~\cite[\S 2.2]{DBLP:books/daglib/0025983} to model
a \emph{crash-stop} semantics, where we assume that roles can crash \emph{at any time} unless assumed
\emph{reliable} (never fail or crash).
For simplicity, we assume the interface $\roleFmt{I}$ and the logger $\roleFmt{L}$ to be reliable.
The above logging behaviour, incorporating crash-stop failures, can be represented by extending
$\gtG[0]$ with a branch handling crashes of $\roleFmt{C}$:
\begin{equation}
\label{ex:overview-global-with-crash}
\gtG =
    \gtCommSingle{\roleFmt{L}}{\roleFmt{I}}{
  \gtMsgFmt{trigger}}{}{
    \gtCommRaw{\roleFmt{C}}{\roleFmt{I}}{
    \begin{array}{@{}l@{}}
    \gtCommChoice{\gtMsgFmt{read}}{}{
    \gtCommSingle{\roleFmt{I}}{\roleFmt{L}}
    {\gtMsgFmt{read}}{}{
    \gtCommSingle{\roleFmt{L}}{\roleFmt{I}}
    {\gtMsgFmt{report}}{\stFmtC{log}}{
    \gtCommSingle{\roleFmt{I}}{\roleFmt{C}}{\gtMsgFmt{report}}{\stFmtC{log}}{
    \gtEnd
    }}}}\\
     \gtCommChoice{\gtMsgFmt{crash}}{}{
    \gtCommSingle{\roleFmt{I}}{\roleFmt{L}}{\gtMsgFmt{fatal}}{}{
    \gtEnd
    }}
    \end{array}
   }}
   \end{equation}

We model crash detection on receiving roles: when a role $\roleFmt{I}$ is
waiting to receive a message from role $\roleFmt{C}$, the receiving role $\roleFmt{I}$ is
able to detect whether $\roleFmt{C}$ has $\gtCrashLab$ed.
Since crashes are detected only by the receiving role, we do not
require a crash handling branch on the communication step between $\roleFmt{I}$ and
$\roleFmt{C}$ -- nor do we require them on any interaction between $\roleFmt{L}$ and $\roleFmt{I}$ (since we
are assuming that $\roleFmt{L}$ and $\roleFmt{I}$ are reliable).
Note that our global type syntax has minimal changes from the standard syntax,
with additional crashing and crash handling semantics. We give details
in Section~\ref{sec:gtype}.

Following the MPST top-down methodology,
a global type is then \emph{projected} onto \emph{local types},
which describe communications from the perspective of a single role.
In our unreliable Simpler Logging example,
$\gtG$ is projected onto three local types (one for each role
$\roleFmt{C}$, $\roleFmt{L}$, $\roleFmt{I}$):

\smallskip
\centerline{\(
 \begin{array}{c}
  \stT[\roleFmt{C}] =
 \roleFmt{I} \stFmt{\oplus}
   \stLabFmt{read}
 \stSeq
  \roleFmt{I}
   \stFmt{\&}
    \stLabFmt{report(\stFmtC{log})}
 \stSeq
\stEnd
  \quad
   \stT[\roleFmt{L}] =
   \roleFmt{I} \stFmt{\oplus}
   \stLabFmt{trigger}
 \stSeq
  \stExtSum{\roleFmt{I}}{}{
  \begin{array}{@{}l@{}}
\stLabFmt{read}
  \stSeq
   \roleFmt{I} \stFmt{\oplus}
    \stLabFmt{report(\stFmtC{log})}
    \stSeq
    \stEnd
    \\
    \stLabFmt{fatal}
  \stSeq
   \stEnd
  \end{array}
  }
\\[3mm]
  \stT[\roleFmt{I}]  =
    \roleFmt{L}
  \stFmt{\&}
  \stLabFmt{trigger}
  \stSeq
  \stExtSum{\roleFmt{C}}{}{
  \begin{array}{@{}l@{}}
 \stLabFmt{read}
 \stSeq
 \roleFmt{L}
  \stFmt{\oplus}
  \stLabFmt{read}
  \stSeq
 \roleFmt{L}
\stFmt{\&}
\stLabFmt{report(\stFmtC{log})}
\stSeq
\roleFmt{C}
\stFmt{\oplus}
\stLabFmt{report(\stFmtC{log})}
\stSeq
\stEnd
 \\
  \stCrashLab
  \stSeq
  \roleFmt{L}
  \stFmt{\oplus}
  \stLabFmt{fatal}
 \stSeq
\stEnd
 \end{array}
 }
\end{array}
\)}
\smallskip

\noindent
Here,
$\stT[\roleFmt{I}]$ states that the interface
 $\roleFmt{I}$ first receives a trigger message from the logger
 $\roleFmt{L}$; then $\roleFmt{I}$ either expects a  $\stLabFmt{read}$ request from
 the client $\roleFmt{C}$, or detects the crash of $\roleFmt{C}$ and handles it (in $\stCrashLab$)
 by sending the $\stLabFmt{fatal}$ message to notify the logger $\roleFmt{L}$.
Again, the syntax of local types does not depart from the standard, but we add
additional crash modelling and introduce a $\stStop$ type for crashed
endpoints, explained in Section~\ref{sec:gtype}. 
We show the
    operational correspondence between global and local type
     semantics in Theorems~\ref{thm:gtype:proj-comp} and~\ref{thm:gtype:proj-sound}, 
     and demonstrate that
a projectable global type always produces a safe, deadlock-free, 
and live typing context in~\autoref{cor:allproperties}.

The next step in this top-down methodology is to use
local types to type-check processes $\mpP[i]$ executed by role $\roleP[i]$
in our session calculus.
For example, $\stT[\roleFmt{I}]$ can be used to type check the interface
$\roleFmt{I}$ that executes the process:

\smallskip
\centerline{\(
\procin{\roleFmt{L}}{\labFmt{trigger}}{\sum \setenum{
  \begin{array}{@{}l@{}}
 \procin{\roleFmt{C}}{\labFmt{read}}{
 \procoutNoVal{\roleFmt{L}}{\labFmt{read}}
 {\procin{\roleFmt{L}}{\labFmt{report}(\mpx)}
 {\procout{\roleFmt{C}}{\labFmt{report}}{\mpx}{\mpNil}}}
 }
\\
  \procin{\roleFmt{C}}{\mpCrashLab}{ \procoutNoVal{\roleFmt{L}}{\labFmt{fatal}}{\mpNil}}
  \end{array}
  }}
\)} %
\smallskip

\noindent
 In our operational semantics (Section~\ref{sec:process}),
we allow active processes executed by unreliable roles
to crash \emph{arbitrarily}.
Therefore, the role executing the crashed process
also crashes, and is assigned the local type $\stStop$.
To ensure that a communicating process is type-safe even in the presence of crashes,
we require that its typing context %
satisfies a \emph{safety property} accounting for possible crashes (\autoref{def:mpst-env-safe}),
which is preserved by
projection (\autoref{cor:allproperties}).
Despite minor changes in the surface syntax of types and processes,
additional semantics surrounding crashes adds subtleties even in standard
results. We prove subject reduction and session fidelity results,  accounting for crashes and
sets of reliable roles, in Theorems~\ref{lem:sr} and~\ref{lem:sf}.

\section{Asynchronous Multiparty Session Calculus with  Crash-Stop Semantics }
\label{sec:process}
\label{sec:session-calculus}
In this section, we formalise the syntax (Section~\ref{sec:process:syntax}) and operational semantics (Section~\ref{sec:process:semantics})
of our asynchronous multiparty session calculus with process failures and crash detection.
For clarity of presentation, delegation is not supported in our system.

\subsection{Syntax of Session Calculus with Crash-Stop Failures}
\label{sec:process:syntax}

Our asynchronous multiparty session calculus models processes that may crash arbitrarily.
Our formalisation is based on~\cite{POPL21AsyncMPSTSubtyping} -- but
additionally follows the \emph{fail-stop} model in \cite[\S 2.7]{DBLP:books/daglib/0025983},
where processes may crash and never recover, and where process failures can be detected by failure detectors~\cite[\S 2.6.2]{DBLP:books/daglib/0025983} \cite{JACM96FailureDetector} when attempting to receive messages.

The syntax of our calculus is presented in~\cref{fig:processes:syntax}.  %
Values and expressions are standard: a \emph{value}, ranged over by $\mpV, \mpVi, \ldots$,
can be a natural number $\mpNat$,
an integer $\mpInt$,
a boolean $\mpTrue$ or $\mpFalse$,
a string $\mpString{}$,
a unit $\mpUnit{}$~(often omitted for brevity),
or any other specifically tailored value;
an \emph{expression}, ranged over by $\mpE, \mpEi, \ldots$,
can be a variable, a value,
or a term built from expressions by applying operators, such as
$\text{\texttt{succ}}, \text{\texttt{neg}}, \neg$, and  $<$.
We fix a set of \emph{roles} $\roleSet$, ranged over by $\roleP, \roleQ,
\roleR, \ldots$.

\begin{figure}
 \noindent
  \centerline{\(
  \begin{array}{c}
  \mpV
   \bnfdef
  \mpNat \bnfsep \mpInt \bnfsep \mpTrue \bnfsep \mpFalse \bnfsep \mpString{} \bnfsep \mpUnit{} \bnfsep \cdots
  \\[.7mm]
  \mpE
  \coloncolonequals
  \mpFmt{x} \bnfsep \mpV \bnfsep \mpSucc{\mpE} \bnfsep \mpNeg{\mpE} \bnfsep
   \neg \mpE \bnfsep  \mpE < \mpEi \bnfsep \cdots
  \\[1.5mm]
  \begin{array}{@{}r@{\;}c@{\;}ll@{}}
    \mpP,\mpQ & \bnfdef & & \text{\em\bf Processes} \\
            & & \sum_{i\in I}\procin{\roleP}{\mpLab_i(\mpx_i)}{\mpP_i} & \text{\em external choice}\\
            & \bnfsep & \procout{\roleP}{\mpLab}{\mpE}{\mpP} \quad
            \highlight{\text{(\em where } \mpLab \neq \mpCrashLab\text{)}}
            & \text{\em output} \\
            & \bnfsep & \mpIf{\mpE} \mpP \mpQ & \text{\em conditional} \\
            & \bnfsep & \mpX & \text{\em process variable} \\
            & \bnfsep & \mu \mpX.\mpP & \text{\em recursion} \\
            & \bnfsep &  \mpNil  & \text{\em inaction} \\
            & \bnfsep & \highlight{\mpCrash} & \highlight{\text{\em crashed}}
  \end{array}%
  \quad%
  \begin{array}{@{}r@{\;}c@{\;}ll@{}}
    \mpM & \bnfdef &  & \text{\em\bf Sessions} \\
       &         & \mpPart\roleP\mpP \mpPar \mpPart\roleP \mpH & \text{\em
       participant} \\
       & \bnfsep & \mpM\mpPar\mpM  & \text{\em parallel} \\
    \mpH & \bnfdef & & \text{\em\bf %
    Queues} \\
       &         & \mpQEmpty & \text{\em empty %
       } \\
       & \bnfsep & \highlight{\mpQUnavail} & \highlight{\text{\em unavailable %
       }} \\
       & \bnfsep & \left(\roleP,\mpLab(\mpV)\right) & \text{\em message} \\
       & \bnfsep & \mpH\cdot \mpH  &\text{\em concatenation}
  \end{array}
  \end{array}
\)}
\caption{Syntax of values, expressions, sessions, processes, and queues.
Appreciable changes w.r.t.\@ the standard session calculus~\cite{POPL21AsyncMPSTSubtyping} are $\highlight{\text{highlighted}}$.}
\label{fig:processes:syntax}
\end{figure}

A process, ranged over by $\mpP, \mpQ, \ldots$, is a communication agent within a
session.
An \emph{output} process
$\procout{\roleP}{\mpLab}{\mpE}{\mpP}$
sends a message to another role $\roleP$ in the session,
where the message is labelled $\mpLab$, and carries a payload expression $\mpE$,
then the process continues as $\mpP$.
An \emph{external choice} (\emph{input}) process
$\sum_{i\in I}\procin{\roleP}{\mpLab_i(\mpx_i)}{\mpP_i}$
receives a message from another role $\roleP$ in the session,
among a finite set of indices $I$,
if the message is labelled $\mpLab[i]$, then the payload would be received as
$\mpx[i]$, and the process continues as $\mpP[i]$.
Note that our calculus uses $\mpCrashLab$ as a special message label
denoting that a participant (role) has crashed.
Such a label cannot be sent by any process, but a process can implement
crash detection and handling by receiving it.
Consequently, an output process \emph{cannot} send
 a $\mpCrashLab$ message (side condition $\mpLab \neq \mpCrashLab$),
 whereas an input process
may include a \emph{crash handling branch} of the form $\mpCrashLab.\mpPi$
meaning that $\mpPi$ is executed when the sending role has crashed.
A \emph{conditional} process
$\mpIf{\mpE} \mpP \mpQ$
continues as either $\mpP$ or $\mpQ$ depending on the evaluation of $\mpE$
(which is detailed later in Section~\ref{sec:process:semantics}).
We allow \emph{recursion} at the process level using
$\mu \mpX.\mpP$ and $\mpX$, requiring process recursion variables to be guarded by an input or
output action.
Finally, we write $\mpNil$ for an \emph{inactive} process, representing
a successful termination; and $\mpCrash$ for a \emph{crashed} process, representing a
termination due to failure.

An \emph{incoming queue}\footnotemark, ranged over by $\mpH, \mpHi, \mpH[i], \ldots$, is a sequence of messages
tagged with their origin.
We write $\mpQEmpty$ for an \emph{empty} queue;
$\mpQUnavail$ for an \emph{unavailable} queue;
and $(\roleP,\mpLab(\val))$
for a message sent from $\roleP$, labelled $\mpLab$, and containing a payload
value $\mpV$.
We write $\mpH[1] \cdot \mpH[2]$ to denote the concatenation of two queues
$\mpH[1]$ and $\mpH[2]$.
When describing incoming queues, we consider two messages from different
origins as swappable:
\(
  \mpH_1
  \cdot
  \msg{\roleQ[1]}{\mpLab_1}{\val_1}
  \cdot
  \msg{\roleQ[2]}{\mpLab_2}{\val_2}
  \cdot
  \mpH_2
  \)
 is structurally equivalent to  %
 \(
  \mpH_1
  \cdot
  \msg{\roleQ[2]}{\mpLab_2}{\val_2}
  \cdot
  \msg{\roleQ[1]}{\mpLab_1}{\val_1}
  \cdot
  \mpH_2
\)
whenever $\roleQ[1] \neq \roleQ[2]$.
Moreover, we consider concatenation $(\cdot)$ as associative, and the empty
queue $\mpQEmpty$ as the identity element for concatenation.

\footnotetext{In~\cite{POPL21AsyncMPSTSubtyping}, the queues are outgoing
instead of incoming.
We use incoming queues to model our crashing semantics more easily.}

A session, ranged over by $\mpM, \mpMi, \mpM[i], \ldots$, consists of processes and their
respective incoming queue, indexed by their roles.
A single entry for a role $\roleP$ is denoted $\mpPart{\roleP}{\mpP} \mpPar
\mpPart{\roleP}{\mpH}$, where $\mpP$ is the process for $\roleP$ and $\mpH$ is
the incoming queue.
Entries are composed together in parallel as $\mpM \mpPar \mpMi$, where the
roles in $\mpM$ and $\mpMi$ are disjoint.
We consider parallel composition as commutative and associative, with
$\mpPart{\roleP}{\mpNil} \mpPar \mpPart{\roleP}{\mpQEmpty}$ as a neutral
element of the operator.
We write
$\prod_{i\in I} (\mpPart{\roleP[i]}{\mpP[i]} \mpPar \mpPart{\roleP[i]}{\mpH[i]})$
for the parallel composition of multiple entries in a set.

\subsection{Operational Semantics of Session Calculus with Crash-Stop Failures}
\label{sec:process:semantics}

The \emph{evaluation} of an expression is illustrated in~\Cref{fig:expression_eval},
where $\eval{\mpE}{\mpV}$ indicates that the expression $\mpE$ evaluates to the value $\mpV$,
and an \emph{evaluation context} $\mathcal{E}$ is an expression containing exactly one hole.
The successor operation $\text{\texttt{succ}}$ is defined only for
natural numbers,
the negation $\text{\texttt{neg}}$ is defined for integers, the logical negation $\neg$ is defined only
for boolean values, and the less-than operator $<$ is applied only to integers.

We give the operational semantics of our session calculus in~\autoref{def:session:red},
using a \emph{structural precongruence} $\Rrightarrow$ defined in \Cref{tab:precongruence}.
Structural precongruence is a \emph{preorder} that relates sessions  based on inconsequential structural modifications.
A standard structural congruence $\equiv$ can be defined as the symmetric closure of $\Rrightarrow$:  $\Rrightarrow \cup \Rrightarrow^{-}$.

Our semantics parameterises on a (possibly empty) set of \emph{reliable} roles $\rolesR$,
\ie roles assumed to \emph{never crash}.
This approach enables us to model \emph{optional reliability}, allowing for the representation of
a \emph{spectrum} of reliability assumptions
that range from total process reliability~(as seen in standard session types
work, \ie $\rolesR = \roleSet$)
to total unreliability (\ie $\rolesR = \emptyset$).

\begin{defi}[Session Reductions]
\label{def:session:red}
The session reduction relation $\mpMove[\rolesR]$
  is inductively
  defined by the rules in \cref{fig:processes:reduction},
  parameterised by a fixed set $\rolesR$ of reliable roles.
  We write $\mpMove$
whenever $\rolesR$ is not significant.
We write $\mpMoveStar[\rolesR]$ (resp.  $\mpMoveStar$)
for the reflexive and transitive closure of $\mpMove[\rolesR]$ (resp. $\mpMove$).
\end{defi}

Our operational semantics retains the basic rules
in~\cite{POPL21AsyncMPSTSubtyping}, but also includes ($\highlight{\text{highlighted}}$)
rules for crash-stop
failures and crash handling, adapted from~\cite{CONCUR22MPSTCrash}.
Rules \inferrule{r-send} and \inferrule{r-rcv} model ordinary message delivery
and reception:
an output process located at $\roleP$ sending to $\roleQ$ would append a
message to the incoming queue of $\roleQ$; and an input process located at
$\roleP$ receiving from $\roleQ$ would consume the first message from the
incoming queue.
Rules \inferrule{r-cond-T} and \inferrule{r-cond-F} model the conditional
process; and rule \inferrule{r-struct} permits reductions up to structural
precongruence.

With regard to crashes and related behaviour, rule \inferrule{r-$\lightning$}
models process crashes: an active ($\mpP \neq \mpNil$) process located at an
unreliable role ($\roleP \notin \rolesR$) may reduce to a crashed process
$\mpPart{\roleP}{\mpCrash}$, with its incoming queue becoming unavailable
$\mpPart{\roleP}{\mpQUnavail}$.
Rule \inferrule{r-send-$\lightning$} models a
message delivery to a crashed role (and thus an unavailable queue),
and the message becomes lost and would not be added to the queue.
Rule \inferrule{r-rcv-$\odot$} models crash detection, which activates as a
`last resort':
an input process at $\roleP$ receiving from $\roleQ$ would first attempt find a
message from $\roleQ$ in the incoming queue, which engages the usual rule
\inferrule{r-recv};
if none exists and $\roleQ$ has crashed
($\mpPart{\roleQ}{\mpCrash}$), then the crash handling branch in the input
process at $\roleP$ can activate.
We draw attention to the interesting fact that \inferrule{r-recv} may engage
even if $\roleQ$ has crashed, in cases where a message from $\roleQ$ in the
incoming queue may be consumed.

\begin{figure}[t]
  \centerline{\(
    \begin{array}{@{\hskip 0mm}c@{\hskip 0mm}}
    \eval{\mpV}{\mpV} 
    \quad 
    \eval{\mpSucc{\mpNat}}{(\mpNat + 1)}
    \quad 
    \eval{\mpNeg{\mpInt}}{(-\mpInt)}
    \quad 
    \eval{\neg \mpTrue}{\mpFalse}
    \quad 
    \eval{\neg \mpFalse}{\mpTrue}
     \\[1.5mm]%
   \eval{\mpInt_{1} < \mpInt_{2}}{\left\{%
    \begin{array}{@{}l@{\hskip 5mm}l@{}}
      \mpTrue
      &\text{\small %
        if\, $\mpInt_{1} < \mpInt_{2}$}%
      \\[1mm]%
     \mpFalse
      &\text{\small%
      otherwise
      }
    \end{array}
    \right.
}
\qquad 
  \inference[]{\eval{\mpE}{\mpV} & \eval{\mathcal{E}(\mpV)}{\mpVi}}{\eval{\mathcal{E}(\mpE)}{\mpVi}} 
        \end{array}
    \)}%
\caption{%
Expression evaluation. 
}%
  \label{fig:expression_eval}%
\end{figure}

\begin{figure}
 \noindent
  \centerline{\(
\begin{array}{@{}llr@{}}
\highlight{\inferrule{r-$\lightning$}} &
\highlight{{
\mpPart\roleP{\mpP}
\mpPar
\mpPart\roleP{\mpH[\roleP]}
\mpPar
\mpM
\;\redCrash{\roleP}{\rolesR}\;
\mpPart\roleP{\mpCrash}
\mpPar
\mpPart\roleP{\mpQUnavail}
\mpPar
\mpM
}}
&
\hspace{-11em}
\highlight{(\mpP \neq \mpNil, \roleP \notin \rolesR)}
\\
\inferrule{r-send} &
{
\mpPart\roleP{\procout{\roleQ}{\mpLab}{\mpE}{\mpP}}
\mpPar
\mpPart\roleP{\mpH[\roleP]}
\mpPar
\mpPart\roleQ{\mpQ}
\mpPar
\mpPart\roleQ{\mpH[\roleQ]}
\mpPar
\mpM
}
\\
&
{%
\redSend{\roleP}{\roleQ}{\mpLab}\;
\mpPart\roleP{\mpP}
\mpPar
\mpPart\roleP{\mpH[\roleP]}
\mpPar
\mpPart\roleQ{\mpQ}
\mpPar
\mpPart{\roleQ}{\mpH[\roleQ]}\cdot(\roleP,\mpLab(\val))
\mpPar
\mpM
}
&
\hspace{-11em}
(\eval{\mpE} \val, \mpH[\roleQ] \neq \mpQUnavail)
\\[1mm]
\highlight{\inferrule{r-send-\mpCrash}} &
{\highlight{
\mpPart\roleP{\procout{\roleQ}{\mpLab}{\mpE}{\mpP}}
\mpPar
\mpPart\roleP{\mpH[\roleP]}
\mpPar
\mpPart\roleQ{\mpCrash}
\mpPar
\mpPart\roleQ{\mpQUnavail}
\mpPar
\mpM
\;\redSend{\roleP}{\roleQ}{\mpLab}\;
\mpPart\roleP{\mpP}
\mpPar
\mpPart\roleP{\mpH[\roleP]}
\mpPar
\mpPart\roleQ{\mpCrash}
\mpPar
\mpPart{\roleQ}{\mpQUnavail}
\mpPar
\mpM}
\hspace{-16em}
}
&
\\[1mm]
\inferrule{r-rcv}
&
\mpPart\roleP{\sum_{i\in I} \procin\roleQ{\mpLab_i(\mpx_i)}\mpP_i}
\mpPar
\mpPart\roleP{(\roleQ,\mpLab_k(\val)) \cdot \mpH[\roleP]}
\mpPar
\mpM
\redRecv{\roleP}{\roleQ}{\mpLab_k}\;
\mpPart\roleP \mpP_k\subst{\mpx_k}{\val}
\mpPar
\mpPart\roleP{\mpH_{\roleP}}
\mpPar
\mpM
\hspace{-21em}
&
(k \in I)
\\[1mm]
\highlight{\inferrule{r-rcv-$\odot$}}
&
\highlight{\mpPart\roleP{\sum_{i\in I} \procin\roleQ{\mpLab_i(\mpx_i)}\mpP_i}
\mpPar
\mpPart\roleP{\mpH[\roleP]}
\mpPar
\mpPart\roleQ\mpCrash
\mpPar
\mpPart\roleQ\mpQUnavail
\mpPar
\mpM}
\\
&
\highlight{\redRecv{\roleP}{\roleQ}{\mpLab_k}\;
\mpPart\roleP \mpP_k
\mpPar
\mpPart\roleP{\mpH_{\roleP}}
\mpPar
\mpPart\roleQ\mpCrash
\mpPar
\mpPart\roleQ\mpQUnavail
\mpPar
\mpM}
&
\hspace{-28em}
\highlight{(k \in I, \mpLab[k] = \mpCrashLab, \nexists \mpLab, \mpV: (\roleQ,
\mpLab(\mpV)) \in \mpH[\roleP])}
\\[1mm]
\inferrule{r-cond-T} &
\mpPart\roleP{\mpIf{\mpE}{\mpP}{\mpQ}}
\mpPar
\mpPart\roleP\mpH
\mpPar
\mpM
\;\redIf{\roleP}\;
\mpPart\roleP\mpP
\mpPar
\mpPart\roleP\mpH
\mpPar
\mpM
&
\hspace{-28em}
(\eval{\mpE}{\mpTrue})
\\[1mm]
\inferrule{r-cond-F} &
\mpPart\roleP{\mpIf{\mpE}{\mpP}{\mpQ}}
\mpPar
\mpPart\roleP\mpH
\mpPar
\mpM
\;\redIf{\roleP}\;
\mpPart\roleP\mpQ
\mpPar
\mpPart\roleP\mpH
\mpPar
\mpM
& 
\hspace{-28em}
(\eval{\mpE}{\mpFalse})
\\[1mm]
\inferrule{r-struct} &
\mpM_1\Rrightarrow \mpM_1'
\;\;\;\text{and}\;\;\;
\mpM_1'\mpMove \mpM_2'
\;\;\;\text{and}\;\;\;
\mpM_2' \Rrightarrow \mpM_2
\quad\implies\quad
\mpM_1\mpMove\mpM_2
\end{array}
\)}
\caption{Reduction relation on sessions with crash-stop failures.}%
\label{fig:processes:reduction}
\end{figure}

\begin{figure}
 \noindent
  \centerline{\(
\begin{array}{c}
\mpH_1
\cdot
\msg{\roleQ_i}{\mpLab_i}{\val_i}
\cdot
\msg{\roleQ_j}{\mpLab_j}{\val_j}
\cdot
\mpH_2
\Rrightarrow
\mpH_1
\cdot
\msg{\roleQ_j}{\mpLab_j}{\val_j}
\cdot
\msg{\roleQ_i}{\mpLab_i}{\val_i}
\cdot
\mpH_2
\\[1mm]
\text{(if $i \neq j$, $i, j \in \setenum{1, 2}$, $\roleQ_1 \neq \roleQ_2$)}
\\[1.5mm]
\mpQEmpty \cdot \mpH \Rrightarrow  \mpH \quad
\mpH \cdot \mpQEmpty \Rrightarrow  \mpH \quad
\mpH[1] \cdot (\mpH[2] \cdot \mpH[3])
\Rrightarrow
(\mpH[1] \cdot \mpH[2]) \cdot \mpH[3] \quad 
(\mpH[1] \cdot \mpH[2]) \cdot \mpH[3] 
\Rrightarrow
\mpH[1] \cdot (\mpH[2] \cdot \mpH[3])
\\[1.5mm]
\mpPart\roleP\mpNil
\mpPar
\mpPart\roleP\mpQEmpty
\mpPar
\mpM
\Rrightarrow
\mpM
\quad
\mu X.\mpP \Rrightarrow \mpP\subst{X}{\mu X.\mpP}
\quad
\mpM \Rrightarrow \mpMi
\;\text{ and }\;
\mpMi \Rrightarrow \mpMii
\;\implies\;
\mpM
\Rrightarrow
\mpMii
\\[1.5mm]
\prod_{i\in I} (\mpPart{\roleP[i]}{\mpP[i]} \mpPar \mpPart{\roleP[i]}{\mpH[i]}) 
\Rrightarrow
\prod_{j\in J} (\mpPart{\roleP[j]}{\mpP[j]} \mpPar \mpPart{\roleP[j]}{\mpH[j]}) 
\quad \text{(if $I$ is a permutation of $J$)}
\\[1.5mm]
\mpP \Rrightarrow \mpQ
\;\text{ and }\;
\mpH_1 \Rrightarrow \mpH_2
\;\implies\;
\mpPart\roleP\mpP
\mpPar
\mpPart\roleP\mpH_1
\mpPar
\mpM
\Rrightarrow
\mpPart\roleP\mpQ
\mpPar
\mpPart\roleP\mpH_2
\mpPar
\mpM
\end{array}
\)}
\caption{Structural precongruence rules for queues, processes, and sessions.}
\label{tab:precongruence}
\end{figure}

\begin{exa}
\label{ex:process_syntax_semantics}
We now illustrate our operational semantics of sessions with an example.
Consider the session:

\smallskip
\centerline{\(
\mpPart\roleP\mpP \mpPar  \mpPart\roleP \mpQEmpty
     \mpPar \mpPart \roleQ \mpQ \mpPar  \mpPart \roleQ \mpQEmpty
\)}

\smallskip
\noindent
where
$ \mpP = \procout{\roleQ}{\mpLab}{\text{``\texttt{abc}''}}{\mpPi}$
with
$\mpPi = \sum
\setenum{
\begin{array}{l}
\procin{\roleQ}{\mpLabi(\mpx)}{\mpNil}
\\
\procin{\roleQ}{\mpCrashLab}{\mpNil}
\end{array}
}
$,
 and
$
 \mpQ =
  \sum \setenum{
 \begin{array}{l}
   \procin{\roleP}{\mpLab(\mpx)}{\mpQi}
 \\
  \procin{\roleP}{\mpCrashLab}{\mpNil}
  \end{array}
  }
  $
  with
  $
  \mpQi = \procout{\roleP}{\mpLabi}{42}{\mpNil}
  $.

In this session,  the process $\mpQ$ for $\roleQ$ receives a message sent from $\roleP$ to $\roleQ$;
the process $\mpP$ for $\roleP$ sends a message from $\roleP$ to $\roleQ$, and then receives a message
sent from $\roleQ$ to  $\roleP$.

 On a successful reduction (without crashes),  %
     we have:

 \smallskip
 \centerline{\(
 \begin{array}{rll}
  \mpPart\roleP\mpP \mpPar  \mpPart\roleP \mpQEmpty
     \mpPar \mpPart \roleQ \mpQ \mpPar  \mpPart \roleQ \mpQEmpty
     & \redSend{\roleP}{\roleQ}{\mpLab} &
  \mpPart \roleP \mpPi
\mpPar
\mpPart\roleP \mpQEmpty  \mpPar
\mpPart \roleQ \mpQ \mpPar  \mpPart \roleQ (\roleP,\mpLab(\text{``\texttt{abc}''}))
\\
& \redSend{\roleP}{\roleQ}{\mpLab} &
  \mpPart \roleP \mpPi
\mpPar
\mpPart\roleP \mpQEmpty
\mpPar
\mpPart \roleQ \mpQi
\mpPar
\mpPart\roleQ \mpQEmpty
\\
& \redSend{\roleP}{\roleQ}{\mpLab} &
  \mpPart \roleP \mpPi
\mpPar
\mpPart\roleP (\roleQ, \mpLabi(42))
\mpPar
\mpPart \roleQ \mpNil
\mpPar
\mpPart\roleQ \mpQEmpty
\\
& \redSend{\roleP}{\roleQ}{\mpLab} &
\mpPart\roleP \mpNil
\mpPar
\mpPart\roleP \mpQEmpty
\mpPar
\mpPart \roleQ \mpNil
\mpPar
\mpPart\roleQ \mpQEmpty
\end{array}
\)}

\smallskip
\noindent
Let $\rolesR = \emptyset$ (\ie each role is unreliable).
Suppose that
$\mpP$ crashes before sending, which leads to the reduction:

\smallskip
 \centerline{\(
 \begin{array}{rll}
  \mpPart\roleP\mpP \mpPar  \mpPart\roleP \mpQEmpty
     \mpPar \mpPart \roleQ \mpQ \mpPar  \mpPart \roleQ \mpQEmpty
    &\redCrash{\roleP}{\rolesR} &
     \mpPart\roleP{\mpCrash}
\mpPar
\mpPart\roleP{\mpQUnavail}
\mpPar
\mpPart \roleQ \mpQ \mpPar  \mpPart \roleQ \mpQEmpty \\
    &\redSend{\roleP}{\roleQ}{\mpLab}&
\mpPart\roleP{\mpCrash}
\mpPar
\mpPart\roleP{\mpQUnavail}
\mpPar
\mpPart \roleQ \mpNil
\mpPar
\mpPart\roleQ \mpQEmpty
\end{array}
\)}

\smallskip
\noindent
We can observe that when the output (sending) process $\mpP$ located at an unreliable role $\roleP$ crashes (by
\inferrule{r-$\lightning$}), $\roleP$ also crashes ($\mpPart\roleP{\mpCrash}$),  with an unavailable incoming message
queue ($\mpPart\roleP{\mpQUnavail}$). Subsequently, the input (receiving) process $\mpQ$ located at $\roleQ$ can detect and
handle the crash by \inferrule{r-rcv-$\odot$} via its handling branch.
\end{exa}

\section{Asynchronous Multiparty Session Types with Crash-Stop Semantics}
\label{sec:gtype}
\label{SEC:GTYPE}

In this section, we present our asynchronous multiparty session types
with crash-stop semantics.
We give an overview of global and local types with crashes in Section~\ref{sec:gtype:syntax},
including syntax, projection, and subtyping.
Our key additions to the classic theory are \emph{crash handling branches} in
both global and local types,
and a special local type $\stStop$ to denote crashed processes.
We give a Labelled Transition System (LTS) semantics to both global types
(Section~\ref{sec:gtype:lts-gt}) and configurations (\ie a collection of local types
and point-to-point communication queues, Section~\ref{sec:gtype:lts-context}).
We discuss alternative design options of modelling crash-stop
failures in Section~\ref{sec:gtype:alternative}.
We relate the two semantics in Section~\ref{sec:gtype:relating}, and show that
a configuration obtained via projection is safe, deadlock-free, and live in Section~\ref{sec:gtype:pbp}.

\subsection{Global and Local Types with Crash-Stop Failures}\label{sec:gtype:syntax}

The top-down methodology begins with \emph{global types} to provide an
overview of
the communication between a number of \emph{roles}
($\roleP, \roleQ, \roleS, \roleT, \ldots$),
belonging to a (fixed) set $\roleSet$.
On the other end, we use \emph{local types} to describe how a \emph{single}
role communicates with other roles from a local perspective, and they are
obtained via \emph{projection} from a global type.
We give the syntax of both global and local types in \cref{fig:syntax-mpst},
which are similar to syntax used in~\cite{POPL19LessIsMore,CONCUR22MPSTCrash}.

\begin{figure}[t]%
   \noindent
  \centerline{\(
    \begin{array}{r@{\quad}c@{\quad}l@{\quad}l}
      \tyGround & \bnfdef & \tyNat \bnfsep \tyInt \bnfsep \tyBool \bnfsep \tyString \bnfsep%
      \tyUnit \bnfsep \ldots
        & \text{\small Basic types} 
        \\[.7mm]
      \gtG & \bnfdef &
        \gtComm{\roleP}{\roleFmt{q^{\runtime{\roleMaybeCrashedSym}}}}{i \in
        I}{\gtLab[i]}{\tyGround[i]}{\gtG[i]}
        &
        {\small \text{Transmission}} 
        \\[.7mm]
        & \bnfsep &
        \runtime{
          \gtCommTransit{\rolePMaybeCrashed}{\roleQ}{i \in
          I}{\gtLab[i]}{\tyGround[i]}{\gtG[i]}{j}
        }~(j \in I)
        &
        \runtime{\small \text{Transmission en route (Runtime)}} 
        \\[.7mm]
        & \bnfsep & \gtRec{\gtRecVar}{\gtG} \quad \bnfsep \quad \gtRecVar \quad
        \bnfsep \quad \gtEnd &
        \text{\small Recursion, Type variable, Termination} 
        \\[.7mm]
      \runtime{\roleMaybeCrashedSym} & \bnfdef & \cdot \quad
      \bnfsep \quad \runtime{\roleCrashedSym} & \runtime{\text{\small Crash
      annotation (Runtime)}}
        \\[1ex]
      \stS, \stT
        & \bnfdef & \stExtSum{\roleP}{i \in I}{\stChoice{\stLab[i]}{\tyGround[i]} \stSeq \stT[i]}
          & \text{\small External choice (Receive)} 
          \\[.7mm]
        & \bnfsep & \stIntSum{\roleP}{i \in I}{\stChoice{\stLab[i]}{\tyGround[i]} \stSeq \stT[i]}
          & \text{\small Internal choice (Send)}
          \\[.7mm]
          & \bnfsep & \stRec{\stRecVar}{\stT} \ \bnfsep \ \stRecVar 
          &\text{\small Recursion, Type variable} 
          \\[.7mm]
        & \bnfsep & \stEnd \ \bnfsep \ \runtime{\stStop} &
        \text{\small Termination, $\runtime{\text{Crash (Runtime)}}$} 
        \\
    \end{array}
  \)}
  \caption{Syntax of basic types, global types,  and local types. Runtime types are
  $\runtime{\text{shaded}}$.}
  \label{fig:syntax-global-type}%
  \label{fig:syntax-local-type}%
  \label{fig:syntax-mpst}
\end{figure}

\paragraph*{Basic Types}
\emph{Basic types} (types for payloads) are taken from a set $\tyGroundSet$,
and describe the types of values such as natural numbers, integers, booleans, strings, and units.

\paragraph*{Global Types} %
\emph{Global types}
are ranged over $\gtG, \gtGi, \gtG[i], \ldots$,
and describe the behaviour for all roles from a bird's eye view.
The syntactic constructs shown in  $\runtime{\text{shade}}$ are \emph{runtime} constructs, which are not used for
describing a system at design-time, but for describing the state of a system
during execution. %

We explain each global type syntactic construct:
a transmission
$\gtComm{\roleP}{\roleQMaybeCrashed}{i \in I}{\gtLab[i]}{\tyGround[i]}{\gtG[i]}$
denotes a message from role $\roleP$ to role $\roleQ$ (with possible crash
annotations), with labels $\gtLab[i]$, payload types $\tyGround[i]$,
and continuations $\gtG[i]$, where $i$ is taken from an index set $I$.
We require that the index set be non-empty ($I \neq \emptyset$), labels
$\gtLab[i]$ be pair-wise distinct and taken from a fixed set of  labels $\labSet$,
and self receptions be excluded (\ie
$\roleP \neq \roleQ$), as is standard in session type literature. %
Additionally, we require that the special $\gtCrashLab$ label (explained later)
not be the only label in a transmission, \ie $\setcomp{\gtLab[i]}{i \in I} \neq
\setenum{\gtCrashLab}$.
A transmission en route
$\gtCommTransit{\rolePMaybeCrashed}{\roleQ}{i \in
I}{\gtLab[i]}{\tyGround[i]}{\gtG[i]}{j}$
is a runtime construct representing a message $\gtLab[j]$ (index $j$) sent by
$\roleP$, and yet to be received by $\roleQ$.
Recursive types are represented via $\gtRec{\gtRecVar}{G}$ and $\gtRecVar$,
where contractive requirements apply~\cite[\S 21.8]{PierceTAPL}.
The type $\gtEnd$ describes a terminated type (omitted where
unambiguous).

To model crashes and crash handling, we use crash annotations $\roleCrashedSym$
and crash handling branches:
a \emph{crash annotation} $\roleCrashedSym$, a new addition in this work, marks
a \emph{crashed} role (only used in the \emph{runtime syntax}), and we omit
annotations for \emph{live}~(or \emph{active}) roles, \ie $\roleP$ is a live role, $\rolePCrashed$ is a
crashed role, and $\rolePMaybeCrashed$ represents a possibly crashed role,
namely either $\roleP$ or $\rolePCrashed$.
We use a special label $\gtCrashLab$ for handling crashes: this
continuation denotes the protocol to follow when the sender of a
message is detected to have crashed by the receiver.
The special label acts as a `pseudo'-message: when a sender role crashes, the
receiver can select the `pseudo'-message to enter crash handling.

\begin{defi}[Set of Live and Crashed Roles]
\label{def:active_crashed_roles}
The set of \emph{live} %
roles in a global
type $\gtG$, denoted $\gtRoles{\gtG}$,
and the set of \emph{crashed} roles, denoted $\gtRolesCrashed{\gtG}$,
are defined inductively as:

\smallskip
\centerline{\(
\footnotesize
\begin{array}{r@{\;}l@{\;}lr@{\;}l@{\;}l}
    \gtRoles{\gtComm{\roleP}{\roleQ}{i \in I}{\gtLab[i]}{\tyGround[i]}{\gtG[i]}
    }
    &= &
     \setenum{\roleP, \roleQ} \cup \bigcup\limits_{i \in I}{\gtRoles{\gtG[i]}}
    &
     \gtRolesCrashed{\gtComm{\roleP}{\roleQ}{i \in I}{\gtLab[i]}{\tyGround[i]}{\gtG[i]}
    }
    &=&
     \bigcup\limits_{i \in I}{\gtRolesCrashed{\gtG[i]}}
    \\[0.5mm]
    \gtRoles{
      \gtComm{\roleP}{\roleQCrashed}{i \in I}{\gtLab[i]}{\tyGround[i]}{\gtG[i]}
    }
    &=&
    \setenum{\roleP} \cup \bigcup\limits_{i \in I}{\gtRoles{\gtG[i]}}
    &
    \gtRolesCrashed{
      \gtComm{\roleP}{\roleQCrashed}{i \in I}{\gtLab[i]}{\tyGround[i]}{\gtG[i]}
    }
    &=&
     \setenum{\roleQ} \cup \bigcup\limits_{i \in I}{\gtRolesCrashed{\gtG[i]}}
    \\[0.5mm]
    \gtRoles{
      \gtCommTransit{\rolePMaybeCrashed}{\roleQ}{i \in I}{\gtLab[i]}{\tyGround[i]}{\gtG[i]}{j}
    }
     &=&
     \setenum{\roleQ} \cup \bigcup\limits_{i \in I}{\gtRoles{\gtG[i]}}
    &
    \gtRolesCrashed{
      \gtCommTransit{\rolePMaybeCrashed}{\roleQ}{i \in I}{\gtLab[i]}{\tyGround[i]}{\gtG[i]}{j}
    }
     &=&
    \bigcup\limits_{i \in I}{\gtRolesCrashed{\gtG[i]}}
    \\[0.5mm]
    \gtRoles{\gtEnd} = \gtRoles{\gtRecVar}
     &=&
     \emptyset
    &
     \gtRolesCrashed{\gtEnd} = \gtRolesCrashed{\gtRecVar}
     &=&
     \emptyset
    \\[0.5mm]
    \gtRoles{\gtRec{\gtRecVar}{\gtG}}
     &=&
     \gtRoles{\gtG\subst{\gtRecVar}{\gtRec{\gtRecVar}{\gtG}}}
    &
     \gtRolesCrashed{\gtRec{\gtRecVar}{\gtG}}
     &=&
     \gtRolesCrashed{\gtG\subst{\gtRecVar}{\gtRec{\gtRecVar}{\gtG}}}
  \end{array}
 \)}
\end{defi}

\begin{rem}
  Even if the sender $\roleP$ occurs as crashed during a communication in
  transit, it is not considered
 as crashed unless it appears
 crashed in a continuation:
  \[
  \gtRolesCrashed{
    \gtCommTransit{\rolePCrashed}{\roleQ}{i \in I}{\gtLab[i]}{\tyGround[i]}{\gtG[i]}{j}
  }
  = \bigcup\limits_{i \in I}{\gtRolesCrashed{\gtG[i]}}.
\]
\end{rem}

\paragraph*{Local Types}  \emph{Local types} (or \emph{session types}) are ranged over by $\stS, \stT, \stU,
\ldots$, and describe the behaviour of a single role.
An internal choice (selection)
$\stIntSum{\roleP}{i \in I}{\stChoice{\stLab[i]}{\tyGround[i]} \stSeq \stT[i]}$
(resp.~an external choice (branching)
$\stExtSum{\roleP}{i \in I}{\stChoice{\stLab[i]}{\tyGround[i]} \stSeq \stT[i]}$%
)
indicates that the \emph{current} role is to \emph{send} to (resp.~%
\emph{receive} from) the role $\roleP$.
Similarly to global types, we require pairwise-distinct,
non-empty labels in local types.
Moreover, we require that the $\stCrashLab$ label not appear in \emph{internal}
choices, reflecting that a $\stCrashLab$ `pseudo'-message can never be
sent; and that singleton $\stCrashLab$ labels are not permitted in external choices.
The type $\stEnd$ indicates a \emph{successful} termination (omitted where
unambiguous), and recursive types follow a similar fashion to global types.
We use a new \emph{runtime} type $\stStop$ to denote crashes.

\paragraph*{Typographical Conventions}
Whenever a payload type $\tyGround$ is insignificant, we choose to omit it from
the global or local type.
This is frequently the case when we discuss the branches with $\gtCrashLab$ (in
a global type) or $\stCrashLab$ (in a local types) labels,
where the payload type is irrelevant.

\paragraph*{Projection}  \emph{Projection} gives the local type of a participating role in a global
type,
defined as a \emph{partial} function that takes a global type
$\gtG$ and a
role $\roleP$ to project on, and returns a local type, given by~\autoref{def:global-proj}.

\begin{defi}[Global Type Projection]%
  \label{def:global-proj}%
  \label{def:local-type-merge}%
  \label{def:removing-crash-label}%
  The \emph{projection of a global type $\gtG$ onto a role $\roleP$}, %
  with respect to a set of \emph{reliable} roles $\rolesR$,
  written \;$\gtProj[\rolesR]{\gtG}{\roleP}$,\; %
  is:

  \smallskip%
  \centerline{\(%
  \small%
  \begin{array}{c}
    \gtProj[\rolesR]{\left(%
      \gtCommSmall{\roleQ}{\roleRMaybeCrashed}
      {i \in I}{\gtLab[i]}{\tyGround[i]}{\gtG[i]}%
      \right)}{\roleP}%
    =\!%
    \left\{%
    \begin{array}{@{}l@{\hskip 3mm}l@{}}
      \stIntSum{\roleR}{\highlight{i \in \setcomp{j \in I}{\stFmt{\stLab[j]} \neq \stCrashLab}}}{ %
        \stChoice{\stLab[i]}{\tyGround[i]} \stSeq (\gtProj[\rolesR]{\gtG[i]}{\roleP})%
      }%
      \hspace{-1.5cm}
      &\text{\footnotesize%
        if\, $\roleP = \roleQ$%
      }%
      \\[3mm]%
      \stExtSum{\roleQ}{i \in I}{%
        \stChoice{\stLab[i]}{\tyGround[i]} \stSeq (\gtProj[\rolesR]{\gtG[i]}{\roleP})%
      }%
      &
      \begin{array}{@{}r@{}}
        \text{\footnotesize if \,} \roleP = \roleR,\text{\footnotesize
        \,and \,}
       \highlight{\roleQ \notin \rolesR \text{\footnotesize \, implies}}
        \\
        \highlight{\exists k \in I : \gtLab[k] = \gtCrashLab}
      \end{array}
      \\[3mm]%
      \stMerge{i \in I}{\gtProj[\rolesR]{\gtG[i]}{\roleP}}%
      &
      \text{\footnotesize%
        if \,} \roleP \neq \roleQ,\text{\footnotesize \,and \,} \roleP \neq \roleR%
    \end{array}
    \right.
    \\[12mm]%
     \gtProj[\rolesR]{\left(%
      \gtCommTransit{\roleQMaybeCrashed}{\roleR}
      {i \in I}{\gtLab[i]}{\tyGround[i]}{\gtG[i]}{j}%
      \right)}{\roleP}%
    =\!%
    \left\{%
    \begin{array}{@{}l@{\hskip 12mm}l@{}}
      \gtProj[\rolesR]{\gtG[j]}{\roleP}%
      &\text{\footnotesize%
        if\, $\roleP = \roleQ$%
      }%
      \\[3mm]%
      \stExtSum{\roleQ}{i \in I}{%
        \stChoice{\stLab[i]}{\tyGround[i]} \stSeq (\gtProj[\rolesR]{\gtG[i]}{\roleP})%
      }%
      &
      \begin{array}{@{}r@{}}
        \text{\footnotesize if \,} \roleP = \roleR,\text{\footnotesize
        \,and \,}
       \highlight{\roleQ \notin \rolesR \text{\footnotesize \, implies}}
        \\
       \highlight{\exists k \in I : \gtLab[k] = \gtCrashLab}
      \end{array}
      \\[3mm]%
     \stMerge{i \in I}{\gtProj[\rolesR]{\gtG[i]}{\roleP}}%
      &
      \text{\footnotesize%
        if \,} \roleP \neq \roleQ, \text{\footnotesize \,and \,} \roleP \neq \roleR%
    \end{array}
    \right.
    \\[12mm]%
    \gtProj[\rolesR]{(\gtRec{\gtRecVar}{\gtG})}{\roleP}%
    \;=\;%
    \left\{%
    \begin{array}{@{\hskip 0.5mm}l@{\hskip 5mm}l@{}}
      \stRec{\stRecVar}{(\gtProj[\rolesR]{\gtG}{\roleP})}%
      &%
      \text{\footnotesize%
        if\,
        $\roleP \in \gtG$ \,or\,
        $\fv{\gtRec{\gtRecVar}{\gtG}} \neq \emptyset$%
      }%
      \\%
      \stEnd%
      &%
      \text{\footnotesize%
        otherwise}
    \end{array}
    \right.%
    \quad\qquad%
    \begin{array}{@{}r@{\hskip 1mm}c@{\hskip 1mm}l@{}}
      \gtProj[\rolesR]{\gtRecVar}{\roleP}%
      &=&%
      \stRecVar%
      \\%
      \gtProj[\rolesR]{\gtEnd}{\roleP}%
      &=&%
      \stEnd%
    \end{array}
  \end{array}
  \)}%
  \smallskip%

  \noindent%
  where
  $\stMerge{}{}$ is %
  the \emph{merge operator for session types} %
  (\emph{full merging}):

  \smallskip%
  \centerline{\(%
  \small%
  \begin{array}{c}%
    \textstyle%
    \stExtSum{\roleP}{i \in I}{\stChoice{\stLab[i]}{\tyGround[i]} \stSeq \stSi[i]}%
    \!\stBinMerge\!%
    \stExtSum{\roleP}{\!j \in J}{\stChoice{\stLab[j]}{\tyGround[j]} \stSeq \stTi[j]}%
   \\
    \;=\;%
    \stExtSum{\roleP}{k \in I \cap J}{\stChoice{\stLab[k]}{\tyGround[k]} \stSeq%
      (\stSi[k] \!\stBinMerge\! \stTi[k])%
    }%
    \stExtC%
    \stExtSum{\roleP}{i \in I \setminus J}{\stChoice{\stLab[i]}{\tyGround[i]} \stSeq \stSi[i]}%
    \stExtC%
    \stExtSum{\roleP}{\!j \in J \setminus I}{\stChoice{\stLab[j]}{\tyGround[j]} \stSeq \stTi[j]}%
    \\[3mm]%
    \stIntSum{\roleP}{i \in I}{\stChoice{\stLab[i]}{\tyGround[i]} \stSeq \stSi[i]}%
    \,\stBinMerge\,%
    \stIntSum{\roleP}{i \in I}{\stChoice{\stLab[i]}{\tyGround[i]} \stSeq \stTi[i]}%
    \;=\;%
    \stIntSum{\roleP}{i \in I}{\stChoice{\stLab[i]}{\tyGround[i]} \stSeq (\stSi[i]
    \stBinMerge \stTi[i])}%
    \\[1mm]%
    \stRec{\stRecVar}{\stS} \,\stBinMerge\, \stRec{\stRecVar}{\stT}%
    \,=\,%
    \stRec{\stRecVar}{(\stS \stBinMerge \stT)}%
    \qquad%
    \stRecVar \,\stBinMerge\, \stRecVar%
    \,=\,%
    \stRecVar%
    \qquad%
    \stEnd \,\stBinMerge\, \stEnd%
    \,=\,%
    \stEnd%
  \end{array}
  \)}%
\end{defi}

We parameterise our theory on a (fixed) set of \emph{reliable} roles
$\rolesR$, \ie roles assumed to \emph{never crash}, for modelling optional reliability assumptions:
if $\rolesR = \emptyset$, we assume every role to be unreliable and
susceptible to crash;
if $\gtRoles{\gtG} \subseteq \rolesR$, we assume every role in $\gtG$ to be
reliable, and we simulate the results from the original MPST theory without
crashes\footnote{%
  Here we consider the original MPST theory without delegation (session
  passing).
}.
We base our definition of projection on the standard
definition~\cite{POPL19LessIsMore}, but include more ($\highlight{\text{highlighted}}$) cases to account for reliable
roles, $\stCrashLab$ branches, and runtime global types.

When projecting a transmission from $\roleQ$ to $\roleR$, we remove the
$\stCrashLab$ label from the internal choice at $\roleQ$, reflecting our
model that a $\stCrashLab$ `pseudo'-message cannot be sent.
Dually, we require a $\stCrashLab$ label to be present in the external choice
at $\roleR$ -- unless the sender role $\roleQ$ is assumed to be reliable.
Our definition of projection enforces that transmissions, whenever
an unreliable role is the sender ($\roleQ \notin \rolesR$),
\emph{must include} a crash handling branch
($\exists k \in I: \gtLab[k] = \gtCrashLab$).
This requirement ensures that the receiving role $\roleR$ can \emph{always} handle
crashes whenever it happens, so that processes are not stuck when crashes occur.
We explain how these requirements help us achieve various properties by
projection, such as safety, deadlock-freedom, and liveness,  in Section~\ref{sec:gtype:pbp}.
The rest of the rules, including the definition of the merge operator, are
taken from the literature~\cite{POPL19LessIsMore, VanGlabbeekLICS2021},
without much modification.

\paragraph*{Subtyping}  We define a \emph{subtyping} relation $\stSub$ on local types in~\autoref{def:subtyping},
which will be used later to %
relate the semantics of global types and configurations in Section~\ref{sec:gtype:relating}.
Our subtyping relation is mostly standard~\cite[Def.\@ 2.5]{POPL19LessIsMore},
except for the ($\highlight{\text{highlighted}}$) addition of the rule $\inferrule{\iruleStSubStop}$ and
extra requirements in \inferrule{\iruleStSubIn}.
In \inferrule{\iruleStSubIn}, we add two additional requirements:
\emph{(1)} the supertype cannot be a `pure' crash handling branch;
and \emph{(2)} if the subtype has a crash handling branch, then the supertype
must also have one.
For simplicity, we do not consider subtyping on basic types $\tyGround$.

\begin{defi}[Subtyping]
\label{def:subtyping}
The subtyping relation $\stSub$ is coinductively defined:

\smallskip
\centerline{\(
\begin{array}{c}
\cinference[\iruleStSubEnd]{}{
  \stEnd \stSub \stEnd
}
\quad

\cinference[\iruleStSubIn]{
  \forall i \in I
  &
  \stT[i] \stSub \stTi[i]
  &
  \highlight{\setcomp{\stLab[k]}{k \in I} \neq \setenum{\stCrashLab}}
  &
 \highlight{\nexists j \in J: \stLab[j] = \stCrashLab}
}{
  \stExtSum{\roleP}{i \in I \cup J}{\stChoice{\stLab[i]}{\tyGround[i]} \stSeq \stT[i]}%
  \stSub
  \stExtSum{\roleP}{i \in I}{\stChoice{\stLab[i]}{\tyGround[i]} \stSeq \stTi[i]}%
}

\\[2ex]
\highlight{
\cinference[\iruleStSubStop]{}{
  \stStop \stSub \stStop
}}
\quad

\cinference[\iruleStSubOut]{
  \forall i \in I
  &
  \stT[i] \stSub \stTi[i]
}{
  \stIntSum{\roleP}{i \in I}{\stChoice{\stLab[i]}{\tyGround[i]} \stSeq \stT[i]}
  \stSub
  \stIntSum{\roleP}{i \in I \cup J}{\stChoice{\stLab[i]}{\tyGround[i]} \stSeq \stTi[i]}
}

\\[2ex]
\cinference[\iruleStSubRecL]{
  \stT{}\subst{\stRecVar}{\stRec{\stRecVar}{\stT}} \stSub \stTi
}
{
  \stRec{\stRecVar}{\stT} \stSub \stTi
}
\quad

\cinference[\iruleStSubRecR]{
  \stT \stSub \stTi{}\subst{\stRecVar}{\stRec{\stRecVar}{\stTi}}
}{
  \stT \stSub \stRec{\stRecVar}{\stTi}
}
\end{array}
\)}
\end{defi}

As standard, our subtyping relation is reflexive and transitive.
Additionally,  the properties of subtyping related to merges are
demonstrated in Lemmas~\ref{lem:merge-subtyping},~\ref{lem:merge-upper-bound},
and~\ref{lem:subtype:merge-subty}.  The proofs are available in Appendix~\ref{sec:app:subtyping}.

\begin{restatable}[Reflexivity and Transitivity of Subtyping]{lem}{propSubtyping}%
\label{prop:subtyping-ref-tran}
The subtyping relation $\stSub$ is reflexive and transitive.
\end{restatable}

\begin{restatable}{lem}{mergeSubtyping}%
\label{lem:merge-subtyping}
  Given a collection of mergable local types $\stT[i]$ ($i \in I$).
  For all $j \in I$, $\stMerge{i \in I}{\stT[i]} \stSub \stT[j]$ holds.
\end{restatable}

\begin{restatable}{lem}{mergeUpperSubtyping}%
\label{lem:merge-upper-bound}
  Given a collection of mergable local types $\stT[i]$ ($i \in I$).
  If for all $i \in I$, $\stS \stSub \stT[i]$ for some local type $\stS$,
  then $\stS \stSub \stMerge{i \in I}{\stT[i]}$.
\end{restatable}

\begin{restatable}{lem}{submergeSubtyping}%
\label{lem:subtype:merge-subty}
  Given two collections of mergable local types $\stS[i],  \stT[i]$ ($i \in I$).
  If for all $i \in I$, $\stS[i] \stSub \stT[i]$, then
  $\stMerge{i \in I} {\stS[i]} \stSub \stMerge{i \in I}{\stT[i]}$.
\end{restatable}

\subsection{Crash-Stop Semantics of Global Types}
\label{sec:gtype:lts-gt}
We now give a Labelled Transition System (LTS) semantics to global types,
with crash-stop semantics.
To this end, we first introduce some
auxiliary definitions.
We define the transition labels in~\autoref{def:mpst-env-reduction-label},
which are also used in the LTS semantics of configurations
(later in Section~\ref{sec:gtype:lts-context}).

\begin{defi}[Transition Labels]
  \label{def:mpst-env-reduction-label}%
  \label{def:mpst-label-subject}%
  Let $\stEnvAnnotGenericSym$ %
  be a transition label of the form:

\smallskip
\centerline{\(
  \begin{array}{rclllll}
  \stEnvAnnotGenericSym &\bnfdef&
    \stEnvInAnnotSmallLab{\roleP}{\roleQ}{\stChoice{\stLab}{\tyGround}}&
    (\text{$\roleP$ receives $\stChoice{\stLab}{\tyGround}$ from $\roleQ$})
    &\bnfsep&\stEnvOutAnnotSmallLab{\roleP}{\roleQ}{\stChoice{\stLab}{\tyGround}}&
    (\text{$\roleP$ sends $\stChoice{\stLab}{\tyGround}$ to $\roleQ$})
    \\
    &\bnfsep&\ltsCrash{\mpS}{\roleP}&%
    (\text{$\roleP$ crashes}) %
    &\bnfsep&\ltsCrDe{\mpS}{\roleP}{\roleQ}&%
    (\text{$\roleP$ detects the crash of $\roleQ$})
  \end{array}
  \)}

\smallskip
\noindent
The subject of a transition label, written
\;$\ltsSubject{\stEnvAnnotGenericSym}$,\; is defined as:

\smallskip
\centerline{\(
    \ltsSubject{\stEnvInAnnotSmallLab{\roleP}{\roleQ}{\stChoice{\stLab}{\tyGround}}}
    =
    \ltsSubject{\stEnvOutAnnotSmallLab{\roleP}{\roleQ}{\stChoice{\stLab}{\tyGround}}}
    =
    \ltsSubject{\ltsCrash{\mpS}{\roleP}}
    =
    \ltsSubject{\ltsCrDe{\mpS}{\roleP}{\roleQ}}
    =
    \roleP. 
\)}
\end{defi}

The labels
$\stEnvOutAnnotSmallLab{\roleP}{\roleQ}{\stChoice{\stLab}{\tyGround}}$
and
$\stEnvInAnnotSmallLab{\roleP}{\roleQ}{\stChoice{\stLab}{\tyGround}}$
describe sending and receiving actions respectively.
The crash of a role $\roleP$ is denoted by the %
label
$\ltsCrashSmall{\mpS}{\roleP}$, and
the detection of a crash by %
label $\ltsCrDe{\mpS}{\roleP}{\roleQ}$:
we model crash detection at \emph{reception}, the label
contains a \emph{detecting} role $\roleP$ and a \emph{crashed} role $\roleQ$.
We use these labels when giving the semantics for global types and
configurations.

We then define an operator to \emph{remove} a role from a
global type in~\autoref{def:gtype:remove-role}: the intuition is to remove any interaction of a
crashed role from the given global type.
When a role has crashed, we remove it by attaching a \emph{crashed annotation},
and removing infeasible actions, \eg when the sender and receiver of a
transmission have both crashed.
The removal operator is a partial function that takes a global type $\gtG$ and a
live role $\roleR$ ($\roleR \in \gtRoles{\gtG}$) and gives a global type
$\gtCrashRole{\gtG}{\roleR}$.

\begin{defi}[Role Removal]%
  \label{def:gtype:remove-role}
  The removal of a live role $\roleP$ in a global type
  $\gtG$, written \;$\gtCrashRole{\gtG}{\roleP}$,\; is defined as follows:
\newlength{\gtRemovalLHSWidth}
\settowidth{\gtRemovalLHSWidth}{$
\gtCommTransit{\rolePCrashed}{\roleQ}{i \in I}{\gtLab[i]}{\tyGround[i]}{
(\gtCrashRole{\gtG[i]}{\roleR})}{j}
$}

\smallskip
  \centerline{\(
\small
  \begin{array}{@{}r@{~}c@{~}l@{}}
    \gtCrashRole{
      (\gtCommSmall{\roleP}{\roleQ}{i \in I}{\gtLab[i]}{\tyGround[i]}{\gtG[i]})
    }{
      \roleR
    }
    &=&
    \left\{
      \begin{array}{@{}ll@{}}
        \gtCommTransit{\rolePCrashed}{\roleQ}{i \in I}{\gtLab[i]}{\tyGround[i]}{
          (\gtCrashRole{\gtG[i]}{\roleR})}{j}
          &
        \text{if~} \roleP = \roleR \text{~and~} \exists j \in I: \gtLab[j] =
        \gtCrashLab \\
        \gtCommSmall{\roleP}{\roleQCrashed}{i \in I}{\gtLab[i]}{\tyGround[i]}{
          (\gtCrashRole{\gtG[i]}{\roleR})
        } &
        \text{if~} \roleQ = \roleR \\
        \makebox[\gtRemovalLHSWidth][l]{%
          \gtCommSmall{\roleP}{\roleQ}{i \in I}{\gtLab[i]}{\tyGround[i]}{
            (\gtCrashRole{\gtG[i]}{\roleR})
          }
        }
        &
        \text{if~} \roleP \neq \roleR \text{~and~} \roleQ \neq \roleR
      \end{array}
    \right.
    \\[8mm]
    \gtCrashRole{
      (\gtCommTransit{\roleP}{\roleQ}{i \in
      I}{\gtLab[i]}{\tyGround[i]}{\gtG[i]}{j})
    }{
      \roleR
    }
    &=&
    \left\{
      \begin{array}{@{}ll@{}}
        \gtCommTransit{\rolePCrashed}{\roleQ}{i \in I}{\gtLab[i]}{\tyGround[i]}{
          (\gtCrashRole{\gtG[i]}{\roleR})}{j}
          &
        \text{if~} \roleP = \roleR
        \\
        \makebox[\gtRemovalLHSWidth][l]{\gtCrashRole{\gtG[j]}{\roleR}}
          &
        \text{if~} \roleQ = \roleR
        \\
        \gtCommTransit{\roleP}{\roleQ}{i \in I}{\gtLab[i]}{\tyGround[i]}{
          (\gtCrashRole{\gtG[i]}{\roleR})}{j}
          &
        \text{if~} \roleP \neq \roleR \text{~and~} \roleQ \neq \roleR
      \end{array}
    \right.
    \\[8mm]
    \gtCrashRole{
      (\gtCommSmall{\roleP}{\roleQCrashed}{i \in I}{\gtLab[i]}{\tyGround[i]}{\gtG[i]})
    }{
      \roleR
    }
    &=&
    \left\{
      \begin{array}{@{}ll@{}}
        \makebox[\gtRemovalLHSWidth][l]{\gtCrashRole{\gtG[j]}{\roleR}}
        &
        \text{if~} \roleP = \roleR \text{~and~} \exists j \in I: \gtLab[j] = \gtCrashLab \\
        \gtCommSmall{\roleP}{\roleQCrashed}{i \in I}{\gtLab[i]}{\tyGround[i]}{
          (\gtCrashRole{\gtG[i]}{\roleR})
        } &
        \text{if~} \roleP \neq \roleR \text{~and~} \roleQ \neq \roleR
      \end{array}
    \right.
    \\[6mm]
    \gtCrashRole{
      (\gtCommTransit{\rolePCrashed}{\roleQ}{i \in
      I}{\gtLab[i]}{\tyGround[i]}{\gtG[i]}{j})
    }{
      \roleR
    }
    &=&
    \left\{
      \begin{array}{@{}ll@{}}
        \makebox[\gtRemovalLHSWidth][l]{\gtCrashRole{\gtG[j]}{\roleR}} &
        \text{if~} \roleQ = \roleR
        \\
        \gtCommTransit{\rolePCrashed}{\roleQ}{i \in I}{\gtLab[i]}{\tyGround[i]}{
          (\gtCrashRole{\gtG[i]}{\roleR})}{j}
          &
        \text{if~} \roleP \neq \roleR \text{~and~} \roleQ \neq \roleR
      \end{array}
    \right.
    \\[6mm]
    \gtCrashRole{
      (\gtRec{\gtRecVar}{\gtG})
    }{
      \roleR
    }
    &=&
    \left\{
      \begin{array}{@{}ll@{}}
        \makebox[\gtRemovalLHSWidth][l]{\gtRec{\gtRecVar}{(\gtCrashRole{\gtG}{\roleR})}} & \text{if~}
        \fv{\gtRec{\gtRecVar}{\gtG}} \neq \emptyset \text{~or~}
        \gtRoles{\gtCrashRole{\gtG}{\roleR}}
        \neq \emptyset \\
        \gtEnd &
        \text{otherwise}
      \end{array}
    \right.
    \\[6mm]
    \gtCrashRole{
      \gtRecVar
    }{
      \roleR
    }
    &=&
    \gtRecVar
    \hspace{4cm}
    \gtCrashRole{\gtEnd}{\roleR} = \gtEnd
  \end{array}
\)}
\end{defi}

We now explain the definition in detail.
For simple cases, the removal of a role $\gtCrashRole{\gtG}{\roleR}$ attaches
crash annotations $\roleCrashedSym$ on all occurrences of the removed role
$\roleR$ throughout global type $\gtG$ inductively.

We draw attention to some interesting cases:
when we remove the sender role $\roleP$ from a transmission prefix
$\gtFmt{\roleP \!\to\! \roleQ}$,
the result is a `pseudo'-transmission en route prefix
$\gtFmt{\rolePCrashed \!\rightsquigarrow\! \roleQ : j}$
where $\gtLab[j] = \gtCrashLab$.
This enables the receiver $\roleQ$ to `receive' the special $\gtCrashLab$
after the crash of $\roleP$, hence triggering the crash handling branch.
Recall that our definition of projection requires that a crash handling branch
be present whenever a crash may occur ($\roleQ \notin \rolesR$).

When we remove the sender role $\roleP$ from a transmission en route prefix
$\gtFmt{\roleP \!\rightsquigarrow\! \roleQ : j}$,
the result \emph{retains} the index $j$ that was selected by $\roleP$,
instead of the index associated with $\gtCrashLab$ handling.
This is crucial to our crash modelling: when a role crashes, the messages that
the role \emph{has sent} to other roles are still available.
We discuss alternative models later in Section~\ref{sec:gtype:alternative}.

In other cases, where removing the role $\roleR$ would render a transmission
(regardless of being en route or not) meaningless, \eg both sender and receiver
have crashed, we simply remove the prefix entirely.

\begin{exa}
\label{ex:role_removal}
We remove role $\roleFmt{C}$ in the global type $\gtG$ in Equation~\eqref{ex:overview-global-with-crash} (defined in Section~\ref{sec:overview}).

\smallskip
\centerline{\(
\gtCrashRole{\gtG}{\roleFmt{C}} =
    \gtCommSingle{\roleFmt{L}}{\roleFmt{I}}
    {\gtMsgFmt{trigger}}{}{\gtCommTransitSingle{\roleFmt{C^\lightning}}{\roleFmt{I}}
     {\gtCrashLab}{}
     {\gtCommSingle{\roleFmt{I}}{\roleFmt{L}}{\gtMsgFmt{fatal}}{}{
    \gtEnd
    }
}}
\)}

\smallskip
\noindent
Role $\roleFmt{C}$ now carries a crash annotation $\roleCrashedSym$ in the
 resulting global type, denoting it has crashed.
 Crash annotations change the reductions available for global types.
  \end{exa}

We now give a Labelled Transition System (LTS) semantics to a
global type $\gtG$,
by defining the semantics with a tuple $\gtWithCrashedRoles{\rolesC}{\gtG}$,
where
$\rolesC$ is a set of \emph{crashed} roles,
with all roles in $\rolesC$ belonging to the fixed role set  $\roleSet$, \ie $\rolesC \subseteq \roleSet$.
The transition system is parameterised by reliability assumptions, in the form
of a fixed set of reliable roles $\rolesR$.
Where it is not ambiguous, we write
$\gtG$ as an abbreviation of $\gtWithCrashedRoles{\rolesEmpty}{\gtG}$.
We define the reduction rules of global types in~\autoref{def:gtype:lts-gt}.

\begin{figure}[t]
 \noindent
  \centerline{\(
\begin{array}{c}
 \highlight{\inference[\iruleGtMoveCrash]{
    \roleP \notin \rolesR
      &
    \roleP \in \gtRoles{\gtG}
    &
    \gtG \neq \gtRec{\gtRecVar}{\gtGi}
  }{
    \gtWithCrashedRoles{\rolesC}{\gtG}
    \gtMove[\ltsCrashSmall{\mpS}{\roleP}]{\rolesR}
    \gtWithCrashedRoles{\rolesC \cup \setenum{\roleP}}{\gtCrashRole{\gtG}{\roleP}}
  }}
  \qquad

  \inference[\iruleGtMoveRec]{
    \gtWithCrashedRoles{\rolesC}{\gtG{}\subst{\gtRecVar}{\gtRec{\gtRecVar}{\gtG}}}
    \gtMove[\stEnvAnnotGenericSym]{\rolesR}
    \gtWithCrashedRoles{\rolesCi}{\gtGi}}
  {
    \gtWithCrashedRoles{\rolesC}{\gtRec{\gtRecVar}{\gtG}}
    \gtMove[\stEnvAnnotGenericSym]{\rolesR}
    \gtWithCrashedRoles{\rolesCi}{\gtGi}
  }
  \\[2ex]

  \inference[\iruleGtMoveOut]{
    j \in I
    &
    \highlight{\gtLab[j] \neq \gtCrashLab}
  }{
    \gtWithCrashedRoles{\rolesC}{
      \gtCommSmall{\roleP}{\roleQ}{i \in I}{\gtLab[i]}{\tyGround[i]}{\gtGi[i]}
    }
    \gtMove[
      \stEnvOutAnnotSmall{\roleP}{\roleQ}{\stChoice{\gtLab[j]}{\tyGround[j]}}
    ]{
      \rolesR
    }
    \gtWithCrashedRoles{\rolesC}{
      \gtCommTransit{\roleP}{\roleQ}{i \in I}{\gtLab[i]}{\tyGround[i]}{\gtGi[i]}{j}
    }
  }\\[2ex]

  \inference[\iruleGtMoveIn]{
    j \in I
    &
    \highlight{\gtLab[j] \neq \gtCrashLab}
  }{
    \gtWithCrashedRoles{\rolesC}{
      \gtCommTransit{\rolePMaybeCrashed}{\roleQ}{i \in I}{\gtLab[i]}{\tyGround[i]}{\gtGi[i]}{j}
    }
    \gtMove[
      \stEnvInAnnotSmall{\roleQ}{\roleP}{\stChoice{\gtLab[j]}{\tyGround[j]}}
    ]{
      \rolesR
    }
    \gtWithCrashedRoles{\rolesC}{\gtGi[j]}
  }\\[2ex]

\highlight{\inference[\iruleGtMoveCrDe]{
    j \in I
    &
    \gtLab[j] = \gtCrashLab
  }{
    \gtWithCrashedRoles{\rolesC}{
      \gtCommTransit{\rolePCrashed}{\roleQ}{i \in I}{\gtLab[i]}{\tyGround[i]}{\gtGi[i]}{j}
    }
    \gtMove[\ltsCrDe{\mpS}{\roleQ}{\roleP}]{\rolesR}
    \gtWithCrashedRoles{\rolesC}{\gtGi[j]}
  }}\\[2ex]

  \highlight{\inference[\iruleGtMoveOrph]{
    j \in I
    &
    \gtLab[j] \neq \gtCrashLab
  }{
    \gtWithCrashedRoles{\rolesC}{\gtCommSmall{\roleP}{\roleQCrashed}{i \in
    I}{\gtLab[i]}{\tyGround[i]}{\gtGi[i]}}
    \gtMove[\stEnvOutAnnotSmall{\roleP}{\roleQ}{\stChoice{\gtLab[j]}{\tyGround[j]}}]{
      \rolesR
    }
    \gtWithCrashedRoles{\rolesC}{\gtGi[j]}
  }}\\[2ex]

  \inference[\iruleGtMoveCtx]{
    \forall i \in I :
    \gtWithCrashedRoles{\rolesC}{\gtGi[i]}
    \gtMove[\stEnvAnnotGenericSym]{\rolesR}
    \gtWithCrashedRoles{\rolesCi}{\gtGii[i]}
    &
    \ltsSubject{\stEnvAnnotGenericSym} \notin \setenum{\roleP, \roleQ}
  }{
    \gtWithCrashedRoles{\rolesC}{
      \gtCommSmall{\roleP}{\roleQMaybeCrashed}{i \in
      I}{\gtLab[i]}{\tyGround[i]}{\gtGi[i]}
    }
    \gtMove[\stEnvAnnotGenericSym]{\rolesR}
    \gtWithCrashedRoles{\rolesCi}{
      \gtCommSmall{\roleP}{\roleQMaybeCrashed}{i \in
      I}{\gtLab[i]}{\tyGround[i]}{\gtGii[i]}
    }
  }\\[2ex]

  \inference[\iruleGtMoveCtxi]{
    \forall i \in I :
    \gtWithCrashedRoles{\rolesC}{\gtGi[i]}
    \gtMove[\stEnvAnnotGenericSym]{\rolesR}
    \gtWithCrashedRoles{\rolesCi}{\gtGii[i]}
    &
    \ltsSubject{\stEnvAnnotGenericSym} \neq \roleQ
  }{
    \gtWithCrashedRoles{\rolesC}{
      \gtCommTransit{\rolePMaybeCrashed}{\roleQ}{i \in
      I}{\gtLab[i]}{\tyGround[i]}{\gtGi[i]}{j}
    }
    \gtMove[\stEnvAnnotGenericSym]{\rolesR}
    \gtWithCrashedRoles{\rolesCi}{
      \gtCommTransit{\rolePMaybeCrashed}{\roleQ}{i \in
      I}{\gtLab[i]}{\tyGround[i]}{\gtGii[i]}{j}
    }
  }
\end{array}
\)}
\caption{Global type reduction rules.}
\label{fig:gtype:red-rules}
\end{figure}

\begin{defi}[Global Type Reductions]
  \label{def:gtype:lts-gt}
  The global type (annotated with a set of crashed roles $\rolesC$)
  transition relation
  $\gtMove[\stEnvAnnotGenericSym]{\rolesR}$
  is inductively
  defined by the rules in~\Cref{fig:gtype:red-rules},
  parameterised by a fixed set $\rolesR$ of reliable roles.
  We write
  $\gtWithCrashedRoles{\rolesC}{\gtG} \gtMove{\rolesR}
  \gtWithCrashedRoles{\rolesCi}{\gtGi}$
  if there
  exists $\stEnvAnnotGenericSym$ such that
  $\gtWithCrashedRoles{\rolesC}{\gtG}
  \gtMove[\stEnvAnnotGenericSym]{\rolesR}
  \gtWithCrashedRoles{\rolesCi}{\gtGi}$;
  we write
  $\gtWithCrashedRoles{\rolesC}{\gtG} \gtMove{\rolesR}$
  if there
  exists $\rolesCi$, $\gtGi$, and $\stEnvAnnotGenericSym$ such that
  $\gtWithCrashedRoles{\rolesC}{\gtG}
  \gtMove[\stEnvAnnotGenericSym]{\rolesR}
  \gtWithCrashedRoles{\rolesCi}{\gtGi}$,
  and $\gtMoveStar[\rolesR]$ for the transitive and reflexive closure of
  $\gtMove{\rolesR}$.
\end{defi}

Rules \inferrule{\iruleGtMoveOut} and \inferrule{\iruleGtMoveIn} model sending
and receiving messages respectively, as are standard in existing
works~\cite{ICALP13CFSM}.
We add a ($\highlight{\text{highlighted}}$) extra condition that the message exchanged not be a
`pseudo'-message carrying the $\gtCrashLab$ label.
$\inferrule{\iruleGtMoveRec}$ is a standard rule that deals with recursion.

We introduce ($\highlight{\text{highlighted}}$) rules to account for crash and consequential behaviour.
\begin{itemize}[leftmargin=*, nosep]
\item
Rule $\inferrule{\iruleGtMoveCrash}$ models crashes, where a live ($\roleP \in
\gtRoles{\gtG}$), but unreliable ($\roleP \notin \rolesR$) role $\roleP$ may crash.
The crashed role $\roleP$ is added into the set of crashed roles ($\rolesC \cup
\setenum{\roleP}$), and removed
from the global type, resulting in a global type $\gtCrashRole{\gtG}{\roleP}$.
\item
Rule $\inferrule{\iruleGtMoveCrDe}$ is for \emph{crash detection}, where a live
role $\roleQ$ may detect that $\roleP$ has crashed at reception,
and then continues with the crash handling continuation labelled $\gtCrashLab$.
This rule only applies when the message en route is a `pseudo'-message, since
otherwise a message rests in the queue %
of the receiver and can be received
despite the crash of the sender (\cf~\inferrule{\iruleGtMoveIn}).
\item
Rule $\inferrule{\iruleGtMoveOrph}$ models the orphaning of a message sent from a
live role $\roleP$ to a crashed role $\roleQ$.
Similar to the requirement in \inferrule{\iruleGtMoveOut}, we add the side
condition that the message sent is not a `pseudo'-message.
\end{itemize}

Finally, rules $\inferrule{\iruleGtMoveCtx}$ and $\inferrule{\iruleGtMoveCtxi}$
allow non-interfering reductions of (intermediate) global types
under prefix, provided that all of the continuations can be reduced by
that label.

\begin{rem}[Necessity of $\rolesC$ in Semantics]
  While we can obtain the set of crashed roles in any global type $\gtG$ via
  $\gtRolesCrashed{\gtG}$, we need a separate $\rolesC$ for bookkeeping
  purposes.

  Let
  $\gtG = \gtCommSingleErr{\roleP}{\roleQ}{\gtLab}{}{\gtEnd}{\gtEnd}$,
  we can have the following reductions:

\smallskip
 \centerline{\(
    \gtWithCrashedRoles{\rolesEmpty}{\gtG}
    \gtMove[\ltsCrash{\mpS}{\roleQ}]{\rolesEmpty}
    \gtWithCrashedRoles{\setenum{\roleQ}}{
      \gtCommSingleErr{\roleP}{\roleQCrashed}{\gtLab}{}{\gtEnd}{\gtEnd}
    }
    \gtMove[\stEnvOutAnnot{\roleP}{\roleQ}{\gtLab}{}]{\rolesEmpty}
    \gtWithCrashedRoles{\setenum{\roleQ}}{\gtEnd}
\)}

\smallskip
\noindent
  While we can deduce $\roleQ$ is a crashed role in the interim global type,
  the same information cannot be recovered from the final global type $\gtEnd$.
\end{rem}

Both live and crashed roles in a global type remain consistent throughout a transition,
except for those directly involved in the crash transition action, as demonstrated in~\autoref{lem:no-revival-roles}.

\begin{restatable}[No Revival Or Unexpected Crashes]{lem}{lemNorevival}%
\label{lem:no-revival-roles}
  Assume
  \;$\gtWithCrashedRoles{\rolesC}{\gtG}
  \gtMove[\stEnvAnnotGenericSym]{\rolesR}
  \gtWithCrashedRoles{\rolesCi}{\gtGi}
  $.\;
  \begin{enumerate}
    \item If $\roleP \in \gtRolesCrashed{\gtGi}$ and $\stEnvAnnotGenericSym \neq
      \ltsCrash{}{\roleP}$, then $\roleP \in \gtRolesCrashed{\gtG}$;
      \label{item:crash-roles-remain-crash}
    \item If $\roleP \in \gtRoles{\gtGi}$ and $\stEnvAnnotGenericSym \neq
      \ltsCrash{}{\roleP}$, then $\roleP \in \gtRoles{\gtG}$;
      \label{item:live-roles-remain-live}
    \item If $\roleP \in \gtRolesCrashed{\gtGi}$ and $\stEnvAnnotGenericSym =
      \ltsCrash{}{\roleP}$, then $\roleP \in \gtRoles{\gtG}$.
      \label{item:crash-role-crash}
  \end{enumerate}
\end{restatable}
\begin{proof}
By induction on global type reductions. See Appendix~\ref{sec:proof:semantics:gty} for details.
\end{proof}

We introduce an auxiliary concept of \emph{well-annotated} global
types in~\autoref{def:globaltypes:well-anno}, as a consistency requirement for crash
annotations $\roleCrashedSym$ in a global type $\gtG$, and the set of crashed
roles $\rolesC$, and a fixed set of reliable roles $\rolesR$.
We show that well-annotatedness \wrt \rolesR  is preserved by global type
reductions in~\autoref{lem:well-annotated-preserve}.
It follows that, a global type $\gtG$ without runtime constructs is trivially
well-annotated, and all reducta
 $\gtG \gtMoveStar[\rolesR]{}
\gtWithCrashedRoles{\rolesC}{\gtGi}$ are also well-annotated.
\iftoggle{full}{The proof of \cref{lem:well-annotated-preserve} is available in~\Cref{sec:proof:semantics:gty}.
}{}

\begin{defi}[Well-Annotated Global Types]
\label{def:globaltypes:well-anno}
  A global type $\gtG$ with crashed roles $\rolesC$ is \emph{well-annotated}
  \wrt %
  a (fixed) set of reliable roles $\rolesR$, iff:
  \begin{enumerate}[label=(WA\arabic*), leftmargin=12mm, nosep]
    \item No reliable roles are crashed, \;$\gtRolesCrashed{\gtG} \cap \rolesR =
      \emptyset$;\; and,
      \label{item:wa:reliable-no-crash}
    \item All roles with crash annotations are in the crashed set,
      \;$\gtRolesCrashed{\gtG} \subseteq \rolesC$;\; and,
      \label{item:wa:crash-annot-crash}
    \item A role cannot be live and crashed simultaneously, \;$\gtRoles{\gtG}
      \cap \gtRolesCrashed{\gtG} = \emptyset$.
      \label{item:wa:live-no-crash}
  \end{enumerate}
\end{defi}

\begin{restatable}[Preservation of Well-Annotated Global Types]{lem}{lemWellAnnoPreserve}%
  \label{lem:well-annotated-preserve}
  If
  \;$\gtWithCrashedRoles{\rolesC}{\gtG}
  \gtMove[\stEnvAnnotGenericSym]{\rolesR}
  \gtWithCrashedRoles{\rolesCi}{\gtGi}
  $,\; and
  \;$\gtWithCrashedRoles{\rolesC}{\gtG}$\; is well-annotated \wrt $\rolesR$, then
  \;$\gtWithCrashedRoles{\rolesCi}{\gtGi}$\; is also well-annotated \wrt
  $\rolesR$.
\end{restatable}
\begin{proof}
By induction on global type reductions~(\autoref{def:gtype:lts-gt}).
See Appendix~\ref{sec:proof:semantics:gty} for details.
\qedhere
\end{proof}

\subsection{Crash-Stop Semantics of Configurations}\label{sec:gtype:lts-context}
After giving semantics to global types, %
we now give an
LTS semantics to \emph{configurations},
\ie a collection of local types and communication queues %
across roles.
We first give a definition of configurations
in~\autoref{def:mpst-env}, followed by their reduction rules in~\autoref{def:mpst-env-reduction}.

\begin{defi}[Configurations]%
  \label{def:mpst-env}%
  \label{def:mpst-env-closed}%
  \label{def:mpst-env-comp}%
  \label{def:mpst-env-subtype}%
  A configuration is a tuple \;$\stEnv; \qEnv$,\; where
  $\stEnv$ is a \emph{typing context}, denoting a partial mapping
  from roles to local types,
  defined as:

  \smallskip
  \centerline{\(%
  \stEnv
  \,\coloncolonequals\,
  \stEnvEmpty
  \bnfsep
  \stEnv \stEnvComp \stEnvMap{\roleP}{\stT}
  \)}%

  \smallskip
  \noindent
   The \emph{context composition} $\stEnv[1] \stEnvComp \stEnv[2]$
   is defined iff $\dom{\stEnv[1]} \cap \dom{\stEnv[2]} = \emptyset$.
  A typing context  $\stEnv$ can be \emph{decomposed} (or \emph{split}) into sub-contexts
   $\stEnv[1]$ and $\stEnv[2]$, written $\stEnv = \stEnv[1] \stEnvComp \stEnv[2]$, if $\dom{\stEnv} = \dom{\stEnv[1]} \cup \dom{\stEnv[2]}$,
   and $\forall \roleP \in \dom{\stEnv}$, $\stEnvApp{\stEnv}{\roleP} = \stEnvApp{(\stEnv[1] \stEnvComp \stEnv[2])}{\roleP}$.
   We write \;$\stEnvUpd{\stEnv}{\roleP}{\stT}$\; for typing context updates,
  namely
  \;$\stEnvApp{\stEnvUpd{\stEnv}{\roleP}{\stT}}{\roleP} = \stT$\; and
  \;$\stEnvApp{\stEnvUpd{\stEnv}{\roleP}{\stT}}{\roleQ} =
  \stEnvApp{\stEnv}{\roleQ}$\; (where \;$\roleP \neq \roleQ$).

  A \emph{queue}, denoted $\stQ$, is either a (possibly empty) sequence of
  messages
  \;$\stM[1] \stFmt{\cdot} \stM[2] \stFmt{\cdot} \cdots \stFmt{\cdot}
  \stM[n]$,\;
  or an unavailable queue $\stQUnavail$.
  We write $\stQEmpty$ for an empty queue, and
  \;$\stQCons{\stM}{\stQi}$\;
  for a non-empty queue with message $\stM$ at the beginning.
  A \emph{queue message} $\stM$ is of form
  \;$\stQMsg{\stLab}{\tyGround}$,\;
  denoting a message with label $\stLab$ and payload $\tyGround$.
  We sometimes omit $\tyGround$ when the payload is not of specific
  interest.

  We write $\qEnv$ to denote a \emph{queue environment}, a collection of
  peer-to-peer queues.
  A queue from $\roleP$ to $\roleQ$ at environment $\qEnv$ is denoted
  \;$\qApp{\qEnv}{\roleP}{\roleQ}$.\;
  We define queue environment updates \;$\stEnvUpd{\qEnv}{\roleP,
  \roleQ}{\stQ}$\;
  similarly.

  We also write \;$\stQCons{\stQi}{\stM}$\; for appending a message at the end
  of a queue:
  the message is appended to the sequence when $\stQi$ is available, or
  discarded when $\stQi$ is unavailable
  (\ie \;$\stQCons{\stQUnavail}{\stM} = \stQUnavail$).\;
  Additionally, we write
  \;$\stEnvUpd{\qEnv}{\cdot, \roleQ}{\stQUnavail}$\;
  for
  making all the queues to $\roleQ$ unavailable: \ie
  \;$\stEnvUpd{
    \stEnvUpd{
      \stEnvUpd{\qEnv}{\roleP[1], \roleQ}{\stQUnavail}
    }{\roleP[2], \roleQ}{\stQUnavail} \cdots
  }{\roleP[n], \roleQ}{\stQUnavail}$.

  We write $\qEnv[\stQEmpty]$
  to denote an \emph{empty} queue environment,
  where  \;$\qApp{\qEnv[\stQEmpty]}{\roleP}{\roleQ} = \stQEmpty$\;  %
  for any
  $\roleP$ and $\roleQ$ in the domain.
\end{defi}

We give an LTS semantics of configurations in~\autoref{def:mpst-env-reduction}.
Similar to that of global types, we model the semantics of configurations
in an asynchronous (\aka message-passing) fashion, using a queue environment to
represent the communication queues %
among all roles.

\begin{figure}[t]
  \noindent
  \centerline{\(
  \begin{array}{@{}c@{}}
    \inference[\iruleTCtxOut]{%
      \stEnvApp{\stEnv}{%
        \roleP%
      } =
      \stIntSum{\roleQ}{i \in I}{\stChoice{\stLab[i]}{\tyGround[i]} \stSeq \stT[i]}%
      &
      k \in I%
    }{%
      \stEnv; \qEnv
      \,\stEnvMoveOutAnnot{\roleP}{\roleQ}{\stChoice{\stLab[k]}{\tyGround[k]}}\,%
      \stEnvUpd{\stEnv}{\roleP}{\stT[k]};
      \stEnvUpd{\qEnv}{\roleP, \roleQ}{
        \stQCons{
          \stEnvApp{\qEnv}{\roleP, \roleQ}
        }{
          \stQMsg{\stLab[k]}{\tyGround[k]}
        }
      }%
    }%
    \\[2mm]%
    \inference[\iruleTCtxIn]{%
      \stEnvApp{\stEnv}{%
        \roleP%
      } =
      \stExtSum{\roleQ}{i \in I}{\stChoice{\stLab[i]}{\tyGround[i]} \stSeq \stT[i]}%
      &
      k \in I%
      &
      \stEnvApp{\qEnv}{\roleQ, \roleP}
      =
      \stQCons{\stQMsg{\stLab[k]}{\tyGround[k]}}{\stQi}
      \neq
      \stQUnavail
    }{%
      \stEnv; \qEnv
      \,\stEnvMoveInAnnot{\roleP}{\roleQ}{\stChoice{\stLab[k]}{\tyGround[k]}}\,%
      \stEnvUpd{\stEnv}{\roleP}{\stT[k]};
      \stEnvUpd{\qEnv}{\roleQ, \roleP}{\stQi}
    }%
    \\[2mm]%
    \inference[\iruleTCtxRec]{%
      \stEnvApp{\stEnv}{\roleP} = \stRec{\stRecVar}{\stT}%
      &
      \stEnvUpd{\stEnv}{\roleP}{
        \stT\subst{\stRecVar}{\stRec{\stRecVar}{\stT}}%
      }; \qEnv%
      \stEnvMoveGenAnnot
      \stEnvi; \qEnvi%
    }{%
      \stEnv; \qEnv
      \stEnvMoveGenAnnot
      \stEnvi; \qEnvi%
    }%
    \qquad
    \highlight{%
    \inference[\iruleTCtxCrash]{%
      \stEnvApp{\stEnv}{\roleP} \neq \stEnd
      &
      \stEnvApp{\stEnv}{\roleP} \neq \stStop
    }{
      \stEnv; \qEnv
      \stEnvMoveAnnot{\ltsCrash{\mpS}{\roleP}}
      \stEnvUpd{\stEnv}{\roleP}{\stStop};
      \stEnvUpd{\qEnv}{\cdot, \roleP}{\stQUnavail}
    }
    }%
    \\
    \highlight{%
    \inference[\iruleTCtxCrashDetect]{%
      \stEnvApp{\stEnv}{\roleQ} =
      \stExtSum{\roleP}{i \in I}{\stChoice{\stLab[i]}{\tyGround[i]} \stSeq \stT[i]}%
      &
      \stEnvApp{\stEnv}{\roleP} = \stStop
      &
      k \in I
      &
      \stLab[k] = \stCrashLab
      &
      \stEnvApp{\qEnv}{\roleP, \roleQ} = \stQEmpty
    }{%
      \stEnv; \qEnv
      \,\stEnvMoveAnnot{\ltsCrDe{\mpS}{\roleQ}{\roleP}}\,%
      \stEnvUpd{\stEnv}{\roleQ}{\stT[k]}; \qEnv
    }%
    }%
  \end{array}
  \)}%
 \caption{Configuration semantics.}
  \label{fig:gtype:tc-red-rules}
\end{figure}

\begin{defi}[Configuration Semantics]%
  \label{def:mpst-env-reduction}%
  The \emph{configuration transition relation \;$\stEnvMoveGenAnnot$\;} %
  is defined in \cref{fig:gtype:tc-red-rules}.
  We write \;$\stEnvMoveGenAnnotP{\stEnv; \qEnv}$\; %
  iff \;$\stEnv; \qEnv \!\stEnvMoveGenAnnot\! \stEnvi; \qEnvi$\; for some
  \;$\stEnvi$\;
  and \;$\qEnvi$.\; %
  We define two \emph{reductions} \;$\stEnvMove$\; and
  \;$\stEnvMoveMaybeCrash[\rolesR]$\; (where \;$\rolesR$\; is a fixed set of
  reliable roles)
  as follows.
   \begin{itemize}[left=0pt, topsep=0pt]
   \item We write \;$\stEnv; \qEnv \!\stEnvMove\! \stEnvi; \qEnvi$\;
     for
     \;$\stEnv; \qEnv \stEnvMoveGenAnnot \stEnvi; \qEnvi$\;
     with
     \;\(\stEnvAnnotGenericSym \!\in\! \setenum{\stEnvInAnnotSmall{\roleP}{\roleQ}{\stChoice{\stLab}{\tyGround}},
     \stEnvOutAnnotSmall{\roleP}{\roleQ}{\stChoice{\stLab}{\tyGround}},
     \ltsCrDe{\mpS}{\roleP}{\roleQ}}\).\;
     We write
     \;$\stEnvMoveP{\stEnv; \qEnv}$\;
     iff
     \;$\stEnv; \qEnv \!\stEnvMove\! \stEnvi; \qEnvi$\;
     for some
     \;$\stEnvi; \qEnvi$,\; %
     and
     \;$\stEnvNotMoveP{\stEnv; \qEnv}$\; for its negation,  %
     and \;$\stEnvMoveStar$\; %
     for the reflexive and transitive closure of \;$\stEnvMove$;

   \item We write
    \;$\stEnv; \qEnv \!\stEnvMoveMaybeCrash[\rolesR]\! \stEnvi; \qEnvi$\;
    for
    \;$\stEnv; \qEnv \stEnvMoveGenAnnot \stEnvi; \qEnvi$\;
    with
    \;\(
    \stEnvAnnotGenericSym \!\notin\! \setcomp{ \ltsCrash{\mpS}{\roleR} }{
      \roleR \!\in\!  \rolesR}
    \).\;
    We write \;$\stEnvMoveMaybeCrashP[\rolesR]{\stEnv; \qEnv}$ %
    iff
    \;$\stEnv; \qEnv \!\stEnvMoveMaybeCrash[\rolesR]\! \stEnvi; \qEnvi$\;
    for some
    \;$\stEnvi; \qEnvi$,\;
    and
    \;$\stEnvNotMoveMaybeCrashP[\rolesR]{\stEnv; \qEnvi}$\;
    for its negation.   %
    We define
    \;$\stEnvMoveMaybeCrashStar[\rolesR]$\;
    as the reflexive and transitive closure of \;$\stEnvMoveMaybeCrash[\rolesR]$.

   \end{itemize}
\end{defi}

We first explain the standard rules:
rule $\inferrule{\iruleTCtxOut}$ (resp.~$\inferrule{\iruleTCtxIn}$)
says that a role can perform an output (resp.~input) transition by appending
(resp.~consuming) a message at the corresponding queue.
Recall that whenever a queue is unavailable, the resulting queue remains
unavailable after appending ($\stQCons{\stQUnavail}{\stM} = \stQUnavail$).
Therefore, the rule $\inferrule{\iruleTCtxOut}$ covers delivery to both crashed
and live roles, whereas two separate rules are used in modelling global type
semantics (\inferrule{\iruleGtMoveOut} and \inferrule{\iruleGtMoveOrph}).
We also include a standard rule $\inferrule{\iruleTCtxRec}$ for recursive types.

The key innovations are the ($\highlight{\text{highlighted}}$) rules modelling crashes and crash
detection:
by rule~\inferrule{\iruleTCtxCrash}, a role $\roleP$ may crash
and become $\stStop$ at any time
(unless it is already $\stEnd$ed or $\stStop$ped).
All of $\roleP$'s receiving queues become unavailable $\stQUnavail$, so that
future messages to $\roleP$ would be discarded.
Note that any pending messages in the receiving queue are discarded (since the
queue becomes unavailable $\stQUnavail$),
since the messages cannot be processed after the crash.

Rule \inferrule{\iruleTCtxCrashDetect} models crash detection and handling:
if ${\roleP}$ is crashed and stopped,
another role ${\roleQ}$ attempting to receive from $\roleP$
can then take its $\stCrashLab$ handling branch. %
However, this rule only applies when the corresponding queue is empty: it is
still possible to receive messages sent before crashing via
\inferrule{\iruleTCtxIn}.

\subsection{Alternative Modellings for Crash-Stop Failures}\label{sec:gtype:alternative}
Before we dive into the relation between two semantics, let us have a short
digression to discuss our modelling choices and alternatives.
In this work, we mostly follow the assumptions laid out in
\cite{CONCUR22MPSTCrash},
where a crash is detected at reception. However, they opt to use a synchronous
(rendez-vous) semantics, whereas we give an asynchronous (message-passing)
semantics, which entails interesting scenarios that would not arise in
a synchronous semantics.

Specifically,
consider the case where a role $\roleP$ sends a message to $\roleQ$, and
then $\roleP$ crashes after sending, but before $\roleQ$ receives the message.
The situation does not arise under a synchronous semantics, since sending and
receiving actions are combined into a single transmission action.

Intuitively, there are two possibilities to handle this scenario.
The questions are
whether the message sent immediately before crashing is deliverable to $\roleQ$,
and consequentially,
at what time $\roleQ$ detects the crash of $\roleP$.

In our semantics (\cref{fig:gtype:red-rules,fig:gtype:tc-red-rules}), we
opt to answer the first question positively:
we argue that this model is more consistent with our `passive' crash detection
design.
For example, if a role $\roleP$ never receives from another role $\roleQ$, then
$\roleP$ does not need to react in the event of $\roleQ$'s crash.
Following a similar line of reasoning,
if the message sent by $\roleP$ arrives in the receiving queue of
$\roleQ$,
then $\roleQ$ should be able to receive the message, without triggering
a crash detection (although it may be triggered later).
As a consequence, we require in \inferrule{\iruleTCtxCrashDetect} that the
queue $\stEnvApp{\qEnv}{\roleP, \roleQ}$ be empty, to reflect the idea that
crash detection should be a `last resort'.

For an alternative model,
we can opt to detect the crash after it has occurred.
This is possibly better modelled with using outgoing queues
(\cf\cite{FSTTCS15MPSTExpressiveness}),
instead of incoming queues in the semantics presented.
Practically, this may be the scenario that a TCP connection is closed (or
reset) when a peer has crashed,
and the content in the queue %
is lost.
 It is worth noting that this kind of alternative model
will not affect our main theoretical results: the operational correspondence
between global and local type semantics, and furthermore, global type properties guaranteed by projection.

\subsection{Relating Global Type and Configuration Semantics}%
\label{sec:gtype:relating}
We have given LTS semantics for both global types
(\autoref{def:gtype:lts-gt}) and configurations (\autoref{def:mpst-env-reduction}).
We will now relate these two semantics with the help of the projection operator
$\gtProj[]{}{}$~(\autoref{def:global-proj})
and the subtyping relation $\stSub$~(\autoref{def:subtyping}).

We associate configurations $\stEnv; \qEnv$ with global types $\gtG$ (as
annotated with a set of crashed roles $\rolesC$) by projection, written
$\stEnvAssoc{\gtWithCrashedRoles{\rolesC}{\gtG}}{\stEnv; \qEnv}{\rolesR}$.
Naturally, there are two components of the association: \emph{(1)} the local
types in $\stEnv$ need to correspond to the projections of the global type
$\gtG$ and the set of crashed roles $\rolesC$; and \emph{(2)} the queues in
$\qEnv$ corresponds to the transmissions en route in the global type $\gtG$ and
also the set of crashed roles $\rolesC$.

\begin{defi}[Association of Global Types and Configurations]%
  \label{def:assoc}
  \label{def:assoc-queue}
  A configuration \;$\stEnv; \qEnv$\; is associated to a
  (well-annotated \wrt$\rolesR$) global type
  \;$\gtWithCrashedRoles{\rolesC}{\gtG}$,\;
  written
  \;$\stEnvAssoc{\gtWithCrashedRoles{\rolesC}{\gtG}}{\stEnv; \qEnv}{\rolesR}$,\;
  iff
  \begin{enumerate}[leftmargin=*, nosep]
    \item
      \label{item:assoc:stenv}
      $\stEnv$ can be split into disjoint (possibly empty) sub-contexts
      \;$\stEnv = \stEnv[\gtG] \stEnvComp \stEnv[\stStopSym] \stEnvComp
      \stEnv[\stEnd]$\; where:
      \begin{enumerate}[label={(A\arabic*)}, leftmargin=*, ref={(A\arabic*)},
        nosep]
        \item\label{item:assoc:alive-sub}
          $\stEnv[\gtG]$\; contains projections of \;$\gtG$:
          \;$ \dom{\stEnv[\gtG]}
            =
            \gtRoles{\gtG}
          $,\; and
          \;$ \forall \roleP \in \dom{\stEnv[\gtG]}: \linebreak
            \stEnvApp{\stEnv}{{\roleP}}
            \stSub
            \gtProj[\rolesR]{\gtG}{\roleP}
          $;
        \item\label{item:assoc:crash-stop}
          $\stEnv[\stStopSym]$\; contains crashed roles:
          \;$ \dom{\stEnv[\stStopSym]} =
            \rolesC$,\;
            and
            \;$ \forall \roleP \in \dom{\stEnv[\stStopSym]}:
            \stEnvApp{\stEnv}{{\roleP}} = \stStop$;
        \item\label{item:assoc:end}
          $\stEnv[\stEnd]$\; contains only \;$\stEnd$\; endpoints:
          \;$\forall \roleP \in \stEnv[\stEnd]: \stEnvApp{\stEnv}{\roleP}
          = \stEnd$.
        \end{enumerate}
    \item
        \label{item:assoc:qenv}
        \begin{enumerate}[label={(A\arabic*)}, leftmargin=*, ref={(A\arabic*)},
          resume, nosep]
        \item\label{item:assoc:queue}
          $\qEnv$ is associated with global type
          \;$\gtWithCrashedRoles{\rolesC}{\gtG}$,\; given as follows:
          \begin{enumerate}[leftmargin=0pt, nosep]
            \item
              Receiving queues for a role is unavailable if and only if
              it has crashed:
              \;$\forall \roleQ : \linebreak
              \roleQ \in \rolesC \iff
              \stEnvApp{\qEnv}{\cdot, \roleQ} = \stQUnavail$;
            \item
              If \;$\gtG = \gtEnd$\; or \;$\gtG = \gtRec{\gtRecVar}{\gtGi}$,\;
              then queues between all roles are empty (expect receiving queue
              for crashed roles):
              \;$\forall \roleP, \roleQ: \roleQ \notin \rolesC \implies
              \stEnvApp{\qEnv}{\roleP, \roleQ} = \stQEmpty$;
            \item
              If
              \;$\gtG = \gtComm{\roleP}{\roleQMaybeCrashed}{i \in I}{\gtLab[i]}{\tyGround[i]}{\gtGi[i]}$,\;
              or
              \;$\gtG =
              \gtCommTransit{\rolePMaybeCrashed}{\roleQ}{i \in I}{\gtLab[i]}{\tyGround[i]}{\gtGi[i]}{j}
              $\; with \;$\gtLab[j] = \gtCrashLab$\;
              (\ie a `pseudo'-message is en route),
              then
              \begin{enumerate*}[label=\emph{(\roman*)}]
                \item
                  if $\roleQ$ is live, then the queue from $\roleP$ to $\roleQ$ is
                  empty:
                  \;$\roleQMaybeCrashed \neq \roleQCrashed \implies
                  \stEnvApp{\qEnv}{\roleP, \roleQ} = \stQEmpty$, and
                \item $\forall i \in I: \qEnv$\; is associated with
                  \;$\gtWithCrashedRoles{\rolesC}{\gtGi[i]}$;
              \end{enumerate*}
              and,
            \item
              If
              \;$\gtG =
              \gtCommTransit{\rolePMaybeCrashed}{\roleQ}{i \in I}{\gtLab[i]}{\tyGround[i]}{\gtGi[i]}{j}
              $\; with \;$\gtLab[j] \neq \gtCrashLab$,\;
              then
              \begin{enumerate*}[label=\emph{(\roman*)}]
                \item
                  the queue from $\roleP$ to $\roleQ$ begins with the message
                  \;$\stQMsg{\gtLab[j]}{\tyGround[j]}$:\;
                  $\stEnvApp{\qEnv}{\roleP, \roleQ} =
                  \stQCons{\stQMsg{\gtLab[j]}{\tyGround[j]}}{\stQ}$;
                \item
                  $\forall i \in I:$\; removing the message from the head of the
                  queue,
                  \;$\stEnvUpd{\qEnv}{\roleP, \roleQ}{\stQ}$\; is associated with
                  \;$\gtWithCrashedRoles{\rolesC}{\gtGi[i]}$.
              \end{enumerate*}
        \end{enumerate}
      \end{enumerate}
  \end{enumerate}
  We write \;$\stEnvAssoc{\gtG}{\stEnv}{\rolesR}$\; as an abbreviation of
  \;$\stEnvAssoc{\gtWithCrashedRoles{\emptyset}{\gtG}}{\stEnv; \qEnv[\stQEmpty]}{\rolesR}$. %
  We sometimes say a $\stEnv$ (resp.~$\qEnv$) is associated with
  \;$\gtWithCrashedRoles{\rolesC}{\gtG}$\; for stating \cref{item:assoc:stenv}
  (resp.~\cref{item:assoc:qenv}) is
  satisfied.
\end{defi}

We demonstrate %
the relation between the two semantics via association, by showing two
main theorems:
all possible reductions of a configuration have a corresponding action in
reductions of the associated global type (\autoref{thm:gtype:proj-comp});
and the reducibility of a global type is the same as its associated
configuration (\autoref{thm:gtype:proj-sound}).

\begin{restatable}[Completeness of Association]{thm}{thmProjCompleteness}%
  \label{thm:gtype:proj-comp}
  Given associated global type $\gtG$ and configuration $\stEnv; \qEnv$:
  \;$\stEnvAssoc{\gtWithCrashedRoles{\rolesC}{\gtG}}{\stEnv; \qEnv}{\rolesR}$.\;
  If \;$\stEnv; \qEnv \stEnvMoveGenAnnot \stEnvi; \qEnvi$,\;
  where \;$\stEnvAnnotGenericSym \neq \ltsCrash{\mpS}{\roleP}$\;
  for all \;$\roleP \in \rolesR$,\;
  then there exists \;$\gtWithCrashedRoles{\rolesCi}{\gtGi}$\; such that
  \;$\stEnvAssoc{\gtWithCrashedRoles{\rolesCi}{\gtGi}}{\stEnvi;
  \qEnvi}{\rolesR}$\;
  and
  \;$\gtWithCrashedRoles{\rolesC}{\gtG} \gtMove[\stEnvAnnotGenericSym]{\rolesR}
    \gtWithCrashedRoles{\rolesCi}{\gtGi}$.
\end{restatable}
\begin{proof}
  By induction on configuration reductions (\autoref{def:mpst-env-reduction}).
  See Appendix~\ref{sec:proof:relating} for details.
  \qedhere
\iftoggle{full}{See \cref{sec:proof:relating} for detailed proof.}{}
\end{proof}
\begin{restatable}[Soundness of Association]{thm}{thmProjSoundness}%
  \label{thm:gtype:proj-sound}
  Given associated global type $\gtG$ and configuration $\stEnv; \qEnv$:%
  \;$\stEnvAssoc{\gtWithCrashedRoles{\rolesC}{\gtG}}{\stEnv; \qEnv}{\rolesR}$.\;
  If
  \;$\gtWithCrashedRoles{\rolesC}{\gtG} \gtMove[]{\rolesR}$,\;
  then there exists \;$\stEnvi; \qEnvi$, \;$\stEnvAnnotGenericSym$\; and
  \;$\gtWithCrashedRoles{\rolesCi}{\gtGi}$,\; such that
  \;$\gtWithCrashedRoles{\rolesC}{\gtG} \gtMove[\stEnvAnnotGenericSym]{\rolesR}
   \gtWithCrashedRoles{\rolesCi}{\gtGi}$,\;
  $\stEnvAssoc{\gtWithCrashedRoles{\rolesCi}{\gtGi}}{\stEnvi;
  \qEnvi}{\rolesR}$,\;
  and
  \;$\stEnv; \qEnv \stEnvMoveGenAnnot \stEnvi; \qEnvi$.
\end{restatable}
\begin{proof}
  By induction on global type reductions (\autoref{def:gtype:lts-gt}).
  See Appendix~\ref{sec:proof:relating} for details.
  \qedhere
 \iftoggle{full}{See \cref{sec:proof:relating} for detailed proof.}{}
\end{proof}

By Theorems~\ref{thm:gtype:proj-comp} and~\ref{thm:gtype:proj-sound}, we obtain, as a corollary, that
a global type $\gtG$ is in operational correspondence with the typing context
$\stEnv = \setenum{\stEnvMap{\roleP}{\gtProj[\rolesR]{\gtG}{\roleP}}}_{\roleP \in \gtRoles{\gtG}}$,
which contains the projections of all roles in $\gtG$.

\begin{rem}[Sufficiency of Soundness Theorem]
  Curious readers may wonder why we proved a soundness theorem that is not
 the dual of the completeness theorem,  \eg as seen in the literature~\cite{ICALP13CFSM}.
  This is a consequence of using `full' subtyping (\autoref{def:subtyping}, notably
  \inferrule{\iruleStSubOut}).
  A local type in the typing context may have fewer branches to choose from
  than the projected local type,
  resulting in uninhabited sending actions in the global type.

  For example, let $\gtG = \gtCommRaw{\roleP}{\roleQ}{
    \gtCommChoice{\gtLab[1]}{}{\gtEnd}; \; \;
    \gtCommChoice{\gtLab[2]}{}{\gtEnd}
  }$.
  An associated typing context $\stEnv$ (assuming $\roleP$ reliable) may have
  $\stEnvApp{\stEnv}{\roleP} =
  \stIntSum{\roleQ}{}{\stChoice{\stLab[1]}{}
  \stSeq {\stEnd}}
  \stSub
  \stIntSum{\roleQ}{}{
    \stChoice{\stLab[1]}{} \stSeq {\stEnd}; \; \;
    \stChoice{\stLab[2]}{} \stSeq {\stEnd}
  }$ (via \inferrule{\iruleStSubOut}).
  The global type $\gtG$ may make a transition
  $\stEnvOutAnnot{\roleP}{\roleQ}{\gtLab[2]}$, where an associated
  configuration $\stEnv; \qEnv[\stQEmpty]$ cannot.

  Our soundness theorem is nevertheless \emph{sufficient} for concluding
  that desired properties are guaranteed via association, \eg safety, deadlock-freedom, and liveness,
  as illustrated in Section~\ref{sec:gtype:pbp}.
\end{rem}

\begin{rem}[Relation between Well-Annotated and Well-Formed Global Types]
In the multiparty session type literature, a global type is \emph{well-formed} if it can be projected onto
every declared protocol participant.
Readers may wonder how well-formedness is applied to the global type in this paper
and how well-annotated global types~(\autoref{def:globaltypes:well-anno}) relate to well-formed ones.
Our definition of association~(\autoref{def:assoc}) ensures that a global type associated with a configuration,
with respect to a set of reliable roles $\rolesR$, is also well-formed with respect to $\rolesR$.
This is because:
  \begin{enumerate*}
  \item every global type is closed, \ie it has no free type variables;
  \item every global type is contractive; and
  \item condition (A1) %
  in~\autoref{def:assoc} ensures that the associated
    global type is projectable onto all roles with respect to a set of reliable
    roles, which is a key requirement for well-formedness.
  \end{enumerate*}

Since the soundness~(\autoref{thm:gtype:proj-sound}) and completeness~(\autoref{thm:gtype:proj-comp}) of association, along with the results on typed session properties (discussed in Section~\ref{sec:typing_system}),
depend on the concept of association, it follows that all involved global types are well-formed.

Furthermore, there is no direct relationship between well-formedness and well-annotation; a well-annotated global type may not be well-formed, and vice versa. All main results apply to global types that are both well-formed and well-annotated.
\end{rem}

\subsection{Properties Guaranteed by Projection}\label{sec:gtype:pbp}

A key benefit of our top-down approach of multiparty protocol design is that
desirable properties are guaranteed by the methodology.
As a consequence, processes
following the local types obtained from projections are correct \emph{by
construction}.
In this subsection, we focus on three properties: \emph{communication safety},
\emph{deadlock-freedom},  and \emph{liveness}, and show that the
three properties are guaranteed from
configurations associated with global types.

\paragraph*{Communication Safety}
\label{sec:type-system-safety}

We %
begin by defining communication safety for configurations
(\autoref{def:mpst-env-safe}).
We focus on the following two safety requirements:
\begin{enumerate*}[label=(\roman*)]
  \item each role must be able to handle any message that may end up in their
    receiving queue (so that there are no label mismatches); and
 \item each receiver must be able to handle the potential crash of the
    sender, unless the sender is reliable.
\end{enumerate*}

\begin{defi}[Configuration Safety]
\label{def:mpst-env-safe}%
  Given a fixed set of reliable roles $\rolesR$, we say that
  $\predP$ is an \emph{$\rolesR$-safety property} of configurations %
  iff, whenever \;$\predPApp{\stEnv; \qEnv}$,\; we have:

  \noindent%
  \begin{tabular}{@{}r@{\hskip 2mm}l}
    \inferrule{\iruleSafeComm}%
    &%
    $ \stEnvApp{\stEnv}{\roleQ} =
        \stExtSum{\roleP}{i \in I}{\stChoice{\stLab[i]}{\tyGround[i]} \stSeq \stSi[i]}%
    $
    \,and\,
    $\stEnvApp{\qEnv}{\roleP, \roleQ} \neq \stQUnavail
    $
    \,and\,
    $
      \stEnvApp{\qEnv}{\roleP, \roleQ} \neq \stQEmpty
    $
    \,implies\, %
    $\stEnvMoveAnnotP{\stEnv;
    \qEnv}{\stEnvInAnnot{\roleQ}{\roleP}{\stChoice{\stLabi}{\tyGroundi}}}$;
    \\%
    \inferrule{\iruleSafeCrash}%
    &%
    $\stEnvApp{\stEnv}{\roleP} = \stStop$
    \,and\,
    $ \stEnvApp{\stEnv}{\roleQ} =
        \stExtSum{\roleP}{i \in I}{\stChoice{\stLab[i]}{\stS[i]} \stSeq \stSi[i]}%
    $
    \,and\,
    $
      \stEnvApp{\qEnv}{\roleP, \roleQ} = \stQEmpty
    $
    \,implies\, %
    $\stEnvMoveAnnotP{\stEnv; \qEnv}{\ltsCrDe{\mpS}{\roleQ}{\roleP}}$;
    \\[1mm]
    \inferrule{\iruleSafeRec}%
    &%
    $
      \stEnvApp{\stEnv}{%
        \roleP%
      } =
      \stRec{\stRecVar}{\stS}%
    $ %
    \,implies\, %
    $\predPApp{%
      \stEnvUpd{\stEnv}{%
        \roleP%
      }{%
        \stS\subst{\stRecVar}{\stRec{\stRecVar}{\stS}}%
      }; \qEnv%
    }$;
    \\[1mm]
    \inferrule{\iruleSafeMove}%
    &%
    $\stEnv; \qEnv \stEnvMoveMaybeCrash[\rolesR] \stEnvi; \qEnvi$
    \,implies\, %
    $\predPApp{\stEnvi; \qEnvi}$.
  \end{tabular}

\smallskip
 \noindent%
  We say \emph{\;$\stEnv; \qEnv$\; is $\rolesR$-safe}, %
  if \;$\predPApp{\stEnv; \qEnv}$\; holds %
  for some $\rolesR$-safety property $\predP$. %
\end{defi}

We use a coinductive view of the safety property~\cite{SangiorgiBiSimCoInd},
where the predicate of $\rolesR$-safe configurations is the %
largest $\rolesR$-safety property, by taking the union of all safety properties
$\predP$.
For a configuration $\stEnv; \qEnv$ to be $\rolesR$-safe, it has to satisfy
all clauses defined in \autoref{def:mpst-env-safe}.

By clause~\inferrule{\iruleSafeComm}, whenever a role $\roleQ$
receives from another role $\roleP$,
and a message is present in the queue,
the receiving action must be possible for some label $\gtLabi$,
\ie the receiver $\roleQ$ must support all output messages that may appear
at the head of the queue sent from $\roleP$.

Clause~\inferrule{\iruleSafeCrash} states that if a role $\roleQ$ receives
from a crashed role $\roleP$, and there is nothing in the queue,
then $\roleQ$ must have a $\stCrashLab$ branch, and a crash detection action
can be fired.
(Note that $\inferrule{\iruleSafeComm}$ applies when the queue is non-empty,
despite the crash of sender $\roleP$.)

Finally,
clause~\inferrule{\iruleSafeRec} extends the previous clauses
by unfolding any recursive entries; and
clause \inferrule{\iruleSafeMove}
states that any configuration $\stEnvi; \qEnvi$ which $\stEnv; \qEnv$ transitions to
must also be $\rolesR$-safe.
By using transition $\stEnvMoveMaybeCrash[\rolesR]$, we
ignore crash transitions $\ltsCrash{\mpS}{\roleP}$ for
any reliable role $\roleP \in \rolesR$.%

\begin{exa}
\label{ex:configuration_safety}
Recall the local types of the Simpler Logging example in Section~\ref{sec:overview}:

\smallskip
\centerline{\(
{\small{
 \begin{array}{c}
   \begin{array}{l}
     \stT[\roleFmt{C}]=
     \roleFmt{I} \stFmt{\oplus}
       \stLabFmt{read}
     \stSeq
      \stTi[\roleFmt{C}]
      \\[2mm]
      \stTi[\roleFmt{C}] =
      \roleFmt{I}
       \stFmt{\&}
        \stLabFmt{report(\stFmtC{log})}
     \stSeq
    \stEnd
  \end{array}
  \qquad
  \begin{array}{l}
     \stT[\roleFmt{L}]
       =
     \roleFmt{I} \stFmt{\oplus} \stLabFmt{trigger} \stSeq \stTi[\roleFmt{L}]
     \\[2mm]
     \stTi[\roleFmt{L}]
       =
       \stExtSum{\roleFmt{I}}{}{
         \begin{array}{@{}l@{}}
           \stLabFmt{fatal} \stSeq \stEnd \\
           \stLabFmt{read} \stSeq \roleFmt{I} \stFmt{\oplus} \stLabFmt{report(\stFmtC{log})} \stSeq \stEnd
         \end{array}
        }
  \end{array}
\\[6mm]
\begin{array}{l}
  \stT[\roleFmt{I}]
  =
  \roleFmt{L} \stFmt{\&} \stLabFmt{trigger} \stSeq \stTi[\roleFmt{I}]
  \\[2mm]
  \stTi[\roleFmt{I}]
  =
  \stExtSum{\roleFmt{C}}{}{
    \begin{array}{@{}l@{}}
      \stLabFmt{read} \stSeq
      \roleFmt{L} \stFmt{\oplus} \stLabFmt{read} \stSeq \roleFmt{L} \stFmt{\&} \stLabFmt{report(\stFmtC{log})} \stSeq \roleFmt{C} \stFmt{\oplus} \stLabFmt{report(\stFmtC{log})} \stSeq \stEnd
       \\
      \stCrashLab \stSeq
      \stTii[\roleFmt{I}]
    \end{array}
  }
  \\[4mm]
  \stTii[\roleFmt{I}]
  =
  \roleFmt{L} \stFmt{\oplus} \stLabFmt{fatal} \stSeq \stEnd
\end{array}
\end{array}
 }
}
\)}

\smallskip
\noindent
The configuration $\stEnv; \qEnv$, where
$\stEnv =  \stEnvMap{\roleFmt{C}}{\stT[\roleFmt{C}]}  \stEnvComp
\stEnvMap{\roleFmt{L}}{\stT[\roleFmt{L}]}   \stEnvComp
\stEnvMap{\roleFmt{I}}{\stT[\roleFmt{I}]}$ and
$\qEnv =  \qEnv[\stQEmpty]$, is
$\setenum{\roleFmt{L}, \roleFmt{I}}$-safe.
This can be verified by checking its possible
reductions.  For example,
in the case where $\roleFmt{C}$ crashes immediately, we have:

\smallskip
{\small{
\centerline{\(
\begin{array}{@{}r@{}cl}
\stEnv; \qEnv
&
\stEnvMoveAnnot{\ltsCrash{\mpS}{\roleFmt{C}}}
&
 \stEnvUpd{\stEnv}{\roleFmt{C}}{
\stStop};
       \stEnvUpd{\qEnv}{\cdot, \roleFmt{C}}{\stQUnavail}
       \\
&
\stEnvMoveOutAnnot{\roleFmt{L}}{\roleFmt{I}}{\stChoice{\stLabFmt{trigger}}{}}
&
\stEnvUpd{\stEnvUpd{\stEnv}{\roleFmt{C}}{\stStop}}{
  \roleFmt{L}}{\stTi[\roleFmt{L}]
};  \stEnvUpd{\stEnvUpd{\qEnv}{\cdot, \roleFmt{C}}{\stQUnavail}}{\roleFmt{L}, \roleFmt{I}}{\stLabFmt{trigger}}
\\
    &
     \stEnvMoveInAnnot{\roleFmt{I}}{\roleFmt{L}}{\stChoice{\stLabFmt{trigger}}{}}
     &
 \stEnvUpd{
   \stEnvUpd{
     \stEnvUpd{\stEnv}{\roleFmt{C}}{\stStop}
   }{
     \roleFmt{L}
   }{
     \stTi[\roleFmt{L}]
   }
 }{
   \roleFmt{I}
 }{
    \stTi[\roleFmt{I}]
 }
; \stEnvUpd{\qEnv}{\cdot, \roleFmt{C}}{\stQUnavail}
 \\
 &
 \stEnvMoveAnnot{\ltsCrDe{\mpS}{\roleFmt{I}}{\roleFmt{C}}}
 &
  \stEnvUpd{\stEnvUpd{\stEnvUpd{\stEnv}{\roleFmt{C}}{\stStop}}{
 \roleFmt{L}}{
 \stTi[\roleFmt{L}]
  }}{\roleFmt{I}}{\stTii[\roleFmt{I}]}; \stEnvUpd{\qEnv}{\cdot, \roleFmt{C}}{\stQUnavail}
 \\
  &
     \stEnvMoveOutAnnot{\roleFmt{I}}{\roleFmt{L}}{\stChoice{\stLabFmt{fatal}}{}}
     &
    \stEnvUpd{\stEnvUpd{\stEnvUpd{\stEnv}{\roleFmt{C}}{\stStop}}{
 \roleFmt{L}}{
    \stTi[\roleFmt{L}]
  }}{\roleFmt{I}}{\stEnd};
   \stEnvUpd{\stEnvUpd{\qEnv}{\cdot, \roleFmt{C}}{\stQUnavail}}{\roleFmt{I}, \roleFmt{L}}{\stLabFmt{fatal}}
   \\
   &
    \stEnvMoveInAnnot{\roleFmt{L}}{\roleFmt{I}}{\stChoice{\stLabFmt{fatal}}{}}
     &
    \stEnvUpd{\stEnvUpd{\stEnvUpd{\stEnv}{\roleFmt{C}}{\stStop}}{
 \roleFmt{L}}{
  \stEnd
  }}{\roleFmt{I}}{\stEnd};
  \stEnvUpd{\qEnv}{\cdot, \roleFmt{C}}{\stQUnavail}
 \end{array}
\)}
}}

\smallskip
\noindent
and each reductum satisfies all clauses of~\autoref{def:mpst-env-safe}.
The cases where $\roleFmt{C}$ crashes after sending the $\stLabFmt{read}\text{ing}$ message
to $\roleFmt{I}$ are similar.  There are no other crash reductions to consider, since both
$\roleFmt{L}$ and $\roleFmt{I}$ are assumed to be reliable.
The cases where no crashes occur are similar as well,
except that \inferrule{\iruleTCtxCrashDetect} and \inferrule{\iruleTCtxCrash} are not applied in the non-crash reductions.
\end{exa}

\paragraph*{Deadlock-Freedom}
\label{sec:type-system-deadlock-free}

The property of deadlock-freedom, sometimes also known as progress, describes
whether a configuration can keep reducing unless it is a terminal configuration.
We give its formal definition in~\autoref{def:mpst-env-deadlock-free}.

\begin{defi}[Configuration Deadlock-Freedom]
\label{def:mpst-env-deadlock-free}%
Given a set of reliable roles $\rolesR$, we say that a configuration \;$\stEnv;
\qEnv$\; is
\emph{$\rolesR$-deadlock-free} iff:
\begin{enumerate}[leftmargin=*, nosep]
  \item $\stEnv; \qEnv$\; is $\rolesR$-safe; and,
  \item
    \label{item:df:reduces}
    If \;$\stEnv; \qEnv$\; can reduce to a configuration \;$\stEnvi; \qEnvi$\;
    without further reductions:
    \;$\stEnv; \qEnv \!\stEnvMoveMaybeCrashStar[\rolesR]\! \stEnvi; \qEnvi
    \!\not\stEnvMoveMaybeCrash[\rolesR]$,\; then:
    \begin{enumerate}[leftmargin=*, nosep]
    \item
      \label{item:df:context}
      $\stEnvi$ can be split into two disjoint contexts,
      one with only $\stEnd$
      entries, and one with only $\stStop$ entries:
      \;$\stEnvi =  \stEnvi[\stEnd] \stEnvComp \stEnvi[\stStopSym]$,\; where
      \;$\dom{\stEnvi[\stEnd]} =
      \setcomp{\roleP}{\stEnvApp{\stEnvi}{\roleP} = \stEnd}$\; and
      \;$\dom{\stEnvi[\stStopSym]} =
      \setcomp{\roleP}{\stEnvApp{\stEnvi}{\roleP} = \stStop}$;\; and,
    \item
      \label{item:df:queues}
      $\qEnvi$ is empty for all pairs of roles, except for the receiving queues
      of crashed roles, which are unavailable:
      \;$\forall \roleP, \roleQ:  \stEnvApp{\qEnvi}{\cdot, \roleQ} =
      \stQUnavail$\;
      if
       \;$\stEnvApp{\stEnvi}{\roleQ} = \stStop$,\; and
       \;$\stEnvApp{\qEnvi}{\roleP, \roleQ} = \stQEmpty$,\; otherwise.
    \end{enumerate}
\end{enumerate}
\end{defi}

It is worth noting that a (safe) configuration that reduces infinitely
satisfies deadlock-freedom, as \cref{item:df:reduces} in the premise does not hold.
Otherwise, whenever a terminal configuration is reached, it must satisfy
\cref{item:df:context} that all local types in the typing context be
terminated (either successfully $\stEnd$, or crashed $\stStop$), and
\cref{item:df:queues} that all queues be empty (unless unavailable due to
crash).
As a consequence,
a deadlock-free configuration $\stEnv; \qEnv$ either does not stop reducing, or
terminates in a stable configuration.

\paragraph*{Liveness}
\label{sec:type-system-live}
The property of liveness describes that every pending internal/external choice
is eventually triggered by means of a message transmission or crash detection.
Our liveness property is based on \emph{fairness},
which guarantees that every enabled message transmission, including crash detection,
is performed successfully.
We give the definitions of non-crashing, fair,  and live paths of configurations respectively in~\autoref{def:non-crash-fair-live-path}, and use these paths to
formalise the liveness for configurations in~\autoref{def:mpst-env-live}.

\begin{defi}[Non-crashing, Fair, Live Paths]
\label{def:non-crash-fair-live-path}
 A \emph{non-crashing path} is a possibly infinite
sequence of configurations \;$(\stEnv[n]; \qEnv[n])_{n \in N}$,\; where
\;$N = \setenum{0, 1, 2, \ldots}$\;
is a set of consecutive natural numbers, and
\;$\forall n \in N$, $\stEnv[n]; \qEnv[n] \!\stEnvMove\!  \stEnv[n+1]; \qEnv[n+1]$.

We say that a non-crashing path \;$(\stEnv[n]; \qEnv[n])_{n \in N}$\; is
\emph{fair} iff, \;$\forall n \in N$:
 \begin{enumerate}[label={(F\arabic*)}, leftmargin=*, ref={(F\arabic*)},nosep]
 \item
 \label{item:fairness_send}
 $\stEnv[n]; \qEnv[n] \stEnvMoveOutAnnot{\roleP}{\roleQ}{\stChoice{\stLab}{\tyGround}}$\;
 implies \;$\exists k, \stLabi, \tyGroundi$\; such that \;$n \leq k \in N$\;
 and
 \;$\stEnv[k]; \qEnv[k] \stEnvMoveOutAnnot{\roleP}{\roleQ}{\stChoice{\stLabi}{\tyGroundi}} \stEnv[k+1]; \qEnv[k+1]$;
 \item
 \label{item:fairness_receive}
 $\stEnv[n]; \qEnv[n] \stEnvMoveInAnnot{\roleP}{\roleQ}{\stChoice{\stLab}{\tyGround}}$\;
 implies \;$\exists k$\; such that \;$n \leq k \in N$\;
 and
 \;$\stEnv[k]; \qEnv[k] \stEnvMoveInAnnot{\roleP}{\roleQ}{\stChoice{\stLab}{\tyGround}} \stEnv[k+1]; \qEnv[k+1]$;
 \item
 \label{item:fairness_crash_detection}
 $\stEnv[n]; \qEnv[n] \stEnvMoveAnnot{\ltsCrDe{\mpS}{\roleP}{\roleQ}}$\;
 implies \;$\exists k$\; such that \;$n \leq k \in N$\;
 and
 \;$\stEnv[k]; \qEnv[k] \stEnvMoveAnnot{\ltsCrDe{\mpS}{\roleP}{\roleQ}}\stEnv[k+1]; \qEnv[k+1]$.
 \end{enumerate}

We say that a non-crashing path \;$(\stEnv[n]; \qEnv[n])_{n \in N}$\;
is \emph{live} iff, \;$\forall n \in N$:
 \begin{enumerate}[label={(L\arabic*)}, leftmargin=*, ref={(L\arabic*)},nosep]
 \item
 \label{item:liveness_consume}
 $\stEnvApp{\qEnv[n]}{\roleP, \roleQ}
      =
      \stQCons{\stQMsg{\stLab}{\tyGround}}{\stQ}
      \neq
      \stQUnavail$\; and \;$\stLab \neq \stCrashLab$\;
 implies \;$\exists k$\; such that \;$n \leq k \in N$\;
 and \;
 $\stEnv[k]; \qEnv[k] \stEnvMoveInAnnot{\roleQ}{\roleP}{\stChoice{\stLab}{\tyGround}} \stEnv[k+1]; \qEnv[k+1]$;
\item
\label{item:liveness_receiving_cd}
$\stEnvApp{\stEnv[n]}{%
        \roleP%
      } =
      \stExtSum{\roleQ}{i \in I}{\stChoice{\stLab[i]}{\tyGround[i]} \stSeq \stT[i]}$\;    %
 implies \;$\exists k, \stLabi, \tyGroundi$\; such that \;$n \leq k \in N$\; and\\
 $\stEnv[k]; \qEnv[k]
 \stEnvMoveInAnnot{\roleP}{\roleQ}{\stChoice{\stLabi}{\tyGroundi}} \stEnv[k+1];
 \qEnv[k+1]$\;
 or
\;$\stEnv[k]; \qEnv[k]
      \;\stEnvMoveAnnot{\ltsCrDe{\mpS}{\roleP}{\roleQ}}\;%
      \stEnv[k+1]; \qEnv[k+1]$.
 \end{enumerate}
 \end{defi}
A non-crashing path is a (possibly infinite) sequence of reductions of
a configuration without crashes. A non-crashing path is fair if along the path, every internal choice eventually
sends a message~\ref{item:fairness_send}, every external choice eventually receives a message~\ref{item:fairness_receive},
and every crash detection is eventually performed~\ref{item:fairness_crash_detection}.  A non-crashing path is live if
along the path, every non-crash message in the queue is eventually consumed~\ref{item:liveness_consume}, and every
hanging external choice eventually consumes a message or performs a crash detection~\ref{item:liveness_receiving_cd}.

 \begin{defi}[Configuration Liveness]
 \label{def:mpst-env-live}%
Given a set of reliable roles $\rolesR$, we say that a configuration $\stEnv;
\qEnv$ is
\emph{$\rolesR$-live} iff:
\begin{enumerate*}[leftmargin=*, nosep]
  \item $\stEnv; \qEnv$ is $\rolesR$-safe; and,
  \item
    \label{item:live:reduces}
  $\stEnv; \qEnv \stEnvMoveMaybeCrashStar[\rolesR]
  \stEnvi; \qEnvi$\; implies all non-crashing paths starting with
  $\stEnvi; \qEnvi$ that are fair are also live.
   \end{enumerate*}
\end{defi}
A configuration $\stEnv; \qEnv$ is $\rolesR$-live when it is $\rolesR$-safe and
any reductum of $\stEnv; \qEnv$ (via transition $\stEnvMoveMaybeCrashStar[\rolesR]$)
consistently leads to a live path if it is fair.

\begin{exa} 
\label{ex:dl-liveness}
We illustrate safety, deadlock-freedom, and liveness over configurations via a series of small examples. 
We consider the configuration $\stEnv[\stFmt{A}]; \qEnv[\stFmt{A}]$, where 
$\stEnv[\stFmt{A}] = \stEnv[\stFmt{A\roleP}] \stEnvComp \stEnv[\stFmt{A\roleQ}] \stEnvComp \stEnv[\stFmt{A\roleR}]$ 
and 
$\qEnv[\stFmt{A}] =  \qEnv[\stQEmpty]$ with:  

\smallskip
\centerline{\(
\begin{array}{rcl}
\stEnv[\stFmt{A\roleP}] &=&
\stEnvMap{\roleP}{
  \stRec{\stRecVar[\roleP]}{
    \stIntSum{\roleQ}{}{
      \stChoice{\stLabOK}{} \stSeq
      \stExtSum{\roleQ}{}{
        \stChoice{\stLabOK}{} \stSeq
        \stRecVar[\roleP]
        \stEnvComp\;
        \stChoice{\stLabKO}{} \stSeq
        \stEnd
        \stEnvComp\;
        \stChoice{\stCrashLab}{} \stSeq
        \stEnd
      }
      \stEnvComp\;
      \stChoice{\stLabKO}{} \stSeq \stEnd
    }
  }
}
\\
\stEnv[\stFmt{A\roleQ}] &=&
\stEnvMap{\roleQ}{
  \stRec{\stRecVar[\roleQ]}{
    \stExtSum{\roleP}{}{
      \stChoice{\stLabOK}{} \stSeq
      \stIntSum{\roleP}{}{
        \stChoice{\stLabOK}{} \stSeq
        \stRecVar[\roleQ]
        \stEnvComp\;
        \stChoice{\stLabKO}{} \stSeq
        \stEnd
      }
      \stEnvComp\;
      \stChoice{\stLabKO}{} \stSeq
      \stEnd
      \stEnvComp
      \stChoice{\stCrashLab}{} \stSeq
      \stOut{\roleR}{\stLabOK}{} \stSeq
      \stEnd
    }
  }
}
\\
\stEnv[\stFmt{A\roleR}] &=&
\stEnvMap{\roleR}{
  \stIn{\roleP}{\stCrashLab}{}{
    \stIn{\roleQ}{\stLabOK}{}{
    \stEnd
    \stEnvComp\;
    \stChoice{\stCrashLab}{} \stSeq
    \stEnd
    }
  }
}
\end{array}
\)}\smallskip

\noindent
If we assume that all roles %
are unreliable, \ie $\rolesR = \emptyset$, 
$\stEnv[\stFmt{A}]; \qEnv[\stFmt{A}]$ is $\emptyset$-safe 
since the inputs/outputs in the typing context $\stEnv[\stFmt{A}]$ are dual and the queue environment $\qEnv[\stFmt{A}]$ is empty. 
However, $\stEnv[\stFmt{A}]; \qEnv[\stFmt{A}]$ is \emph{neither} $\emptyset$-deadlock-free \emph{nor} $\emptyset$-live since it is possible for $\roleP$ to crash immediately before $\roleQ$ sends $\stLabKO$ to $\roleP$. In such cases, $\roleQ$ will \emph{not} detect that $\roleP$ has crashed (since we only detect crashes on receive actions) and terminate \emph{without} sending a message to the backup process $\roleR$. This results in a deadlock because  $\roleR$ \emph{will} detect that $\roleP$ has crashed, and \emph{will} expect a message from $\roleQ$.

We observe that changing the reliability assumptions, without changing the configuration, may influence whether a configuration property holds. For example, in the case of $\stEnv[\stFmt{A}]; \qEnv[\stFmt{A}]$, 
we \emph{can} obtain liveness by adjusting the reliability assumptions: in fact, if we assume $\roleR \in \rolesR$, then 
$\stEnv[\stFmt{A}]; \qEnv[\stFmt{A}]$ is both $\rolesR$-deadlock-free and $\rolesR$-live.

Then consider the configuration $\stEnv[\stFmt{B}];  \qEnv[\stFmt{B}]$, where 
$\stEnv[\stFmt{B}] = \stEnv[\stFmt{B\roleP}] \stEnvComp \stEnv[\stFmt{B\roleQ}] \stEnvComp \stEnv[\stFmt{B\roleR}]$ 
and $\qEnv[\stFmt{B}] = \qEnv[\stQEmpty]$ with:

\smallskip
\centerline{\(
\begin{array}{rcl}
\stEnv[\stFmt{B\roleP}] &=&
\stEnvMap{\roleP}{
  \stRec{\stRecVar[\roleP]}{
    \stOut{\roleQ}{\stLabOK}{} \stSeq
    \stRecVar[\roleP]
  }
}
\\
\stEnv[\stFmt{B\roleQ}] &=&
\stEnvMap{\roleQ}{
  \stRec{\stRecVar[\roleQ]}{
    \stIn{\roleP}{\stLabOK}{}{
      \stRecVar[\roleQ]
      \stEnvComp\;
      \stChoice{\stCrashLab}{} \stSeq
      \stRec{\stRecVari[\roleQ]}{
        \stIn{\roleR}{\stLabOK}{}{
          \stRecVari[\roleQ]
          \stEnvComp\;
          \stChoice{\stCrashLab}{} \stSeq
          \stEnd
        }
      }
    }
 }
}
\\
\stEnv[\stFmt{B\roleR}] &=&
\stEnvMap{\roleR}{
  \stRec{\stRecVar[\roleR]}{
    \stOut{\roleQ}{\stLabOK}{} \stSeq
    \stRecVar[\roleR]
  }
}
\end{array}
\)}\smallskip

\noindent
If we assume that all roles are unreliable, \ie $\rolesR = \emptyset$, 
$\stEnv[\stFmt{B}]; \qEnv[\stFmt{B}]$ is $\emptyset$-safe and $\emptyset$-deadlock-free but \emph{not} 
$\emptyset$-live -- because 
$\roleP$ may never crash, and in this case, $\roleR$'s outputs are never received by $\roleQ$.  %
Notice that, in the case of $\stEnv[\stFmt{B}]; \qEnv[\stFmt{B}]$, we are unable to make liveness hold purely via combinations of reliable roles: this is because (unless $\roleP$ crashes) $\roleR$'s output will never be received by 
$\roleQ$, irrespective of reliability assumptions. The configuration itself must instead be adapted accordingly, \eg 
in $\stEnv[\stFmt{B}]$, $\roleR$ should be permitted to send only once it has detected that $\roleP$ has crashed.

Finally, consider the configuration $\stEnv[\stFmt{C}]; \qEnv[\stFmt{C}]$, where $\stEnv[\stFmt{C}] = \stEnv[\stFmt{C\roleP}] \stEnvComp \stEnv[\stFmt{C\roleQ}] \stEnvComp \stEnv[\stFmt{C\roleR}]$ and $\qEnv[\stFmt{C}] =  \qEnv[\stQEmpty]$ with: 

\smallskip
\centerline{\(
\begin{array}{rcl}
\stEnv[\stFmt{C\roleP}] &=&
\stEnvMap{\roleP}{
  \stOut{\roleQ}{\stLab[1]}{} \stSeq
  \stIn{\roleQ}{\stLab[2]}{}{
    \stEnd
    \stEnvComp\;
    \stChoice{\stCrashLab}{} \stSeq
    \stRec{\stRecVar[\roleP]}{
      \stOut{\roleR}{\stLabOK}{} \stSeq
      \stRecVar[\roleP]
    }
  }
}
\\
\stEnv[\stFmt{C\roleQ}] &=&
\stEnvMap{\roleQ}{
  \stIn{\roleP}{\stLab[1]}{}{
    \stOut{\roleP}{\stLab[2]}{} \stSeq
    \stEnd
  }
}
\\
\stEnv[\stFmt{C\roleR}] &=&
\stEnvMap{\roleR}{
  \stIn{\roleP}{\stCrashLab}{}{
    \stRec{\stRecVar[\roleQ]}{
      \stIn{\roleP}{\stLabOK}{}{
        \stRecVar[\roleQ]
      }
    }
  }
}
\end{array}
\)}\smallskip

\noindent
 When \emph{all} roles are assumed to be reliable, \ie $\rolesR = \setenum{\roleP, \roleQ, \roleR}$, 
 $\stEnv[\stFmt{C}]; \qEnv[\stFmt{C}]$ satisfies $\setenum{\roleP, \roleQ, \roleR}$-safety and $\setenum{\roleP, \roleQ, \roleR}$-deadlock-freedom. %
However, should we instead assume that no roles are reliable, \ie $\rolesR = \emptyset$, $\stEnv[\stFmt{C}]; \qEnv[\stFmt{C}]$ satisfies only $\emptyset$-safety since external choices in $\stEnv[\stFmt{C}]$ do not feature a crash handling branch when receiving from $\roleP$.

\end{exa}

\paragraph*{Properties by Projection}
We conclude by showing the guarantee of safety,
deadlock-freedom, and liveness in configurations
associated with global types in~\autoref{lem:ext-proj}.
Furthermore, as a corollary,~\autoref{cor:allproperties} demonstrates that
a typing context projected from a global type (without runtime
constructs) is inherently safe, deadlock-free, and live by construction.  %
\iftoggle{full}{The detailed proofs for~\Cref{lem:ext-proj,cor:allproperties}
are available in~\cref{sec:proof:propbyproj}.}{}

\begin{restatable}{lem}{lemProj}
 \label{lem:ext-proj}
  If $\stEnvAssoc{\gtWithCrashedRoles{\rolesC}{\gtG}}{\stEnv; \qEnv}{\rolesR}$,
  then $\stEnv; \qEnv$ is $\rolesR$-safe, $\rolesR$-deadlock-free, and $\rolesR$-live.%
\end{restatable}
\begin{proof}
The results follow from the operational correspondence between
global types and configurations~(Theorems \ref{thm:gtype:proj-comp}
and~\ref{thm:gtype:proj-sound}),
and the respective definitions for each property~(Definitions \ref{def:mpst-env-safe},
\ref{def:mpst-env-deadlock-free}, and \ref{def:mpst-env-live}).
See Appendix~\ref{sec:proof:safety},~\ref{sec:proof:deadlockfree}, and~\ref{sec:proof:liveness} for details.
\qedhere
\end{proof}

\begin{restatable}[Safety, Deadlock-Freedom, and Liveness by Projection]{thm}{colProjAll}%
  \label{cor:allproperties}
  Let $\gtG$ be a global type without runtime constructs,
  and $\rolesR$ be a set of reliable roles.
  If $\stEnv$ is a typing context
  associated with the global type $\gtG$:
  $\stEnvAssoc{\gtG}{\stEnv}{\rolesR}$,
  then $\stEnv; \qEnv[\stQEmpty]$ is $\rolesR$-safe, $\rolesR$-deadlock-free, and $\rolesR$-live.
\end{restatable}
\begin{proof}
Note $\stEnvAssoc{\gtG}{\stEnv}{\rolesR}$ is an abbreviation of
  $\stEnvAssoc{\gtWithCrashedRoles{\emptyset}{\gtG}}{\stEnv; \qEnv[\stQEmpty]}{\rolesR}$,
  apply~\autoref{lem:ext-proj}.
  \qedhere
\end{proof}

\section{A Type System with Crash-Stop Semantics}
\label{sec:typing_system}
In this section, we present a type system for our asynchronous multiparty session
calculus. Our typing system %
extends the one in~\cite{POPL21AsyncMPSTSubtyping} with
crash-stop failures -- typing rules for crashed process and unavailable queues, respectively.
We introduce the typing rules in Section~\ref{sec:typingrules}, and show various
properties of typed sessions: subject reduction (\autoref{lem:sr}), session
fidelity (\autoref{lem:sf}), deadlock-freedom (\autoref{lem:session_deadlock_free}),
and liveness (\autoref{lem:session_live}) in Section~\ref{sec:type:system:results}.

\subsection{Typing Rules}
\label{sec:typingrules}

Our type system involves four kinds of typing judgements:
\emph{(1)} for expressions;
\emph{(2)} for processes;
\emph{(3)} for queues;
and \emph{(4)} for sessions,
and is defined inductively by the typing rules in~\cref{fig:processes:typesystem}.

\begin{figure}
 \noindent
  \centerline{\(
\begin{array}{@{}c@{}}
\inference{}{
\Theta \vdash \mpNat : \tyNat}
  \quad 
   \inference{}{\Theta \vdash \mpInt : \tyInt}
   \quad 
    \inference{}{\Theta \vdash \mpTrue : \tyBool}
    \quad 
    \inference{}{\Theta \vdash \mpFalse : \tyBool}
    \quad 
    \inference{}{\Theta \vdash \mpString{} : \tyString}
    \\[3mm]
    \inference{}{\Theta \vdash \mpUnit{} : \tyUnit}
 \quad 
     \inference{}{\Theta , \mpx : \tyGround \vdash \mpx : \tyGround}
     \quad 
     \inference{\Theta \vdash \mpE : \tyNat}{\Theta \vdash \mpSucc{\mpE} : \tyNat}
     \quad 
     \inference{\Theta \vdash \mpE : \tyInt}{\Theta \vdash \mpNeg{\mpE} : \tyInt}
     \\[3mm]
      \inference{\Theta \vdash \mpE : \tyBool}{\Theta \vdash \neg \mpE : \tyBool}
      \qquad 
       \inference{\Theta \vdash \mpE[1] : \tyInt & \Theta \vdash \mpE[2] : \tyInt}{\Theta \vdash \mpE[1] < \mpE[2]: \tyBool}
     \\[3mm]
  \inference[{t-$\mpQEmpty$}]{}{\vdash \mpQEmpty : \stQEmpty}
  \qquad
  \highlight{\inference[{t-$\mpQUnavail$}]{ }{\vdash \mpQUnavail : \stQUnavail}}
\qquad
  \inference[{t-$\cdot$}]{
    \vdash \mpH_1:\qEnvPartial[1] \quad \vdash \mpH_2:\qEnvPartial[2]
  }{
    \vdash \mpH_1 \cdot \mpH_2: \qEnvPartial[1] \cdot \qEnvPartial[2]
  }
\\[3mm]
    \inference[{t-msg}]{
    \vdash \mpV:\tyGround
    &
    \stEnvApp{\qEnvPartial}{\roleQ} = \stQMsg\stLab\tyGround
    &
    \forall \roleR \neq \roleQ:
    \stEnvApp{\qEnvPartial}{\roleR} = \stQEmpty
  }{
    \vdash (\roleQ , \mpLab(\mpV)) : \qEnvPartial
  }
\\[3mm]
\highlight{\inference[{t-$\mpCrash$}]{$\;$}{\Theta \vdash \mpCrash : \stStop}}
  \qquad
\inference[{t-$\mpNil$}]{$\;$}{\Theta \vdash \mpNil : \stEnd}
  \qquad
\inference[{t-out}]{
  \Theta \vdash \mpE:\tyGround
  \quad
  \Theta \vdash \mpP:\stT
}{
  \Theta  \vdash \procout\roleQ{\mpLab}{\mpE}{\mpP}:
  \stOut{\roleQ}{\stLab}{\tyGround}\stSeq{\stT}
}
  \\[3mm]%
\inference[{t-ext}]{
  \forall i\in I\;\;\; \Theta, x_i:\tyGround_i \vdash \mpP_i:\stT_i
}{
  \Theta \vdash  \sum_{i\in I}\procin{\roleQ}{\mpLab_i(\mpx_i)}{\mpP_i}:
  \stExtSum{\roleQ}{i\in I}{\stChoice{\stLab[i]}{\tyGround_i}\stSeq{\stT_i}}
}
  \quad
\inference[{t-cond}]{\Theta \vdash \mpE:\tyBool
\quad \Theta  \vdash \mpP_i:\stT \ \text{\footnotesize $(i=1,2)$}
}{\Theta \vdash \mpIf\mpE{\mpP_1}{\mpP_2}:\stT}
\\[3mm]
\inference[{t-rec}]{\Theta, X:\stT  \vdash \mpP:\stT}
{\Theta  \vdash  \mu X.\mpP: \stT}
\qquad
\inference[{t-var}]{$\;$}{\Theta, X:\stT  \vdash X:\stT}
\qquad
\inference[{t-sub}]{\Theta \vdash \mpP:\stT \quad \stT\stSub \stT' }{\Theta \vdash \mpP:\stT'}
\\[3mm]
\inference[{t-sess}]{
  \highlight{\stEnvAssoc{\gtWithCrashedRoles{\rolesC}{\gtG}}{\stEnv; \qEnv}{\rolesR}}
  \quad
  \forall i\in I %
  \quad 
  \highlight{\vdash \mpP_i:\stEnvApp{\stEnv}{\roleP[i]}}
   \quad  %
 \highlight{\vdash \mpH[i]: \stEnvApp{\qEnv}{-, \roleP[i]}} %
  &
 \highlight{ \dom{\stEnv} \subseteq
  \setcomp{\roleP[i]}{i \in I}}
}{
  \gtWithCrashedRoles{\rolesC}{\gtG}
  \vdash \prod_{i\in I} (\mpPart{\roleP[i]}{\mpP[i]} \mpPar
  \mpPart{\roleP[i]}{\mpH[i]})
}
\end{array}
\)}
\caption{
  Typing rules for expressions, queues,
  processes, and sessions.%
}%
\label{fig:processes:typesystem}
\end{figure}

Typing judgments for expressions and processes take the forms
$\Theta \vdash \mpE : \tyGround$ and $\Theta \vdash \mpP : \stT$, respectively, where $\Theta$ is a typing context for expression and process variables,
defined as $\Theta \bnfdef
\emptyset \bnfsep \Theta , \mpx : \tyGround \bnfsep \Theta , \mpX : \stT$.

With regard to queues, we use judgments of form $\vdash \mpH : \qEnvPartial$,
where we use $\qEnvPartial$ denote a partially applied queue lookup function.
We write $\qEnvPartial = \stEnvApp{\qEnv}{-, \roleP}$ to describe the incoming
queue for a role $\roleP$, as a partially
applied function $\qEnvPartial = \stEnvApp{\qEnv}{-, \roleP}$ such that
$\stEnvApp{\qEnvPartial}{\roleQ} = \stEnvApp{\qEnv}{\roleQ,\roleP}$.
We write $\qEnvPartial[1] \cdot \qEnvPartial[2]$ to denote the point-wise
application of concatenation.
We lift the process-level constructs for empty queues~($\mpQEmpty$), unavailable queues~($\mpQUnavail$), 
queue concatenations~($\cdot$), 
and the structural precongruence relation~($\Rrightarrow$) on queues to their corresponding type-level counterparts.
For a singleton message $(\roleQ, \mpLab(\mpV))$, the appropriate partial queue
$\qEnvPartial$ would be a singleton of $\stQMsg{\stLab}{\tyGround}$ (where
$\tyGround$ is the type of $\mpV$) for $\roleQ$, and an empty queue
($\stQEmpty$) for any other role.

Finally, we use judgments of form $\gtWithCrashedRoles{\rolesC}{\gtG} \vdash
\mpM$ for sessions.
We use a global type-guided judgment, effectively asserting that all
participants in the session respect the prescribed global type, as is the case
in~\cite{JLAMP19SyncSubtyping}.
As $\highlight{\text{highlighted}}$, the global type with crashed roles $\gtWithCrashedRoles{\rolesC}{\gtG}$ must
have some associated configuration $\stEnv; \qEnv$, which are used to type the
processes and the queues respectively.
Moreover, all the entries in the configuration must be present in the session.

We explain the rules in~\cref{fig:processes:typesystem}
that assign session types based on process behaviour~(other rules are mostly self-explanatory).
Rule \inferrule{t-$\mpQUnavail$} ($\highlight{\text{highlighted}}$)
assigns the unavailable queue type $\stQUnavail$
 to an unavailable queue
$\mpQUnavail$.
Rules $\inferrule{t-out}$ and $\inferrule{t-ext}$ assign internal
and external choice types to input and output processes, respectively.
Additionally, rule \inferrule{t-$\mpCrash$} ($\highlight{\text{highlighted}}$)
assigns the crash termination
$\stStop$ to a crashed process $\mpCrash$,
while rule \inferrule{t-$\mpNil$} assigns the successful termination $\stEnd$ to an inactive process $\mpNil$.

\begin{exa}
\label{ex:typing_system}
Consider our Simpler Logging example (Section~\ref{sec:overview} and \autoref{ex:configuration_safety}).
Specifically, consider the processes that act as the roles $\roleFmt{C}$, $\roleFmt{L}$, and
$\roleFmt{I}$ respectively:

\smallskip
 \centerline{\(
 \begin{array}{c}
\mpP[\roleFmt{C}]  =
\procoutNoVal{\roleFmt{I}}{\labFmt{read}}{\procin{\roleFmt{I}}{\labFmt{report}(\mpx)}{\mpNil}}
\quad
\mpP[\roleFmt{L}] =
\procoutNoVal{\roleFmt{I}}{\labFmt{trigger}}{\sum \setenum{
 \begin{array}{l}
 \procin{\roleFmt{I}}{\labFmt{fatal}}{\mpNil}
 \\
 \procin{\roleFmt{I}}{\labFmt{read}}{
 \procout{\roleFmt{I}}{\labFmt{report}}{\mpFmt{log}}{\mpNil}}
  \end{array}
  }
  }
  \\[3mm]
 \mpP[\roleFmt{I}] =
\procin{\roleFmt{L}}{\labFmt{trigger}}{\sum \setenum{
 \begin{array}{l}
 \procin{\roleFmt{C}}{\labFmt{read}}{
 \procoutNoVal{\roleFmt{L}}{\labFmt{read}}
 {\procin{\roleFmt{L}}{\labFmt{report}(\mpx)}
 {\procout{\roleFmt{C}}{\labFmt{report}}{\mpFmt{log}}{\mpNil}}}
 }
 \\
  \procin{\roleFmt{C}}{\mpCrashLab}{ \procoutNoVal{\roleFmt{L}}{\labFmt{fatal}}{\mpNil}}
  \end{array}
  }
  }
\end{array}
\)}

\smallskip
\noindent
and message queues $\mpH[\roleFmt{C}] = \mpH[\roleFmt{L}] = \mpH[\roleFmt{I}] = \mpQEmpty$,
and %
 the configuration $\stEnv; \qEnv$ as in \autoref{ex:configuration_safety}.

Process $\mpP[\roleFmt{C}]$ (resp. $\mpP[\roleFmt{L}]$, $\mpP[\roleFmt{I}]$)
has the type $\stEnvApp{\stEnv}{\roleFmt{C}}$ (resp. $\stEnvApp{\stEnv}{\roleFmt{L}}$,
$\stEnvApp{\stEnv}{\roleFmt{I}}$),
and
queue $\mpH[\roleFmt{C}]$ (resp. $\mpH[\roleFmt{L}]$, $\mpH[\roleFmt{I}]$) %
has the type
$\stEnvApp{\qEnv}{-, \roleFmt{C}}$ (resp. $\stEnvApp{\qEnv}{-, \roleFmt{L}}$, $\stEnvApp{\qEnv}{-, \roleFmt{I}}$),
which can be verified in the standard way. Then, together with the association
 $\stEnvAssoc{\gtWithCrashedRoles{\emptyset}{\gtG[s]}}{\stEnv; \qEnv}{\setenum{\roleFmt{L}, \roleFmt{I}}}$,
we can use \inferrule{{t-sess}} to assert that the session
$\mpPart{\roleFmt{C}}{\mpP[\roleFmt{C}]} \mpPar
  \mpPart{\roleFmt{C}}{\mpH[\roleFmt{C}]} \mpPar
  \mpPart{\roleFmt{L}}{\mpP[\roleFmt{L}]} \mpPar
  \mpPart{\roleFmt{L}}{\mpH[\roleFmt{L}]} \mpPar
  \mpPart{\roleFmt{I}}{\mpP[\roleFmt{I}]} \mpPar
 \mpPart{\roleFmt{I}}{\mpH[\roleFmt{I}]}$
 is governed by the global type
  $\gtWithCrashedRoles{\emptyset}{\gtG[s]}$.
  If we follow a crash reduction, \eg by the rule \inferrule{r-$\lightning$},
  the session evolves as $\mpPart{\roleFmt{C}}{\mpP[\roleFmt{C}]} \mpPar
  \mpPart{\roleFmt{C}}{\mpH[\roleFmt{C}]} \mpPar
  \mpPart{\roleFmt{L}}{\mpP[\roleFmt{L}]} \mpPar
  \mpPart{\roleFmt{L}}{\mpH[\roleFmt{L}]} \mpPar
  \mpPart{\roleFmt{I}}{\mpP[\roleFmt{I}]} \mpPar
 \mpPart{\roleFmt{I}}{\mpH[\roleFmt{I}]}
\;\redCrash{\roleP}{\rolesR}\;
\mpPart{\roleFmt{C}}{\mpCrash} \mpPar
  \mpPart{\roleFmt{C}}{\mpQUnavail} \mpPar
  \mpPart{\roleFmt{L}}{\mpP[\roleFmt{L}]} \mpPar
  \mpPart{\roleFmt{L}}{\mpH[\roleFmt{L}]} \mpPar
  \mpPart{\roleFmt{I}}{\mpP[\roleFmt{I}]} \mpPar
 \mpPart{\roleFmt{I}}{\mpH[\roleFmt{I}]}$, where, by
 \inferrule{t-$\mpCrash$}, $\mpP[\roleFmt{C}]$ is typed by $\stStop$,  and
 $\mpH[\roleFmt{C}]$ is typed by $\stQUnavail$.
  \end{exa}

\subsection{Properties of Typed Sessions}
\label{sec:type:system:results}
We present the main properties of typed sessions: \emph{subject reduction} (\autoref{lem:sr}),
\emph{session fidelity} (\autoref{lem:sf}), \emph{deadlock-freedom} (\autoref{lem:session_deadlock_free}), and
\emph{liveness} (\autoref{lem:session_live}). \iftoggle{full}{The proofs for~\Cref{lem:sr,lem:sf,lem:session_deadlock_free,lem:session_live} are available in~\cref{sec:proof:typesystem}.}{}

\emph{Subject reduction} states that well-typedness of sessions is preserved by reduction, \ie a session that is governed by a global type continues to be
governed by a global type.

\begin{restatable}[Subject Reduction]{thm}{lemSubjectReduction}%
  \label{lem:sr}
  If
  \;$\gtWithCrashedRoles{\rolesC}{\gtG} \vdash \mpM$\;
  and
  \;$\mpM \mpMove[\rolesR] \mpMi$,\;
  then either
  \;$\gtWithCrashedRoles{\rolesC}{\gtG} \vdash \mpMi$,\;
  or there exists \;$\gtWithCrashedRoles{\rolesCi}{\gtGi}$\; such that
  \;$\gtWithCrashedRoles{\rolesC}{\gtG} \gtMove{\rolesR}
  \gtWithCrashedRoles{\rolesCi}{\gtGi}$\;
  and
  \;$\gtWithCrashedRoles{\rolesCi}{\gtGi} \vdash \mpMi$.
\end{restatable}
\begin{proof}
By induction on the derivation of $\mpM \mpMove[\rolesR] \mpMi$.
See Appendix~\ref{sec:proof:typesystem} for details.
\qedhere
\end{proof}

\emph{Session fidelity} states the opposite implication with regard to subject reduction:
sessions respect the progress of the governing global type.

\begin{restatable}[Session Fidelity]{thm}{lemSessionFidelity}
  \label{lem:sf}
  If
  \;$\gtWithCrashedRoles{\rolesC}{\gtG} \vdash \mpM$\; and
  \;$\gtWithCrashedRoles{\rolesC}{\gtG} \gtMove{\rolesR}$,\;
  then there exists $\mpMi$ and \;$\gtWithCrashedRoles{\rolesCi}{\gtGi}$\; such that
  \;$\gtWithCrashedRoles{\rolesC}{\gtG} \gtMove{\rolesR}
  \gtWithCrashedRoles{\rolesCi}{\gtGi}$,
  \;$\mpM \mpMoveStar[\rolesR] \mpMi$\; and
  \;$\gtWithCrashedRoles{\rolesCi}{\gtGi} \vdash \mpMi$.
\end{restatable}
\begin{proof}
By induction on the derivation of $\gtWithCrashedRoles{\rolesC}{\gtG} \gtMove{\rolesR}$.
See Appendix~\ref{sec:proof:typesystem} for details.
\qedhere
\end{proof}

Session \emph{deadlock-freedom} means that
the `successful' termination of a session may include crashed processes and their respective unavailable
incoming queues -- but reliable roles (which cannot crash) can only successfully terminate by reaching
inactive processes with empty incoming queues.
We formalise the definition of deadlock-free sessions in~\autoref{def:session_df} and show that a well-typed session is deadlock-free
in~\autoref{lem:session_deadlock_free}.

\begin{defi}[Deadlock-Free Sessions]
\label{def:session_df}
A session $\mpM$ is \emph{deadlock-free} iff
\;$\mpM \mpMoveStar[\rolesR] \mpMi \mpNotMove[\rolesR]$\;
 implies
either \;$\mpMi \Rrightarrow \mpPart\roleP\mpNil \mpPar \mpPart\roleP \mpQEmpty$ for some $\roleP$,\;
or
\;$\mpMi \Rrightarrow \prod_{i\in I} (\mpPart{\roleP[i]}{\mpCrash} \mpPar \mpPart{\roleP[i]}{\mpQUnavail})$ with $I \neq \emptyset$.
\end{defi}

\begin{restatable}[Session Deadlock-Freedom]{thm}{lemSessionDF}%
\label{lem:session_deadlock_free}
If
 \;$\gtWithCrashedRoles{\rolesC}{\gtG} \vdash \mpM$,\;
 then $\mpM$ is deadlock-free.
\end{restatable}
\begin{proof}
The result follows by~\autoref{lem:ext-proj}~(deadlock-freedom by projection),
\autoref{lem:sr}~(subject reduction), and~\autoref{lem:sf}~(session fidelity).
See Appendix~\ref{sec:proof:typesystem} for details.
\qedhere
\end{proof}

Finally, we show that well-typed sessions guarantee the property of \emph{liveness}:
a session is \emph{live} when all its input processes will be performed eventually,
and all its queued messages will be consumed eventually.
We formalise the definition of live sessions in~\autoref{def:session_live} and conclude by showing
that a well-typed session is live
in~\autoref{lem:session_live}.

\begin{defi}[Live Sessions]
\label{def:session_live}
A session $\mpM$ is \emph{live} iff
\;$\mpM \mpMoveStar[\rolesR] \mpMi
\Rrightarrow \mpPart\roleP\mpP \mpPar \mpPart\roleP \mpH[\roleP] \mpPar \mpMii$\;
 implies:
 \begin{enumerate}[label=\emph{(\roman*)}]
 \item if \;$\mpH[\roleP] = \left(\roleQ,\mpLab(\mpV)\right)\cdot\mpHi[\roleP]$, \;then $\exists \mpPi, \mpMiii:
 \mpMi \mpMoveStar[\rolesR] \mpPart\roleP\mpPi \mpPar \mpPart\roleP \mpHi[\roleP] \mpPar \mpMiii$;
 \item if \;$\mpP =  \sum_{i\in I}\procin{\roleQ}{\mpLab_i(\mpx_i)}{\mpP_i}$,  %
 \;then $\exists k \in I, \mpW, \mpHi[\roleP], \mpMiii:
 \mpMi \mpMoveStar[\rolesR] \mpPart\roleP \mpP_k\subst{\mpx_k}{\mpW} \mpPar \mpPart\roleP \mpHi[\roleP]
 \mpPar \mpMiii$.
 \end{enumerate}
\end{defi}

\begin{restatable}[Session Liveness]{thm}{lemSessionLive}%
\label{lem:session_live}
If
 \;$\gtWithCrashedRoles{\rolesC}{\gtG} \vdash \mpM$,\;
 then $\mpM$ is live.
\end{restatable}
\begin{proof}
The result follows by~\autoref{lem:ext-proj}~(liveness by projection), \autoref{lem:sr}~(subject reduction), and~\autoref{lem:sf}~(session fidelity).
See Appendix~\ref{sec:proof:typesystem} for details.
\qedhere
\end{proof}

\section{Case Study: Non-Blocking Atomic Commits}
\label{sec:casestudy}

\begin{figure}
  \[
    \begin{array}{rcl}
    \gtG &=& \gtRec{\gtRecVar}{
      \gtCommSingle{\roleC}{\roleL,\roleMhd,\roleMtl}{\gtLabProp}{}{
        \gtCommRawNB{\roleMhd}{\roleL}{
          \begin{cases}
          \gtCommChoice{\gtLabCommit}{}{
            \gtCommRawNB{\roleMtl}{\roleL}{
              \begin{cases}
              \gtCommChoice{\gtLabCommit}{}{
                \gtCommRawNB{\roleL}{\roleC}{
                  \begin{cases}
                  \gtCommChoice{\gtLabCommit}{}{\gtG[1]}
                  \\
                  \gtCommChoice{\gtLabVeto}{}{\gtG[2]}
                  \\
                  \gtCommChoice{\gtCrashLab}{}{\gtG[3]}
                  \end{cases}
                }
              }
              \\[7mm]
              \gtCommChoice{\gtLabVeto}{}{
                \gtCommRawNB{\roleL}{\roleC}{
                  \begin{cases}
                  \gtCommChoice{\gtLabVeto}{}{\gtG[2]}
                  \\
                  \gtCommChoice{\gtCrashLab}{}{\gtG[5]}
                  \end{cases}
                }
              }
              \end{cases}
            }
          }
          \\[15mm]
          \gtCommChoice{\gtLabVeto}{}{
            \gtCommRawNB{\roleMtl}{\roleL}{
              \begin{cases}
              \gtCommChoice{\gtLabCommit}{}{
                \gtCommRawNB{\roleL}{\roleC}{
                  \begin{cases}
                  \gtCommChoice{\gtLabVeto}{}{\gtG[2]}
                  \\
                  \gtCommChoice{\gtCrashLab}{}{\gtG[4]}
                  \end{cases}
                }
              }
              \\[5mm]
              \gtCommChoice{\gtLabVeto}{}{
                \gtCommRawNB{\roleL}{\roleC}{
                  \begin{cases}
                  \gtCommChoice{\gtLabVeto}{}{\gtG[2]}
                  \\
                  \gtCommChoice{\gtCrashLab}{}{\gtG[5]}
                  \end{cases}
                }
              }
              \end{cases}
            }
          }
          \end{cases}
        }
      }
    }
    \\[10mm]
    \gtG[1] &=& \gtCommSingle{\roleC}{\roleL,\roleMhd,\roleMtl}{\gtLabSucc}{}{
      \gtRecVar
    }
    \\
    \gtG[2] &=& \gtCommSingle{\roleC}{\roleL,\roleMhd,\roleMtl}{\gtLabAbort}{}{
      \gtRecVar
    }
    \\
    \gtG[3] &=& \gtCommSingle{\roleC}{\roleMhd}{\gtLabPromote}{}{
      \gtCommSingle{\roleC}{\roleMtl}{\gtLabNext}{}{
        \gtCommSingle{\roleMtl}{\roleMhd}{\gtLabCommit}{}{
          \gtCommSingle{\roleMhd}{\roleC}{\gtLabCommit}{}{
            \gtCommSingle{\roleC}{\roleMhd,\roleMtl}{\gtLabSucc}{}{
              \gtEnd
            }
          }
        }
      }
    }
    \\
    \gtG[4] &=& \gtCommSingle{\roleC}{\roleMhd}{\gtLabPromote}{}{
      \gtCommSingle{\roleC}{\roleMtl}{\gtLabNext}{}{
        \gtCommSingle{\roleMtl}{\roleMhd}{\gtLabCommit}{}{
          \gtCommSingle{\roleMhd}{\roleC}{\gtLabVeto}{}{
            \gtCommSingle{\roleC}{\roleMhd,\roleMtl}{\gtLabAbort}{}{
              \gtEnd
            }
          }
        }
      }
    }
    \\
    \gtG[5] &=& \gtCommSingle{\roleC}{\roleMhd}{\gtLabPromote}{}{
      \gtCommSingle{\roleC}{\roleMtl}{\gtLabNext}{}{
        \gtCommSingle{\roleMtl}{\roleMhd}{\gtLabVeto}{}{
          \gtCommSingle{\roleMhd}{\roleC}{\gtLabVeto}{}{
            \gtCommSingle{\roleC}{\roleMhd,\roleMtl}{\gtLabAbort}{}{
              \gtEnd
            }
          }
        }
      }
    }
    \end{array}
  \]
  \caption{A global type $\gtG$ for a variant of the NBAC abstraction.}
  \label{fig:case:gtype}
\end{figure}

\begin{figure}
  \setlength{\aboverulesep}{2.5mm}
 \setlength{\belowrulesep}{2.5mm}
  \[
    \begin{tabular}{>{$}r<{$}c>{$}l<{$}}
      \stT[\roleC] &=& \stRec{\stRecVar}{
        \stOut{\roleL,\roleMhd,\roleMtl}{\stLabProp}{}
          \stSeq \stExtSumNB{\roleL}{}{
            \begin{cases}
            \stChoice{\stLabCommit}{}
              \stSeq \stOut{\roleL,\roleMhd,\roleMtl}{\stLabSucc}{}
              \stSeq \stRecVar
            \\
            \stChoice{\stLabVeto}{}
              \stSeq \stOut{\roleL,\roleMhd,\roleMtl}{\stLabAbort}{}
              \stSeq \stRecVar
            \\
            \stChoice{\stCrashLab}{}
              \stSeq \stS[\roleC]
            \end{cases}
          }
      }
      \\[8mm]
      \stS[\roleC] &=& \stOut{\roleMhd}{\stLabPromote}{}
        \stSeq \stOut{\roleMtl}{\stLabNext}{}
        \stSeq \stExtSumNB{\roleMhd}{}{
          \begin{cases}
          \stChoice{\stLabCommit}{}
            \stSeq \stOut{\roleMhd,\roleMtl}{\stLabSucc}{}
            \stSeq \stEnd
          \\
          \stChoice{\stLabVeto}{}
            \stSeq \stOut{\roleMhd,\roleMtl}{\stLabAbort}{}
            \stSeq \stEnd
          \end{cases}
        }
      \\
      \midrule
      \stT[\roleL] &=& \stRec{\stRecVar}{
        \stInNB{\roleC}{\stLabProp}{}{
          \stExtSumNB{\roleMhd}{}{
            \begin{cases}
            \stChoice{\stLabCommit}{}
              \stSeq \stExtSumNB{\roleMtl}{}{
                \begin{cases}
                \stChoice{\stLabCommit}{}
                  \stSeq \stIntSumNB{\roleC}{}{
                    \begin{cases}
                    \stChoice{\stLabCommit}{}
                      \stSeq \stInNB{\roleC}{\stLabSucc}{}{\stRecVar}
                    \\
                    \stChoice{\stLabVeto}{}
                      \stSeq \stInNB{\roleC}{\stLabAbort}{}{\stRecVar}
                    \end{cases}
                  }
                \\
                \stChoice{\stLabVeto}{}
                  \stSeq \stOut{\roleC}{\stLabVeto}{}
                  \stSeq \stInNB{\roleC}{\stLabAbort}{}{\stRecVar}
                \end{cases}
              }
            \\[8mm]
            \stChoice{\stLabVeto}{}
              \stSeq \stExtSumNB{\roleMtl}{}{
                \begin{cases}
                \stChoice{\stLabCommit}{}
                  \stSeq \stOut{\roleC}{\stLabVeto}{}
                  \stSeq \stInNB{\roleC}{\stLabAbort}{}{\stRecVar}
                \\
                \stChoice{\stLabVeto}{}
                  \stSeq \stOut{\roleC}{\stLabVeto}{}
                  \stSeq \stInNB{\roleC}{\stLabAbort}{}{\stRecVar}
                \end{cases}
              }
            \end{cases}
          }
        }
      }
    \\\midrule
    \stT[\roleMhd] &=& \stRec{\stRecVar}{
      \stInNB{\roleC}{\stLabProp}{}{
        \stIntSumNB{\roleL}{}{
          \begin{cases}
          \stChoice{\stLabCommit}{}
            \stSeq \stExtSum{\roleC}{}{
              \stChoice{\stLabAbort}{} \stSeq \stRecVar
              ,\;
              \stChoice{\stLabSucc}{} \stSeq \stRecVar
              ,\;
              \stChoice{\stLabPromote}{}
                \stSeq \stS[\roleMhd]
            }
          \\
          \stChoice{\stLabVeto}{}
            \stSeq \stExtSum{\roleC}{}{
              \stChoice{\stLabAbort}{} \stSeq \stRecVar
              ,\;
              \stChoice{\stLabSucc}{} \stSeq \stRecVar
              ,\;
              \stChoice{\stLabPromote}{}
                \stSeq \stSi[\roleMhd]
            }
          \end{cases}
        }
      }
    }
    \\[6.5mm]
    \stS[\roleMhd] &=& \stExtSumNB{\roleMtl}{}{
      \begin{cases}
      \stChoice{\stLabCommit}{}
        \stSeq \stOut{\roleC}{\stLabCommit}{}
        \stSeq \stInNB{\roleC}{\stLabSucc}{}{\stEnd}
      \\
      \stChoice{\stLabVeto}{}
        \stSeq \stOut{\roleC}{\stLabVeto}{}
        \stSeq \stInNB{\roleC}{\stLabAbort}{}{\stEnd}
      \end{cases}
    }
    \\[6mm]
    \stSi[\roleMhd] &=& \stExtSumNB{\roleMtl}{}{
      \begin{cases}
      \stChoice{\stLabCommit}{}
        \stSeq \stOut{\roleC}{\stLabVeto}{}
        \stSeq \stInNB{\roleC}{\stLabAbort}{}{\stEnd}
      \\
      \stChoice{\stLabVeto}{}
        \stSeq \stOut{\roleC}{\stLabVeto}{}
        \stSeq \stInNB{\roleC}{\stLabAbort}{}{\stEnd}
      \end{cases}
    }
    \\\midrule
    \stT[\roleMtl] &=& \stRec{\stRecVar}{
      \stInNB{\roleC}{\stLabProp}{}{
        \stIntSumNB{\roleL}{}{
          \begin{cases}
          \stChoice{\stLabCommit}{}
            \stSeq \stExtSumNB{\roleC}{}{
              \begin{cases}
              \stChoice{\stLabAbort}{} \stSeq \stRecVar
              \\
              \stChoice{\stLabSucc}{} \stSeq \stRecVar
              \\
              \stChoice{\stLabNext}{}
                \stSeq \stOut{\roleMhd}{\stLabCommit}{}
                \stSeq \stExtSumNB{\roleC}{}{
                  \begin{cases}
                  \stChoice{\stLabAbort}{} \stSeq \stEnd
                  \\
                  \stChoice{\stLabSucc}{} \stSeq \stEnd
                  \end{cases}
                }
              \end{cases}
            }
          \\[11mm]
          \stChoice{\stLabVeto}{}
            \stSeq \stExtSumNB{\roleC}{}{
              \begin{cases}
              \stChoice{\stLabAbort}{} \stSeq \stRecVar
              \\
              \stChoice{\stLabNext}{}
                \stSeq \stOut{\roleMhd}{\stLabVeto}{}
                \stSeq \stInNB{\roleC}{\stLabAbort}{}{\stEnd}
              \end{cases}
            }
          \end{cases}
        }
      }
    }
    \end{tabular}
  \]
  \caption{Projected local types for the NBAC abstraction $\gtG$.} %
  \label{fig:case:ltype}
\end{figure}

We demonstrate our approach via the Non-Blocking Atomic Commits (NBAC)
abstraction~\cite{DBLP:books/daglib/0025983}, which enables consistent
application of transactions in distributed database systems. The abstraction is
straightforward: a given transaction is applied iff all data managers -- each
representing part of the database -- accept the transaction.
We present a variant of this abstraction via the global type $\gtG$ in
\cref{fig:case:gtype}, which involves a co-ordinator $\roleC$, two data managers
$\roleMhd, \roleMtl$, and one unreliable leading data manager $\roleL$.  
The
protocol exhibits failover behaviour, where $\roleMhd$ is the failover role for
$\roleL$. For clarity, we opt to eliminate the right-hand braces from global types.

The recursive protocol in \cref{fig:case:gtype} begins with the co-ordinator $\roleC$
broadcasting a transaction proposal, denoted $\gtLabProp$, to all three data
managers. For brevity, we omit payload types and write
$\gtCommSingle{\roleC}{\roleL,\roleMhd,\roleMtl}{\gtLabProp}{}{\gtGi}$ for
$\gtCommSingle{\roleC}{\roleL}{\gtLabProp}{}{\gtCommSingle{\roleC}{\roleMhd}{\gtLabProp}{}{\gtCommSingle{\roleC}{\roleMtl}{\gtLabProp}{}{\gtGi}}}$.
Data managers $\roleMhd$ and $\roleMtl$ send their votes to the leading data manager $\roleL$.
In cases where $\roleL$ does not crash, it sends $\gtLabVeto$ to $\roleC$ when
it receives a $\gtLabVeto$ion from its peers, or if $\roleL$ itself
$\gtLabVeto$s the transaction; $\roleC$ subsequently broadcasts an $\gtLabAbort$
message (in $\gtG[2]$) prior to the protocol recursing. In cases where $\roleL$,
$\roleMhd$, and $\roleMtl$ all vote to $\gtLabCommit$ the transaction, $\roleC$
broadcasts a success message (in $\gtG[1]$).
In cases where $\roleC$ detects that $\roleL$ has crashed, it promotes
$\roleMhd$ to leader. Once promoted, $\roleMhd$ receives a vote from $\roleMtl$,
and informs $\roleC$ of the data managers' decision. Should the transaction be
rejected (resp.\@ accepted), $\roleC$ broadcasts an abort (resp.\@ a success)
message to all live data managers, before the protocol terminates.
The global types $\gtG[3],\gtG[4]$, and $\gtG[5]$ define this crash handling behaviour.

All projected local types for the NBAC protocol $\gtG$ are listed in
\cref{fig:case:ltype}. 
The type obtained by projecting onto $\roleC$, $\stT[\roleC]$, is the sole local
type that contains a crash handling branch since $\roleL$ is the only unreliable
role, and both $\roleMhd$ and $\roleMtl$ do not receive messages from $\roleL$.
Its crash handling branch is defined by $\stS[\roleC]$. Similarly to the global
type above, broadcast operations are denoted as 
$\stOut{\roleP[1],\dots,\roleP[n]}{\stLab}{}$, and stand for 
$\stOut{\roleP[1]}{\stLab}{} \stSeq \dots \stSeq \stOut{\roleP[n]}{\stLab}{}.$ 
Additionally, we remove the right-hand braces from local types as well.
Projections for roles $\roleL$, $\roleMhd$, and $\roleMtl$ are as standard,
where labels $\stLabPromote$ and $\stLabNext$ trigger the crash handling
behaviour in $\roleMhd$ and $\roleMtl$ respectively.

\begin{figure}
  \[
    \begin{array}{rcl}
      \mpP[\roleC] &=& \mpRec{\mpX}{
        \procoutNoVal{\roleFmt{\roleL,\roleMhd,\roleMtl}}{\mpLabProp}{
          \sum\setenum{
            \begin{array}{l}
            \procin{\roleL}{\mpLabCommit}{
              \procoutNoVal{\roleFmt{\roleL,\roleMhd,\roleMtl}}{\mpLabSucc}{
                \mpX
              }
            }
            \\
            \procin{\roleL}{\mpLabVeto}{
              \procoutNoVal{\roleFmt{\roleL,\roleMhd,\roleMtl}}{\mpLabAbort}{
                \mpX
              }
            }
            \\
            \procin{\roleL}{\mpCrashLab}{
              \mpQ[\roleC]
            }
            \end{array}
          }
        }
      }
      \\[6mm]
      \mpQ[\roleC] &=& \procoutNoVal{\roleMhd}{\mpLabPromote}{
        \procoutNoVal{\roleMtl}{\mpLabNext}{
          \sum\setenum{
            \begin{array}{l}
            \procin{\roleMhd}{\mpLabCommit}{
              \procoutNoVal{\roleFmt{\roleMhd,\roleMtl}}{\mpLabSucc}{
                \mpNil
              }
            }
            \\
            \procin{\roleMhd}{\mpLabVeto}{
              \procoutNoVal{\roleFmt{\roleMhd,\roleMtl}}{\mpLabAbort}{
                \mpNil
              }
            }
            \end{array}
          }
        }
      }
      \\[6mm]
      \mpP[\roleMhd] &=& \mpRec{\mpX}{
        \procin{\roleC}{\mpLabProp}{
          \procoutNoVal{\roleL}{\mpLabCommit}{
            \sum\setenum{
              \begin{array}{l}
                \procin{\roleC}{\mpLabSucc}{\mpX}
                ,\;
                \procin{\roleC}{\mpLabAbort}{\mpX}
                ,\;
                \procin{\roleC}{\mpLabPromote}{\mpQ[\roleMhd]}
              \end{array}
            }
          }
        }
      }
      \\[2mm]
      \mpQ[\roleMhd] &=& \sum\setenum{
        \begin{array}{l}
          \procin{\roleMtl}{\mpLabCommit}{
            \procoutNoVal{\roleC}{\mpLabCommit}{
              \procin{\roleC}{\mpLabSucc}{\mpNil}
            }
          }
          \\
          \procin{\roleMtl}{\mpLabVeto}{
            \procoutNoVal{\roleC}{\mpLabVeto}{
              \procin{\roleC}{\mpLabAbort}{\mpNil}
            }
          }
        \end{array}
      }
    \end{array}
  \]
  \caption{Process definitions for roles $\roleC$ and $\roleMhd$ in $\gtG$.}
  \label{fig:case:prcss}
\end{figure}

Let $\stEnv; \qEnv$ be the $\{\roleC,\roleMhd,\roleMtl\}$-safe configuration,
where
\(
  \stEnv =
    \stEnvMap{\roleC}{\stT[\roleC]}
    \stEnvComp \stEnvMap{\roleL}{\stT[\roleL]}
    \stEnvComp \stEnvMap{\roleMhd}{\stT[\roleMhd]}
    \stEnvComp \stEnvMap{\roleMtl}{\stT[\roleMtl]}
\), 
and $\qEnv =  \qEnv[\stQEmpty]$.
The processes $\mpP[\roleC]$ and $\mpP[\roleMhd]$ in \cref{fig:case:prcss}, with
message queues $\mpH[\roleC] =  \mpH[\roleMhd] =
\mpQEmpty$ 
act as roles $\roleC$ and $\roleMhd$ in $\gtG$. The former,
$\mpP[\roleC]$, has type $\stEnvApp{\stEnv}{\roleC}$; the latter,
$\mpP[\roleMhd]$ will always accept a proposed transaction, and has type
$\stEnvApp{\stEnv}{\roleMhd}$. Processes $\mpP[\roleL]$ and $\mpP[\roleMtl]$ can
be defined similarly such that $\mpP[\roleL]$ has type
$\stEnvApp{\stEnv}{\roleL}$, and $\mpP[\roleMtl]$ has type
$\stEnvApp{\stEnv}{\roleMtl}$.
Since the local types in $\stEnv$ are derived via projection from $\gtG$, we
have the association
\(
  \stEnvAssoc{\gtWithCrashedRoles{\emptyset}{\gtG}}{
    \stEnv; \qEnv
  }{\setenum{\roleC, \roleMhd, \roleMtl}}
\). 
This allows us to assert that the session 
\(
  \mpPart{\roleC}{\mpP[\roleC]} 
  \mpPar \mpPart{\roleC}{\mpH[\roleC]} 
  \mpPar \mpPart{\roleL}{\mpP[\roleL]} 
  \mpPar \mpPart{\roleL}{\mpH[\roleL]} 
  \mpPar \mpPart{\roleMhd}{\mpP[\roleMhd]} 
  \mpPar \mpPart{\roleMhd}{\mpH[\roleMhd]}
  \mpPar \mpPart{\roleMtl}{\mpP[\roleMtl]} 
  \mpPar \mpPart{\roleMtl}{\mpH[\roleMtl]}
\), where  $\mpH[\roleL] =  \mpH[\roleMtl] = \mpQEmpty$, 
is governed by the global type $\gtWithCrashedRoles{\emptyset}{\gtG}$. 
Consequently, this session is both deadlock-free (by~\autoref{lem:session_deadlock_free}) 
and live (by~\autoref{lem:session_live}).

\section{Related Work}\label{sec:related}

We summarise related work on session types with failure handling. 
While most of the literature discussed includes delegation, our framework does not. 
However, this distinction does not impact the core comparison, which focuses on failure handling within session types.
The discussion starts with
the closest related 
work~\cite{OOPSLA21FaultTolerantMPST,FORTE22FaultTolerant,CONCUR22MPSTCrash,ESOP23MAGPi}
in which multiparty session types are extended to model crashes or failures. 
\paragraph{Multiparty Session Types with Failures} 
Peters \etal~\cite{FORTE22FaultTolerant} propose an MPST framework to
model fine-grained unreliability.
In their work, each transmission in a global type is parameterised by a
reliability annotation,
which can be unreliable (sender/receiver can crash, and messages can be
lost),
weakly reliable (sender/receiver can crash, messages are not lost), or
reliable (no crashes or message losses).
The design choice taken in our work roughly falls under weakly
reliable in their work, where a $\gtCrashLab$ handling branch (in their work, a
default branch) needs to be present to handle failures.
In our work, the reliability assumptions operate on a coarse level, but
nonetheless are \emph{consistent} within a given global type -- if a role
$\roleP$ is assumed reliable, $\roleP \in \rolesR$, then it does not crash for
the duration of the protocol, and vice versa.
Therefore, in a transmission $\roleP \to \roleQ$, our model allows \emph{one}
of the two roles to be unreliable, whereas their work does not permit the
`mixing' of reliability of sending and receiving roles. 

Viering \etal~\cite{OOPSLA21FaultTolerantMPST} utilise MPST as a guidance for
fault-tolerant distributed systems with recovery mechanisms.
Their framework includes various features, such as sub-sessions,
event-driven programming, dynamic role assignments, and most importantly
failure handling.
Our work handles unreliability in distributed programming, with the
following differences.
\begin{enumerate}%
  \item Failure detection \emph{assumptions} and \emph{models} are different:
    in our work, we assume a perfect failure detector, where all detected
    crashes are genuine.
    Their work uses a less strict assumption to allow \emph{false suspicions}.
    This difference subsequently gives rise to how failures are handled in both
    approaches:
    in our work, we use a special $\gtCrashLab$ handling branch in %
    global types to specify how the global protocol should progress after a
    crash has been detected.
    In contrast, they use a \emph{try-catch} construct in global types.
    In such a try-catch construct, crash detection within a sub-session
    $\gtG[1]$ is specified
    $\gtG[1]~\mathsf{with}~\roleP @ \roleQ~.~\gtG[2]$, where $\gtG[2]$ is a
    global type for failure handling (where $\roleP$ cannot occur), and
    $\roleQ$ is a \emph{monitor} that monitors $\roleP$ for possible failures.
    Moreover, the well-formedness condition (2) in~\cite[\S
    4.1]{OOPSLA21FaultTolerantMPST} requires the first message in $\gtG[2]$ to
    be a message broadcast of failure notification from the monitor $\roleQ$
    to all roles participating in the sub-session (except the crashed role
    $\roleP$).

    On this matter, we consider our framework more \emph{flexible} when
    detecting and handling crashes: every communication construct can have a
    crash handling branch %
    (when the sender is not assumed reliable), and the 
    failure broadcast is not necessary (failure detection only occurs when
    receiving from a crashed role).

  \item The \emph{merge} operators, used when projecting global types to obtain
    local types, are different: we use a more expressive \emph{full} merge
    operator~(\autoref{def:local-type-merge}), whereas they use a \emph{plain}
    merge operator, \ie requiring all continuations to project to the
    \emph{same} local type.

  \item  \emph{Reliability assumptions} are different: in our work, we
    support a range of assumptions from every role being unreliable, to totally
    reliable (as in the literature).
    In their work, they require \emph{at least one} reliable role, because
    they use a \emph{monitoring tree} for detecting crashes.
    Our work allows a role to detect the crash of its communication partner
    during reception, thus requiring neither such trees nor reliable roles.
\end{enumerate}

Barwell \etal~\cite{CONCUR22MPSTCrash} %
develop a theory of multiparty session types with crash-stop
failures. Their theory models crash-stop failures in the semantics of processes
and session
types, where the type system uses a model checker to validate type safety.
Our theory follows a similar model of crash-stop failures, but differs in the
following ways.
\begin{enumerate}%
  \item We model an asynchronous (message-passing) semantics, whereas they
    model a synchronous (rendez-vous) semantics.
    We focus on asynchronous systems, where a message can be buffered while in
    transit, since most of the interactions in the real distributed world are asynchronous.
  \item We follow a top-down methodology, beginning with protocol specification
    using global types, whereas they follow~\cite{POPL19LessIsMore} to analyse
    only local types.
    Our method dispenses with the need to use a model checker.
    More specifically, it is not feasible to model check asynchronous systems
    with buffers, since the model may be infinite~\cite[Appendix \S
    G]{POPL19LessIsMore}. 
    \end{enumerate}

Le Brun and Dardha~\cite{ESOP23MAGPi} 
follow a similar framework
to~\cite{CONCUR22MPSTCrash}.
They model an asynchronous semantics, and support more patterns of failure,
including message losses, delays, reordering, as well as link failures and
network partitioning.
However, their typing system suffers from its genericity, when type-level
properties become undecidable~\cite[\S 4.4]{ESOP23MAGPi}.
Our work uses global types for guidance, and recovers decidability of
properties at a small expense of expressivity.

\paragraph{Other Session Types with Failures} 
Most other session type works on modelling failures can be briefly categorised into
two: using \emph{affine types} or \emph{exceptions}~\cite{ECOOP22AffineMPST, LMCS18Affine,
DBLP:journals/pacmpl/FowlerLMD19}, %
and using \emph{coordinators} or \emph{supervision}%
~\cite{DBLP:conf/forte/AdameitPN17,ESOP18CrashHandling}.
The former adapts session types to an \emph{affine}
representation, in which endpoints may cease prematurely;
the latter, instead, are usually reliant on one or more \emph{reliable} processes that \emph{coordinate}
in the event of failure. 

\subparagraph{\emph{Affine Failure Handling}}
Mostrous and Vasconcelos~\cite{LMCS18Affine} 
first proposed the affine approach
to failure handling.
They present a $\pi$-calculus with affine binary sessions (\ie limited to two participants)
and concomitant reduction rules that represent a
minimal extension to standard (binary) session types.
Their extension is primarily comprised of a \emph{cancel operator}, which is
semantically similar to our crash construct: it represents a process
that has terminated early.
Besides these similarities, our work differs from \cite{LMCS18Affine} in several ways.
\begin{enumerate}%
  \item We address multiparty protocols and sessions rather than binary session types.
  \item In \cite{LMCS18Affine}, a cancel operator can have an arbitrary session type;
  consequently, crashes are not visible at the type level.
  Instead, we type crashed session endpoints with the special type $\stStop$,
  which lets us model crashes in the type semantics, and
  helps us in ensuring that a process implements its failure handling as expected in its (global or local) type.
  \item The reduction rules of \cite{LMCS18Affine}  do not permit a process to
    terminate early arbitrarily: cancellations must be raised explicitly by the
    programmer (or automatically by attempting to receive messages from crashed
    endpoints).
  \item Finally, cancellations in \cite{LMCS18Affine} may be caught and handled via a \emph{do-catch}
construct. This construct catches only the first cancellation and cannot be nested, thus providing little help in handling failure across
multiple roles. Our global and local types seamlessly support protocols where the failure of a role is detected (and handled) while handling the failure of another role.
\end{enumerate}

Fowler \etal~\cite{DBLP:journals/pacmpl/FowlerLMD19} 
present a concurrent
$\lambda$-calculus, \emph{EGV}, with asynchronous session-typed communication
and exception handling. %
Their approach is based on~\cite{LMCS18Affine}, and
therefore shares many of the same differences to our approach:
the use of the cancel operator and binary
session types, and the lack of a reduction rule enabling a process to crash
arbitrarily;
the cancel operator used in EGV takes an arbitrary session type
-- whereas we reflect the crashed status with the dedicated $\stStop$ type.
Similar to  our work, \cite{DBLP:journals/pacmpl/FowlerLMD19}
has asynchronous communication channels:
messages are queued when sent, and delivered at a later stage.

Lagaillardie \etal~\cite{ECOOP22AffineMPST} 
propose a framework of
\emph{affine} multiparty session types. 
They utilise the affine type
system to enforce 
that failures are handled.
In their system, a multiparty session can terminate prematurely.
While their theory can be used to model crash-stop failures, such failures
are not built into the semantics, so manual encoding of failures is necessary.
Moreover, there is no way to recover from a cancellation (\ie failure) besides
propagating the cancellation.
In our work, we provide the ability to follow a \emph{different} protocol
when a crash is detected, which gives rise to more flexibility and
expressivity.

There are some other works adopting \emph{exception-based} approaches, \eg 
Capecchi \etal~\cite{DBLP:journals/mscs/CapecchiGY16} and
Carbone \etal~\cite{DBLP:conf/concur/CarboneHY08}
model failure
at the \emph{application} level (similar to~\cite{LMCS18Affine,DBLP:journals/pacmpl/FowlerLMD19})
by equipping processes with constructs to actively
produce or catch failures. Conversely, our approach models
arbitrary failures (such as hardware failures).

\subparagraph{\emph{Coordinator Model Failure Handling}}
Coordinator model approaches find their roots in work by
Demangeon \etal~\cite{DBLP:journals/fmsd/DemangeonHHNY15}. 
In their model, global types are extended by \emph{interrupt blocks}
wherein one role may interrupt the protocol.
Interruptions are broadcast to all active roles within the block, and, once
interrupted, the roles follow a continuation specified separately.
Although their original intention was primarily to interrupt streaming behaviour, 
interrupt blocks (or similar constructs) have been used
to model crashes and failure handling in~\cite{DBLP:conf/forte/AdameitPN17,ESOP18CrashHandling}.

Adameit \etal~\cite{DBLP:conf/forte/AdameitPN17} 
extend the standard MPST syntax with  \emph{optional blocks}, representing regions of a protocol that are susceptible to communication failures. 
In their approach, if a process $\mpP$ expects a value from an optional block which fails, then a
\emph{default value} is provided to $\mpP$, so $\mpP$ can continue running.
This ensures termination and deadlock-freedom.
Although this approach does not feature an explicit reliable coordinator
process, we describe it here due to the inherent coordination required for multiple processes to start and end an optional block. %
Our approach differs in three key ways.
\begin{enumerate}%
  \item We model crash-stop process failures instead of impermanent link failures.
  \item We extend the semantics of communications in lieu of introducing a
    new syntactic construct to enclose the potentially crashing regions of a protocol --
    our global type projections and typing context safety ensure that crash detection is performed at every pertinent communication point.
  \item We allow crashes to significantly affect the evolution of protocols:
    our global and local types can have crash detection branches specifying
    significantly different behaviours \wrt non-crashing executions.
    Conversely, the approach in~\cite{DBLP:conf/forte/AdameitPN17} does not
    discriminate between the presence and absence of failures: both have the
    same protocol in the optional block's continuation. %
\end{enumerate}

Viering \etal~\cite{ESOP18CrashHandling} 
similarly extend the standard
global type syntax with a \emph{try-handle} construct, which is facilitated by the
presence of a reliable coordinator process, and via a construct to specify
reliable processes.
When the coordinator detects a failure, it broadcasts notifications to all remaining live processes;
then, the protocol proceeds according to the failure handling continuation
specified as part of the try-handle construct.
Our approach and \cite{ESOP18CrashHandling} share several modelling choices: crash-stop
semantics, perfect links, and the possibility of specifying reliable processes.
However, unlike \cite{ESOP18CrashHandling}, our approach does \emph{not} depend on a reliable coordinator that broadcasts failure notifications: all roles in a protocol can be unreliable, all processes may crash.

Besides the differences discussed above, we decided not to adopt coordinator processes nor failure broadcasts in order to avoid their inherent drawbacks.
The use of a coordinator requires additional run-time resources and increases the overall complexity of a distributed system.  Furthermore, the broadcasting of failure notifications introduces an effective synchronisation point for all roles, with additional overheads. Such synchronisation points may also make it harder to extend the theory to support scenarios
with unreliable communication.

\subparagraph{\emph{Alternative Session Types with Failure Handling}} 
Caires and P{\'{e}}rez~\cite{DBLP:conf/esop/CairesP17} present a linearly-typed
calculus that
supports non-determinism and process failures.
They present a core calculus with binary session types (based on linear logic), 
extended with constructs permitting non-determinism and control effects (\eg exceptions).
Failure is simulated via a non-deterministic choice operator.
The main difference is that we model arbitrary failures, and support
multiparty sessions.

Neykova and Yoshida~\cite{NY2017Recovery} build upon the \emph{Let it Crash} 
and supervisor-based design of
the \Erlang programming language.
Using an MPST specification, they build a dependency graph
between the processes interacting in a system; in case of failure, 
this information is then used by supervisor processes
to determine which subset of processes should be restarted, and which messages
should be re-sent.
As in the coordinator model approaches above, \Erlang supervisors result in
additional overhead compared with our approach.
Furthermore, as in the affine session types
approach~\cite{LMCS18Affine,DBLP:journals/pacmpl/FowlerLMD19}, failure states
are not reflected in the type system. Instead, there is an underlying
assumption that the supervisors will restart (a subset of) the system
until it terminates successfully.

Chen \etal~\cite{DBLP:conf/forte/ChenVBZE16} handle partial failures in
MPST by transforming programs to include synchronisation
points, at which failures can be detected and handled.
Our approach is less rigid due to the handling of
failures: synchronisation points are not needed (unless required by the
protocol), and the original protocols do not need to meet certain criteria to
permit their transformation.

\section{Conclusion and Future Work}
\label{sec:conclusion}
To overcome the challenge of accounting for failure handling in 
distributed systems using session types, 
we present an asynchronous multiparty session type 
framework with the ability to model and handle
crash-stop failures, \eg caused by hardware faults. 
In our session calculus, processes may crash
arbitrarily, and other processes can detect their crash and handle the event. 
Additionally, our system supports the specification of \emph{optional
reliability assumptions}:
by influencing type-checking, they allow our approach
to cover the spectrum between idealised scenarios where no process ever fails
(as typically assumed in standard session type works),
and more realistic scenarios
where any process can fail at any time. 
With only minimal changes to 
the syntax of standard asynchronous MPST in our theory, 
desirable global type properties such as deadlock-freedom, protocol conformance, 
and liveness are preserved by construction in typed processes, 
even in the presence of crashes.

This work is a new step towards modelling and handling real-world failures 
by utilising  session types, 
effectively bridging the gap between session type theory and 
applications. 
As future work, we plan to study different crash models (\eg crash-recover) and
failures of other components (\eg link failures). 
These further steps would lead us to our longer term goal, \ie to eventually
model and type-check implementations of well-known consensus algorithms used
in large-scale distributed systems.

\bigskip
\noindent
\emph{Acknowledgments}. 
{\small 
We thank the reviewers for their detailed and helpful comments. 
Additionally, we thank Yoshinori Kamegai for highlighting an issue related to structural congruence.
The work is partially supported by
  EPSRC grants EP/T006544/2, EP/K011715/1,
EP/K034413/1, EP/L00058X/1, EP/N027833/2,
EP/N028201/1, EP/T014709/2, EP/V000462/1, 
EP/X015955/1,  EP/Y005244/1, NCSS/EPSRC VeTSS, Horizon EU TaRDIS 101093006, 
Advanced Research and Invention Agency (ARIA) Safeguarded AI, and a grant from the Simons Foundation.}

\bibliographystyle{alphaurl}
\bibliography{references}

\newpage
\appendix
\section*{Table of Contents for Appendix}
\startcontents[sections]
\printcontents[sections]{l}{1}{\setcounter{tocdepth}{2}}
\section{Proofs for Section \ref{sec:gtype}}
\label{sec:app:types}

With regard to recursive global types, we define their \emph{unfolding} as
$\unfoldOne{\gtRec{\gtRecVar}{\gtG}} =
\unfoldOne{\gtG\subst{\gtRecVar}{\gtRec{\gtRecVar}{\gtG}}}$, and
$\unfoldOne{\gtG} = \gtG$ otherwise.
A recursive type $\gtRec{\gtRecVar}{\gtG}$ must be guarded (or
contractive), \ie the unfolding leads to a progressive prefix, \eg a
transmission.
Unguarded types, such as $\gtRec{\gtRecVar}{\gtRecVar}$ and
$\gtRec{\gtRecVar}{\gtRec{\gtRecVari}{\gtRecVar}}$, are excluded.
Similar definitions and requirements apply for local types.

\subsection{Live and Crashed Roles}
\begin{lem}
\label{lem:gtype:crash-remove-role}
  If \;$\roleP \in \gtRoles{\gtG}$, \;then \;$\roleP \notin
  \gtRoles{\gtCrashRole{\gtG}{\roleP}}$ and
  $\gtRoles{\gtCrashRole{\gtG}{\roleP}} \subseteq \gtRoles{\gtG}$.
\end{lem}
\begin{proof}
  By induction on \autoref{def:gtype:remove-role}. We detail interesting cases
  here:
\begin{enumerate}[leftmargin=*]
  \item
\[
  \textstyle
  \gtRoles{
    \gtCrashRole{
      (\gtCommSmall{\roleP}{\roleQ}{i \in I}{\gtLab[i]}{\tyGround[i]}{\gtG[i]})
    }{\roleP}
  }
  =
  \gtRoles{
    \gtCommTransit{\rolePCrashed}{\roleQ}{i \in I}{\gtLab[i]}{\tyGround[i]}{
      (\gtCrashRole{\gtG[i]}{\roleP})
    }{j}
  }
  =
  \setenum{\roleQ}
  \cup
  \bigcup\limits_{i \in I}{\gtRoles{\gtCrashRole{\gtG[i]}{\roleP}}}.
\]
\noindent
The required result follows by inductive hypothesis that $\roleP \notin
\gtRoles{\gtCrashRole{\gtG[i]}{\roleP}}$, and \linebreak
$\gtRoles{\gtCrashRole{\gtG[i]}{\roleP}} \subseteq \gtRoles{\gtG[i]}$.

  \item
\[
  \gtRoles{
    \gtCrashRole{
      (\gtCommTransit{\rolePCrashed}{\roleQ}{i \in
      I}{\gtLab[i]}{\tyGround[i]}{\gtG[i]}{j})
    }{\roleQ}
  }
  =
  \gtRoles{
    \gtCrashRole{\gtG[j]}{\roleQ}
  }
\]
The required result follows by inductive hypothesis that $\roleQ \notin
\gtRoles{\gtCrashRole{\gtG[j]}{\roleQ}}$, and \linebreak
$\gtRoles{\gtCrashRole{\gtG[j]}{\roleQ}} \subseteq \gtRoles{\gtG[j]}$.

\end{enumerate}
The rest of the cases are similar or straightforward.
\qedhere
\end{proof}

\begin{lem}
\label{lem:gtype:crashed-crash-remove-role}
  If \;$\roleP \in \gtRoles{\gtG}$, \;then\; %
  $\gtRolesCrashed{\gtCrashRole{\gtG}{\roleP}} \setminus \setenum{\roleP} \subseteq \gtRolesCrashed{\gtG}$.
\end{lem}
\begin{proof}
By induction on \autoref{def:gtype:remove-role}. We detail interesting cases here:
\begin{enumerate}[leftmargin=*]
 \item
\[
  \textstyle
  \gtRolesCrashed{
  \gtCrashRole{
      (\gtCommSmall{\roleP}{\roleQ}{i \in I}{\gtLab[i]}{\tyGround[i]}{\gtG[i]})
  }{\roleP}
  }
 =
 \gtRolesCrashed{
  \gtCommTransit{\rolePCrashed}{\roleQ}{i \in I}{\gtLab[i]}{\tyGround[i]}{
    (\gtCrashRole{\gtG[i]}{\roleP})
   }{j}
 }
  =
  \bigcup\limits_{i \in I}{\gtRolesCrashed{\gtCrashRole{\gtG[i]}{\roleP}}}.
\]
\noindent
The required result follows by inductive hypothesis that
$\gtRolesCrashed{\gtCrashRole{\gtG[i]}{\roleP}} \setminus \setenum{\roleP} \subseteq \gtRolesCrashed{\gtG[i]}$.
 \item
\[
  \gtRolesCrashed{
   \gtCrashRole{
         (\gtCommTransit{\rolePCrashed}{\roleQ}{i \in
      I}{\gtLab[i]}{\tyGround[i]}{\gtG[i]}{j})
    }{\roleQ}
 }
  =
  \gtRolesCrashed{
    \gtCrashRole{\gtG[j]}{\roleQ}
  }
\]
The required result follows by inductive hypothesis that
$\gtRolesCrashed{\gtCrashRole{\gtG[j]}{\roleQ}} \setminus \setenum{\roleQ} \subseteq \gtRolesCrashed{\gtG[j]}$.
\end{enumerate}
The rest of the cases are similar or straightforward.
\qedhere
\end{proof}

\begin{lem}
\label{lem:in-roles-not-end}
If \;$\gtG \neq \gtRec{\gtRecVar}{\gtGi}$ and\;
$\roleP \in  \gtRoles{\gtG}$, then
\;$\gtProj[\rolesR]{\gtG}{\roleP} \neq \stEnd$.
\end{lem}
\begin{proof}
We know that $\gtG \neq \gtEnd$; otherwise, we may have $\gtRoles{\gtG} = \emptyset$, a contradiction
to $\roleP \in  \gtRoles{\gtG}$. By induction on the structure of $\gtG$:
\begin{itemize}[leftmargin=*]
\item Case $\gtG = \gtComm{\roleQ}{\roleR}{i \in I}{\gtLab[i]}{\tyGround[i]}{\gtG[i]}$:

We perform case analysis on $\roleP$:
\begin{itemize}[leftmargin=*]
\item $\roleP = \roleQ$: we have  $\gtProj[\rolesR]{\gtG}{\roleP} =
 \stIntSum{\roleR}{i \in \setcomp{j \in I}{\stFmt{\stLab[j]} \neq \stCrashLab}}{ %
        \stChoice{\stLab[i]}{\tyGround[i]} \stSeq (\gtProj[\rolesR]{\gtG[i]}{\roleP})%
      } \neq \stEnd$.
 \item $\roleP = \roleR$: we have $\gtProj[\rolesR]{\gtG}{\roleP} =
 \stExtSum{\roleQ}{i \in I}{%
        \stChoice{\stLab[i]}{\tyGround[i]} \stSeq (\gtProj[\rolesR]{\gtG[i]}{\roleP})%
      }
      \neq \stEnd$.
  \item $\roleP \neq \roleQ$ and $\roleP \neq \roleR$: we have
  $\gtProj[\rolesR]{\gtG}{\roleP} =
   \stMerge{i \in I}{\gtProj[\rolesR]{\gtG[i]}{\roleP}}$.
   Since $\roleP \in \gtRoles{\gtG}$, $\roleP \neq \roleQ$, and $\roleP \neq \roleR$,
   there exists $j \in I$ such that $\roleP \in \gtRoles{\gtG[j]}$. Then, by applying inductive hypothesis,
   $\gtProj[\rolesR]{\gtG[j]}{\roleP} \neq \gtEnd$, and therefore, we have
   $\gtProj[\rolesR]{\gtG}{\roleP} =
    \stMerge{i \in I}{\gtProj[\rolesR]{\gtG[i]}{\roleP}} = \linebreak
    \gtProj[\rolesR]{\gtG[j]}{\roleP} \,\stBinMerge\,
    \stMerge{i \in I \setminus \{j\} }{\gtProj[\rolesR]{\gtG[i]}{\roleP}} \neq \stEnd$.
  \end{itemize}
  \item Case $\gtG =  \gtCommTransit{\roleQ}{\roleR}{i \in
          I}{\gtLab[i]}{\tyGround[i]}{\gtG[i]}{j}$:

          We perform case analysis on $\roleP$:
\begin{itemize}[leftmargin=*]
  \item $\roleP \neq \roleQ$ and $\roleP \neq \roleR$: we have
  $\gtProj[\rolesR]{\gtG}{\roleP} =
   \stMerge{i \in I}{\gtProj[\rolesR]{\gtG[i]}{\roleP}}$.
   Since $\roleP \in \gtRoles{\gtG}$, $\roleP \neq \roleQ$, and $\roleP \neq \roleR$,
   there exists $k \in I$ such that $\roleP \in \gtRoles{\gtG[k]}$. Then, by applying inductive hypothesis,
   $\gtProj[\rolesR]{\gtG[k]}{\roleP} \neq \gtEnd$, and therefore, we have
   $\gtProj[\rolesR]{\gtG}{\roleP} =
    \stMerge{i \in I}{\gtProj[\rolesR]{\gtG[i]}{\roleP}} = \linebreak
    \gtProj[\rolesR]{\gtG[k]}{\roleP} \,\stBinMerge\,
    \stMerge{i \in I \setminus \{k\} }{\gtProj[\rolesR]{\gtG[i]}{\roleP}} \neq \stEnd$, as desired. Meanwhile, we obtain that
    $\forall l \in I: \gtProj[\rolesR]{\gtG[l]}{\roleP} \neq \stEnd$.

\item $\roleP = \roleQ$: we have  $\gtProj[\rolesR]{\gtG}{\roleP} = \gtProj[\rolesR]{\gtG[j]}{\roleP}$, which follows that
$\gtProj[\rolesR]{\gtG}{\roleP} \neq \stEnd$ by the fact that $\forall l \in I: \gtProj[\rolesR]{\gtG[l]}{\roleP} \neq \stEnd$.

\item $\roleP = \roleR$: we have $\gtProj[\rolesR]{\gtG}{\roleP} =
 \stExtSum{\roleQ}{i \in I}{%
        \stChoice{\stLab[i]}{\tyGround[i]} \stSeq (\gtProj[\rolesR]{\gtG[i]}{\roleP})%
      }
      \neq \stEnd$.
  \end{itemize}
  \end{itemize}
Other cases are similar.
  \qedhere
\end{proof}

\begin{lem}
\label{lem:not-in-roles-end}
If \;$\roleP \notin  \gtRoles{\gtG}$ and\;
$\roleP \notin \gtRolesCrashed{\gtG}$, then
\;$\gtProj[\rolesR]{\gtG}{\roleP} = \stEnd$.
\end{lem}
\begin{proof}
By induction on the structure of $\gtG$:
\begin{itemize}[leftmargin=*]
     \item Case $\gtG = \gtComm{\roleQ}{\roleR}{i \in I}{\gtLab[i]}{\tyGround[i]}{\gtG[i]}$:
     since $\roleP \notin \gtRoles{\gtG}$ and $\roleP \notin \gtRolesCrashed{\gtG}$, we have $\roleP \neq \roleQ$,
     $\roleP \neq \roleR$, and for all $i \in I$, $\roleP \notin \gtRoles{\gtG[i]}$ and
     $\roleP \notin \gtRolesCrashed{\gtG[i]}$ by~\autoref{def:active_crashed_roles}.
     Thus, $\gtProj[\rolesR]{\gtG}{\roleP} = \stMerge{i \in I}{\gtProj[\rolesR]{\gtG[i]}{\roleP}} = \stEnd$ by applying inductive hypothesis
     and $ \stEnd \,\stBinMerge\, \stEnd%
    \,=\,%
    \stEnd$.
         \item Case $\gtG = \gtRec{\gtRecVar}{\gtGi}$: since $\roleP \notin \gtRoles{\gtG}$ and $\roleP \notin \gtRolesCrashed{\gtG}$, we have
      $\roleP \notin \gtRoles{\gtGi}$ and $\roleP \notin \gtRolesCrashed{\gtGi}$ by~\autoref{def:active_crashed_roles}.
      We have two further subcases to consider:
       \begin{itemize}[leftmargin=*]
       \item If $\fv{\gtRec{\gtRecVar}{\gtGi}} \neq \emptyset$, we have $\gtProj[\rolesR]{\gtG}{\roleP} =
       \stRec{\stRecVar}{(\gtProj[\rolesR]{\gtGi}{\roleP})} =
       \stRec{\stRecVar}{\stEnd} = \stEnd$ by applying inductive hypothesis.
       \item Otherwise, we have $\gtProj[\rolesR]{\gtG}{\roleP} = \stEnd$ immediately.
       \end{itemize}
\end{itemize}
Other cases are similar or trivial.
\qedhere
\end{proof}

\subsection{Subtyping}
\label{sec:app:subtyping}

\begin{lem}[Subtyping is Reflexive]
\label{lem:reflexive-subtyping}
For any closed, well-guarded local type $\stT$, $\stT \stSub \stT$ holds.
\end{lem}
\begin{proof}
We construct a relation $R = \setenum{(\stT, \stT)}$. 
It is straightforward 
that $R$ satisfies all clauses of~\autoref{def:subtyping}. 
 Hence, since $\stSub$ is the largest relation satisfying such rules, $R \subseteq \stSub$. 
\qedhere 
\end{proof}

\begin{lem}[Subtyping is Transitive]
\label{transitive-subtyping}
  For any closed, well-guarded local type $\stS$, $\stT$, $\stU$,
  if $\stS \stSub \stT$ and $\stT \stSub \stU$ hold, then $\stS \stSub \stU$
  holds.
\end{lem}
\begin{proof}
We construct a relation 
$R = \setcomp{(\stS, \stU)}{\exists \stT \text{ such that } \stS \stSub \stT \text{ and } \stT \stSub \stU}$. 
By showing that $R$ satisfies all clauses of~\autoref{def:subtyping}, it follows that $R \subseteq \stSub$. 
\qedhere
\end{proof}

\propSubtyping*
\begin{proof}
Immediate consequence of \autoref{lem:reflexive-subtyping} and~\autoref{transitive-subtyping}. 
\qedhere 
\end{proof}

\begin{lem}
\label{lem:unfold-subtyping}
  For any closed, well-guarded local type $\stT$,
  \begin{enumerate*}
    \item $\unfoldOne{\stT} \stSub \stT$; and
    \item $\stT \stSub \unfoldOne{\stT}$.
  \end{enumerate*}
\end{lem}
 \begin{proof}
 \begin{enumerate*}
     \item If $\stT = \stRec{\stRecVar}{\stTi}$,  $\unfoldOne{\stT} \stSub \stT$ holds by $\inferrule{\iruleStSubRecR}$. Otherwise,
      by \autoref{lem:reflexive-subtyping}.
    \item If $\stT = \stRec{\stRecVar}{\stTi}$, $\stT \stSub \unfoldOne{\stT}$ holds by $\inferrule{\iruleStSubRecL}$.  Otherwise,
       by \autoref{lem:reflexive-subtyping}. 
       \qedhere 
  \end{enumerate*}
 \end{proof}

\mergeSubtyping*
\begin{proof}
 By constructing a relation 
$R = \setcomp{(\stMerge{i \in I}{\stT[i]}, \stT[j])}{j \in I}$, and showing 
that $R$ satisfies all clauses of~\autoref{def:subtyping}. 
\qedhere 
\end{proof}

\mergeUpperSubtyping*
\begin{proof}
 By constructing a relation $R = \setenum{(\stS, \stMerge{i \in I}{\stT[i]})}$, and showing 
that $R$ satisfies all clauses of~\autoref{def:subtyping}. %
\qedhere 
\end{proof}

\submergeSubtyping*
\begin{proof}
 By constructing a relation $R = \setenum{(\stMerge{i \in I}{\stS[i]}, \stMerge{i \in I}{\stT[i]})}$, and showing 
that $R$ satisfies all clauses of~\autoref{def:subtyping}. %
\qedhere
\end{proof}

\begin{lem}
\label{lem:proj-non-crashing-role-preserve}
  If \; $\roleP, \roleQ \in \gtRoles{\gtG}$ \; with \; $\roleP \neq \roleQ$ \;
  and \; $\roleQ \notin \rolesR$, \;
  then \;
  $
    \gtProj[\rolesR]{\gtG}{\roleP}
    \stSub
    \gtProj[\rolesR]{(\gtCrashRole{\gtG}{\roleQ})}{\roleP}
  $.
\end{lem}
\begin{proof}
  We construct a relation 
  $
  R = \setcomp{
    (\gtProj[\rolesR]{\gtG}{\roleP}
    ,
    \gtProj[\rolesR]{(\gtCrashRole{\gtG}{\roleQ})}{\roleP})
  }{
    \roleP, \roleQ \in \roleSet, \roleP \neq \roleQ, \roleQ \in \rolesR
  }$, and show that $R \subseteq \;\stSub$. 
  Consider all possible shapes of $\gtG$: 
  \begin{itemize}[leftmargin=*]
      \item Case $\gtG = \gtComm{\roleP}{\roleQ}{i \in I}{\gtLab[i]}{\tyGround[i]}{\gtG[i]}$:

      We perform case analysis on the role being projected upon:
      \begin{itemize}[leftmargin=*]
        \item
          On LHS,
          we have $\gtProj[\rolesR]{\gtG}{\roleP} =
          \stIntSum{\roleQ}{i \in \setcomp{j \in I}{\stFmt{\stLab[j]} \neq \stCrashLab}}{ %
            \stChoice{\stLab[i]}{\tyGround[i]} \stSeq (\gtProj[\rolesR]{\gtG[i]}{\roleP})%
          }%
          $.

          On RHS, we perform case analysis on the role being removed:

          \begin{enumerate}[leftmargin=*]
            \item we have
          $\gtCrashRole{\gtG}{\roleQ} =
          \gtComm{\roleP}{\roleQCrashed}{i \in I}{\gtLab[i]}{\tyGround[i]}{
            (\gtCrashRole{\gtG[i]}{\roleQ})
          }
          $, and thus 
          $\gtProj[\rolesR]{(\gtCrashRole{\gtG}{\roleQ})}{\roleP} = $ \linebreak
          $\stIntSum{\roleQ}{i \in \setcomp{j \in I}{\stFmt{\stLab[j]} \neq \stCrashLab}}{ %
            \stChoice{\stLab[i]}{\tyGround[i]} \stSeq
            (\gtProj[\rolesR]{(\gtCrashRole{\gtG[i]}{\roleQ})}{\roleP})%
          }%
          $,
          apply \inferrule{\iruleStSubOut} and coinductive
          hypothesis.

            \item ($\roleR \neq \roleQ$) we have
          $\gtCrashRole{\gtG}{\roleR} =
          \gtComm{\roleP}{\roleQ}{i \in I}{\gtLab[i]}{\tyGround[i]}{
            (\gtCrashRole{\gtG[i]}{\roleR})
          }
          $, and thus 
          $
          \gtProj[\rolesR]{(\gtCrashRole{\gtG}{\roleR})}{\roleP} = $ \linebreak
          $\stIntSum{\roleQ}{i \in \setcomp{j \in I}{\stFmt{\stLab[j]} \neq \stCrashLab}}{ %
            \stChoice{\stLab[i]}{\tyGround[i]} \stSeq
            (\gtProj[\rolesR]{(\gtCrashRole{\gtG[i]}{\roleR})}{\roleP})%
          }
          $, 
          apply $\inferrule{\iruleStSubOut}$ and coinductive
          hypothesis.
          \end{enumerate}
        \item
          On LHS,
          we have $\gtProj[\rolesR]{\gtG}{\roleQ} =
          \stExtSum{\roleP}{i \in I}{%
            \stChoice{\stLab[i]}{\tyGround[i]} \stSeq (\gtProj[\rolesR]{\gtG[i]}{\roleQ})%
          }%
          $.

          On RHS, we perform case analysis on the role being removed:
          \begin{enumerate}[leftmargin=*]
          \item we have $
          \gtCrashRole{\gtG}{\roleP} =
          \gtCommTransit{\rolePCrashed}{\roleQ}{i \in I}{\gtLab[i]}{\tyGround[i]}{
            (\gtCrashRole{\gtG[i]}{\roleP})
          }{j}
          $, and thus 
          $
          \gtProj[\rolesR]{(\gtCrashRole{\gtG}{\roleP})}{\roleQ} = $ \linebreak
          $\stExtSum{\roleP}{i \in I}{
            \stChoice{\stLab[i]}{\tyGround[i]} \stSeq
            (\gtProj[\rolesR]{(\gtCrashRole{\gtG[i]}{\roleP})}{\roleQ})%
          }
          $,
          apply \inferrule{\iruleStSubIn} and coinductive
          hypothesis.

          \item ($\roleR \neq \roleP$) we have $
          \gtCrashRole{\gtG}{\roleR} =
          \gtComm{\roleP}{\roleQ}{i \in I}{\gtLab[i]}{\tyGround[i]}{
            (\gtCrashRole{\gtG[i]}{\roleR})
          }
          $, and thus 
          $
          \gtProj[\rolesR]{(\gtCrashRole{\gtG}{\roleR})}{\roleQ} = $ \linebreak 
          $\stExtSum{\roleP}{i \in I}{
            \stChoice{\stLab[i]}{\tyGround[i]} \stSeq
            (\gtProj[\rolesR]{(\gtCrashRole{\gtG[i]}{\roleR})}{\roleQ})%
          }
          $,
          apply $\inferrule{\iruleStSubIn}$ and coinductive
          hypothesis.
          \end{enumerate}
        \item ($\roleR \notin \setenum{\roleP, \roleQ}$)
          On LHS, we have $
          \gtProj[\rolesR]{\gtG}{\roleR} =
          \stMerge{i \in I}{\gtProj[\rolesR]{\gtG[i]}{\roleR}}%
          $.

          On RHS, we perform case analysis on the role being removed:

          \begin{enumerate}[leftmargin=*]
          \item we have $
          \gtCrashRole{\gtG}{\roleP} =
          \gtCommTransit{\rolePCrashed}{\roleQ}{i \in I}{\gtLab[i]}{\tyGround[i]}{
            (\gtCrashRole{\gtG[i]}{\roleP})
          }{j}
          $, and thus 
          $
          \gtProj[\rolesR]{(\gtCrashRole{\gtG}{\roleP})}{\roleR} = $ \linebreak
          $\stMerge{i \in I}{
            (\gtProj[\rolesR]{(\gtCrashRole{\gtG[i]}{\roleP})}{\roleR})%
          }
          $,
          apply \autoref{lem:subtype:merge-subty} and coinductive
          hypothesis.

          \item we have
          $\gtCrashRole{\gtG}{\roleQ} =
          \gtComm{\roleP}{\roleQCrashed}{i \in I}{\gtLab[i]}{\tyGround[i]}{
            (\gtCrashRole{\gtG[i]}{\roleQ})
          }
          $, and thus $
          \gtProj[\rolesR]{(\gtCrashRole{\gtG}{\roleQ})}{\roleR} = $ \linebreak 
          $\stMerge{i \in I}{
            (\gtProj[\rolesR]{(\gtCrashRole{\gtG[i]}{\roleQ})}{\roleR})%
          }
          $,
          apply \autoref{lem:subtype:merge-subty} and coinductive
          hypothesis.

          \item ($\roleS \notin \setenum{\roleP, \roleQ, \roleR}$) we have $
          \gtCrashRole{\gtG}{\roleS} =
          \gtComm{\roleP}{\roleQ}{i \in I}{\gtLab[i]}{\tyGround[i]}{
            (\gtCrashRole{\gtG[i]}{\roleS})
          }
          $, and thus $
          \gtProj[\rolesR]{(\gtCrashRole{\gtG}{\roleS})}{\roleR} =
          \stMerge{i \in I}{
            (\gtProj[\rolesR]{(\gtCrashRole{\gtG[i]}{\roleS})}{\roleR})%
          }
          $,
          apply \autoref{lem:subtype:merge-subty} and coinductive
          hypothesis.

          \end{enumerate}
      \end{itemize}
          \item Case $\gtG = \gtRec{\gtRecVar}{\gtGi}$: 
      
      By coinductive hypothesis.
  \end{itemize}
  Other cases are similar or trivial. 
    \qedhere
\end{proof}

\begin{lem}[Inversion of Subtyping]
\label{lem:subtyping-invert}
  ~
  \begin{enumerate}
    \item If
      $\stS \stSub
       \stIntSum{\roleP}{i \in I}{\stChoice{\stLab[i]}{\tyGround[i]} \stSeq \stT[i]}
      $, then
      $\unfoldOne{\stS} =
       \stIntSum{\roleP}{j \in J}{\stChoice{\stLabi[j]}{\tyGroundi[j]} \stSeq
       \stTi[j]}
      $, and $J \subseteq I$,
      and $\forall i \in J: \stLab[i] = \stLabi[i], \tyGround[i] =
      \tyGroundi[i]$ and $\stTi[i] \stSub \stT[i]$.
    \item If
      $\stS \stSub
       \stExtSum{\roleP}{i \in I}{\stChoice{\stLab[i]}{\tyGround[i]} \stSeq \stT[i]}
      $, then
      $\unfoldOne{\stS} =
       \stExtSum{\roleP}{j \in J}{\stChoice{\stLabi[j]}{\tyGroundi[j]} \stSeq
       \stTi[j]}
      $, and $I \subseteq J$,
      and $\forall i \in I: \stLab[i] = \stLabi[i], \tyGround[i] =
      \tyGroundi[i]$ and $\stTi[i] \stSub \stT[i]$.
  \end{enumerate}
\end{lem}
\begin{proof}
By \autoref{lem:unfold-subtyping}, the transitivity of subtyping, and \autoref{def:subtyping} (\inferrule{\iruleStSubIn}, \inferrule{\iruleStSubOut}).
\end{proof}

\subsection{Semantics of Global Types}
\label{sec:proof:semantics:gty}

\lemNorevival*
\begin{proof}
\begin{enumerate}[leftmargin=*]
  \item By induction on global type reductions: since $\stEnvAnnotGenericSym \neq
      \ltsCrash{}{\roleP}$, we start from \inferrule{\iruleGtMoveRec}.
      \begin{itemize}[leftmargin=*]
      \item Case \inferrule{\iruleGtMoveRec}:
     we have
      $\gtG = \gtRec{\gtRecVar}{\gtGii}$ and
      $\gtWithCrashedRoles{\rolesC}{\gtGii{}\subst{\gtRecVar}{\gtRec{\gtRecVar}{\gtGii}}}
      \gtMove[\stEnvAnnotGenericSym]{\rolesR} \gtWithCrashedRoles{\rolesCi}{\gtGi}$
    by \inferrule{\iruleGtMoveRec} and its inversion.
    Hence, by $\gtWithCrashedRoles{\rolesC}{\gtGii{}\subst{\gtRecVar}{\gtRec{\gtRecVar}{\gtGii}}}
    \gtMove[\stEnvAnnotGenericSym]{\rolesR} \gtWithCrashedRoles{\rolesCi}{\gtGi}$, $\roleP \in  \gtRolesCrashed{\gtGi}$, and inductive hypothesis,
    we have $\roleP \in \gtRolesCrashed{\gtGii{}\subst{\gtRecVar}{\gtRec{\gtRecVar}{\gtGii}}}$.
    Therefore,
    by $\gtRolesCrashed{\gtRec{\gtRecVar}{\gtGii}} =  \gtRolesCrashed{\gtGii{}\subst{\gtRecVar}{\gtRec{\gtRecVar}{\gtGii}}}$,
    we conclude with $\roleP \in \gtRolesCrashed{\gtG}$, as desired.
     \item Case \inferrule{\iruleGtMoveIn}:
  we have $\gtG =  \gtCommTransit{\rolePMaybeCrashed}{\roleQ}{i \in I}{\gtLab[i]}{\tyGround[i]}{\gtGi[i]}{j}$ and $\gtGi = \gtGi[j]$ by \inferrule{\iruleGtMoveIn}. It follows that $\gtRolesCrashed{\gtG} = \bigcup\limits_{i \in I}{\gtRolesCrashed{\gtGi[i]}}$
       and $\gtRolesCrashed{\gtGi} = \gtRolesCrashed{\gtGi[j]}$ with $j \in I$, and hence, $\gtRolesCrashed{\gtGi}  \subseteq \gtRolesCrashed{\gtG}$.
       Therefore, by $\roleP \in \gtRolesCrashed{\gtGi}$,
       we conclude with $\roleP \in \gtRolesCrashed{\gtG}$, as desired.
       
    \item Case \inferrule{\iruleGtMoveCtx}: we have $\gtG =  \gtCommSmall{\roleP}{\roleQMaybeCrashed}{i \in
      I}{\gtLab[i]}{\tyGround[i]}{\gtGi[i]}$, $\gtGi =  \gtCommSmall{\roleP}{\roleQMaybeCrashed}{i \in
      I}{\gtLab[i]}{\tyGround[i]}{\gtGii[i]}$,  $\forall i \in I :
    \gtWithCrashedRoles{\rolesC}{\gtGi[i]}
    \gtMove[\stEnvAnnotGenericSym]{\rolesR}
    \gtWithCrashedRoles{\rolesCi}{\gtGii[i]}$, and $\ltsSubject{\stEnvAnnotGenericSym} \notin \setenum{\roleP, \roleQ}$
    by \inferrule{\iruleGtMoveCtx} and its inversion. It follows that
    $\gtRolesCrashed{\gtG} = \bigcup\limits_{i \in I}{\gtRolesCrashed{\gtGi[i]}}$, $\gtRolesCrashed{\gtGi} = \bigcup\limits_{i \in I}{\gtRolesCrashed{\gtGii[i]}}$,
    and $\stEnvAnnotGenericSym \neq
      \ltsCrash{}{\roleP}$. Then by $\forall i \in I :
    \gtWithCrashedRoles{\rolesC}{\gtGi[i]}
    \gtMove[\stEnvAnnotGenericSym]{\rolesR}
    \gtWithCrashedRoles{\rolesCi}{\gtGii[i]}$, $\stEnvAnnotGenericSym \neq
      \ltsCrash{}{\roleP}$, and inductive hypothesis, we have $\forall i \in I : \text{if } \roleP \in \gtRolesCrashed{\gtGii[i]}, \text{then } \roleP \in \gtRolesCrashed{\gtGi[i]}$.
      Therefore, by $\roleP \in \bigcup\limits_{i \in I}{\gtRolesCrashed{\gtGii[i]}}$, we conclude with $\roleP \in \bigcup\limits_{i \in I}{\gtRolesCrashed{\gtGi[i]}} = \gtRolesCrashed{\gtG}$, as desired.
      \end{itemize}
       Other cases are similar.
 \item Similar to the proof of (1).
 \item The proof is trivial by \inferrule{\iruleGtMoveCrash} and its inversion.
 \qedhere
\end{enumerate}
\end{proof}

\lemWellAnnoPreserve*
\begin{proof}
  By induction on global type reductions:

  \begin{itemize}[leftmargin=*]
    \item Case \inferrule{\iruleGtMoveCrash}: we have $\rolesCi = \rolesC \cup \setenum{\roleP}$, $\gtGi = \gtCrashRole{\gtG}{\roleP}$,  $\roleP \notin \rolesR$,
    $\roleP \in \gtRoles{\gtG}$, and $\gtG \neq \gtRec{\gtRecVar}{\gtGi}$ by \inferrule{\iruleGtMoveCrash} and its inversion.

    \cref{item:wa:reliable-no-crash}:
     from the premise, we have
    $\gtRolesCrashed{\gtG} \cap \rolesR = \emptyset$. Since $\roleP \in \gtRoles{\gtG}$, by \autoref{lem:gtype:crashed-crash-remove-role},
    we have $\gtRolesCrashed{\gtCrashRole{\gtG}{\roleP}} \setminus \setenum{\roleP} \subseteq \gtRolesCrashed{\gtG}$. Then we consider two cases:
    \begin{itemize}[leftmargin=*]
    \item if $\roleP \in \gtRolesCrashed{\gtCrashRole{\gtG}{\roleP}}$, then
    $\gtRolesCrashed{\gtCrashRole{\gtG}{\roleP}} = \setenum{\roleP}
    \cup (\gtRolesCrashed{\gtCrashRole{\gtG}{\roleP}} \setminus \setenum{\roleP})$. Hence, by $\roleP \notin \rolesR$,
    $\gtRolesCrashed{\gtCrashRole{\gtG}{\roleP}} \setminus \setenum{\roleP} \subseteq \gtRolesCrashed{\gtG}$, and $\gtRolesCrashed{\gtG} \cap \rolesR = \emptyset$,
    we have $\gtRolesCrashed{\gtCrashRole{\gtG}{\roleP}} \cap \rolesR  = \emptyset$.
    \item  if $\roleP \notin \gtRolesCrashed{\gtCrashRole{\gtG}{\roleP}}$, then
    $\gtRolesCrashed{\gtCrashRole{\gtG}{\roleP}} = \gtRolesCrashed{\gtCrashRole{\gtG}{\roleP}} \setminus \setenum{\roleP}$. Hence, by
    \linebreak $\gtRolesCrashed{\gtCrashRole{\gtG}{\roleP}} \setminus \setenum{\roleP} \subseteq \gtRolesCrashed{\gtG}$ and $\gtRolesCrashed{\gtG} \cap \rolesR = \emptyset$,
    we have $\gtRolesCrashed{\gtCrashRole{\gtG}{\roleP}} \cap \rolesR  = \emptyset$.
     \end{itemize}
  Therefore, by  $\gtGi = \gtCrashRole{\gtG}{\roleP}$, we conclude with $\gtRolesCrashed{\gtGi} \cap \rolesR = \emptyset$, as desired.

    \cref{item:wa:crash-annot-crash}: from the premise, we have
    $\gtRolesCrashed{\gtG} \subseteq \rolesC$. Since $\roleP \in \gtRoles{\gtG}$, by \autoref{lem:gtype:crashed-crash-remove-role},
     we have $\gtRolesCrashed{\gtCrashRole{\gtG}{\roleP}} \setminus \setenum{\roleP} \subseteq \gtRolesCrashed{\gtG}$. Then we consider two cases:
    \begin{itemize}[leftmargin=*]
    \item if $\roleP \in \gtRolesCrashed{\gtCrashRole{\gtG}{\roleP}}$, then
    $\gtRolesCrashed{\gtCrashRole{\gtG}{\roleP}} = \setenum{\roleP}
    \cup (\gtRolesCrashed{\gtCrashRole{\gtG}{\roleP}} \setminus
    \setenum{\roleP})$. Hence, by \linebreak
    $\gtRolesCrashed{\gtCrashRole{\gtG}{\roleP}} \setminus \setenum{\roleP} \subseteq \gtRolesCrashed{\gtG}$ and $\gtRolesCrashed{\gtG} \subseteq \rolesC$,
    we have $\gtRolesCrashed{\gtCrashRole{\gtG}{\roleP}} \subseteq \rolesC \cup \setenum{\roleP}$.
    \item  if $\roleP \notin \gtRolesCrashed{\gtCrashRole{\gtG}{\roleP}}$, then
    $\gtRolesCrashed{\gtCrashRole{\gtG}{\roleP}} =
    \gtRolesCrashed{\gtCrashRole{\gtG}{\roleP}} \setminus \setenum{\roleP}$.
    Hence, by \linebreak
    $\gtRolesCrashed{\gtCrashRole{\gtG}{\roleP}} \setminus \setenum{\roleP} \subseteq \gtRolesCrashed{\gtG}$ and $\gtRolesCrashed{\gtG} \subseteq \rolesC$,
    we have $\gtRolesCrashed{\gtCrashRole{\gtG}{\roleP}} \subseteq \rolesC \cup \setenum{\roleP}$.
     \end{itemize}
   Therefore, by $\gtGi = \gtCrashRole{\gtG}{\roleP}$ and
    $\rolesCi = \rolesC \cup \{\roleP\}$, we conclude with $\gtRolesCrashed{\gtGi} \subseteq \rolesCi$, as desired.

    \cref{item:wa:live-no-crash}:
     from the premise, we have $\gtRoles{\gtG}
      \cap \gtRolesCrashed{\gtG} = \emptyset$. Since $\roleP \in \gtRoles{\gtG}$, by \autoref{lem:gtype:crash-remove-role},  \autoref{lem:gtype:crashed-crash-remove-role},
      we have $\gtRoles{\gtCrashRole{\gtG}{\roleP}} \subseteq \gtRoles{\gtG}$,
      $\roleP \notin \gtRoles{\gtCrashRole{\gtG}{\roleP}}$,
      and $\gtRolesCrashed{\gtCrashRole{\gtG}{\roleP}} \setminus \setenum{\roleP} \subseteq \gtRolesCrashed{\gtG}$. Then we consider two cases:
       \begin{itemize}[leftmargin=*]
    \item if $\roleP \in \gtRolesCrashed{\gtCrashRole{\gtG}{\roleP}}$, then
    $\gtRolesCrashed{\gtCrashRole{\gtG}{\roleP}} = \setenum{\roleP}
    \cup (\gtRolesCrashed{\gtCrashRole{\gtG}{\roleP}} \setminus
    \setenum{\roleP})$. Hence, by \linebreak $\gtRoles{\gtG}
      \cap \gtRolesCrashed{\gtG} = \emptyset$, $\roleP \notin \gtRoles{\gtCrashRole{\gtG}{\roleP}}$, $\gtRolesCrashed{\gtCrashRole{\gtG}{\roleP}} \setminus \setenum{\roleP} \subseteq \gtRolesCrashed{\gtG}$,
      and \linebreak $\gtRoles{\gtCrashRole{\gtG}{\roleP}} \subseteq \gtRoles{\gtG}$, we have $\gtRolesCrashed{\gtCrashRole{\gtG}{\roleP}}  \cap \gtRoles{\gtCrashRole{\gtG}{\roleP}} = \emptyset$.
     \item if $\roleP \notin \gtRolesCrashed{\gtCrashRole{\gtG}{\roleP}}$, then
    $\gtRolesCrashed{\gtCrashRole{\gtG}{\roleP}} =
    \gtRolesCrashed{\gtCrashRole{\gtG}{\roleP}} \setminus \setenum{\roleP}$.
    Hence, by \linebreak $\gtRoles{\gtG}
      \cap \gtRolesCrashed{\gtG} = \emptyset$,
    $\gtRolesCrashed{\gtCrashRole{\gtG}{\roleP}} \setminus \setenum{\roleP} \subseteq \gtRolesCrashed{\gtG}$, and
     $\gtRoles{\gtCrashRole{\gtG}{\roleP}} \subseteq \gtRoles{\gtG}$,  we have $\gtRolesCrashed{\gtCrashRole{\gtG}{\roleP}}  \cap \gtRoles{\gtCrashRole{\gtG}{\roleP}} = \emptyset$.
    \end{itemize}
  Therefore, by $\gtGi = \gtCrashRole{\gtG}{\roleP}$, we conclude with
      $\gtRoles{\gtGi}
      \cap \gtRolesCrashed{\gtGi} = \emptyset$, as desired.

      \item Case \inferrule{\iruleGtMoveRec}: we have
      $\gtG = \gtRec{\gtRecVar}{\gtGii}$ and
      $\gtWithCrashedRoles{\rolesC}{\gtGii{}\subst{\gtRecVar}{\gtRec{\gtRecVar}{\gtGii}}}
      \gtMove[\stEnvAnnotGenericSym]{\rolesR} \gtWithCrashedRoles{\rolesCi}{\gtGi}$
    by \inferrule{\iruleGtMoveRec} and its inversion.
     From the premise, we also have
    $\gtWithCrashedRoles{\rolesC}{\gtRec{\gtRecVar}{\gtGii}}$ is well-annotated.
    Hence, by $\gtRolesCrashed{\gtRec{\gtRecVar}{\gtGii}} =
    \gtRolesCrashed{\gtGii{}\subst{\gtRecVar}{\gtRec{\gtRecVar}{\gtGii}}}$,
    $\gtRoles{\gtRec{\gtRecVar}{\gtGii}} =
    \gtRoles{\gtGii{}\subst{\gtRecVar}{\gtRec{\gtRecVar}{\gtGii}}}$,
    and $\gtWithCrashedRoles{\rolesC}{\gtRec{\gtRecVar}{\gtGii}}$ is well-annotated,
    we have
    $\gtRolesCrashed{\gtGii{}\subst{\gtRecVar}{\gtRec{\gtRecVar}{\gtGii}}} \cap \rolesR = \emptyset$,
    $\gtRolesCrashed{\gtGii{}\subst{\gtRecVar}{\gtRec{\gtRecVar}{\gtGii}}} \subseteq \rolesC$,
    and $\gtRoles{\gtGii{}\subst{\gtRecVar}{\gtRec{\gtRecVar}{\gtGii}}} \cap
    \gtRolesCrashed{\gtGii{}\subst{\gtRecVar}{\gtRec{\gtRecVar}{\gtGii}}}= \emptyset$.
    It follows directly that
     $\gtWithCrashedRoles{\rolesC}{\gtGii{}\subst{\gtRecVar}{\gtRec{\gtRecVar}{\gtGii}}}$ is well-annotated.
     Therefore, by $\gtWithCrashedRoles{\rolesC}{\gtGii{}\subst{\gtRecVar}{\gtRec{\gtRecVar}{\gtGii}}}
    \gtMove[\stEnvAnnotGenericSym]{\rolesR} \gtWithCrashedRoles{\rolesCi}{\gtGi}$
    and inductive hypothesis, we conclude with $\gtWithCrashedRoles{\rolesCi}{\gtGi}$ is well-annotated, as desired.
      \end{itemize}
 Other cases are similar.
  \qedhere
\end{proof}

\begin{lem}
\label{lem:gt-lts-unfold}
  $\gtWithCrashedRoles{\rolesC}{\gtG}
  \gtMove[\stEnvAnnotGenericSym]{\rolesR}
  \gtWithCrashedRoles{\rolesCi}{\gtGi}$, \;iff
  \;$
  \gtWithCrashedRoles{\rolesC}{\unfoldOne{\gtG}}
  \gtMove[\stEnvAnnotGenericSym]{\rolesR}
  \gtWithCrashedRoles{\rolesCi}{\gtGi}$.
\end{lem}
\begin{proof}
  By inverting or applying $\inferrule{\iruleGtMoveRec}$ when necessary.
  \qedhere
\end{proof}

\begin{lem}[Progress of Global Types]
\label{lem:gtype:progress}
  If $\gtWithCrashedRoles{\rolesC}{\gtG}$ (where $\gtG$ is a projectable
  global type) is well-annotated, and $\gtG \neq
  \gtEnd$, then there exists $\gtGi, \rolesCi$ such that
  $\gtWithCrashedRoles{\rolesC}{\gtG}
  \gtMove{\rolesR}
  \gtWithCrashedRoles{\rolesCi}{\gtGi}$.
\end{lem}
\begin{proof}
  By \autoref{lem:gt-lts-unfold}, we only consider unfoldings.
\begin{itemize}[leftmargin=*]
\item  Case $\unfoldOne{\gtG} = \gtEnd$: the premise does not hold.
\item Case $
    \unfoldOne{\gtG} =
    \gtComm{\roleP}{\roleQ}{i \in I}{\gtLab[i]}{\tyGround[i]}{\gtG[i]}
  $:  apply \inferrule{\iruleGtMoveOut}.
  We can
  pick any $\gtLab[i] \neq \gtCrashLab$ to reduce the global type (note that
  our syntax prohibits singleton $\gtCrashLab$ branches).
\item Case $
    \unfoldOne{\gtG} =
    \gtComm{\roleP}{\roleQCrashed}{i \in I}{\gtLab[i]}{\tyGround[i]}{\gtG[i]}
  $: apply \inferrule{\iruleGtMoveOrph}.
  We can
  pick any $\gtLab[i] \neq \gtCrashLab$ to reduce the global type (note that
  our syntax prohibits singleton $\gtCrashLab$ branches).
\item Case $
    \unfoldOne{\gtG} =
    \gtCommTransit{\rolePMaybeCrashed}{\roleQ}{i \in I}{\gtLab[i]}{\tyGround[i]}{\gtG[i]}{j}
  $:
  if $\gtLab[j] = \gtCrashLab$, then it $\roleP$ must have crashed, apply \inferrule{\iruleGtMoveCrDe}.
  Otherwise, apply \inferrule{\iruleGtMoveIn}.
  \qedhere
\end{itemize}
\end{proof}

\subsection{Semantics of Configurations}
\label{sec:proof:semantics:conf}

\begin{lem}
\label{lem:stenv-red:trivial-2}
  If \,$\stEnv; \qEnv \stEnvMoveMaybeCrash \stEnvi; \qEnvi$ \,with\,
  $\stEnv; \qEnv
  \stEnvMoveGenAnnot \stEnvi; \qEnvi$, \,then\,
  \begin{enumerate}
    \item $\dom{\stEnv} = \dom{\stEnvi}$; \,and\,
    \item for all \, $\roleP \in \dom{\stEnv}$ \,with\,
      $\roleP \neq \ltsSubject{\stEnvAnnotGenericSym}$, \,we have\,
      $\stEnvApp{\stEnv}{\roleP} =
      \stEnvApp{\stEnvi}{\roleP}$.
  \end{enumerate}
\end{lem}
\begin{proof}
Trivial by induction on reductions of configuration.
\qedhere
\end{proof}

\begin{lem}
\label{lem:stenv-red:trivial-3}
 If \,$\stEnv; \qEnv \stEnvMoveGenAnnot \stEnvi; \qEnvi$, \,then\,
 for any $\roleP, \roleQ \in \dom{\stEnv}$ with
 $\ltsSubject{\stEnvAnnotGenericSym} \notin \setenum{\roleP, \roleQ}$, we have
 $\stEnvApp{\qEnv}{\roleP, \roleQ} = \stEnvApp{\qEnvi}{\roleP, \roleQ}$.
\end{lem}
\begin{proof}
Trivial by induction on reductions of configuration.
\qedhere
\end{proof}

\begin{lem}[Inversion of Typing Context Reduction]
\label{lem:stenv-red:inversion-basic}
  ~
  \begin{enumerate}
    \item
      If \,$
        \stEnv; \qEnv
        \stEnvMoveOutAnnot{\roleP}{\roleQ}{\stChoice{\stLab[k]}{\tyGround[k]}}
        \stEnvi; \qEnvi
      $\,, then \,
      $
        \unfoldOne{\stEnvApp{\stEnv}{\roleP}} =
        \stIntSum{\roleQ}{i \in I}{\stChoice{\stLab[i]}{\tyGround[i]} \stSeq \stT[i]}%
      $\,, \,$k \in I$\,, and \,
      $
        \stEnvApp{\stEnvi}{\roleP} = \nolinebreak \stT[k]
      $;
    \item
      If \,$
        \stEnv; \qEnv
        \stEnvMoveInAnnot{\roleQ}{\roleP}{\stChoice{\stLab[k]}{\tyGround[k]}}
        \stEnvi; \qEnvi
      $\,, then \,
      $
        \unfoldOne{\stEnvApp{\stEnv}{\roleQ}} =
        \stExtSum{\roleP}{i \in I}{\stChoice{\stLab[i]}{\tyGround[i]} \stSeq \stT[i]}%
      $\,, \,$k \in I$\,, and \,
      $
        \stEnvApp{\stEnvi}{\roleQ} = \nolinebreak \stT[k]
      $.
  \end{enumerate}
\end{lem}
\begin{proof}
By applying and inverting \inferrule{\iruleTCtxOut} and \inferrule{\iruleTCtxIn}.
\qedhere
\end{proof}

\begin{lem}[Determinism of Configuration Reduction]
\label{lem:stenv-red:det}
  If \,$\stEnv; \qEnv \stEnvMoveGenAnnot \stEnvi; \qEnvi$ \,and\, \linebreak
  $\stEnv; \qEnv \stEnvMoveGenAnnot
  \stEnvii; \qEnvii$, then $\stEnvi = \stEnvii$ and $\qEnvi = \qEnvii$.
\end{lem}
\begin{proof}
Trivial by induction on reductions of configuration.
\qedhere
\end{proof}

\begin{lem}
\label{lem:stenv-queue-red:det}
 If \,$\stEnv[1]; \qEnv \stEnvMoveGenAnnot \stEnvi[1]; \qEnvi[1]$ \,and\,
  $\stEnv[2]; \qEnv \stEnvMoveGenAnnot
  \stEnvi[2]; \qEnvi[2]$, then $\qEnvi[1]= \qEnvi[2]$.
\end{lem}
\begin{proof}
Trivial by induction on reductions of configuration.
\qedhere
\end{proof}

\subsection{Relating Semantics}%
\label{sec:proof:relating}
\begin{prop}
\label{prop:gt-lts-unfold-once}
  $\stEnvAssoc{
    \gtWithCrashedRoles{\rolesC}{\gtRec{\gtRecVar}{\gtG}}
  }{\stEnv; \qEnv}{\rolesR}$
  \,if and only if\,
 $\stEnvAssoc{
    \gtWithCrashedRoles{\rolesC}{\gtG{}\subst{\gtRecVar}{\gtRec{\gtRecVar}{\gtG}}}
    }{\stEnv; \qEnv}{\rolesR}$.
\end{prop}
\begin{proof}
  By \autoref{lem:unfold-subtyping}.
  \qedhere
\end{proof}

\begin{prop}\label{prop:lt-lts-unfold-once}
  If\,
  $\stEnvAssoc{
    \gtWithCrashedRoles{\rolesC}{\gtG}
  }{\stEnv; \qEnv}{\rolesR}$\,
  and
  \,$\stEnvApp{\stEnv}{\roleP} = \stRec{\stRecVar}{\stT}$,
  \,then we have\, \linebreak
  $\stEnvAssoc{
    \gtWithCrashedRoles{\rolesC}{\gtG}
  }{
      \stEnvUpd{\stEnv}{\roleP}{
        \stT\subst{\stRecVar}{\stRec{\stRecVar}{\stT}}%
      }
  ; \qEnv}{\rolesR}$.
\end{prop}
\begin{proof}
  By \autoref{lem:unfold-subtyping}.
  \qedhere
\end{proof}

\begin{prop}
\label{prop:gt-lts-unfold}
  $\stEnvAssoc{\gtWithCrashedRoles{\rolesC}{\gtG}}{\stEnv; \qEnv}{\rolesR}$
  \,if and only if\,
  $\stEnvAssoc{\gtWithCrashedRoles{\rolesC}{\unfoldOne{\gtG}}}{\stEnv; \qEnv}{\rolesR}$.
\end{prop}
\begin{proof}
  By applying \autoref{prop:gt-lts-unfold-once} as many times as necessary.
\end{proof}

\begin{lem}[Inversion of Projection]
\label{lem:inv-proj}
  Given a local type $\stS$, which is a subtype of projection from a global type
  $\gtG$ on a role $\roleP$ with respect to a set of reliable roles $\rolesR$,
  \ie $\stS \stSub (\gtProj[\rolesR]{\gtG}{\roleP})$,
  then:
  \begin{enumerate}
    \item
      \label{item:proj-inv:send}
      If\;
      $\unfoldOne{\stS}
      =
      \stIntSum{\roleQ}{i \in I}{\stChoice{\stLab[i]}{\tyGround[i]} \stSeq
        \stSi[i]}$, then either
        \begin{enumerate}
        \item
          $\unfoldOne{\gtG} =
            \gtComm{\roleP}{\roleQMaybeCrashed}{i \in I'}{\gtLab[i]}{\tyGroundi[i]}{\gtG[i]}$,
          where
          $I \subseteq I'$, and
          for all $i \in I$:
          $\stLab[i] = \gtLab[i]$, \linebreak
          $\stSi[i] \stSub (\gtProj[\rolesR]{\gtG[i]}{\roleP})$,
          and
          $\tyGround[i] = \tyGroundi[i]$;
          or,
        \item
          $\unfoldOne{\gtG} =
            \gtComm{\roleS}{\roleTMaybeCrashed}{j \in J}{\gtLab[j]}{\tyGroundi[j]}{\gtG[j]}$,
          or
          $\unfoldOne{\gtG} =
            \gtCommTransit{\roleSMaybeCrashed}{\roleT}{j \in
            J}{\gtLab[j]}{\tyGroundi[j]}{\gtG[j]}{k}$,
          where for all $j \in J$:
          $\stS \stSub (\gtProj[\rolesR]{\gtG[j]}{\roleP})$,
          with $\roleP \neq \roleS$ and $\roleP \neq \roleT$.
        \end{enumerate}
    \item
      \label{item:proj-inv:recv}
      If\;
      $\unfoldOne{\stS}
      =
      \stExtSum{\roleQ}{i \in I}{\stChoice{\stLab[i]}{\tyGround[i]} \stSeq
        \stSi[i]}$, then either
      \begin{enumerate}
        \item\label{item:proj-inv:recv-equal}
          $\unfoldOne{\gtG} =
            \gtComm{\roleQ}{\rolePMaybeCrashed}{i \in I'}{\gtLab[i]}{\tyGroundi[i]}{\gtG[i]}$,
          or
          $\unfoldOne{\gtG} =
            \gtCommTransit{\roleQMaybeCrashed}{\roleP}{i \in
            I'}{\gtLab[i]}{\tyGroundi[i]}{\gtG[i]}{j}$,
          where \linebreak
          $I' \subseteq I$, and
          for all $i \in I'$:
          $\stLab[i] = \gtLab[i]$,
          $\stSi[i] \stSub (\gtProj[\rolesR]{\gtG[i]}{\roleP})$,
          $\tyGroundi[i] = \tyGround[i]$,
          and $\roleQ \notin \rolesR$ implies $\exists k \in I': \gtLab[k] =
          \gtCrashLab$;
          or,
        \item
          $\unfoldOne{\gtG} =
            \gtComm{\roleS}{\roleTMaybeCrashed}{j \in J}{\gtLab[j]}{\tyGroundi[j]}{\gtG[j]}$,
          or
          $\unfoldOne{\gtG} =
            \gtCommTransit{\roleSMaybeCrashed}{\roleT}{j \in
            J}{\gtLab[j]}{\tyGroundi[j]}{\gtG[j]}{k}$,
          where for all $j \in J$:
          $\stS \stSub (\gtProj[\rolesR]{\gtG[j]}{\roleP})$,
          with $\roleP \neq \roleS$ and $\roleP \neq \roleT$.
      \end{enumerate}
    \item
      \label{item:proj-inv:end}
      If\;
      $\stS = \stEnd$,
      then $\roleP \notin \gtRoles{\gtG}$.
  \end{enumerate}
\end{lem}
\begin{proof}
By the definition of global type projection~(\autoref{def:global-proj}).
\qedhere
\end{proof}

\begin{lem}[Existence of Crash Handling Branch Under Projection]
\label{lem:crash-lab-exists}
  If role $\roleQ$ is unreliable, $\roleQ \notin \rolesR$, and a global type
  projects onto $\roleP$ with an external choice from $\roleQ$,\linebreak
  $ \unfoldOne{\gtProj[\rolesR]{\gtG}{\roleP}} =
    \stExtSum{\roleQ}{i \in I}{\stChoice{\stLab[i]}{\stS[i]} \stSeq \stSi[i]}
  $, then there must be a crash handling branch,
  $\exists j \in I : \stLab[j] = \stCrashLab$.
\end{lem}
\begin{proof}
  By induction on \cref{item:proj-inv:recv} of \autoref{lem:inv-proj}.
  \qedhere
\end{proof}

\begin{lem}
\label{lem:eventual-global-form}
If a local type $\stT$ is a subtype of an external choice with a matching role, obtained
via projection from a global type $\gtG$, i.e.,
$\stT = \stExtSum{\roleP}{i \in I}{\stChoice{\stLab[i]}{\tyGround[i]} \stSeq \stT[i]}
\stSub
\gtProj[\rolesR]{\gtG}{\roleQ}$, $\unfoldOne{\gtG}$ is of the form
$ \gtComm{\roleS}{\roleTMaybeCrashed}{j \in J}{\gtLab[j]}{\tyGroundi[j]}{\gtG[j]}$ or
$  \gtCommTransit{\roleSMaybeCrashed}{\roleT}{j \in
            J}{\gtLab[j]}{\tyGroundi[j]}{\gtG[j]}{k}$, and a queue environment $\qEnv$ is associated
 with $\gtWithCrashedRoles{\rolesC}{\gtG}$, then there exists a global type $\gtWithCrashedRoles{\rolesCi}{\gtGi}$ and a queue environment $\qEnvi$ such that
 $\gtProj[\rolesR]{\gtG}{\roleQ} \stSub
 \gtProj[\rolesR]{\gtGi}{\roleQ}$, $\unfoldOne{\gtGi}$ is of the form
 $ \gtCommTransit{\rolePMaybeCrashed}{\roleQ}{i \in
            I'}{\gtLab[i]}{\tyGroundi[i]}{\gtG[i]}{j}$ or
 $ \gtComm{\roleP}{\roleQMaybeCrashed}{i \in I'}{\gtLab[i]}{\tyGroundi[i]}{\gtG[i]}$, and $\qEnvi$ is
 associated with $\gtWithCrashedRoles{\rolesCi}{\gtGi}$ with $\stEnvApp{\qEnvi}{\roleP, \roleQ} = \stEnvApp{\qEnv}{\roleP, \roleQ}$.
\end{lem}
\begin{proof}
Apply \cref{item:proj-inv:recv} of \autoref{lem:inv-proj} on the premise, we have $\forall j \in J: \stT =
 \stExtSum{\roleP}{i \in I}{\stChoice{\stLab[i]}{\tyGround[i]} \stSeq \stT[i]}
\stSub
\gtProj[\rolesR]{\gtG[j]}{\roleQ}$ with $\roleQ \neq \roleS$ and $\roleQ \neq \roleT$, which follows that
$\gtProj[\rolesR]{\gtG}{\roleQ} = \stMerge{j \in J}{\gtProj[\rolesR]{\gtG[j]}{\roleQ}}$.
Then, by \autoref{lem:merge-subtyping}, we get $\forall j \in J: \gtProj[\rolesR]{\gtG}{\roleQ} \stSub
\gtProj[\rolesR]{\gtG[j]}{\roleQ}$.
 We take an arbitrary $\gtG[j]$ with $j \in J$. By applying \cref{item:proj-inv:recv} of \autoref{lem:inv-proj} again,
we have two cases.
\begin{itemize}[leftmargin=*]
\item Case (1):
$\unfoldOne{\gtG[j]} = \gtCommTransit{\rolePMaybeCrashed}{\roleQ}{i \in
            I'}{\gtLab[i]}{\tyGroundi[i]}{\gtG[i]}{l}$ or
$\unfoldOne{\gtG[j]} =  \gtComm{\roleP}{\roleQMaybeCrashed}{i \in I'}{\gtLab[i]}{\tyGroundi[i]}{\gtG[i]}$.

Since $\gtProj[\rolesR]{\gtG}{\roleQ} \stSub \gtProj[\rolesR]{\gtG[j]}{\roleQ} $,
we only need to show that there exists $\qEnvi$ such that
$\qEnvi$ is associated with $\gtWithCrashedRoles{\rolesC}{\gtG[j]}$ and
$\stEnvApp{\qEnvi}{%
      \roleP,  \roleQ%
      }  = \stEnvApp{\qEnv}{%
     \roleP,   \roleQ%
      }$. We consider two subcases:
      \begin{itemize}[leftmargin=*]
      \item $\unfoldOne{\gtG} =
      \gtComm{\roleS}{\roleTMaybeCrashed}{j \in J}{\gtLab[j]}{\tyGroundi[j]}{\gtG[j]}$: %
      by~\autoref{def:assoc-queue}, we have that  $\qEnv$ is associated with $\gtG[j]$. We take $\qEnvi$ as $\qEnv$.
      \item $\unfoldOne{\gtG} =  \gtCommTransit{\roleSMaybeCrashed}{\roleT}{j \in
            J}{\gtLab[j]}{\tyGroundi[j]}{\gtG[j]}{k}$, we perform case analysis on $\gtLab[k] = \gtCrashLab$ and
            \linebreak $\gtLab[k] \neq \gtCrashLab$:
            \begin{itemize}[leftmargin=*]
            \item $\gtLab[k] = \gtCrashLab$: by~\autoref{def:assoc-queue}, we have that  $\qEnv$ is associated with $\gtWithCrashedRoles{\rolesC}{\gtG[j]}$.
            We take $\qEnvi$ as $\qEnv$.
            \item $\gtLab[k] \neq \gtCrashLab$: by~\autoref{def:assoc-queue}, we have that  $\stEnvApp{\qEnv}{\roleS, \roleT} =
        \stQCons{\stQMsg{\gtLab[k]}{\tyGroundi[k]}}{\stQ}$ and
        $\stEnvUpd{\qEnv}{\roleS, \roleT}{\stQ}$ is associated with
        $\gtWithCrashedRoles{\rolesC}{\gtG[j]}$. We take $\qEnvi = \stEnvUpd{\qEnv}{\roleS, \roleT}{\stQ}$. Since $\roleQ \neq \roleS$
        and $\roleQ \neq \roleT$, it is straightforward that $\stEnvApp{\qEnvi}{\roleP, \roleQ} = \stEnvApp{\qEnv}{\roleP, \roleQ}$, as required.
            \end{itemize}
       \end{itemize}
\item Case (2): $\unfoldOne{\gtG[j]} =  \gtComm{\roleR}{\roleUMaybeCrashed}{l \in L}{\gtLab[l]}{\tyGroundi[l]}{\gtG[l]}$ or
$  \gtCommTransit{\roleRMaybeCrashed}{\roleU}{l \in
            L}{\gtLab[l]}{\tyGroundi[l]}{\gtG[l]}{k}$.

            We can construct a queue environment $\qEnvi$ as in case (1), which is associated with $ \gtWithCrashedRoles{\rolesC}{\gtG[j]}$.
            The thesis is then proved by applying inductive hypothesis on $\gtWithCrashedRoles{\rolesC}{\gtG[j]}$ and $\qEnvi$.
\qedhere
\end{itemize}
\end{proof}

\begin{lem}
\label{lem:queue-assoc-crash}
  If $\qEnv$ is associated with $\gtWithCrashedRoles{\rolesC}{\gtG}$ and
  $\roleP \in \gtRoles{\gtG}$, then
  $\stEnvUpd{\qEnv}{\cdot, \roleP}{\stQUnavail}$ is associated with
  $\gtWithCrashedRoles{\rolesC \cup
  \setenum{\roleP}}{\gtCrashRole{\gtG}{\roleP}}$.
\end{lem}

\begin{proof}
  We denote $\qEnvi = \stEnvUpd{\qEnv}{\cdot, \roleP}{\stQUnavail}$ in the
  subsequent proof.
  To show association, there are two parts: namely a shape-dependent
  part, and a shape-independent part.

  We first show shape-independent part, which shared for all cases:
  that a crashed role $\roleR$ is in $\rolesC \cup
  \setenum{\roleP}$ iff $\stEnvApp{\qEnvi}{\cdot, \roleR} = \stQUnavail$.
  It follows that the roles $\roleR \in \rolesC$ have the requirements
  satisfied from the premise, and we set $\qEnvi = \stEnvUpd{\qEnv}{\cdot,
  \roleP}{\stQUnavail}$, and that $\roleP$ is in the new set of crashed
  roles.

  Shape-dependent part are by induction on the definition of
  $\gtCrashRole{\gtG}{\roleP}$:
  \begin{itemize}[leftmargin=*]
    \item Case
    \(
    \gtCrashRole{
      (\gtCommSmall{\roleP}{\roleQ}{i \in I}{\gtLab[i]}{\tyGround[i]}{\gtG[i]})
    }{
      \roleP
    }
    =
    \gtCommTransit{\rolePCrashed}{\roleQ}{i \in I}{\gtLab[i]}{\tyGround[i]}{
      (\gtCrashRole{\gtG[i]}{\roleP})}{j}
    \) where $j \in I$ and $\gtLab[j] = \gtCrashLab$.

    Since $\qEnv$ is associated with $\gtWithCrashedRoles{\rolesC}{\gtG}$, we that
    $\forall i \in I:$ $\qEnv$ is associated with
    $\gtWithCrashedRoles{\rolesC}{\gtGi[i]}$, and $\stEnvApp{\qEnv}{\roleP, \roleQ} =
    \stQEmpty$.

    By inductive hypothesis, we have $\qEnvi$ is associated with
    $\gtWithCrashedRoles{\rolesC}{(\gtCrashRole{\gtG[i]}{\roleP})}$.
    Moreover, since $\gtLab[j] = \gtCrashLab$, we can show that
    $\stEnvApp{\qEnvi}{\roleP, \roleQ} = \stEnvApp{\qEnv}{\roleP, \roleQ} =
    \stQEmpty$.
       \item Case
    \(
    \gtCrashRole{
      (\gtCommTransit{\roleP}{\roleQ}{i \in
      I}{\gtLab[i]}{\tyGround[i]}{\gtG[i]}{j})
    }{
      \roleQ
    }
    =
    \gtCrashRole{\gtG[j]}{\roleQ}. 
    \)

    Subcase $\gtLab[j] = \gtCrashLab$:

    By inductive hypothesis, we know $\qEnvi$ is associated with
    $\gtWithCrashedRoles{\rolesC \cup \setenum{\roleQ}}{\gtGi[j]}$, as
    required.

    Subcase $\gtLab[j] \neq \gtCrashLab$:

    From association, we have
    $\stEnvApp{\qEnv}{\roleP, \roleQ} =
    \stQCons{\stQMsg{\gtLab[j]}{\tyGround[j]}}{\stQ}$
    and
    $\stEnvUpd{\qEnv}{\roleP, \roleQ}{\stQ}$
    is associated with
    $\gtWithCrashedRoles{\rolesC}{\gtGi[j]}$.

    By inductive hypothesis, we know $
    \stEnvUpd{
      \stEnvUpd{\qEnv}{\roleP, \roleQ}{\stQ}
    }{\cdot, \roleQ}{\stQUnavail}$
    is associated with
    $\gtWithCrashedRoles{\rolesC \cup
    \setenum{\roleQ}}{\gtCrashRole{\gtGi[j]}{\roleQ}}$.

    Since
    $ \stEnvUpd{
      \stEnvUpd{\qEnv}{\roleP, \roleQ}{\stQ}
    }{\cdot, \roleQ}{\stQUnavail} =
    \stEnvUpd{\qEnv}{\cdot, \roleQ}{\stQUnavail} = \qEnvi$,
    so $\qEnvi$ is associated with $\gtWithCrashedRoles{\rolesC \cup
    \setenum{\roleQ}}{\gtCrashRole{\gtGi[j]}{\roleQ}}$.
    \item Case
    \(
    \gtCrashRole{
      (\gtCommSmall{\roleP}{\roleQCrashed}{i \in I}{\gtLab[i]}{\tyGround[i]}{\gtG[i]})
    }{
      \roleP
    }
    =
    \gtCrashRole{\gtG[j]}{\roleP}
    \), where $j \in I$ and $\gtLab[j] = \gtCrashLab$. 

    Since $\qEnv$ is associated with $\gtWithCrashedRoles{\rolesC}{\gtG}$, we
    have that $\qEnv$ is associated with
    $\gtWithCrashedRoles{\rolesC}{\gtG[j]}$.

    By inductive hypothesis, we have $\qEnvi$ is associated with
    $\gtWithCrashedRoles{\rolesC \cup
    \setenum{\roleP}}{\gtCrashRole{\gtG[j]}{\roleP}}$, as required.
    
    Other cases follows directly by inductive hypothesis or trivially.
    \qedhere
  \end{itemize}
\end{proof}

\begin{lem}
\label{lem:stenv-red:non-trivial}
  If\,
   $\stEnv; \qEnv
  \stEnvMoveGenAnnot \stEnvi; \qEnvi$,\,
  $\stEnvAssoc{\gtWithCrashedRoles{\rolesC}{\gtG}}{\stEnv; \qEnv}{\rolesR}$,\,
  $\stEnvAssoc{\gtWithCrashedRoles{\rolesC}{\gtG[1]}}{\stEnv[1]; \qEnv}{\rolesR}$,\,
  $\stEnvApp{\stEnv}{\ltsSubject{\stEnvAnnotGenericSym}} = \stEnvApp{\stEnv[1]}{\ltsSubject{\stEnvAnnotGenericSym}}$,\,
  and\,
  $\stEnvi[1] =  \stEnvUpd{\stEnv[1]}{\ltsSubject{\stEnvAnnotGenericSym}}{\stEnvApp{\stEnvi}{\ltsSubject{\stEnvAnnotGenericSym}}}$,\,
   then\,
$\stEnv[1]; \qEnv
  \stEnvMoveGenAnnot \stEnvi[1]; \qEnvi$.
 \end{lem}
  \begin{proof}
 By induction on reductions of configuration.
\begin{itemize}[leftmargin=*]
\item Case \inferrule{\iruleTCtxOut}:   %
$\stEnvAnnotGenericSym = \stEnvOutAnnotSmall{\roleP}{\roleQ}{\stChoice{\stLab[k]}{\tyGround[k]}}$.
We apply and invert \inferrule{\iruleTCtxOut} on 
 $\stEnv; \qEnv
  \stEnvMoveGenAnnot \stEnvi; \qEnvi$ to get 
 $\stEnvApp{\stEnv}{\roleP}  = 
   \stIntSum{\roleQ}{i \in I}{\stChoice{\stLab[i]}{\tyGround[i]} \stSeq \stT[i]}$,  
  $k \in I$, 
   $ \stEnvi = \stEnvUpd{\stEnv}{\roleP}{\stT[k]}$, and 
  $\qEnvi = \stEnvUpd{\qEnv}{\roleP, \roleQ}{
        \stQCons{
          \stEnvApp{\qEnv}{\roleP, \roleQ}
        }{
          \stQMsg{\stLab[k]}{\tyGround[k]}
        }}$.
 Then by applying \inferrule{\iruleTCtxOut} on 
 $\stEnvApp{\stEnv[1]}{\roleP} = 
 \stEnvApp{\stEnv}{\roleP} = \stIntSum{\roleQ}{i \in I}{\stChoice{\stLab[i]}{\tyGround[i]} \stSeq \stT[i]}$ and 
 $k \in I$, we obtain that  $\stEnv[1]; \qEnv
  \stEnvMoveGenAnnot \stEnvUpd{\stEnv[1]}{\roleP}{\stT[k]}; \qEnvi$, which follows that 
  $\stEnv[1]; \qEnv
  \stEnvMoveGenAnnot \stEnvi[1]; \qEnvi$, as desired. 

\item Case \inferrule{\iruleTCtxIn}:  %
$\stEnvAnnotGenericSym = \stEnvInAnnotSmall{\roleP}{\roleQ}{\stChoice{\stLab[k]}{\tyGround[k]}}$.
We apply and invert \inferrule{\iruleTCtxIn} on 
 $\stEnv; \qEnv
  \stEnvMoveGenAnnot \stEnvi; \qEnvi$ to get 
 $\stEnvApp{\stEnv}{\roleP}  = 
   \stExtSum{\roleQ}{i \in I}{\stChoice{\stLab[i]}{\tyGround[i]} \stSeq \stT[i]}$,  
  $k \in I$, 
  $\stEnvApp{\qEnv}{\roleQ, \roleP}
      =
      \stQCons{\stQMsg{\stLab[k]}{\tyGround[k]}}{\stQ}$, 
   $\stEnvi = \stEnvUpd{\stEnv}{\roleP}{\stT[k]}$, and 
  $\qEnvi = \stEnvUpd{\qEnv}{\roleQ, \roleP}{
      \stQ}$.
 Then by applying \inferrule{\iruleTCtxIn} on 
 $\stEnvApp{\stEnv[1]}{\roleP} = 
 \stEnvApp{\stEnv}{\roleP} = \stExtSum{\roleQ}{i \in I}{\stChoice{\stLab[i]}{\tyGround[i]} \stSeq \stT[i]}$, 
 $k \in I$, and  $\stEnvApp{\qEnv}{\roleQ, \roleP}
      =
      \stQCons{\stQMsg{\stLab[k]}{\tyGround[k]}}{\stQ}$, 
 we obtain that  $\stEnv[1]; \qEnv
  \stEnvMoveGenAnnot \stEnvUpd{\stEnv[1]}{\roleP}{\stT[k]}; \qEnvi$, which follows that 
  $\stEnv[1]; \qEnv
  \stEnvMoveGenAnnot \stEnvi[1]; \qEnvi$, as desired.

  \item Case \inferrule{\iruleTCtxCrash}: %
  $\stEnvAnnotGenericSym = \ltsCrashSmall{\mpS}{\roleP}$.
We apply and invert \inferrule{\iruleTCtxCrash} on 
 $\stEnv; \qEnv
  \stEnvMoveGenAnnot \stEnvi; \qEnvi$ to get 
 $\stEnvApp{\stEnv}{\roleP}  \neq \stStop$, 
 $\stEnvApp{\stEnv}{\roleP}  \neq \stEnd$, 
  $\stEnvi = \stEnvUpd{\stEnv}{\roleP}{\stStop}$, and 
  $\qEnvi = \stEnvUpd{\qEnv}{\cdot, \roleP}{
      \stQUnavail}$.
 Then by applying \inferrule{\iruleTCtxCrash} on 
 $\stEnvApp{\stEnv[1]}{\roleP} = 
 \stEnvApp{\stEnv}{\roleP} \neq \stStop$ and 
  $\stEnvApp{\stEnv[1]}{\roleP} = 
 \stEnvApp{\stEnv}{\roleP} \neq \stEnd$, 
  we obtain that  $\stEnv[1]; \qEnv
  \stEnvMoveGenAnnot \stEnvUpd{\stEnv[1]}{\roleP}{\stStop}; \qEnvi$, which follows that 
  $\stEnv[1]; \qEnv
  \stEnvMoveGenAnnot \stEnvi[1]; \qEnvi$, as desired. 
  
 \item Case \inferrule{\iruleTCtxCrashDetect}:  %
 $\stEnvAnnotGenericSym = 
\ltsCrDe{\mpS}{\roleQ}{\roleP}$. 
We apply and invert \inferrule{\iruleTCtxCrashDetect} on 
 $\stEnv; \qEnv
  \stEnvMoveGenAnnot \stEnvi; \qEnvi$ to get
   $\stEnvApp{\stEnv}{\roleQ}  = 
   \stExtSum{\roleP}{i \in I}{\stChoice{\stLab[i]}{\tyGround[i]} \stSeq \stT[i]}$,  
   $\stEnvApp{\stEnv}{\roleP} = \stStop$, 
  $k \in I$, $\stLab[k] = \stCrashLab$, 
  $\stEnvApp{\qEnv}{\roleP, \roleQ}
      =
     \stQEmpty$, and 
   $\stEnvi = \stEnvUpd{\stEnv}{\roleQ}{\stT[k]}$.   
   Since $\stEnvAssoc{\gtWithCrashedRoles{\rolesC}{\gtG}}{\stEnv; \qEnv}{\rolesR}$ and 
   $\stEnvApp{\stEnv}{\roleP} = \stStop$, it holds that $\roleP \in \rolesC$. Then by 
   $\stEnvAssoc{\gtWithCrashedRoles{\rolesC}{\gtG[1]}}{\stEnv[1]; \qEnv}{\rolesR}$,
   $\stEnvApp{\stEnv[1]}{\roleP} = \stStop$. 
We apply \inferrule{\iruleTCtxCrashDetect} on 
  $\stEnvApp{\stEnv[1]}{\roleQ} = \stEnvApp{\stEnv}{\roleQ}  = 
   \stExtSum{\roleP}{i \in I}{\stChoice{\stLab[i]}{\tyGround[i]} \stSeq \stT[i]}$,  
   $\stEnvApp{\stEnv[1]}{\roleP} = \stStop$, 
  $k \in I$, $\stLab[k] = \stCrashLab$, and 
  $\stEnvApp{\qEnv}{\roleP, \roleQ}
      =
     \stQEmpty$ to obtain 
 $\stEnv[1]; \qEnv
  \stEnvMoveGenAnnot \stEnvUpd{\stEnv[1]}{\roleQ}{\stT[k]}; \qEnvi$, which follows that 
  $\stEnv[1]; \qEnv
  \stEnvMoveGenAnnot \stEnvi[1]; \qEnvi$, as desired. 

\item Case \inferrule{\iruleTCtxRec}: by inductive hypothesis. 
\qedhere
\end{itemize}
\end{proof}

We define
$\stIdxRemoveCrash{I}{\stLab[i]}$ to be an index set with the special label
$\stCrashLab$ removed, \ie
$\stIdxRemoveCrash{I}{\stLab[i]} = \setcomp{i \in I}{\stLab[i] \neq \stCrashLab}$.

\thmProjSoundness*
\begin{proof}
  By induction on reductions of global type
  \(
    \gtWithCrashedRoles{\rolesC}{\gtG}
    \gtMove[\stEnvAnnotGenericSym]{\rolesR}
    \gtWithCrashedRoles{\rolesCi}{\gtGi}
  \).
  \begin{itemize}[leftmargin=*]
    \item Case \inferrule{\iruleGtMoveCrash}:

      From the premise we have
      \begin{gather}
        \stEnvAssoc{\gtWithCrashedRoles{\rolesC}{\gtG}}{\stEnv; \qEnv}{\rolesR}
        \label{eq:soundness:movecrash_assoc_pre}
        \\
        \gtWithCrashedRoles{\rolesC}{\gtG}
        \gtMove{\rolesR}
        \\
        \roleP \notin \rolesR
        \\
        \roleP \in \gtRoles{\gtG}
        \\
        \gtG \neq \gtRec{\gtRecVar}{\gtGi}
        \\
        \stEnvAnnotGenericSym = \ltsCrashSmall{\mpS}{\roleP}
        \\
        \rolesCi = \rolesC \cup \setenum{\roleP}
        \label{eq:soundness:movecrash_rolesCi}
        \\
        \gtGi = \gtCrashRole{\gtG}{\roleP}
      \end{gather}

      Let $\stEnvi; \qEnvi = \stEnvUpd{\stEnv}{\roleP}{\stStop};
      \stEnvUpd{\qEnv}{\cdot, \roleP}{\stQUnavail}$.
      We show
        $\stEnvAssoc{\gtWithCrashedRoles{\rolesCii}{\gtGii}}{\stEnvi;
        \qEnvi}{\rolesR}$\;
      and
      \;$\stEnv; \qEnv \stEnvMoveGenAnnot \stEnvi; \qEnvi$.

      For the first part:
      By association~\eqref{eq:soundness:movecrash_assoc_pre}, we know that
      $\stEnvApp{\stEnv}{\roleP} \stSub \gtProj[\rolesR]{\gtG}{\roleP}$.
      By   $\gtG \neq \gtRec{\gtRecVar}{\gtGi}$, $\roleP \in \gtRoles{\gtG}$,
      and~\autoref{lem:in-roles-not-end}, we know that
      $\gtProj[\rolesR]{\gtG}{\roleP} \neq \stEnd$, which gives $\stEnvApp{\stEnv}{\roleP} \neq \stEnd$.
      By $\roleP \in \gtRoles{\gtG}$, we also know that $\roleP \notin \rolesC$, which gives
      $\stEnvApp{\stEnv}{\roleP} \neq \stStop$.
      Thus we apply \inferrule{\iruleTCtxCrash} to obtain the thesis.

      For the second part:
        \begin{enumerate}[label=(A\arabic*)]
          \item
            From association~\eqref{eq:soundness:movecrash_assoc_pre}, we know
            that
            $\forall \roleQ \in \gtRoles{\gtG} : \stEnvApp{\stEnv}{\roleQ}
            \stSub \gtProj[\rolesR]{\gtG}{\roleQ}$.

            By~\autoref{lem:gtype:crash-remove-role}, we know
            $\gtRoles{\gtCrashRole{\gtG}{\roleP}} \subseteq \gtRoles{\gtG}$.

            For any $\roleQ \in \gtRoles{\gtCrashRole{\gtG}{\roleP}}$,
            we apply~\autoref{lem:proj-non-crashing-role-preserve}, to obtain
            \(
            \gtProj[\rolesR]{\gtG}{\roleQ}
            \stSub
            \gtProj[\rolesR]{(\gtCrashRole{\gtG}{\roleP})}{\roleQ}
            \)

            Thus we have, by transitivity of subtyping,
            $\forall \roleQ \in \gtRoles{\gtCrashRole{\gtG}{\roleP}} :
            \stEnvApp{\stEnv}{\roleQ}
            \stSub \gtProj[\rolesR]{\gtG}{\roleQ}
            \stSub \gtProj[\rolesR]{(\gtCrashRole{\gtG}{\roleP})}{\roleQ}
            $, as required.

          \item
            From association~\eqref{eq:soundness:movecrash_assoc_pre}, we know
            that
            $\forall \roleQ \in \rolesC : \stEnvApp{\stEnv}{\roleQ} = \stStop$,
            and they are unchanged in $\stEnvi$.

            Moreover, we have updated $\stEnvApp{\stEnvi}{\roleP} = \stStop$.

            This completes the consideration of the set $\rolesCi$ \eqref{eq:soundness:movecrash_rolesCi}.
          \item
            No change here.
          \item
            By~\autoref{lem:queue-assoc-crash}.
        \end{enumerate}
    \item Case \inferrule{\iruleGtMoveRec}:

      From the premise we have
      \begin{gather}
        \stEnvAssoc{\gtWithCrashedRoles{\rolesC}{\gtG}}{\stEnv; \qEnv}{\rolesR}
        \\
        \gtG = \gtRec{\gtRecVar}{\gtG[0]}
        \\
        \gtWithCrashedRoles{\rolesC}{\gtG}
        \gtMove{\rolesR}
        \\
        \gtWithCrashedRoles{\rolesC}{\gtG[0]{}\subst{\gtRecVar}{\gtRec{\gtRecVar}{\gtG[0]}}}
        \gtMove[\stEnvAnnotGenericSym]{\rolesR}
        \gtWithCrashedRoles{\rolesCi}{\gtGi}
      \end{gather}

      By~\autoref{prop:gt-lts-unfold-once}, we have
      $\stEnvAssoc{
        \gtWithCrashedRoles{\rolesC}{\gtG[0]{}\subst{\gtRecVar}{\gtRec{\gtRecVar}{\gtG[0]}}}
      }{\stEnv; \qEnv}{\rolesR}
      $, and we can apply inductive hypothesis to obtain the desired result.

    \item Case \inferrule{\iruleGtMoveOut}:

      From the premise we have
      \begin{gather}
        \stEnvAssoc{\gtWithCrashedRoles{\rolesC}{\gtG}}{\stEnv; \qEnv}{\rolesR}
        \label{eq:soundness:moveout_assoc_pre}
        \\
        \gtG =
          \gtCommSmall{\roleP}{\roleQ}{i \in I}{\gtLab[i]}{\tyGround[i]}{\gtG[i]}
        \\
        \gtWithCrashedRoles{\rolesC}{\gtG}
        \gtMove{\rolesR}
        \\
        j \in I
        \\
        \gtLab[j] \neq \gtCrashLab
        \\
        \stEnvAnnotGenericSym =
          \stEnvOutAnnotSmall{\roleP}{\roleQ}{\stChoice{\gtLab[j]}{\tyGround[j]}}
        \\
        \rolesCi = \rolesC
        \\
        \gtGi =
          \gtCommTransit{\roleP}{\roleQ}{i \in I}{\gtLab[i]}{\tyGround[i]}{\gtG[i]}{j}
      \end{gather}
      By association \eqref{eq:soundness:moveout_assoc_pre} and $\roleP \in \gtRoles{\gtG}$,
      we know that $\stEnvApp{\stEnv}{\roleP} \stSub
      \gtProj[\rolesR]{\gtG}{\roleP} =$ \linebreak
      $
        \stIntSum{\roleQ}{i \in \stIdxRemoveCrash{I}{\stLab[i]}}{
        \stChoice{\stLab[i]}{\tyGround[i]} \stSeq (\gtProj[\rolesR]{\gtG[i]}{\roleP})%
      }$. Then by~\autoref{lem:subtyping-invert}, we obtain that
      $\stEnvApp{\stEnv}{\roleP} = $ \linebreak
      $\stIntSum{\roleQ}{i \in I'}{
        \stChoice{\stLab[i]}{\tyGround[i]} \stSeq \stT[i]%
      }$, where $I' \subseteq \stIdxRemoveCrash{I}{\stLab[i]}$ and $\forall i \in I': \stT[i] \stSub \gtProj[\rolesR]{\gtG[i]}{\roleP}$.
      Note that here for any $i \in I': \stLab[i] = \gtLab[i]$.

      Since the $\stCrashLab$ label cannot appear in the internal choices, it holds that for any $i \in I'$, $\stLab[i] \neq \stCrashLab$.
      Therefore, with $\forall i \in I': \stLab[i] = \gtLab[i]$, we can set $j \in I'$ with $\stLab[j] = \gtLab[j] \neq \stCrashLab$ and
      $  \stEnvAnnotGenericSym =
          \stEnvOutAnnotSmall{\roleP}{\roleQ}{\stChoice{\gtLab[j]}{\tyGround[j]}}
          =
          \stEnvOutAnnotSmall{\roleP}{\roleQ}{\stChoice{\stLab[j]}{\tyGround[j]}}$.

        Let $\stEnvi; \qEnvi = \stEnvUpd{\stEnv}{\roleP}{\stT[j]};  \stEnvUpd{\qEnv}{\roleP, \roleQ}{ \stQCons{
          \stEnvApp{\qEnv}{\roleP, \roleQ}
        }{
          \stQMsg{\stLab[j]}{\tyGround[j]}
        }}$.
      We show that
      $\stEnv; \qEnv \stEnvMoveGenAnnot \stEnvi; \qEnvi$\;
      and
        \;$\stEnvAssoc{\gtWithCrashedRoles{\rolesCi}{\gtGi}}{\stEnvi;
        \qEnvi}{\rolesR}$.

    For the first part: we apply \inferrule{\iruleTCtxOut} on those we have:
    $\stEnvApp{\stEnv}{\roleP} =
      \stIntSum{\roleQ}{i \in I'}{
        \stChoice{\stLab[i]}{\tyGround[i]} \stSeq \stT[i]%
      }$ and $j \in I'$, to obtain the thesis.

   For the second part:
     \begin{enumerate}[label=(A\arabic*)]
     \item   We want to show $\forall \roleR \in \gtRoles{\gtGi}: \stEnvApp{\stEnvi}{\roleR}
           \stSub \gtProj[\rolesR]{\gtGi}{\roleR}$. We consider three subcases:
           \begin{itemize}[leftmargin=*]
           \item $\roleR = \roleP$: since
           $\stEnvApp{\stEnvi}{\roleP} = \stT[j]$, $j \in I'$, and $\forall i \in I': \stT[i] \stSub \gtProj[\rolesR]{\gtG[i]}{\roleP}$,
           we have $\stEnvApp{\stEnvi}{\roleP} \stSub
           \gtProj[\rolesR]{\gtG[j]}{\roleP} =  \gtProj[\rolesR]{\gtGi}{\roleP}$, as desired.

           \item $\roleR = \roleQ$:
                     by association~\eqref{eq:soundness:moveout_assoc_pre},
                     we have that
                     $\stEnvApp{\stEnvi}{\roleQ} = \stEnvApp{\stEnv}{\roleQ}
                     \stSub  \gtProj[\rolesR]{\gtG}{\roleQ}=$ \linebreak
                     $ \stExtSum{\roleP}{i \in I}{
        \stChoice{\stLab[i]}{\tyGround[i]} \stSeq (\gtProj[\rolesR]{\gtG[i]}{\roleQ})%
      } =  \gtProj[\rolesR]{\gtGi}{\roleQ}
      $, as desired.
           \item $\roleR \neq \roleQ$ and $\roleR \neq \roleP$:
           by association~\eqref{eq:soundness:moveout_assoc_pre},
            it holds that $\stEnvApp{\stEnvi}{\roleR} =  \stEnvApp{\stEnv}{\roleR} \stSub
           \gtProj[\rolesR]{\gtG}{\roleR} =  \stMerge{i \in I}{\gtProj[\rolesR]{\gtG[i]}{\roleR}}
           =
            \gtProj[\rolesR]{\gtGi}{\roleR}$, as desired.
           \end{itemize}

     \item No change here.
     \item No change here.
     \item  We are left to show that
      $\qEnvi =  \stEnvUpd{\qEnv}{\roleP, \roleQ}{ \stQCons{
          \stEnvApp{\qEnv}{\roleP, \roleQ}
        }{
          \stQMsg{\stLab[j]}{\tyGround[j]}
        }}$ is associated with \linebreak $\gtWithCrashedRoles{\rolesC}
     {\gtCommTransit{\roleP}{\roleQ}{i \in I}{\gtLab[i]}{\tyGround[i]}{\gtG[i]}{j}}$.

     Since $\qEnv$ is associated with $\gtWithCrashedRoles{\rolesC}
     {\gtCommSmall{\roleP}{\roleQ}{i \in I}{\gtLab[i]}{\tyGround[i]}{\gtG[i]}}$,
     by~\autoref{def:assoc-queue}, we have $\stEnvApp{\qEnv}{\roleP, \roleQ} =
     \stQEmpty$ and\linebreak
     $\forall i \in I: \qEnv \text{~is associated with~}
      \gtWithCrashedRoles{\rolesC}{\gtG[i]}$.

     Since $\gtLab[j] \neq \gtCrashLab$, we just need to show that
     $\stEnvApp{\qEnvi}{\roleP, \roleQ} = \stQMsg{\gtLab[j]}{\tyGround[j]}$,
     which follows directly from
     $\qEnvi =  \stEnvUpd{\qEnv}{\roleP, \roleQ}{ \stQCons{
          \stEnvApp{\qEnv}{\roleP, \roleQ}
        }{
          \stQMsg{\stLab[j]}{\tyGround[j]}
        }}$, $\gtLab[j] = \stLab[j]$, and $\stEnvApp{\qEnv}{\roleP, \roleQ} = \stQEmpty$,
        and
      $\forall i \in I: \stEnvUpd{\qEnvi}{\roleP, \roleQ}{\stQEmpty}$ is associated with
      $\gtWithCrashedRoles{\rolesC}{\gtG[i]}$, which follows from
      $\stEnvUpd{\qEnvi}{\roleP, \roleQ}{\stQEmpty}  = \qEnv$ and
       $\forall i \in I: \qEnv \text{~is associated with~}
      \gtWithCrashedRoles{\rolesC}{\gtG[i]}$.
     \end{enumerate}

 \item Case \inferrule{\iruleGtMoveIn}:

      From the premise we have
      \begin{gather}
        \stEnvAssoc{\gtWithCrashedRoles{\rolesC}{\gtG}}{\stEnv; \qEnv}{\rolesR}
         \label{eq:soundness:movein_assoc_pre}
        \\
        \gtG =
          \gtCommTransit{\rolePMaybeCrashed}{\roleQ}{i \in I}{\gtLab[i]}{\tyGround[i]}{\gtG[i]}{j}
        \\
        \gtWithCrashedRoles{\rolesC}{\gtG}
        \gtMove{\rolesR}
        \\
        j \in I
        \\
        \gtLab[j] \neq \gtCrashLab
        \\
        \stEnvAnnotGenericSym =
          \stEnvInAnnotSmall{\roleQ}{\roleP}{\stChoice{\gtLab[j]}{\tyGround[j]}}
        \\
        \rolesCi = \rolesC
        \\
        \gtGi = \gtG[j]
      \end{gather}
     By association \eqref{eq:soundness:movein_assoc_pre} and $\roleQ \in \gtRoles{\gtG}$,
      we know that $\stEnvApp{\stEnv}{\roleQ} \stSub \gtProj[\rolesR]{\gtG}{\roleQ} = $ \linebreak
        $\stExtSum{\roleP}{i \in I}{
        \stChoice{\stLab[i]}{\tyGround[i]} \stSeq (\gtProj[\rolesR]{\gtG[i]}{\roleQ})%
      }$. Note that here for any $i \in I: \stLab[i] = \gtLab[i]$. Furthermore, by~\autoref{lem:subtyping-invert},
      we obtain that $\stEnvApp{\stEnv}{\roleQ} =
      \stExtSum{\roleP}{i \in I'}{
        \stChoice{\stLab[i]}{\tyGround[i]} \stSeq \stT[i]%
      }$, where $I \subseteq I'$ and $\forall i \in I: \stT[i] \stSub \gtProj[\rolesR]{\gtG[i]}{\roleQ}$.
      From association \eqref{eq:soundness:movein_assoc_pre}, we also get that
      $\stEnvApp{\qEnv}{\roleP, \roleQ} = \stQCons{\stQMsg{\gtLab[j]}{\tyGround[j]}}{\stQ}
      = \stQCons{\stQMsg{\stLab[j]}{\tyGround[j]}}{\stQ}$.

        Let $\stEnvi; \qEnvi = \stEnvUpd{\stEnv}{\roleQ}{\stT[j]};  \stEnvUpd{\qEnv}{\roleP, \roleQ}{\stQ}$.
      We show
      $\stEnv; \qEnv \stEnvMoveGenAnnot \stEnvi; \qEnvi$\;
      and
        \;$\stEnvAssoc{\gtWithCrashedRoles{\rolesCi}{\gtGi}}{\stEnvi;
        \qEnvi}{\rolesR}$.

    For the first part: we apply \inferrule{\iruleTCtxIn} on those we have:
    $\stEnvApp{\stEnv}{\roleQ} =
      \stExtSum{\roleP}{i \in I'}{
        \stChoice{\stLab[i]}{\tyGroundi[i]} \stSeq \stT[i]
      }$, $j \in I \subseteq I'$, $\stLab[j] = \gtLab[j]$, and $\stEnvApp{\qEnv}{\roleP, \roleQ} =
      \stQCons{\stQMsg{\stLab[j]}{\tyGround[j]}}{\stQ}$,
    to obtain the thesis.

    For the second part:
     \begin{enumerate}[label=(A\arabic*)]
     \item
       We want to show $\forall \roleR \in \gtRoles{\gtG[j]}: \stEnvApp{\stEnvi}{\roleR}
           \stSub \gtProj[\rolesR]{\gtG[j]}{\roleR}$. We consider three subcases:
           \begin{itemize}[leftmargin=*]
           \item $\roleR = \roleQ$ (meaning that $\roleQ \in \gtRoles{\gtG[j]}$): since
           $\stEnvApp{\stEnvi}{\roleQ} = \stT[j]$ and $\forall i \in I: \stT[i] \stSub \gtProj[\rolesR]{\gtG[i]}{\roleQ}$,
           we have $\stEnvApp{\stEnvi}{\roleQ} \stSub
           \gtProj[\rolesR]{\gtG[j]}{\roleQ}$, as desired.
           \item $\roleR = \roleP$ (meaning that $\rolePMaybeCrashed = \roleP$ and $\roleP \in \gtRoles{\gtG[j]}$):
                     by association~\eqref{eq:soundness:movein_assoc_pre},
                     we have
                     $\stEnvApp{\stEnvi}{\roleP} = \stEnvApp{\stEnv}{\roleP}
                     \stSub  \gtProj[\rolesR]{\gtG}{\roleP} =  \gtProj[\rolesR]{\gtG[j]}{\roleP}$, as desired.
           \item $\roleR \neq \roleQ$ and $\roleR \neq \roleP$:
           by association~\eqref{eq:soundness:movein_assoc_pre},
            it holds that $\stEnvApp{\stEnvi}{\roleR} =  \stEnvApp{\stEnv}{\roleR} \stSub
           \gtProj[\rolesR]{\gtG}{\roleR} =  \stMerge{i \in I}{\gtProj[\rolesR]{\gtG[i]}{\roleR}}$.
           Then, by applying~\autoref{lem:merge-subtyping} and transitivity of subtyping,  we conclude with $\stEnvApp{\stEnvi}{\roleR}
        \stSub \gtProj[\rolesR]{\gtG[j]}{\roleR}$, as desired.
           \end{itemize}
     \item No change here.
     \item No change here if $\roleQ \in \gtRoles{\gtG[j]}$ and $\roleP \in \gtRoles{\gtG[j]}$. Otherwise, if $\roleQ \notin \gtRoles{\gtG[j]}$: with
            $\roleQ \notin \gtRolesCrashed{\gtG[j]}$, by~\autoref{lem:not-in-roles-end}, we have $\gtProj[\rolesR]{\gtG[j]}{\roleQ} = \stEnd$.
            Furthermore, by the fact that $\forall i \in I: \stT[i] \stSub \gtProj[\rolesR]{\gtG[i]}{\roleQ} $,
            it holds that $\stT[j] = \stEnd$, and thus,
          $\stEnvApp{\stEnvi}{\roleQ} = \stEnd$, as desired. The argument for $\roleP \notin  \gtRoles{\gtG[j]}$ and
          $\rolePMaybeCrashed = \roleP$ follows similarly.
     \item  Since $\qEnv$ is associated with $\gtWithCrashedRoles{\rolesC}
     {\gtCommTransit{\rolePMaybeCrashed}{\roleQ}{i \in I}{\gtLab[i]}{\tyGround[i]}{\gtG[i]}{j}}$
     and $\gtLab[j] \neq \gtCrashLab$, by~\autoref{def:assoc-queue},  we have that
     $\stEnvApp{\qEnv}{\roleP, \roleQ} = \stQCons{\stQMsg{\gtLab[j]}{\tyGround[j]}}{\stQ}$ and
            $\forall i \in I: \stEnvUpd{\qEnv}{\roleP, \roleQ}{\stQ} \text{~is
            associated with~}\linebreak
      \gtWithCrashedRoles{\rolesC}{\gtG[i]}$, which follows that $\qEnvi = \stEnvUpd{\qEnv}{\roleP, \roleQ}{\stQ}$
      is associated with
       $\gtWithCrashedRoles{\rolesC}{\gtG[j]}$, as desired.
     \end{enumerate}

    \item Case \inferrule{\iruleGtMoveCrDe}:

      From the premise we have
      \begin{gather}
        \stEnvAssoc{\gtWithCrashedRoles{\rolesC}{\gtG}}{\stEnv; \qEnv}{\rolesR}
        \label{eq:soundness:detectcrash_assoc_pre}
        \\
        \gtG =
          \gtCommTransit{\rolePCrashed}{\roleQ}{i \in I}{\gtLab[i]}{\tyGround[i]}{\gtG[i]}{j}
        \\
        \gtWithCrashedRoles{\rolesC}{\gtG}
        \gtMove{\rolesR}
        \\
        j \in I
        \label{eq:soundness:detectcrash_index}
        \\
        \gtLab[j] = \gtCrashLab
        \label{eq:soundness:detectcrash_crash_branch}
        \\
        \stEnvAnnotGenericSym =
          \ltsCrDe{\mpS}{\roleQ}{\roleP}
        \\
        \rolesCi = \rolesC
        \\
        \gtGi = \gtG[j]
      \end{gather}
      By association \eqref{eq:soundness:detectcrash_assoc_pre} and $\roleQ \in \gtRoles{\gtG}$,
      we know that $\stEnvApp{\stEnv}{\roleQ} \stSub \gtProj[\rolesR]{\gtG}{\roleQ} = $ \linebreak 
        $\stExtSum{\roleP}{i \in I}{
        \stChoice{\stLab[i]}{\tyGround[i]} \stSeq (\gtProj[\rolesR]{\gtG[i]}{\roleQ})%
      }$. Note that here for any $i \in I: \stLab[i] = \gtLab[i]$. Furthermore, by~\autoref{lem:subtyping-invert},
      we obtain that $\stEnvApp{\stEnv}{\roleQ} =
      \stExtSum{\roleP}{i \in I'}{
        \stChoice{\stLab[i]}{\tyGroundi[i]} \stSeq \stT[i]%
      }$, where $I \subseteq I'$ and $\forall i \in I: \tyGround[i] =
      \tyGroundi[i]$ and $\stT[i] \stSub \gtProj[\rolesR]{\gtG[i]}{\roleQ}$.

       Let $\stEnvi; \qEnvi = \stEnvUpd{\stEnv}{\roleQ}{\stT[j]}; \qEnv$.
      We show
      $\stEnv; \qEnv \stEnvMoveGenAnnot \stEnvi; \qEnvi$\;
      and
        \;$\stEnvAssoc{\gtWithCrashedRoles{\rolesCi}{\gtGi}}{\stEnvi;
        \qEnvi}{\rolesR}$.

        For the first part: By association~\eqref{eq:soundness:detectcrash_assoc_pre},
        $\roleP \in \rolesC$, and
        $\gtLab[j] = \gtCrashLab$~\eqref{eq:soundness:detectcrash_crash_branch},
        we know that $\stEnvApp{\stEnv}{\roleP} = \stStop$ and
        $\stEnvApp{\qEnv}{\roleP, \roleQ} = \stQEmpty$.
        Since $j \in I$~\eqref{eq:soundness:detectcrash_index},
         $\gtLab[j] = \gtCrashLab$~\eqref{eq:soundness:detectcrash_crash_branch}, and
         $\forall i \in I: \gtLab[i] = \stLab[i]$, we have $\stLab[j] = \stCrashLab$.  We also get
         $j \in I'$ from $j \in I$ and $I \subseteq I'$. Thus, together with $\stEnvApp{\stEnv}{\roleQ} =
      \stExtSum{\roleP}{i \in I'}{
        \stChoice{\stLab[i]}{\tyGroundi[i]} \stSeq \stT[i]%
      }$, we apply \inferrule{\iruleTCtxCrashDetect} to obtain the thesis.

     For the second part:
        \begin{enumerate}[label=(A\arabic*)]
          \item
          We want to show $\forall \roleR \in \gtRoles{\gtG[j]}: \stEnvApp{\stEnvi}{\roleR}
           \stSub \gtProj[\rolesR]{\gtG[j]}{\roleR}$. We consider two subcases:
           \begin{itemize}[leftmargin=*]
           \item $\roleR = \roleQ$ (meaning that $\roleQ \in \gtRoles{\gtG[j]}$): since
           $\stEnvApp{\stEnvi}{\roleQ} = \stT[j]$, $\forall i \in I: \stT[i] \stSub \gtProj[\rolesR]{\gtG[i]}{\roleQ}$, and
           $j \in I$~\eqref{eq:soundness:detectcrash_index}, we have $\stEnvApp{\stEnvi}{\roleQ} \stSub
           \gtProj[\rolesR]{\gtG[j]}{\roleQ}$, as desired.
           \item $\roleR \neq \roleQ$: since $\roleP \in \rolesC$, we know $\roleP \notin \gtRoles{\gtG[j]}$, and thus, $\roleR \neq \roleP$.
           Furthermore, by association~\eqref{eq:soundness:detectcrash_assoc_pre}, it holds that
           $\stEnvApp{\stEnvi}{\roleR} = \stEnvApp{\stEnv}{\roleR} \stSub
           \gtProj[\rolesR]{\gtG}{\roleR} =  \stMerge{i \in I}{\gtProj[\rolesR]{\gtG[i]}{\roleR}}$.
           Then, by applying~\autoref{lem:merge-subtyping} and transitivity of subtyping,  we can conclude with $\stEnvApp{\stEnvi}{\roleR}
        \stSub \gtProj[\rolesR]{\gtG[j]}{\roleR}$, as desired.
           \end{itemize}

          \item
            No change here.
          \item
            No change here if $\roleQ \in \gtRoles{\gtG[j]}$. Otherwise, if $\roleQ \notin \gtRoles{\gtG[j]}$: with
            $\roleQ \notin \gtRolesCrashed{\gtG[j]}$, by~\autoref{lem:not-in-roles-end}, we have $\gtProj[\rolesR]{\gtG[j]}{\roleQ} = \stEnd$.
            Furthermore, by the fact that $\forall i \in I: \stT[i] \stSub \gtProj[\rolesR]{\gtG[i]}{\roleQ} $,
            it holds that $\stT[j] = \stEnd$, and thus,
          $\stEnvApp{\stEnvi}{\roleQ} = \stEnd$, as desired.

          \item
             Since $\qEnv$ is associated with $\gtWithCrashedRoles{\rolesC}{\gtCommTransit{\rolePCrashed}{\roleQ}{i \in
            I}{\gtLab[i]}{\tyGround[i]}{\gtG[i]}{j} }$ and $\gtLab[j] = \gtCrashLab$, by~\autoref{def:assoc-queue},  we have that
            $\forall i \in I: \qEnv \text{~is associated with~}
      \gtWithCrashedRoles{\rolesC}{\gtG[i]}$, which follows that $\qEnvi = \qEnv$ is associated with
       $\gtWithCrashedRoles{\rolesC}{\gtG[j]}$, as desired.
        \end{enumerate}

       \item Case \inferrule{\iruleGtMoveOrph}:

      From the premise we have
      \begin{gather}
        \stEnvAssoc{\gtWithCrashedRoles{\rolesC}{\gtG}}{\stEnv; \qEnv}{\rolesR}
        \label{eq:soundness:moveorph_assoc_pre}
        \\
        \gtG =
          \gtCommSmall{\roleP}{\roleQCrashed}{i \in I}{\gtLab[i]}{\tyGround[i]}{\gtG[i]}
        \\
        \gtWithCrashedRoles{\rolesC}{\gtG}
        \gtMove{\rolesR}
        \\
        j \in I
        \\
        \gtLab[j] \neq \gtCrashLab
        \\
        \stEnvAnnotGenericSym =
          \stEnvOutAnnotSmall{\roleP}{\roleQ}{\stChoice{\gtLab[j]}{\tyGround[j]}}
        \\
        \rolesCi = \rolesC
        \\
        \gtGi = \gtG[j]
      \end{gather}
        By association \eqref{eq:soundness:moveorph_assoc_pre} and $\roleP \in \gtRoles{\gtG}$,
      we know that $\stEnvApp{\stEnv}{\roleP} \stSub
      \gtProj[\rolesR]{\gtG}{\roleP} = \linebreak
        \stIntSum{\roleQ}{i \in \stIdxRemoveCrash{I}{\stLab[i]}}{
        \stChoice{\stLab[i]}{\tyGround[i]} \stSeq (\gtProj[\rolesR]{\gtG[i]}{\roleP})%
      }$. Then by~\autoref{lem:subtyping-invert}, we obtain that
      $\stEnvApp{\stEnv}{\roleP} = $ \linebreak 
      $\stIntSum{\roleQ}{i \in I'}{
        \stChoice{\stLab[i]}{\tyGround[i]} \stSeq \stT[i]%
      }$, where $I' \subseteq \stIdxRemoveCrash{I}{\stLab[i]}$ and $\forall i \in I': \stT[i] \stSub \gtProj[\rolesR]{\gtG[i]}{\roleP}$.
      Note that here for any $i \in I': \stLab[i] = \gtLab[i]$.

      Since the $\stCrashLab$ label cannot appear in the internal choices, it holds that for any $i \in I'$, $\stLab[i] \neq \stCrashLab$.
      Therefore, with $\forall i \in I': \stLab[i] = \gtLab[i]$, we can set $j \in I'$ with $\stLab[j] = \gtLab[j] \neq \stCrashLab$ and
      $  \stEnvAnnotGenericSym =
          \stEnvOutAnnotSmall{\roleP}{\roleQ}{\stChoice{\gtLab[j]}{\tyGround[j]}}
          =
          \stEnvOutAnnotSmall{\roleP}{\roleQ}{\stChoice{\stLab[j]}{\tyGround[j]}}$.

        Let $\stEnvi; \qEnvi = \stEnvUpd{\stEnv}{\roleP}{\stT[j]};  \stEnvUpd{\qEnv}{\roleP, \roleQ}{ \stQCons{
          \stEnvApp{\qEnv}{\roleP, \roleQ}
        }{
          \stQMsg{\stLab[j]}{\tyGround[j]}
        }}$.
      We show that
      $\stEnv; \qEnv \stEnvMoveGenAnnot \stEnvi; \qEnvi$\;
      and\linebreak
        \;$\stEnvAssoc{\gtWithCrashedRoles{\rolesCi}{\gtGi}}{\stEnvi;
        \qEnvi}{\rolesR}$.

    For the first part: we apply \inferrule{\iruleTCtxOut} on those we have:
    $\stEnvApp{\stEnv}{\roleP} =
      \stIntSum{\roleQ}{i \in I'}{
        \stChoice{\stLab[i]}{\tyGround[i]} \stSeq \stT[i]%
      }$ and $j \in I'$, to obtain the thesis.

   For the second part:
     \begin{enumerate}[label=(A\arabic*)]
     \item   We want to show $\forall \roleR \in \gtRoles{\gtGi}: \stEnvApp{\stEnvi}{\roleR}
           \stSub \gtProj[\rolesR]{\gtGi}{\roleR}$. We consider two subcases:
           \begin{itemize}[leftmargin=*]
           \item $\roleR = \roleP$ (meaning that $\roleP \in \gtRoles{\gtG[j]}$): since
           $\stEnvApp{\stEnvi}{\roleP} = \stT[j]$, $\forall i \in I': \stT[i] \stSub \gtProj[\rolesR]{\gtG[i]}{\roleP}$, and
           $j \in I'$, we have $\stEnvApp{\stEnvi}{\roleP} \stSub
           \gtProj[\rolesR]{\gtG[j]}{\roleP}$, as desired.

           \item $\roleR \neq \roleP$: since $\roleQ \in \rolesC$, we know $\roleQ \notin \gtRoles{\gtG[j]}$, and thus, $\roleR \neq \roleQ$.
           Furthermore, by association~\eqref{eq:soundness:moveorph_assoc_pre}, it holds that
           $\stEnvApp{\stEnvi}{\roleR} = \stEnvApp{\stEnv}{\roleR} \stSub
           \gtProj[\rolesR]{\gtG}{\roleR} =  \stMerge{i \in I}{\gtProj[\rolesR]{\gtG[i]}{\roleR}}$.
           Then, by applying~\autoref{lem:merge-subtyping} and transitivity of subtyping,  we can conclude with $\stEnvApp{\stEnvi}{\roleR}
        \stSub \gtProj[\rolesR]{\gtG[j]}{\roleR}$, as desired.
           \end{itemize}

     \item No change here.
     \item
     No change here if $\roleP \in \gtRoles{\gtG[j]}$. Otherwise, if $\roleP \notin \gtRoles{\gtG[j]}$: with
            $\roleP \notin \gtRolesCrashed{\gtG[j]}$, by~\autoref{lem:not-in-roles-end}, we have $\gtProj[\rolesR]{\gtG[j]}{\roleP} = \stEnd$.
            Furthermore, by the fact that $\forall i \in I': \stT[i] \stSub \gtProj[\rolesR]{\gtG[i]}{\roleP} $,
            it holds that $\stT[j] = \stEnd$, and thus,
          $\stEnvApp{\stEnvi}{\roleP} = \stEnd$, as desired.
     \item  We are left to show
      $\qEnvi =  \stEnvUpd{\qEnv}{\roleP, \roleQ}{ \stQCons{
          \stEnvApp{\qEnv}{\roleP, \roleQ}
        }{
          \stQMsg{\stLab[j]}{\tyGround[j]}
        }}$ is associated with $\gtWithCrashedRoles{\rolesC}
     {\gtG[j]}$.

     Since $\qEnv$ is associated with $\gtWithCrashedRoles{\rolesC}
     {\gtCommSmall{\roleP}{\roleQCrashed}{i \in I}{\gtLab[i]}{\tyGround[i]}{\gtG[i]}}$,
     by~\autoref{def:assoc-queue}, we have $\stEnvApp{\qEnv}{\roleP, \roleQ} = \stQUnavail$ and
     $\forall i \in I: \qEnv \text{~is associated with~}
      \gtWithCrashedRoles{\rolesC}{\gtG[i]}$. It follows from $\stEnvApp{\qEnv}{\roleP, \roleQ} = \stQUnavail$
       that
      $\qEnvi = \stEnvUpd{\qEnv}{\roleP, \roleQ}{ \stQCons{
          \stEnvApp{\qEnv}{\roleP, \roleQ}
        }{
          \stQMsg{\stLab[j]}{\tyGround[j]}
        }}
        =
        \stEnvUpd{\qEnv}{\roleP, \roleQ}{ \stQCons{
        \stQUnavail
        }{\stQMsg{\stLab[j]}{\tyGround[j]}
        } }
        =
         \stEnvUpd{\qEnv}{\roleP, \roleQ}{\stQUnavail}
         =
        \qEnv$. Hence, with $\forall i \in I: \qEnv \text{~is associated with~}
      \gtWithCrashedRoles{\rolesC}{\gtG[i]}$, we conclude that $\qEnvi = \qEnv$ is associated with
      $\gtWithCrashedRoles{\rolesC}
     {\gtG[j]}$, as desired.
  \end{enumerate}

 \item Case \inferrule{\iruleGtMoveCtx}:

      From the premise we have
      \begin{gather}
        \stEnvAssoc{\gtWithCrashedRoles{\rolesC}{\gtG}}{\stEnv; \qEnv}{\rolesR}
        \label{eq:soundness:context_1_assoc_pre}
        \\
        \gtG =
          \gtCommSmall{\roleP}{\roleQMaybeCrashed}{i \in I}{\gtLab[i]}{\tyGround[i]}{\gtG[i]}
        \\
        \gtWithCrashedRoles{\rolesC}{\gtG}
        \gtMove{\rolesR}
        \\
        \forall i \in I :
        \gtWithCrashedRoles{\rolesC}{\gtG[i]}
        \gtMove[\stEnvAnnotGenericSym]{\rolesR}
        \gtWithCrashedRoles{\rolesCi}{\gtGi[i]}
        \\
        \ltsSubject{\stEnvAnnotGenericSym} \notin \setenum{\roleP, \roleQ}
        \\
        \gtGi =
          \gtCommSmall{\roleP}{\roleQMaybeCrashed}{i \in I}{\gtLab[i]}{\tyGround[i]}{\gtGi[i]}
      \end{gather}
      We consider two subcases depending on whether $\roleQ$ has crashed.
      \begin{itemize}[leftmargin=*]
      \item $\roleQMaybeCrashed = \roleQ$:

      For $\roleP$, we have that
     $\gtProj[\rolesR]{\gtG}{\roleP} =
       \stIntSum{\roleQ}{i \in \stIdxRemoveCrash{I}{\stLab[i]}}{ %
        \stChoice{\stLab[i]}{\tyGround[i]} \stSeq (\gtProj[\rolesR]{\gtG[i]}{\roleP})%
      }$ and that $ \gtProj[\rolesR]{\gtGi}{\roleP} = \linebreak
       \stIntSum{\roleQ}{i \in \stIdxRemoveCrash{I}{\stLab[i]}}{ %
        \stChoice{\stLab[i]}{\tyGround[i]} \stSeq (\gtProj[\rolesR]{\gtGi[i]}{\roleP})%
      }$. For $\roleQ$, we have $\gtProj[\rolesR]{\gtG}{\roleQ} =
       \stExtSum{\roleP}{i \in I}{%
        \stChoice{\stLab[i]}{\tyGround[i]} \stSeq (\gtProj[\rolesR]{\gtG[i]}{\roleQ})%
      }$ and $\gtProj[\rolesR]{\gtGi}{\roleQ} =
       \stExtSum{\roleP}{i \in I}{%
        \stChoice{\stLab[i]}{\tyGround[i]} \stSeq (\gtProj[\rolesR]{\gtGi[i]}{\roleQ})%
      }$. Take an arbitrary $j \in I$. %
      Let $\stEnvApp{\stEnv[j]}{\roleP} =
      \gtProj[\rolesR]{\gtG[j]}{\roleP}$, $\stEnvApp{\stEnv[j]}{\roleQ} =
      \gtProj[\rolesR]{\gtG[j]}{\roleQ}$, $\stEnvApp{\stEnv[j]}{\roleR} =
     \stEnvApp{\stEnv}{\roleR}$ for $\roleR \in \dom{\stEnv} \setminus \setenum{\roleP, \roleQ}$, $\qEnv[j] = \qEnv$.
     We show $\stEnvAssoc{\gtWithCrashedRoles{\rolesC}{\gtG[j]}}{\stEnv[j]; \qEnv[j]}{\rolesR}$.
      \begin{enumerate}[label=(A\arabic*)]
      \item We want to show $\forall \roleS \in \gtRoles{\gtG[j]}: \stEnvApp{\stEnv[j]}{\roleS} \stSub \gtProj[\rolesR]{\gtG[j]}{\roleS}$. We consider
      three subcases:
      \begin{itemize}[leftmargin=*]
      \item $\roleS = \roleP$ (meaning that $\roleP \in \gtRoles{\gtG[j]}$): trivial by $\stEnvApp{\stEnv[j]}{\roleP} =
      \gtProj[\rolesR]{\gtG[j]}{\roleP}$ and the reflexivity of subtyping.
      \item $\roleS = \roleQ$ (meaning that $\roleQ \in \gtRoles{\gtG[j]}$): trivial by $\stEnvApp{\stEnv[j]}{\roleQ} =
      \gtProj[\rolesR]{\gtG[j]}{\roleQ}$ and the reflexivity of subtyping.
      \item $\roleS \neq \roleP$ and $\roleS \neq \roleQ$: by association \eqref{eq:soundness:context_1_assoc_pre} and
      $\stEnvApp{\stEnv[j]}{\roleS} =
     \stEnvApp{\stEnv}{\roleS}$, we have $\stEnvApp{\stEnv[j]}{\roleS} \stSub
       \stMerge{i \in I}{\gtProj[\rolesR]{\gtG[i]}{\roleS}}$.
        Then, by~\autoref{lem:merge-subtyping} and transitivity of subtyping,  we conclude  $\stEnvApp{\stEnv[j]}{\roleS}
        \stSub \gtProj[\rolesR]{\gtG[j]}{\roleS}$, as desired.
      \end{itemize}
       \item No change here.
      \item No change here if $\roleP \in \gtRoles{\gtG[j]}$ and $\roleQ \in \gtRoles{\gtG[j]}$. Otherwise, consider the case that $\roleP \notin \gtRoles{\gtG[j]}$:
      with $\roleP \notin \gtRolesCrashed{\gtG[j]}$, by~\autoref{lem:not-in-roles-end}, we have $\gtProj[\rolesR]{\gtG[j]}{\roleP} = \stEnd$.
      Therefore, we know from $\stEnvApp{\stEnv[j]}{\roleP} =
      \gtProj[\rolesR]{\gtG[j]}{\roleP}$ that $ \stEnvApp{\stEnv[j]}{\roleP} = \stEnd$, as required. The argument for
      $\roleQ \notin  \gtRoles{\gtG[j]}$ follows similarly.
      \item Trivial by association~\eqref{eq:soundness:context_1_assoc_pre}, \autoref{def:assoc-queue}, and $\qEnv[j] = \qEnv$.
      \end{enumerate}
      By inductive hypothesis, there exists $\stEnvi[j]; \qEnvi[j]$ such that
      $\stEnv[j]; \qEnv[j]  \stEnvMoveGenAnnot
      \stEnvi[j]; \qEnvi[j]$ and \linebreak
      $\stEnvAssoc{\gtWithCrashedRoles{\rolesCi}{\gtGi[j]}}{\stEnvi[j]; \qEnvi[j]}{\rolesR}$. Since
      $\ltsSubject{\stEnvAnnotGenericSym} \notin \setenum{\roleP, \roleQ}$, we apply~\autoref{lem:stenv-red:trivial-2}, which gives
      $\stEnvApp{\stEnv[j]}{\roleP} = \stEnvApp{\stEnvi[j]}{\roleP}$ and
      $\stEnvApp{\stEnv[j]}{\roleQ} = \stEnvApp{\stEnvi[j]}{\roleQ}$.

      We now construct a configuration $\stEnvi; \qEnvi$ and show
      $\stEnv; \qEnv \stEnvMoveGenAnnot \stEnvi; \qEnvi$ and
      $\stEnvAssoc{\gtWithCrashedRoles{\rolesCi}{\gtGi}}{\stEnvi; \qEnvi}{\rolesR}$.

      Let  $\stEnvApp{\stEnvi}{\roleP}
      = \stEnvApp{\stEnv}{\roleP}$,
       $\stEnvApp{\stEnvi}{\roleQ}
      =
      \stEnvApp{\stEnv}{\roleQ}$,
       $\stEnvApp{\stEnvi}{\roleR} =
      \stMerge{i \in I}{\stEnvApp{\stEnvi[i]}{\roleR}}$ for
    $\roleR \in \dom{\stEnv} \setminus
      \setenum{\roleP, \roleQ}$, $\qEnvi = \qEnvi[j]$ with an arbitrary $j \in I$.

      For the first part that $\stEnv; \qEnv \stEnvMoveGenAnnot \stEnvi; \qEnvi$:

      We know that for any $i \in I$, $\stEnvi[i]$ is obtained from $\stEnv[i]$
      by updating $\ltsSubject{\stEnvAnnotGenericSym}$ to a fixed type $\stT$. It follows that
      for any $i, k \in I$ and for any $\roleR \notin \setenum{\roleP, \roleQ}$, $\stEnvApp{\stEnvi[i]}{\roleR}
      =
      \stEnvApp{\stEnvi[k]}{\roleR}$, and hence, $\stMerge{i \in I}{\stEnvApp{\stEnvi[i]}{\roleR}} =
      \stEnvApp{\stEnvi[j]}{\roleR}$ with an arbitrary $j \in J$. Therefore, we have that $\stEnvi$ is obtained from
      $\stEnv$ by updating $\ltsSubject{\stEnvAnnotGenericSym}$ to $\stEnvApp{\stEnvi[j]}{\ltsSubject{\stEnvAnnotGenericSym}}$,
      i.e., $\stEnvi = \stEnvUpd{\stEnv}{\ltsSubject{\stEnvAnnotGenericSym}}{\stEnvApp{\stEnvi[j]}{\ltsSubject{\stEnvAnnotGenericSym}}}$.

      We apply~\autoref{lem:stenv-red:non-trivial} on $\stEnv[j]; \qEnv[j]
  \stEnvMoveGenAnnot \stEnvi[j]; \qEnvi[j]$,
  $\stEnvAssoc{\gtWithCrashedRoles{\rolesC}{\gtG}}{\stEnv; \qEnv}{\rolesR}$,
  $\stEnvAssoc{\gtWithCrashedRoles{\rolesC}{\gtG[j]}}{\stEnv[j];
  \qEnv[j]}{\rolesR}$, \linebreak
  $\stEnvApp{\stEnv}{\ltsSubject{\stEnvAnnotGenericSym}} = \stEnvApp{\stEnv[j]}{\ltsSubject{\stEnvAnnotGenericSym}}$,
  and
  $\stEnvi =  \stEnvUpd{\stEnv}{\ltsSubject{\stEnvAnnotGenericSym}}{\stEnvApp{\stEnvi[j]}{\ltsSubject{\stEnvAnnotGenericSym}}}$ to get the thesis.

      For the second part that  $\stEnvAssoc{\gtWithCrashedRoles{\rolesCi}{\gtGi}}{\stEnvi; \qEnvi}{\rolesR}$:
      \begin{enumerate}[label=(A\arabic*)]
      \item We want to show $\forall \roleR \in \gtRoles{\gtGi}: \stEnvApp{\stEnvi}{\roleR}
           \stSub \gtProj[\rolesR]{\gtGi}{\roleR}$. We consider three subcases:
           \begin{itemize}[leftmargin=*]
           \item $\roleR = \roleP$: by association \eqref{eq:soundness:context_1_assoc_pre} and
           \autoref{lem:subtyping-invert}, we obtain $\stEnvApp{\stEnv}{\roleP} = 
           \stIntSum{\roleQ}{i \in I'}{ %
        \stChoice{\stLab[i]}{\tyGround[i]} \stSeq \stT[i]%
      }$ where $I' \subseteq \stIdxRemoveCrash{I}{\stLab[i]}$ and $\forall i \in I': \stT[i] \stSub
      \gtProj[\rolesR]{\gtG[i]}{\roleP}$. Since $\stEnvApp{\stEnv[i]}{\roleP} =
       \gtProj[\rolesR]{\gtG[i]}{\roleP} =
      \stEnvApp{\stEnvi[i]}{\roleP}$ for each $i \in I$, with association
      $\stEnvAssoc{\gtWithCrashedRoles{\rolesCi}{\gtGi[i]}}{\stEnvi[i]; \qEnvi[i]}{\rolesR}$, we get
      $\gtProj[\rolesR]{\gtG[i]}{\roleP} = \stEnvApp{\stEnvi[i]}{\roleP} \stSub
      \gtProj[\rolesR]{\gtGi[i]}{\roleP}$. Then by transitivity of subtyping,
      it holds that $\forall i \in I': \stT[i] \stSub \gtProj[\rolesR]{\gtGi[i]}{\roleP}$, which follows that
      $ \stIntSum{\roleQ}{i \in I'}{ %
        \stChoice{\stLab[i]}{\tyGround[i]} \stSeq \stT[i]%
      } \stSub
      \stIntSum{\roleQ}{i \in \stIdxRemoveCrash{I}{\stLab[i]}}{ %
        \stChoice{\stLab[i]}{\tyGround[i]} \stSeq (\gtProj[\rolesR]{\gtGi[i]}{\roleP})%
      }$, as desired.

      \item $\roleR = \roleQ$: by association \eqref{eq:soundness:context_1_assoc_pre} and
           \autoref{lem:subtyping-invert}, we obtain $\stEnvApp{\stEnv}{\roleQ} = 
          \stExtSum{\roleP}{i \in I'}{%
        \stChoice{\stLab[i]}{\tyGround[i]} \stSeq \stT[i]%
      }$ where $I \subseteq I'$ and $\forall i \in I: \stT[i] \stSub
      \gtProj[\rolesR]{\gtG[i]}{\roleQ}$. Since $\stEnvApp{\stEnv[i]}{\roleQ} =
       \gtProj[\rolesR]{\gtG[i]}{\roleQ} =
      \stEnvApp{\stEnvi[i]}{\roleQ}$ for each $i \in I$, with association
      $\stEnvAssoc{\gtWithCrashedRoles{\rolesCi}{\gtGi[i]}}{\stEnvi[i]; \qEnvi[i]}{\rolesR}$, we get
      $\gtProj[\rolesR]{\gtG[i]}{\roleQ} = \stEnvApp{\stEnvi[i]}{\roleQ} \stSub
      \gtProj[\rolesR]{\gtGi[i]}{\roleQ}$. Then by transitivity of subtyping,
      it holds that $\forall i \in I: \stT[i] \stSub \gtProj[\rolesR]{\gtGi[i]}{\roleQ}$, which follows that
      $ \stExtSum{\roleP}{i \in I'}{ %
        \stChoice{\stLab[i]}{\tyGround[i]} \stSeq \stT[i]%
      } \stSub
      \stExtSum{\roleP}{i \in I}{ %
        \stChoice{\stLab[i]}{\tyGround[i]} \stSeq (\gtProj[\rolesR]{\gtGi[i]}{\roleQ})%
      }$, as desired.

      \item $\roleR \neq \roleP$ and $\roleR \neq \roleQ$: by association
      $\stEnvAssoc{\gtWithCrashedRoles{\rolesCi}{\gtGi[i]}}{\stEnvi[i]; \qEnvi[i]}{\rolesR}$ for each
      $i \in I$, it holds that $\stEnvApp{\stEnvi[i]}{\roleR} \stSub \gtProj[\rolesR]{\gtGi[i]}{\roleR}$ for each
      $i \in I$. Then by applying~\autoref{lem:subtype:merge-subty}, we conclude with $\stEnvApp{\stEnvi}{\roleR} = \stMerge{i \in I}{\stEnvApp{\stEnvi[i]}{\roleR}}
      \stSub  \stMerge{i \in I}{\gtProj[\rolesR]{\gtGi[i]}{\roleR}} = \gtProj[\rolesR]{\gtGi}{\roleR}$, as desired.
     \end{itemize}

       \item By association $\stEnvAssoc{\gtWithCrashedRoles{\rolesCi}{\gtGi[i]}}{\stEnvi[i]; \qEnvi[i]}{\rolesR}$ for each
      $i \in I$, it holds that $\forall \roleR \in \rolesCi: \stEnvApp{\stEnvi[i]}{\roleR}  = \stStop$, which follows that for any
      $\roleR \in \rolesCi$, $\stEnvApp{\stEnvi}{\roleR} =  \stMerge{i \in I}{\stEnvApp{\stEnvi[i]}{\roleR}} = \stStop$, as desired.
      \item For any endpoint $\roleR$ in $\gtGi$, $\roleR$ is an endpoint in each $\gtGi[i]$ with $i \in I$. Then by association
      $\stEnvAssoc{\gtWithCrashedRoles{\rolesCi}{\gtGi[i]}}{\stEnvi[i]; \qEnvi[i]}{\rolesR}$, we have $\stEnvApp{\stEnvi[i]}{\roleR} =
      \stEnd$, which follows $\stEnvApp{\stEnvi}{\roleR} =  \stMerge{i \in I}{\stEnvApp{\stEnvi[i]}{\roleR}} = \stEnd$, as desired.
      \item By~\autoref{lem:stenv-queue-red:det}, it holds that for any $i, j \in I$, $\qEnvi[i] = \qEnvi[j]$. Take an arbitrary
      $\qEnvi[j]$ with $j \in I$. Note here $\qEnvi = \qEnvi[j]$. With the fact
      that $\qEnvi[i]$ is associated with $\gtGi[i]$ for any $i \in I$, we have that $\forall i \in I: \qEnvi[j] \text{~is associated with~} \gtGi[i]$, and hence,
       $\forall i \in I: \qEnvi \text{~is associated with~} \gtGi[i]$.
       We are left to show that $\stEnvApp{\qEnvi}{\roleP, \roleQ} = \stQEmpty$, which is obtained by applying \autoref{lem:stenv-red:trivial-3}
      on $\stEnvApp{\qEnv}{\roleP, \roleQ} = \stQEmpty$ and  $\stEnv[j]; \qEnv
  \stEnvMoveGenAnnot \stEnvi[j]; \qEnvi[j]$\,
  with\,
  $\ltsSubject{\stEnvAnnotGenericSym} \notin \setenum{\roleP, \roleQ}$.
      \end{enumerate}

 \item $\roleQMaybeCrashed = \roleQCrashed$:

      Note that $\roleQ \in \rolesC$.
      For $\roleP$, we have
     $\gtProj[\rolesR]{\gtG}{\roleP} =
       \stIntSum{\roleQ}{i \in \stIdxRemoveCrash{I}{\stLab[i]}}{ %
        \stChoice{\stLab[i]}{\tyGround[i]} \stSeq (\gtProj[\rolesR]{\gtG[i]}{\roleP})%
      }$ and \linebreak $ \gtProj[\rolesR]{\gtGi}{\roleP} =
       \stIntSum{\roleQ}{i \in \stIdxRemoveCrash{I}{\stLab[i]}}{ %
        \stChoice{\stLab[i]}{\tyGround[i]} \stSeq (\gtProj[\rolesR]{\gtGi[i]}{\roleP})%
      }$.
      Take an arbitrary $j \in I$. %
      Let $\stEnvApp{\stEnv[j]}{\roleP} =
      \gtProj[\rolesR]{\gtG[j]}{\roleP}$, $\stEnvApp{\stEnv[j]}{\roleQ} =
      \stStop$, $\stEnvApp{\stEnv[j]}{\roleR} =
     \stEnvApp{\stEnv}{\roleR}$ for $\roleR \in \dom{\stEnv} \setminus \setenum{\roleP, \roleQ}$, $\qEnv[j] = \qEnv$.
     We show $\stEnvAssoc{\gtWithCrashedRoles{\rolesC}{\gtG[j]}}{\stEnv[j]; \qEnv[j]}{\rolesR}$.
      \begin{enumerate}[label=(A\arabic*)]
      \item We want to show $\forall \roleS \in \gtRoles{\gtG[j]}: \stEnvApp{\stEnv[j]}{\roleS} \stSub \gtProj[\rolesR]{\gtG[j]}{\roleS}$. We consider
      two subcases:
      \begin{itemize}[leftmargin=*]
      \item $\roleS = \roleP$ (meaning that $\roleP \in \gtRoles{\gtG[j]}$): trivial by $\stEnvApp{\stEnv[j]}{\roleP} =
      \gtProj[\rolesR]{\gtG[j]}{\roleP}$ and the reflexivity of subtyping.
      \item $\roleS \neq \roleP$: since $\roleQ \in \rolesC$, we have $\roleS \neq \roleQ$. By association \eqref{eq:soundness:context_1_assoc_pre} and
      $\stEnvApp{\stEnv[j]}{\roleS} =
     \stEnvApp{\stEnv}{\roleS}$, we have $\stEnvApp{\stEnv[j]}{\roleS} \stSub
       \stMerge{i \in I}{\gtProj[\rolesR]{\gtG[i]}{\roleS}}$.
        Then, by~\autoref{lem:merge-subtyping} and transitivity of subtyping,  we conclude  $\stEnvApp{\stEnv[j]}{\roleS}
        \stSub \gtProj[\rolesR]{\gtG[j]}{\roleS}$, as desired.
      \end{itemize}
       \item No change here.
      \item No change here if $\roleP \in \gtRoles{\gtG[j]}$. Otherwise, consider the case that $\roleP \notin \gtRoles{\gtG[j]}$:
      with $\roleP \notin \gtRolesCrashed{\gtG[j]}$, by~\autoref{lem:not-in-roles-end}, we have $\gtProj[\rolesR]{\gtG[j]}{\roleP} = \stEnd$.
      Therefore, we know from $\stEnvApp{\stEnv[j]}{\roleP} =
      \gtProj[\rolesR]{\gtG[j]}{\roleP}$ that $ \stEnvApp{\stEnv[j]}{\roleP} = \stEnd$, as required.
     \item Trivial by association~\eqref{eq:soundness:context_1_assoc_pre}, \autoref{def:assoc-queue}, and $\qEnv[j] = \qEnv$.
      \end{enumerate}
      By inductive hypothesis, there exists $\stEnvi[j]; \qEnvi[j]$ such that
      $\stEnv[j]; \qEnv[j]  \stEnvMoveGenAnnot
      \stEnvi[j]; \qEnvi[j]$ and \linebreak
      $\stEnvAssoc{\gtWithCrashedRoles{\rolesCi}{\gtGi[j]}}{\stEnvi[j]; \qEnvi[j]}{\rolesR}$. Since
      $\ltsSubject{\stEnvAnnotGenericSym} \notin \setenum{\roleP, \roleQ}$, we apply~\autoref{lem:stenv-red:trivial-2}, which gives
      $\stEnvApp{\stEnv[j]}{\roleP} = \stEnvApp{\stEnvi[j]}{\roleP}$ and
      $\stEnvApp{\stEnv[j]}{\roleQ} = \stEnvApp{\stEnvi[j]}{\roleQ} = \stStop$.
      We know from $\stEnvApp{\stEnvi[j]}{\roleQ} = \stStop$ and
      $\stEnvAssoc{\gtWithCrashedRoles{\rolesCi}{\gtGi[j]}}{\stEnvi[j]; \qEnvi[j]}{\rolesR}$ that $\roleQ \in \rolesCi$.

      We now construct a configuration $\stEnvi; \qEnvi$ and show
      $\stEnv; \qEnv \stEnvMoveGenAnnot \stEnvi; \qEnvi$ and
      $\stEnvAssoc{\gtWithCrashedRoles{\rolesCi}{\gtGi}}{\stEnvi; \qEnvi}{\rolesR}$.

      Let  $\stEnvApp{\stEnvi}{\roleP}
      = \stEnvApp{\stEnv}{\roleP}$,
       $\stEnvApp{\stEnvi}{\roleQ}
      =
      \stEnvApp{\stEnv}{\roleQ}
      = \stStop$,
       $\stEnvApp{\stEnvi}{\roleR} =
      \stMerge{i \in I}{\stEnvApp{\stEnvi[i]}{\roleR}}$ for
    $\roleR \in \dom{\stEnv} \setminus
      \setenum{\roleP, \roleQ}$, $\qEnvi = \qEnvi[j]$ with an arbitrary $j \in I$.

      For the first part that $\stEnv; \qEnv \stEnvMoveGenAnnot \stEnvi; \qEnvi$:

        We know that for any $i \in I$, $\stEnvi[i]$ is obtained from $\stEnv[i]$
      by updating $\ltsSubject{\stEnvAnnotGenericSym}$ to a fixed type $\stT$. It follows that
      for any $i, k \in I$ and for any $\roleR \notin \setenum{\roleP, \roleQ}$, $\stEnvApp{\stEnvi[i]}{\roleR}
      =
      \stEnvApp{\stEnvi[k]}{\roleR}$, and hence, $\stMerge{i \in I}{\stEnvApp{\stEnvi[i]}{\roleR}} =
      \stEnvApp{\stEnvi[j]}{\roleR}$ with an arbitrary $j \in J$. Therefore, we have that $\stEnvi$ is obtained from
      $\stEnv$ by updating $\ltsSubject{\stEnvAnnotGenericSym}$ to $\stEnvApp{\stEnvi[j]}{\ltsSubject{\stEnvAnnotGenericSym}}$,
      i.e., $\stEnvi = \stEnvUpd{\stEnv}{\ltsSubject{\stEnvAnnotGenericSym}}{\stEnvApp{\stEnvi[j]}{\ltsSubject{\stEnvAnnotGenericSym}}}$.

      We apply~\autoref{lem:stenv-red:non-trivial} on $\stEnv[j]; \qEnv[j]
  \stEnvMoveGenAnnot \stEnvi[j]; \qEnvi[j]$,
  $\stEnvAssoc{\gtWithCrashedRoles{\rolesC}{\gtG}}{\stEnv; \qEnv}{\rolesR}$,
  \linebreak
  $\stEnvAssoc{\gtWithCrashedRoles{\rolesC}{\gtG[j]}}{\stEnv[j]; \qEnv[j]}{\rolesR}$,
  $\stEnvApp{\stEnv}{\ltsSubject{\stEnvAnnotGenericSym}} = \stEnvApp{\stEnv[j]}{\ltsSubject{\stEnvAnnotGenericSym}}$,
  and
  $\stEnvi =  \stEnvUpd{\stEnv}{\ltsSubject{\stEnvAnnotGenericSym}}{\stEnvApp{\stEnvi[j]}{\ltsSubject{\stEnvAnnotGenericSym}}}$ to get the thesis.

      For the second part that  $\stEnvAssoc{\gtWithCrashedRoles{\rolesCi}{\gtGi}}{\stEnvi; \qEnvi}{\rolesR}$:
      \begin{enumerate}[label=(A\arabic*)]
      \item We want to show $\forall \roleR \in \gtRoles{\gtGi}: \stEnvApp{\stEnvi}{\roleR}
           \stSub \gtProj[\rolesR]{\gtGi}{\roleR}$. We consider two subcases:
           \begin{itemize}[leftmargin=*]
           \item $\roleR = \roleP$: by association \eqref{eq:soundness:context_1_assoc_pre} and
           \autoref{lem:subtyping-invert}, we obtain $\stEnvApp{\stEnv}{\roleP} =
           \stIntSum{\roleQ}{i \in I'}{ %
        \stChoice{\stLab[i]}{\tyGround[i]} \stSeq \stT[i]%
      }$ where $I' \subseteq \stIdxRemoveCrash{I}{\stLab[i]}$ and $\forall i \in I': \stT[i] \stSub
      \gtProj[\rolesR]{\gtG[i]}{\roleP}$. Since $\stEnvApp{\stEnv[i]}{\roleP} =
       \gtProj[\rolesR]{\gtG[i]}{\roleP} =
      \stEnvApp{\stEnvi[i]}{\roleP}$ for each $i \in I$, with association
      $\stEnvAssoc{\gtWithCrashedRoles{\rolesCi}{\gtGi[i]}}{\stEnvi[i]; \qEnvi[i]}{\rolesR}$, we get
      $\gtProj[\rolesR]{\gtG[i]}{\roleP} = \stEnvApp{\stEnvi[i]}{\roleP} \stSub
      \gtProj[\rolesR]{\gtGi[i]}{\roleP}$. Then by transitivity of subtyping,
      it holds that $\forall i \in I': \stT[i] \stSub \gtProj[\rolesR]{\gtGi[i]}{\roleP}$, which follows that
      $ \stIntSum{\roleQ}{i \in I'}{ %
        \stChoice{\stLab[i]}{\tyGround[i]} \stSeq \stT[i]%
      } \stSub
      \stIntSum{\roleQ}{i \in \stIdxRemoveCrash{I}{\stLab[i]}}{ %
        \stChoice{\stLab[i]}{\tyGround[i]} \stSeq (\gtProj[\rolesR]{\gtGi[i]}{\roleP})%
      }$, as desired.

      \item $\roleR \neq \roleP$: since $\roleQ \in \rolesCi$, we have $\roleR \neq \roleQ$. By association
      $\stEnvAssoc{\gtWithCrashedRoles{\rolesCi}{\gtGi[i]}}{\stEnvi[i]; \qEnvi[i]}{\rolesR}$ for each
      $i \in I$, it holds that $\stEnvApp{\stEnvi[i]}{\roleR} \stSub \gtProj[\rolesR]{\gtGi[i]}{\roleR}$ for each
      $i \in I$. Then by applying~\autoref{lem:subtype:merge-subty}, we conclude with $\stEnvApp{\stEnvi}{\roleR} = \stMerge{i \in I}{\stEnvApp{\stEnvi[i]}{\roleR}}
      \stSub  \stMerge{i \in I}{\gtProj[\rolesR]{\gtGi[i]}{\roleR}} = \gtProj[\rolesR]{\gtGi}{\roleR}$, as desired.
     \end{itemize}

       \item We want to show $\forall \roleR \in \rolesCi: \stEnvApp{\stEnvi}{\roleR} = \stStop$.
          We consider two subcases:
           \begin{itemize}[leftmargin=*]
       \item $\roleR = \roleQ$: trivial by $\stEnvApp{\stEnvi}{\roleQ} = \stStop$.
       \item $\roleR \neq \roleQ$:
       by association $\stEnvAssoc{\gtWithCrashedRoles{\rolesCi}{\gtGi[i]}}{\stEnvi[i]; \qEnvi[i]}{\rolesR}$ for each
      $i \in I$, it holds that  $\forall \roleR \in \rolesCi: \stEnvApp{\stEnvi[i]}{\roleR}  = \stStop$, which follows that for any
      $\roleR \in \rolesCi$ with $\roleR \neq \roleQ$, $\stEnvApp{\stEnvi}{\roleR} =  \stMerge{i \in I}{\stEnvApp{\stEnvi[i]}{\roleR}} = \stStop$, as desired.
      \end{itemize}

      \item For any endpoint $\roleR$ in $\gtGi$, $\roleR$ is an endpoint in each $\gtGi[i]$ with $i \in I$. Then by association
      $\stEnvAssoc{\gtWithCrashedRoles{\rolesCi}{\gtGi[i]}}{\stEnvi[i]; \qEnvi[i]}{\rolesR}$, we have $\stEnvApp{\stEnvi[i]}{\roleR} =
      \stEnd$, which follows $\stEnvApp{\stEnvi}{\roleR} =  \stMerge{i \in I}{\stEnvApp{\stEnvi[i]}{\roleR}} = \stEnd$, as desired.

      \item By~\autoref{lem:stenv-queue-red:det}, it holds that for any $i, j \in I$, $\qEnvi[i] = \qEnvi[j]$. Take an arbitrary
      $\qEnvi[j]$ with $j \in I$. Note here $\qEnvi = \qEnvi[j]$. With the fact
      that $\qEnvi[i]$ is associated with $\gtGi[i]$ for any $i \in I$, we have that $\forall i \in I: \qEnvi[j] \text{~is associated with~} \gtGi[i]$, and hence,
       $\forall i \in I: \qEnvi \text{~is associated with~} \gtGi[i]$.
       We are left to show that $\stEnvApp{\qEnvi}{\roleP, \roleQ} = \stQUnavail$, which is obtained by applying \autoref{lem:stenv-red:trivial-3}
      on $\stEnvApp{\qEnv}{\roleP, \roleQ} = \stQUnavail$ and  $\stEnv[j]; \qEnv
  \stEnvMoveGenAnnot \stEnvi[j]; \qEnvi[j]$\,
  with\,
  $\ltsSubject{\stEnvAnnotGenericSym} \notin \setenum{\roleP, \roleQ}$.
      \end{enumerate}

      \end{itemize}

    \item Case \inferrule{\iruleGtMoveCtxi}:

    Similar to the case  \inferrule{\iruleGtMoveCtx}.
     \qedhere

       \end{itemize}
\end{proof}

\thmProjCompleteness*
\begin{proof}
  By induction on reductions of configuration $\stEnv; \qEnv \stEnvMoveGenAnnot
  \stEnvi; \qEnvi$.
  \begin{itemize}[leftmargin=*]
    \item Case $\inferrule{\iruleTCtxOut}$:

    From the premise, we have:
    \begin{gather}
      \stEnvAssoc{\gtWithCrashedRoles{\rolesC}{\gtG}}{\stEnv; \qEnv}{\rolesR}
      \\
      \stEnvApp{\stEnv}{%
        \roleP%
      } =
      \stIntSum{\roleQ}{i \in I}{\stChoice{\stLab[i]}{\tyGround[i]} \stSeq \stT[i]}%
      \label{eq:stenvp-send}
      \\
      \stEnvAnnotGenericSym = \stEnvOutAnnotSmall{\roleP}{\roleQ}{\stChoice{\stLab[k]}{\tyGround[k]}}
      \\
      k \in I
      \\
      \stEnvi =
      \stEnvUpd{\stEnv}{\roleP}{\stT[k]}
      \label{eq:stenvp-send-type-cxt-cons}
      \\
      \qEnvi =
      \stEnvUpd{\qEnv}{\roleP, \roleQ}{%
        \stQCons{%
          \stEnvApp{\qEnv}{\roleP, \roleQ}
        }{%
          \stQMsg{\stLab[k]}{\tyGround[k]}
        }
      }
      \label{eq:stenvp-send-queue-cons}%
    \end{gather}

    Apply~\autoref{lem:inv-proj}~\cref{item:proj-inv:send} on
    \eqref{eq:stenvp-send}, we have two cases.
    \begin{itemize}[leftmargin=*]
      \item Case (1):
    \begin{gather}
      \unfoldOne{\gtG} =
        \gtComm{\roleP}{\roleQMaybeCrashed}{i \in I'}{\gtLab[i]}{\tyGroundi[i]}{\gtG[i]}
      \\
      I \subseteq I'
      \\
      \forall i \in I:
      \stLab[i] = \gtLab[i],
      \stT[i] \stSub (\gtProj[\rolesR]{\gtG[i]}{\roleP}),
      \tyGround[i] = \tyGroundi[i]
      \label{eq:stenvp-send-sub-cons}
    \end{gather}

    We have two further subcases here: namely $\roleQMaybeCrashed =
    \roleQ$ and $\roleQMaybeCrashed = \roleQCrashed$.

      \begin{itemize}[leftmargin=*]
        \item
        In the case of $\roleQMaybeCrashed = \roleQ$,
        we have $k \in I \subseteq I'$ and $\stCrashLab$ does not appear in
        internal choices, we apply $\inferrule{\iruleGtMoveOut}$ (via~\autoref{lem:gt-lts-unfold}) to get:
        \begin{gather}
          \gtWithCrashedRoles{\rolesC}{\gtG}
          \gtMove[\stEnvAnnotGenericSym]{\rolesR}
          \gtWithCrashedRoles{\rolesC}{
            \gtCommTransit{\roleP}{\roleQ}{i \in I'}{\gtLab[i]}{\tyGroundi[i]}{\gtG[i]}{k}
          }
        \end{gather}

        We are now left to show $\stEnvAssoc{
          \gtWithCrashedRoles{\rolesC}{
            \gtCommTransit{\roleP}{\roleQ}{i \in I'}{\gtLab[i]}{\tyGroundi[i]}{\gtG[i]}{k}
          }
        }{\stEnvi; \qEnvi}{\rolesR}$.

        Note that $\gtGi = \gtCommTransit{\roleP}{\roleQ}{i \in I'}{\gtLab[i]}{\tyGroundi[i]}{\gtG[i]}{k}$ here.
        By~\autoref{def:active_crashed_roles}, we have  $\gtRoles{\gtG} = \gtRoles{\gtGi}$ and
        $\gtRolesCrashed{\gtG} =  \gtRolesCrashed{\gtGi}$. Furthermore, with $\dom{\stEnv} = \dom{\stEnvi}$ and
        $\roleP \in \gtRoles{\gtG}$, we can set $\stEnvi =  \stEnvUpd{\stEnv[\gtG]}{\roleP}{\stT[k]}\stEnvComp \stEnv[\stStopSym] \stEnvComp
  \stEnv[\stEnd]$.

        \begin{enumerate}[label=(A\arabic*)]
          \item First, we want to show that $\dom{\stEnvUpd{\stEnv[\gtG]}{\roleP}{\stT[k]}} =  \setcomp{{\roleP}}{\roleP \in \gtRoles{\gtGi})}$, which follows directly
          from the fact that $\dom{\stEnv[\gtG]} = \setcomp{{\roleP}}{\roleP \in \gtRoles{\gtG})}$, $\gtRoles{\gtG} = \gtRoles{\gtGi}$, and
          $\dom{\stEnvUpd{\stEnv[\gtG]}{\roleP}{\stT[k]}} = \dom{\stEnv[\gtG]}$.

          Then, we are left to show $ \forall \roleR \in \gtRoles{\gtGi}:
        \stEnvApp{\stEnvUpd{\stEnv[\gtG]}{\roleP}{\stT[k]}}{{\roleR}}
        \stSub
        \gtProj[\rolesR]{\gtGi}{\roleR}
      $. We consider two cases:
      \begin{itemize}[leftmargin=*]
      \item $\roleR = \roleP$: we have $\gtProj[\rolesR]{\gtGi}{\roleP} = \gtProj[\rolesR]{\gtG[k]}{\roleP}$. By  \eqref{eq:stenvp-send-sub-cons},
      we obtain that $ \stEnvApp{\stEnvUpd{\stEnv[\gtG]}{\roleP}{\stT[k]}}{{\roleP}} = \stT[k] \stSub \gtProj[\rolesR]{\gtG[k]}{\roleP} = \gtProj[\rolesR]{\gtGi}{\roleP}$, as desired.
      \item $\roleR \neq \roleP$: since $\stEnvAssoc{\gtWithCrashedRoles{\rolesC}{\gtG}}{\stEnv; \qEnv}{\rolesR}$,
      we have that $ \stEnvApp{\stEnv[\gtG]}{{\roleR}} \stSub \gtProj[\rolesR]{\gtG}{\roleR} =   \stMerge{i \in I'}{\gtProj[\rolesR]{\gtG[i]}{\roleR}}$.
      Furthermore, by $\stEnvApp{\stEnvUpd{\stEnv[\gtG]}{\roleP}{\stT[k]}}{{\roleR}} = \stEnvApp{\stEnv[\gtG]}{{\roleR}}$ and
      $\gtProj[\rolesR]{\gtGi}{\roleR} =   \stMerge{i \in I'}{\gtProj[\rolesR]{\gtG[i]}{\roleR}}$, we have that
      $\stEnvApp{\stEnvUpd{\stEnv[\gtG]}{\roleP}{\stT[k]}}{{\roleR}} \stSub \gtProj[\rolesR]{\gtGi}{\roleR}$, as desired.
      \end{itemize}

           \item Trivial by $\stEnvAssoc{\gtWithCrashedRoles{\rolesC}{\gtG}}{\stEnv; \qEnv}{\rolesR}$.
          \item Trivial by $\stEnvAssoc{\gtWithCrashedRoles{\rolesC}{\gtG}}{\stEnv; \qEnv}{\rolesR}$.
          \item By \eqref{eq:stenvp-send-queue-cons}, we have  $ \stEnvApp{\qEnvi}{\roleP, \roleQ} =
        \stQCons{%
          \stEnvApp{\qEnv}{\roleP, \roleQ}
        }{%
          \stQMsg{\stLab[k]}{\tyGround[k]}
        }
      $. Then we only need to show that $\forall i \in I':
      \stEnvUpd{\qEnvi}{\roleP, \roleQ}{ \stEnvApp{\qEnv}{\roleP, \roleQ}}$ 
      (note that $\stEnvUpd{\qEnvi}{\roleP, \roleQ}{
      \stEnvApp{\qEnv}{\roleP, \roleQ}} = \qEnv$) is associated with 
     $\gtWithCrashedRoles{\rolesC}{\gtG[i]}$, which follows directly from
       the fact that $\qEnv$ is associated with $\gtWithCrashedRoles{\rolesC}{\gtG}$ and~\autoref{def:assoc-queue}.

        \end{enumerate}

        \item
        In the case of $\roleQMaybeCrashed = \roleQCrashed$,
        we have $k \in I \subseteq I'$, $\roleQ \in \rolesC$, and
        $\stCrashLab$ does not appear in
        internal choices.
        We apply $\inferrule{\iruleGtMoveOrph}$ (via
        \autoref{lem:gt-lts-unfold}) to get:
        \begin{gather}
          \gtWithCrashedRoles{\rolesC}{\gtG}
          \gtMove[\stEnvAnnotGenericSym]{\rolesR}
          \gtWithCrashedRoles{\rolesC}{
            \gtG[k]
          }
        \end{gather}
        We are now left to show $\stEnvAssoc{
          \gtWithCrashedRoles{\rolesC}{
            \gtG[k]
          }
        }{\stEnvi; \qEnvi}{\rolesR}$.

        Note that $\gtGi = \gtG[k]$ here.
 \begin{enumerate}[label=(A\arabic*)]
          \item We need to show that $\forall \roleR \in \gtRoles{\gtG[k]}:
          \stEnvApp{\stEnvi}{\roleR} \stSub
        \gtProj[\rolesR]{\gtG[k]}{\roleR}$.
        We consider two subcases:
        \begin{itemize}[leftmargin=*]
        \item $\roleR = \roleP$, which means that $\roleP \in \gtRoles{\gtG[k]}$: by \eqref{eq:stenvp-send-type-cxt-cons} and \eqref{eq:stenvp-send-sub-cons}, we have
          $\stEnvApp{\stEnvi}{\roleP} =  \stT[k] \stSub \gtProj[\rolesR]{\gtG[k]}{\roleP}$, as desired.
        \item $\roleR \neq \roleP$: with $\roleQ \notin \gtRoles{\gtG[k]}$, we have $\roleR \neq \roleQ$, and hence,
        $\gtProj[\rolesR]{\gtG}{\roleR} = \stMerge{i \in I'}{\gtProj[\rolesR]{\gtG[i]}{\roleR}}$.
        Furthermore, by \eqref{eq:stenvp-send-type-cxt-cons} and $\stEnvAssoc{\gtWithCrashedRoles{\rolesC}{\gtG}}{\stEnv; \qEnv}{\rolesR}$,
        it holds that $\stEnvApp{\stEnvi}{\roleR} = \stEnvApp{\stEnv}{\roleR} \stSub \gtProj[\rolesR]{\gtG}{\roleR} = \stMerge{i \in I'}{\gtProj[\rolesR]{\gtG[i]}{\roleR}}$.
        Then, by \autoref{lem:merge-subtyping} and transitivity of subtyping,  we can conclude that $\stEnvApp{\stEnvi}{\roleR}
        \stSub \gtProj[\rolesR]{\gtG[k]}{\roleR}$, as desired.
        \end{itemize}
         \item Trivial by $\stEnvAssoc{\gtWithCrashedRoles{\rolesC}{\gtG}}{\stEnv; \qEnv}{\rolesR}$.
          \item Trivial by $\stEnvAssoc{\gtWithCrashedRoles{\rolesC}{\gtG}}{\stEnv; \qEnv}{\rolesR}$ and the
                   fact that if $\roleP \notin \gtRoles{\gtG[k]}$, then $\stEnvApp{\stEnvi}{\roleP} = \stEnd$:
                   with $\roleP \notin \gtRolesCrashed{\gtG[k]}$, by~\autoref{lem:not-in-roles-end}, we have $\gtProj[\rolesR]{\gtG[k]}{\roleP} = \stEnd$.  Furthermore, by \eqref{eq:stenvp-send-sub-cons}, it holds that $\stT[k] = \stEnd$, and thus,
          $\stEnvApp{\stEnvi}{\roleP} = \stEnd$, as desired.

          \item Since $\roleQ \in \rolesC$, by~\autoref{def:assoc-queue}, we have
          $\stEnvApp{\qEnv}{\cdot, \roleQ} = \stQUnavail$.
          Hence,  $\qEnvi =\linebreak
      \stEnvUpd{\qEnv}{\roleP, \roleQ}{%
        \stQCons{%
          \stEnvApp{\qEnv}{\roleP, \roleQ}
        }{%
          \stQMsg{\stLab[k]}{\tyGround[k]}
        }
      } =
      \stEnvUpd{\qEnv}{\roleP, \roleQ}{%
        \stQCons{%
         \stQUnavail
        }{%
          \stQMsg{\stLab[k]}{\tyGround[k]}
        }
      } =  \stEnvUpd{\qEnv}{\roleP, \roleQ}{\stQUnavail} = \qEnv$.
      Then we only need to show that $\qEnv$ is associated with
       $\gtWithCrashedRoles{\rolesC}{
            \gtG[k]}$, which follows directly from the fact that
            $\qEnv$ is associated with $\gtWithCrashedRoles{\rolesC}{\gtG}$ and~\autoref{def:assoc-queue}.
        \end{enumerate}
      \end{itemize}
      \item Case (2):
      \begin{gather}
       \unfoldOne{\gtG} =
            \gtComm{\roleS}{\roleTMaybeCrashed}{j \in J}{\gtLab[j]}{\tyGroundi[j]}{\gtG[j]}
          \text{ \;or\; }
          \unfoldOne{\gtG} =
            \gtCommTransit{\roleSMaybeCrashed}{\roleT}{j \in
            J}{\gtLab[j]}{\tyGroundi[j]}{\gtG[j]}{k}
            \\
       \forall j \in J:
        \stEnvApp{\stEnv}{%
        \roleP%
        } \stSub \gtProj[\rolesR]{\gtG[j]}{\roleP}
        \\
        \roleP \neq \roleS \text{ \;and\; } \roleP \neq \roleT
      \end{gather}

  We consider two subcases:  $\unfoldOne{\gtG} =
            \gtComm{\roleS}{\roleTMaybeCrashed}{j \in J}{\gtLab[j]}{\tyGroundi[j]}{\gtG[j]}$ and
          $\unfoldOne{\gtG} = \linebreak
            \gtCommTransit{\roleSMaybeCrashed}{\roleT}{j \in
            J}{\gtLab[j]}{\tyGroundi[j]}{\gtG[j]}{k}$.

  \begin{itemize}[leftmargin=*]
  \item In the case of $\unfoldOne{\gtG} =
            \gtComm{\roleS}{\roleTMaybeCrashed}{j \in J}{\gtLab[j]}{\tyGroundi[j]}{\gtG[j]}$: %

            First, we take an arbitrary index $j \in J$ and construct a configuration $\stEnv[j]; \qEnv[j]$ such that
  $\stEnv[j]; \qEnv[j]
      \,\stEnvMoveOutAnnot{\roleP}{\roleQ}{\stChoice{\stLab[k]}{\tyGround[k]}}\,
      \stEnvi[j]; \qEnvi[j]$ and
       $\stEnvAssoc{\gtWithCrashedRoles{\rolesC}{\gtG[j]}}{\stEnv[j]; \qEnv[j]}{\rolesR}$.

      We know from $\roleS \in \gtRoles{\gtG}$ and
           $\stEnvAssoc{\gtWithCrashedRoles{\rolesC}{\gtG}}{\stEnv; \qEnv}{\rolesR}$, that
           $
        \stEnvApp{\stEnv}{\roleS}
        \stSub
         \gtProj[\rolesR]{\unfoldOne{\gtG}}{\roleS}$ and
         $  \gtProj[\rolesR]{\unfoldOne{\gtG}}{\roleS}
         =
        \stIntSum{\roleT}{j \in \stIdxRemoveCrash{J}{\stLab[j]}}{%
          \stChoice{\stLab[j]}{\tyGroundi[j]}
                 \stSeq (\gtProj[\rolesR]{\gtG[j]}{\roleS})}$.
      By inverting $\inferrule{\iruleStSubOut}$ (applying \linebreak 
      \autoref{lem:unfold-subtyping} where necessary), we have
      $ \unfoldOne{\stEnvApp{\stEnv}{\roleS}}
        =
        \stIntSum{\roleT}{j \in J_{\roleS}}{%
          \stChoice{\stLab[j]}{\tyGroundii[j]} \stSeq \stTii[j]
        }
      $, where $J_{\roleS} \subseteq \stIdxRemoveCrash{J}{\stLab[j]}$,
      and $\forall j \in J_{\roleS}:
       \stTii[j] \stSub (\gtProj[\rolesR]{\gtG[j]}{\roleS})$.

      To construct $\stEnv[j]$, let
      $ \stEnvApp{\stEnv[j]}{\roleS}
        =
        \stTii[j]
      $ if $j \in J_{\roleS}$ and
      $ \stEnvApp{\stEnv[j]}{\roleS}
        =
        \gtProj[\rolesR]{\gtG[j]}{\roleS}
      $ otherwise.
      In either case, we have
      $ \stEnvApp{\stEnv[j]}{\roleS}
      \stSub
      \gtProj[\rolesR]{\gtG[j]}{\roleS}$, as required.

      We have two further subcases here: namely $\roleTMaybeCrashed = \roleT$ and $\roleTMaybeCrashed = \roleTCrashed$.
      \begin{itemize}[leftmargin=*]
       \item If $\roleTMaybeCrashed = \roleT$, we know
      from $\roleT \in \gtRoles{\gtG}$ and
      $\stEnvAssoc{\gtWithCrashedRoles{\rolesC}{\gtG}}{\stEnv; \qEnv}{\rolesR}$, that
      $ \stEnvApp{\stEnv}{\roleT} \stSub
      \gtProj[\rolesR]{\unfoldOne{\gtG}}{\roleT}$ \linebreak 
        $=
        \stExtSum{\roleS}{j \in J}{\stChoice{\stLab[j]}{\tyGroundi[j]}
                 \stSeq (\gtProj[\rolesR]{\gtG[j]}{\roleT})}
      $.
      By inverting $\inferrule{\iruleStSubIn}$ (applying
      \autoref{lem:unfold-subtyping} where necessary), we have
      $ \unfoldOne{\stEnvApp{\stEnv}{\roleT}}
        =
        \stExtSum{\roleS}{j \in J_{\roleT}}{%
          \stChoice{\stLab[j]}{\tyGroundii[j]} \stSeq \stUii[j]}$,
      where $J \subseteq J_{\roleT}$,
      and $\forall j \in J:\linebreak
      \stUii[j] \stSub (\gtProj[\rolesR]{\gtG[j]}{\roleT})$.

      To construct $\stEnv[j]$, let
      $ \stEnvApp{\stEnv[j]}{\roleT}
        =
        \stUii[j]
      $, and we have
      $ \stEnvApp{\stEnv[j]}{\roleT}
        \stSub
        \gtProj[\rolesR]{\gtG[j]}{\roleT}
        $, as required.
        \item If $\roleTMaybeCrashed = \roleTCrashed$, let
        $ \stEnvApp{\stEnv[j]}{\roleT}
        =
        \stEnvApp{\stEnv}{\roleT} = \stStop$,
      as required.
       \end{itemize}
       For roles $\roleR \in (\gtRoles{\gtG[j]} \cup \rolesC)$,
       where $\roleR \notin
  \setenum{\roleS, \roleT}$, their typing context
  entry do not change, \ie $\stEnvApp{\stEnv[j]}{\roleR}
  = \stEnvApp{\stEnv}{\roleR}$.
  For crashed roles $\roleR \in \rolesC$, we have
  $ \stStop =
    \stEnvApp{\stEnv}{\roleR} =
    \stEnvApp{\stEnv[j]}{\roleR} =
    \stStop
  $, as required.
  For non-crashed roles $\roleR \in \gtRoles{\gtG[j]}$, we have
  $\stEnvApp{\stEnv[j]}{\roleR}
    =
     \stEnvApp{\stEnv}{\roleR}
     \stSub
     \gtProj[\rolesR]{\gtG}{\roleR}
     =
     \stMerge{j \in J}(\gtProj[\rolesR]{\gtG[j]}{\roleR})
     \stSub
     \gtProj[\rolesR]{\gtG[j]}{\roleR}
  $ (applying~\autoref{lem:merge-subtyping}).

We know from $\stEnvAssoc{\gtWithCrashedRoles{\rolesC}{\gtG}}{\stEnv; \qEnv}{\rolesR}$ and
$\unfoldOne{\gtG} =
            \gtComm{\roleS}{\roleTMaybeCrashed}{j \in J}{\gtLab[j]}{\tyGroundi[j]}{\gtG[j]}$,  that
$\qEnv$ is associated with $\gtG[j]$. Hence, to construct $\qEnv[j]$, just let $\qEnv[j] = \qEnv$.
Notice that \linebreak $\setenum{\roleP, \roleQ} \in \gtRoles{\gtG[j]}$, so they are still able to
 perform the communication action
$\stEnv[j]; \qEnv[j]
      \,\stEnvMoveOutAnnot{\roleP}{\roleQ}{\stChoice{\stLab[k]}{\tyGround[k]}}\,
      \stEnvi[j]; \qEnvi[j]$.

  We apply inductive hypothesis on $\stEnv[j]; \qEnv[j]$, and obtain
  $\gtWithCrashedRoles{\rolesC}{\gtG[j]}
          \gtMove[\stEnvOutAnnotSmall{\roleP}{\roleQ}{\stChoice{\stLab[k]}{\tyGround[k]}}]{\rolesR}
         \gtWithCrashedRoles{\rolesC}{\gtGi[j]} $    and
    $\stEnvAssoc{\gtWithCrashedRoles{\rolesC}{\gtGi[j]}}{\stEnvi[j]; \qEnvi[j]}{\rolesR}$.

 We apply $\inferrule{\iruleGtMoveCtx}$ (via
        \autoref{lem:gt-lts-unfold}) to get:
        \begin{gather}
          \gtWithCrashedRoles{\rolesC}{\gtG}
          \gtMove[\stEnvOutAnnotSmall{\roleP}{\roleQ}{\stChoice{\stLab[k]}{\tyGround[k]}}]{\rolesR}
          \gtWithCrashedRoles{\rolesC}{\gtComm{\roleS}{\roleTMaybeCrashed}{j \in J}{\gtLab[j]}{\tyGroundi[j]}{\gtGi[j]}}
        \end{gather}
We are now left to show $\stEnvAssoc{\gtWithCrashedRoles{\rolesC}{\gtComm{\roleS}{\roleTMaybeCrashed}{j \in J}{\gtLab[j]}{\tyGroundi[j]}{\gtGi[j]}}}{\stEnvi; \qEnvi}{\rolesR}$.

Note that $\gtGi = \gtComm{\roleS}{\roleTMaybeCrashed}{j \in J}{\gtLab[j]}{\tyGroundi[j]}{\gtGi[j]}$ here. %

        \begin{enumerate}[label=(A\arabic*)]
        \item  For role $\roleS$, %
        we know that
        $ \unfoldOne{\stEnvApp{\stEnv}{\roleS}}
        =
        \stIntSum{\roleT}{j \in J_{\roleS}}{%
          \stChoice{\stLab[j]}{\tyGroundii[j]} \stSeq \stTii[j]
        }
      $, where $J_{\roleS} \subseteq \stIdxRemoveCrash{J}{\stLab[j]}$,
       $\forall j \in J_{\roleS}:
       \stTii[j] \stSub (\gtProj[\rolesR]{\gtG[j]}{\roleS})$, and
       $\tyGroundii[j] = \tyGroundi[j]$.

  Since $\roleS \notin
  \ltsSubject{\stEnvOutAnnotSmall{\roleP}{\roleQ}{\stChoice{\stLab[k]}{\tyGround[k]}}}$,
  we apply
  \autoref{lem:stenv-red:trivial-2} on $\stEnv$ and $\stEnv[j]$ for all
  $j \in J_{\roleS}$.
  For all $j \in J_{\roleS}$,
  we have $
    \stTii[j]
    =
    \stEnvApp{\stEnv[j]}{\roleS}
    =
    \stEnvApp{\stEnvi[j]}{\roleS}
  $ (from~\autoref{lem:stenv-red:trivial-2}) and $
    \stEnvApp{\stEnvi[j]}{\roleS}
    \stSub
    \gtProj[\rolesR]{\gtGi[j]}{\roleS}
    $ (from inductive hypothesis).
  Therefore, we have $\stTii[j] \stSub \gtProj[\rolesR]{\gtGi[j]}{\roleS}$.
  We now apply~\autoref{lem:stenv-red:trivial-2} on $\stEnv$,
  which gives
  $ \unfoldOne{\stEnvApp{\stEnv}{\mpChanRole{\mpS}{\roleS}}}
    =
    \unfoldOne{\stEnvApp{\stEnvi}{\mpChanRole{\mpS}{\roleS}}}
    =
    \stIntSum{\roleT}{j \in J_{\roleS}}{%
      \stChoice{\stLab[j]}{\tyGroundii[j]} \stSeq \stTii[j]
    }
  $.
  We can now apply $\inferrule{\iruleStSubOut}$ to conclude $
 \stEnvApp{\stEnvi}{\roleS}
 \stSub
    \gtProj[\rolesR]{\gtGi}{\roleS}
  $, as required.

   For role $\roleT$ (where $\roleTMaybeCrashed = \roleT$), we know that
  $ \unfoldOne{\stEnvApp{\stEnv}{\roleT}}
    =
    \stExtSum{\roleS}{j \in J_{\roleT}}{%
      \stChoice{\stLab[j]}{\tyGroundii[j]} \stSeq \stUii[j]
    }
  $, where $J \subseteq J_{\roleT}$,
   $\forall j \in J:
     \stUii[j] \stSub (\gtProj[\rolesR]{\gtG[j]}{\roleT})$, and $
    \tyGroundi[j] = \tyGroundii[j]
  $.
  Since $\roleT \notin$ \linebreak 
  $\ltsSubject{\stEnvOutAnnotSmall{\roleP}{\roleQ}{\stChoice{\stLab[k]}{\tyGround[k]}}}$,
  we apply
  \autoref{lem:stenv-red:trivial-2} on $\stEnv$ and $\stEnv[j]$ for all
  $j \in J$.
  For all $j \in J$,
  we have $
    \stUii[j]
    =
    \stEnvApp{\stEnv[j]}{\roleT}
    =
    \stEnvApp{\stEnvi[j]}{\roleT}
  $ (from~\autoref{lem:stenv-red:trivial-2}) and $
  \stEnvApp{\stEnvi[j]}{\roleT}
    \stSub
    \gtProj[\rolesR]{\gtGi[j]}{\roleT}
  $ (from inductive hypothesis).
  Therefore, we have $\stUii[j] \stSub \gtProj[\rolesR]{\gtGi[j]}{\roleT}$.
  We now apply~\autoref{lem:stenv-red:trivial-2} on $\stEnv$,
  which gives
  $ \unfoldOne{\stEnvApp{\stEnv}{\roleT}}
    =
    \unfoldOne{\stEnvApp{\stEnvi}{\roleT}}
    =
    \stExtSum{\roleS}{j \in J_{\roleT}}{%
      \stChoice{\stLab[j]}{\tyGroundii[j]} \stSeq \stUii[j]
    }
  $.
  We can now apply $\inferrule{\iruleStSubIn}$ to conclude $
    \stEnvApp{\stEnvi}{\roleT}
    \stSub
    \gtProj[\rolesR]{\gtGi}{\roleT}
$, as required.

 For other role $\roleR \in \gtRoles{\gtGi}$ (where
  $\roleR \notin \setenum{\roleS, \roleT}$),
  we need to show
  $\stEnvApp{\stEnvi}{\roleR}
  \stSub
   \gtProj[\rolesR]{\gtGi}{\roleR}
  $.
  We know that $
    \gtProj[\rolesR]{\gtGi}{\roleR} =
    \stMerge{j \in J}{\gtProj[\rolesR]{\gtGi[j]}{\roleR}}
  $.
  If $\roleR \notin \setenum{\roleP, \roleQ}$, we apply
  \autoref{lem:stenv-red:trivial-2} on $\stEnv[j]$, obtaining
  $
    \stEnvApp{\stEnv[j]}{\roleR}
    =
    \stEnvApp{\stEnvi[j]}{\roleR}
  $. The inductive hypothesis gives
  $  \stEnvApp{\stEnvi[j]}{\roleR}
  \stSub
  \gtProj[\rolesR]{\gtGi[j]}{\roleR}$,
  we apply~\autoref{lem:merge-upper-bound} to obtain $
    \stEnvApp{\stEnvi[j]}{\roleR}
    \stSub
    \gtProj[\rolesR]{\gtGi}{\roleR}
    =
    \stMerge{j \in J}{\gtProj[\rolesR]{\gtGi[j]}{\roleR}}
  $.
  Note that $
    \stEnvApp{\stEnvi[j]}{\roleR} =
    \stEnvApp{\stEnv[j]}{\roleR}
  $ by~\autoref{lem:stenv-red:trivial-2},
  and $
    \stEnvApp{\stEnv[j]}{\roleR} =
    \stEnvApp{\stEnv}{\roleR}
  $ by construction.
  Therefore, we have $
  \stEnvApp{\stEnvi}{\roleR}
=
    \stEnvApp{\stEnv}{\roleR}
    \stSub
    \gtProj[\rolesR]{\gtGi}{\roleR}
      $, as required.

 We are left to consider the cases of $\roleP$ and $\roleQ$.
  We know what $
    \gtProj[\rolesR]{\gtGi}{\roleP} =
    \stMerge{j \in J}{\gtProj[\rolesR]{\gtGi[j]}{\roleP}}
  $ and $\setenum{\roleP, \roleQ} \subseteq \gtRoles{\gtGi}$.
  The inductive hypothesis gives
  $ \stEnvApp{\stEnvi[j]}{\roleP}
  \stSub
  \gtProj[\rolesR]{\gtGi[j]}{\roleP}
      $,
  we apply~\autoref{lem:merge-upper-bound} to obtain $
    \stEnvApp{\stEnvi[j]}{\roleP}
    \stSub
    \gtProj[\rolesR]{\gtGi}{\roleP}
    =
    \stMerge{j \in J}{\gtProj[\rolesR]{\gtGi[j]}{\roleP}}
  $.
  We now apply~\autoref{lem:stenv-red:det} on $\stEnv$ and all $\stEnv[j]$, which
  gives $\stEnvi[j] = \stEnvi$ for all $j$.
  Therefore, we have $
  \stEnvApp{\stEnvi}{\roleP}
  \stSub
    \gtProj[\rolesR]{\gtGi}{\roleP}
  $.
  Note that  $\stEnvApp{\stEnvi}{\roleP} =  \stEnvApp{\stEnvi[j]}{\roleP} = \stT[k] \;(\text{as in \eqref{eq:stenvp-send-type-cxt-cons}})$.
  The argument $\roleQ$ follows similarly.
\item For crashed roles $\roleR \in \rolesC$, we have
  $ \stEnvApp{\stEnvi}{\roleR}
    =
    \stEnvApp{\stEnv}{\roleR}
    =
    \stStop
  $ (applying~\autoref{lem:stenv-red:trivial-2}). Note that if $\roleTMaybeCrashed = \roleTCrashed$, $\roleT \in \rolesC$.
        \item Trivial by $\stEnvAssoc{\gtWithCrashedRoles{\rolesC}{\gtG}}{\stEnv; \qEnv}{\rolesR}$. %
        \item We know from $\qEnv[j] = \qEnv$ and $\stEnv[j]; \qEnv[j]
      \,\stEnvMoveOutAnnot{\roleP}{\roleQ}{\stChoice{\stLab[k]}{\tyGround[k]}}\,
      \stEnvi[j]; \qEnvi[j]$, that $\qEnvi[j] = \linebreak
      \stEnvUpd{\qEnv[j]}{\roleP, \roleQ}{%
        \stQCons{%
          \stEnvApp{\qEnv[j]}{\roleP, \roleQ}
        }{%
          \stQMsg{\stLab[k]}{\tyGround[k]}
        }
      } =
      \stEnvUpd{\qEnv}{\roleP, \roleQ}{%
        \stQCons{%
          \stEnvApp{\qEnv}{\roleP, \roleQ}
        }{%
          \stQMsg{\stLab[k]}{\tyGround[k]}
        }
      } =
      \qEnvi$. Furthermore, by inductive hypothesis,
      $\qEnvi[j]$ is associated with $\gtWithCrashedRoles{\rolesC}{\gtGi[j]}$, which follows
      that $\forall j \in J: \qEnvi \text{ is associated with } \gtWithCrashedRoles{\rolesC}{\gtGi[j]}$. We are left to show that
      if $\roleTMaybeCrashed \neq \roleTCrashed$, then $\stEnvApp{\qEnvi}{\roleS, \roleT} = \stQEmpty$, which follows
      directly from $\stEnvAssoc{\gtWithCrashedRoles{\rolesC}{\gtG}}{\stEnv; \qEnv}{\rolesR}$ and
      $\stEnvApp{\qEnvi}{\roleS, \roleT} = \stEnvApp{\qEnv}{\roleS, \roleT}$.
        \end{enumerate}
\end{itemize}
\item In the case of   $\unfoldOne{\gtG} =
            \gtCommTransit{\roleSMaybeCrashed}{\roleT}{j \in
            J}{\gtLab[j]}{\tyGroundi[j]}{\gtG[j]}{k}$:

            Similar to that of the previous subcase that $\unfoldOne{\gtG} =
            \gtComm{\roleS}{\roleTMaybeCrashed}{j \in J}{\gtLab[j]}{\tyGroundi[j]}{\gtG[j]}$, applying \inferrule{\iruleGtMoveCtxi} instead.
\end{itemize}
    \item Case $\inferrule{\iruleTCtxIn}$:

      From the premise, we have:
    \begin{gather}
      \stEnvAssoc{\gtWithCrashedRoles{\rolesC}{\gtG}}{\stEnv; \qEnv}{\rolesR}
      \label{eq:stenvp-receive-assoc}
      \\
       \stEnvApp{\stEnv}{%
        \roleP%
      } =
      \stExtSum{\roleQ}{i \in I}{\stChoice{\stLab[i]}{\tyGround[i]} \stSeq \stT[i]}
      \label{eq:stenvp-receive-comp}
      \\
      \stEnvAnnotGenericSym = \stEnvInAnnotSmall{\roleP}{\roleQ}{\stChoice{\stLab[k]}{\tyGround[k]}}
      \\
      k \in I
      \\
        \stEnvApp{\qEnv}{\roleQ, \roleP}
      =
      \stQCons{\stQMsg{\stLab[k]}{\tyGround[k]}}{\stQi}
      \neq \stQUnavail
      \label{eq:stenvp-receive-queue-con}
      \\
      \stEnvi =
      \stEnvUpd{\stEnv}{\roleP}{\stT[k]}
      \label{eq:stenvp-receive-type-cxt-cons}
      \\
      \qEnvi =
     \stEnvUpd{\qEnv}{\roleQ, \roleP}{\stQi}
    \end{gather}

     Apply~\autoref{lem:inv-proj}~\cref{item:proj-inv:recv} on
    \eqref{eq:stenvp-receive-comp}, we have two cases.
    \begin{itemize}[leftmargin=*]
      \item Case (1):
        \begin{gather}
          \unfoldOne{\gtG} =
            \gtCommTransit{\roleQMaybeCrashed}{\roleP}{i \in
            I'}{\gtLab[i]}{\tyGroundi[i]}{\gtG[i]}{j} \text{ or }
            \unfoldOne{\gtG} =
            \gtComm{\roleQ}{\rolePMaybeCrashed}{i \in I'}{\gtLab[i]}{\tyGroundi[i]}{\gtG[i]}
      \\
      I' \subseteq I
      \\
      \forall i \in I':
      \stLab[i] = \gtLab[i],
      \stT[i] \stSub (\gtProj[\rolesR]{\gtG[i]}{\roleP}),
      \tyGroundi[i] = \tyGround[i]
      \label{eq:stenvp-receive-global-label-equvi}
      \\
      \roleQ \notin \rolesR \text{ implies } \exists l \in I': \gtLab[l] =
          \gtCrashLab
    \end{gather}

First, we show that in this case, $\unfoldOne{\gtG}$ cannot be of the form $\gtComm{\roleQ}{\rolePMaybeCrashed}{i \in I'}{\gtLab[i]}{\tyGroundi[i]}{\gtG[i]}$.
There are two subcases to be considered: $\rolePMaybeCrashed =
    \roleP$ and $\rolePMaybeCrashed = \rolePCrashed$.
     \begin{itemize}[leftmargin=*]
\item In the case of $\rolePMaybeCrashed =
    \rolePCrashed$, we have $\roleP \in \rolesC$. Hence,
    by applying~\autoref{def:assoc-queue} on \eqref{eq:stenvp-receive-assoc}, we have that
     $\stEnvApp{\qEnv}{\cdot, \roleP} = \stQUnavail$, a desired contradiction to  \eqref{eq:stenvp-receive-queue-con}.

\item In the case of $\rolePMaybeCrashed =
    \roleP$, by association, it holds that
     $\stEnvApp{\qEnv}{\roleQ, \roleP} = \stQEmpty$, a desired contradiction to \eqref{eq:stenvp-receive-queue-con}.
  \end{itemize}
  Therefore, we only need to consider the case that $\unfoldOne{\gtG} =
            \gtCommTransit{\roleQMaybeCrashed}{\roleP}{i \in
            I'}{\gtLab[i]}{\tyGroundi[i]}{\gtG[i]}{j}$.

 Then we want to show that $\gtLab[j] \neq \gtCrashLab$, which is proved by contradiction.
  Assume that $\gtLab[j] = \gtCrashLab$, by association, we have that
  $\stEnvApp{\qEnv}{\roleQ, \roleP} =  \stQEmpty
        $, a desired contradiction to
        \eqref{eq:stenvp-receive-queue-con}.
 Moreover, we want to show that $j = k$.
 By association and $\gtLab[j] \neq \gtCrashLab$, we
 have   $\stEnvApp{\qEnv}{\roleQ, \roleP}
      =
      \stQCons{\stQMsg{\gtLab[j]}{\tyGround[j]}}{\stQ}
      =
      \stQCons{\stQMsg{\stLab[k]}{\tyGround[k]}}{\stQi}
      $.
 Then by \eqref{eq:stenvp-receive-global-label-equvi}, it holds that $\gtLab[j]
 = \stLab[j] = \stLab[k]$.
  Furthermore, by  $j, k \in I$ and  the requirement that
   labels in local types must be pair-wise distinct, we have $j = k$, as required.
   Note that $k \in I'$ here.

  We can now apply \inferrule{\iruleGtMoveIn} (via~\autoref{lem:gt-lts-unfold}) to get:
 \begin{gather}
          \gtWithCrashedRoles{\rolesC}{\gtG}
          \gtMove[\stEnvAnnotGenericSym]{\rolesR}
          \gtWithCrashedRoles{\rolesC}{
            \gtG[k]
          }
        \end{gather}
        We are left to show $\stEnvAssoc{
          \gtWithCrashedRoles{\rolesC}{
            \gtG[k]
          }
        }{\stEnvi; \qEnvi}{\rolesR}$.

    \begin{enumerate}[label=(A\arabic*)]
            \item We need to show that $\forall \roleR \in \gtRoles{\gtG[k]}:
          \stEnvApp{\stEnvi}{\roleR} \stSub
        \gtProj[\rolesR]{\gtG[k]}{\roleR}$.
        We consider three subcases:
        \begin{itemize}[leftmargin=*]
       \item $\roleR = \roleQ$, which means that $\roleQMaybeCrashed = \roleQ$ and
               $\roleQ \in \gtRoles{\gtG[k]}$: by \eqref{eq:stenvp-receive-type-cxt-cons},
               $\stEnvApp{\stEnvi}{\roleQ} = \stEnvApp{\stEnv}{\roleQ}$.
               Then by association, we have $\stEnvApp{\stEnv}{\roleQ} \stSub
               \gtProj[\rolesR]{\gtG}{\roleQ} = \gtProj[\rolesR]{\gtG[k]}{\roleQ}$, as desired.

       \item $\roleR = \roleP$, which means that $\roleP \in \gtRoles{\gtG[k]}$: by \eqref{eq:stenvp-receive-type-cxt-cons}
        and \eqref{eq:stenvp-receive-global-label-equvi}, we have
          $\stEnvApp{\stEnvi}{\roleP} =  \stT[k] \stSub \gtProj[\rolesR]{\gtG[k]}{\roleP}$, as desired.

        \item $\roleR \neq \roleQ$ and $\roleR \neq \roleP$:
        $\gtProj[\rolesR]{\gtG}{\roleR} = \stMerge{i \in I'}{\gtProj[\rolesR]{\gtG[i]}{\roleR}}$.
        Furthermore, by \eqref{eq:stenvp-receive-type-cxt-cons} and $\stEnvAssoc{\gtWithCrashedRoles{\rolesC}{\gtG}}{\stEnv; \qEnv}{\rolesR}$,
        it holds that $\stEnvApp{\stEnvi}{\roleR} = \stEnvApp{\stEnv}{\roleR} \stSub \gtProj[\rolesR]{\gtG}{\roleR} = \stMerge{i \in I'}{\gtProj[\rolesR]{\gtG[i]}{\roleR}}$.
        Then, by~\autoref{lem:merge-subtyping} and transitivity of subtyping,  we can conclude that $\stEnvApp{\stEnvi}{\roleR}
        \stSub \gtProj[\rolesR]{\gtG[k]}{\roleR}$, as desired.
        \end{itemize}
         \item Trivial by $\stEnvAssoc{\gtWithCrashedRoles{\rolesC}{\gtG}}{\stEnv; \qEnv}{\rolesR}$. Note that in the case of
         $\roleQMaybeCrashed = \roleQCrashed$, we have $\roleQ \in \rolesC$, and hence, $\stEnvApp{\stEnvi}{\roleQ}
         =
         \stEnvApp{\stEnv}{\roleQ} = \stStop$.
          \item Trivial by $\stEnvAssoc{\gtWithCrashedRoles{\rolesC}{\gtG}}{\stEnv; \qEnv}{\rolesR}$ and the
                   fact that if $\roleP \notin \gtRoles{\gtG[k]}$, then $\stEnvApp{\stEnvi}{\roleP} = \stEnd$:
                   with $\roleP \notin \gtRolesCrashed{\gtG[k]}$, by~\autoref{lem:not-in-roles-end}, we have $\gtProj[\rolesR]{\gtG[k]}{\roleP} = \stEnd$.
                   Furthermore, by \eqref{eq:stenvp-receive-global-label-equvi}, it holds that $\stT[k] = \stEnd$, and thus,
          $\stEnvApp{\stEnvi}{\roleP} = \stEnd$, as desired.  The argument for the case that $\roleQMaybeCrashed
          = \roleQ$ and $\roleQ \notin \gtRoles{\gtG[k]}$ follows similarly.

          \item Since $\qEnv$ is associated with $\gtWithCrashedRoles{\rolesC}{\gtCommTransit{\roleQMaybeCrashed}{\roleP}{i \in
            I'}{\gtLab[i]}{\tyGroundi[i]}{\gtG[i]}{j} }$ and $\gtLab[j] \neq \gtCrashLab$, by~\autoref{def:assoc-queue},  we have that
             $\stEnvApp{\qEnv}{\roleQ, \roleP} =
        \stQCons{\stQMsg{\gtLab[j]}{\tyGround[j]}}{\stQ}$ and
            $\forall i \in I':  \stEnvUpd{\qEnv}{\roleQ, \roleP}{\stQ}$  
            is associated with 
      $\gtWithCrashedRoles{\rolesC}{\gtG[i]}$, which follows that $\stQ = \stQi$ (in \eqref{eq:stenvp-receive-queue-con}) and
      $\qEnvi = \stEnvUpd{\qEnv}{\roleQ, \roleP}{\stQi}
      = \stEnvUpd{\qEnv}{\roleQ, \roleP}{\stQ}$ is associated with
       $\gtWithCrashedRoles{\rolesC}{\gtG[k]}$, as required.
        \end{enumerate}

    \item Case (2): similar to the case (2) in Case $\inferrule{\iruleTCtxOut}$.
    \end{itemize}

    \item Case $\inferrule{\iruleTCtxCrash}$:

    From the premise, we have:
     \begin{gather}
      \stEnvAssoc{\gtWithCrashedRoles{\rolesC}{\gtG}}{\stEnv; \qEnv}{\rolesR}
      \\
       \stEnvApp{\stEnv}{\roleP} \neq \stEnd
       \label{eq:stenvp-crash-not-end}
       \\
        \stEnvApp{\stEnv}{\roleP} \neq \stStop
         \label{eq:stenvp-crash-not-stop}
        \\
        \stEnvAnnotGenericSym = \ltsCrash{\mpS}{\roleP}
      \\
      \stEnvi =
      \stEnvUpd{\stEnv}{\roleP}{\stStop}
      \label{eq:stenvp-crash-type-cons}
      \\
      \qEnvi =
     \stEnvUpd{\qEnv}{\cdot, \roleP}{\stQUnavail}
      \label{eq:stenvp-crash-queue-cons}%
    \end{gather}

By \eqref{eq:stenvp-crash-not-stop}, we know that
$\roleP \notin \rolesC$ and  $\roleP \in \gtRoles{\gtG}$.
The premise requires that $ \stEnvAnnotGenericSym \neq \ltsCrash{\mpS}{\roleP}$ for all
$\roleP \in \rolesR$, therefore, $\roleP \notin \rolesR$.
We apply \inferrule{\iruleGtMoveCrash} (via~\autoref{lem:gt-lts-unfold}) to get:
   \begin{gather}
          \gtWithCrashedRoles{\rolesC}{\gtG}
          \gtMove[\stEnvAnnotGenericSym]{\rolesR}
         \gtWithCrashedRoles{\rolesC \cup \setenum{\roleP}}{\gtCrashRole{\gtG}{\roleP}}
        \end{gather}

We are now left to show $\stEnvAssoc{
          \gtWithCrashedRoles{\rolesC \cup \setenum{\roleP}}{
            \gtCrashRole{\gtG}{\roleP}
          }
        }{\stEnvi; \qEnvi}{\rolesR}$.

Note that $\gtGi = \gtCrashRole{\gtG}{\roleP}$ and $\rolesCi = \rolesC \cup \setenum{\roleP}$ here.
By \eqref{eq:stenvp-crash-type-cons}, we can set  $\stEnvi =  \stEnvi[\gtGi] \stEnvComp \stEnvi[\stStopSym] \stEnvComp
  \stEnvi[\stEnd]$, where $\dom{\stEnvi[\gtGi]} = \setcomp{\roleQ}{\roleQ \in \gtRoles{\gtCrashRole{\gtG}{\roleP}}} =
  \setcomp{\roleQ}{\roleQ \in \gtRoles{\gtG}} \setminus \setenum{\roleP} = \dom{\stEnv[\gtG]} \setminus \setenum{\roleP}$,
  $\dom{\stEnvi[\stStopSym]} = \rolesC \cup \setenum{\roleP}$ =
  $\dom{\stEnv[\stStopSym]} \cup \setenum{\roleP}$,
  and $\dom{\stEnvi[\stEnd]} = \dom{\stEnv[\stEnd]}$. Meanwhile, we have $\stEnvApp{\stEnvi}{\roleQ} = \stEnvApp{\stEnv}{\roleQ}$ if
  $\roleQ \neq \roleP$.

        \begin{enumerate}[label=(A\arabic*)]
          \item We want to show that for any $\roleQ \in \gtRoles{\gtCrashRole{\gtG}{\roleP}}$, we have $\stEnvApp{\stEnvi}{\roleQ} =
          \stEnvApp{\stEnv}{\roleQ} \stSub \gtProj[\rolesR]{\gtG}{\roleQ} \stSub  \gtProj[\rolesR]{(\gtCrashRole{\gtG}{\roleP})}{\roleQ}$ by
          \autoref{lem:proj-non-crashing-role-preserve} and $\stEnvAssoc{\gtWithCrashedRoles{\rolesC}{\gtG}}{\stEnv; \qEnv}{\rolesR}$.
          \item Trivial by $\stEnvAssoc{\gtWithCrashedRoles{\rolesC}{\gtG}}{\stEnv; \qEnv}{\rolesR}$ and \eqref{eq:stenvp-crash-type-cons}, i.e.,
          $\stEnvApp{\stEnvi}{\roleP} = \stStop$.
          \item Trivial by $\stEnvAssoc{\gtWithCrashedRoles{\rolesC}{\gtG}}{\stEnv; \qEnv}{\rolesR}$.
          \item Since $\qEnv$ is associated with $\gtWithCrashedRoles{\rolesC}{\gtG}$,
          by~\autoref{def:assoc-queue}, we have that for any $\roleQ \in \rolesC$,
        $\stEnvApp{\qEnv}{\cdot, \roleQ} = \stQUnavail$. Then, with \eqref{eq:stenvp-crash-queue-cons},
        we have that for any $\roleQ \in \rolesC \cup \setenum{\roleP}$,
        $\stEnvApp{\qEnvi}{\cdot, \roleQ} =
           \stEnvApp{\stEnvUpd{\qEnv}{\cdot, \roleP}{\stQUnavail}}{\cdot, \roleQ} =
           \stQUnavail$, which follows directly that $\qEnvi$ is associated with
           $\gtWithCrashedRoles{\rolesC \cup \setenum{\roleP}}{\gtCrashRole{\gtG}{\roleP}}$.
         \end{enumerate}

    \item Case $\inferrule{\iruleTCtxCrashDetect}$:

    From the premise, we have:
    \begin{gather}
     \stEnvAssoc{\gtWithCrashedRoles{\rolesC}{\gtG}}{\stEnv; \qEnv}{\rolesR}
     \\
     \stEnvApp{\stEnv}{\roleQ} =
      \stExtSum{\roleP}{i \in I}{\stChoice{\stLab[i]}{\tyGround[i]} \stSeq \stT[i]}
      \label{eq:stenvp-crash-detect-comp}
      \\
       \stEnvApp{\stEnv}{\roleP} = \stStop
       \\
        \stEnvAnnotGenericSym = \ltsCrDe{\mpS}{\roleQ}{\roleP}
        \\
       k \in I
       \\
        \stLab[k] = \stCrashLab
      \\
      \stEnvApp{\qEnv}{\roleP, \roleQ} = \stQEmpty
      \label{eq:stenvp-crash-detect-queue-empty}
      \\
      \stEnvi = \stEnvUpd{\stEnv}{\roleQ}{\stT[k]}
      \label{eq:stenvp-crash-detect-type-cxt-cons}
      \\
      \qEnvi = \qEnv
    \end{gather}
  Since $\stEnvApp{\stEnv}{\roleP} = \stStop$, we have $\roleP \in \rolesC$ and
  $\roleP \notin \rolesR$.
  Apply~\autoref{lem:inv-proj}~\cref{item:proj-inv:recv} on
    \eqref{eq:stenvp-crash-detect-comp}, we have two cases.
    \begin{itemize}[leftmargin=*]
      \item Case (1):
        \begin{gather}
          \unfoldOne{\gtG} =
            \gtCommTransit{\rolePCrashed}{\roleQ}{i \in
            I'}{\gtLab[i]}{\tyGroundi[i]}{\gtG[i]}{j} %
      \\
      I' \subseteq I
      \\
      \forall i \in I':
      \stLab[i] = \gtLab[i],
      \stT[i] \stSub (\gtProj[\rolesR]{\gtG[i]}{\roleQ}),
      \tyGroundi[i] = \tyGround[i]
      \label{eq:stenvp-crash-detect-global-label-equvi}
      \\
      \roleP \notin \rolesR \text{ implies } \exists l \in I': \gtLab[l] =
          \gtCrashLab
    \end{gather}
  Then we want to show that $\gtLab[j] = \gtCrashLab$, which is proved by contradiction.
  Assume that $\gtLab[j] \neq \gtCrashLab$, by association, we have that
  $\stEnvApp{\qEnv}{\roleP, \roleQ} =
        \stQCons{\stQMsg{\gtLab[j]}{\tyGround[j]}}{\stQ}$, a desired contradiction to
        \eqref{eq:stenvp-crash-detect-queue-empty}.
   Moreover, we want to show that $j = k$. By \eqref{eq:stenvp-crash-detect-global-label-equvi},
   we have $\gtLab[j] = \stLab[j]$, which means that $\stLab[j] = \stCrashLab$.
   Since $\stLab[j] = \stLab[k] = \stCrashLab$ and $j, k \in I$,  by the requirement that
   labels in local types must be pair-wise distinct, we have $j = k$, as required.
   Note that $k \in I'$ here.

  We can now apply \inferrule{\iruleGtMoveCrDe} (via~\autoref{lem:gt-lts-unfold}) to get:
 \begin{gather}
          \gtWithCrashedRoles{\rolesC}{\gtG}
          \gtMove[\stEnvAnnotGenericSym]{\rolesR}
          \gtWithCrashedRoles{\rolesC}{
            \gtG[k]
          }
        \end{gather}
        We are left to show $\stEnvAssoc{
          \gtWithCrashedRoles{\rolesC}{
            \gtG[k]
          }
        }{\stEnvi; \qEnvi}{\rolesR}$.
 \begin{enumerate}[label=(A\arabic*)]
            \item We need to show that $\forall \roleR \in \gtRoles{\gtG[k]}:
          \stEnvApp{\stEnvi}{\roleR} \stSub
        \gtProj[\rolesR]{\gtG[k]}{\roleR}$.
        We consider two subcases:
        \begin{itemize}[leftmargin=*]
        \item $\roleR = \roleQ$, which means that $\roleQ \in \gtRoles{\gtG[k]}$: by \eqref{eq:stenvp-crash-detect-type-cxt-cons}
        and \eqref{eq:stenvp-crash-detect-global-label-equvi}, we have
          $\stEnvApp{\stEnvi}{\roleQ} =  \stT[k] \stSub \gtProj[\rolesR]{\gtG[k]}{\roleQ}$, as desired.
        \item $\roleR \neq \roleQ$: with $\roleP \notin \gtRoles{\gtG[k]}$, we have $\roleR \neq \roleP$, and hence,
        $\gtProj[\rolesR]{\gtG}{\roleR} = \stMerge{i \in I'}{\gtProj[\rolesR]{\gtG[i]}{\roleR}}$.
        Furthermore, by \eqref{eq:stenvp-crash-detect-type-cxt-cons} and $\stEnvAssoc{\gtWithCrashedRoles{\rolesC}{\gtG}}{\stEnv; \qEnv}{\rolesR}$,
        it holds that $\stEnvApp{\stEnvi}{\roleR} = \stEnvApp{\stEnv}{\roleR} \stSub \gtProj[\rolesR]{\gtG}{\roleR} = \stMerge{i \in I'}{\gtProj[\rolesR]{\gtG[i]}{\roleR}}$.
        Then, by~\autoref{lem:merge-subtyping} and transitivity of subtyping,  we can conclude that $\stEnvApp{\stEnvi}{\roleR}
        \stSub \gtProj[\rolesR]{\gtG[k]}{\roleR}$, as desired.
        \end{itemize}
         \item Trivial by $\stEnvAssoc{\gtWithCrashedRoles{\rolesC}{\gtG}}{\stEnv; \qEnv}{\rolesR}$.
          \item Trivial by $\stEnvAssoc{\gtWithCrashedRoles{\rolesC}{\gtG}}{\stEnv; \qEnv}{\rolesR}$ and the
                   fact that if $\roleQ \notin \gtRoles{\gtG[k]}$, then $\stEnvApp{\stEnvi}{\roleQ} = \stEnd$:
                   with $\roleQ \notin \gtRolesCrashed{\gtG[k]}$, by~\autoref{lem:not-in-roles-end}, we have $\gtProj[\rolesR]{\gtG[k]}{\roleQ} = \stEnd$.
                   Furthermore, by \eqref{eq:stenvp-crash-detect-global-label-equvi}, it holds that $\stT[k] = \stEnd$, and thus,
          $\stEnvApp{\stEnvi}{\roleQ} = \stEnd$, as desired.
          \item Since $\qEnv$ is associated with $\gtWithCrashedRoles{\rolesC}{\gtCommTransit{\rolePCrashed}{\roleQ}{i \in
            I'}{\gtLab[i]}{\tyGroundi[i]}{\gtG[i]}{j} }$ and $\gtLab[j] = \gtCrashLab$, by~\autoref{def:assoc-queue},  we have that
            $\forall i \in I': \qEnv \text{~is associated with~}
      \gtWithCrashedRoles{\rolesC}{\gtG[i]}$, which follows that $\qEnvi = \qEnv$ is associated with
       $\gtWithCrashedRoles{\rolesC}{\gtG[k]}$, as required.
        \end{enumerate}

       \item Case (2): similar to the case (2) in Case $\inferrule{\iruleTCtxOut}$.
    \end{itemize}
     \item Case $\inferrule{\iruleTCtxRec}$:

      By induction hypothesis and~\autoref{prop:gt-lts-unfold}.
    \qedhere
  \end{itemize}
\end{proof}

\subsection{Safety by Projection}
\label{sec:proof:safety}

\begin{lem}
\label{lem:safety-by-proj}
\label{lem:ext-proj-safe}
  If \;$\stEnvAssoc{\gtWithCrashedRoles{\rolesC}{\gtG}}{\stEnv; \qEnv}{\rolesR}$,\;
  then\; $\stEnv; \qEnv$ \;is\; $\rolesR$-safe.%
\end{lem}

\begin{proof}
Let $\predP = \setcomp{\stEnvi; \qEnvi}{\exists \rolesCi, \gtGi: \gtWithCrashedRoles{\rolesC}{\gtG}
\gtMoveStar[\rolesR]{}
\gtWithCrashedRoles{\rolesCi}{\gtGi}
\text{ and }
\stEnvAssoc{\gtWithCrashedRoles{\rolesCi}{\gtGi}}{\stEnvi; \qEnvi}{\rolesR}
}$.
Take any $\stEnvi; \qEnvi \in \predP$, we show that $\stEnvi; \qEnvi$ satisfies all
clauses in~\autoref{def:mpst-env-safe}, which means that $\predP$ is an $\rolesR$-safety 
property. Then we can conclude that, since $\predPApp{\stEnv; \qEnv}$ holds, 
$\stEnv; \qEnv$ is $\rolesR$-safe. 

By definition of $\predP$, there exists $\gtWithCrashedRoles{\rolesCi}{\gtGi}$ with
$\gtWithCrashedRoles{\rolesC}{\gtG}
\gtMoveStar[\rolesR]{}
\gtWithCrashedRoles{\rolesCi}{\gtGi}$ and 
$\stEnvAssoc{\gtWithCrashedRoles{\rolesCi}{\gtGi}}{\stEnvi; \qEnvi}{\rolesR}$. 
\begin{itemize}%
\item \inferrule{\iruleSafeComm}: 
from the premise, we have 
\begin{gather}
    \stEnvAssoc{\gtWithCrashedRoles{\rolesCi}{\gtGi}}{\stEnvi; \qEnvi}{\rolesR}
    \label{eq:safety-assoc-1}
    \\
    \stEnvApp{\stEnvi}{%
        \roleQ%
      } =
      \stExtSum{\roleP}{i \in I}{\stChoice{\stLab[i]}{\tyGround[i]} \stSeq \stT[i]}
      \label{eq:stenvp-safety-rec-comp}
    \\
    \stEnvApp{\qEnvi}{\roleP, \roleQ} \neq \stQUnavail
    \label{eq:safety-queue-available}
    \\
    \stEnvApp{\qEnvi}{\roleP, \roleQ} \neq \stQEmpty
    \label{eq:safety-queue-not-empty}
\end{gather}
Apply~\cref{item:proj-inv:recv} of~\autoref{lem:inv-proj} on \eqref{eq:stenvp-safety-rec-comp}, 
we have two cases.
\begin{itemize}[leftmargin=*]
\item Case (1):
\begin{gather}
\unfoldOne{\gtGi} = \gtCommTransit{\rolePMaybeCrashed}{\roleQ}{i \in
            I'}{\gtLab[i]}{\tyGroundi[i]}{\gtG[i]}{j} \text{ or } 
           \unfoldOne{\gtGi} =  \gtComm{\roleP}{\roleQMaybeCrashed}{i \in I'}{\gtLab[i]}{\tyGroundi[i]}{\gtG[i]}
           \\
            I' \subseteq I
 \label{eq:safety-subset}
      \\
      \forall i \in I':
      \stLab[i] = \gtLab[i],
      \stT[i] \stSub (\gtProj[\rolesR]{\gtG[i]}{\roleQ}),
      \tyGroundi[i] = \tyGround[i]
      \label{eq:stenvp-safety-receive-label-equvi}
\end{gather}
First, we show that $\unfoldOne{\gtGi}$ cannot be the form of 
$\gtComm{\roleP}{\roleQMaybeCrashed}{i \in I'}{\gtLab[i]}{\tyGroundi[i]}{\gtG[i]}$. We prove this by contradiction on two subcases:
$\roleQMaybeCrashed = \roleQ$ and $\roleQMaybeCrashed = \roleQCrashed$:
\begin{itemize}[leftmargin=*]
\item $\roleQMaybeCrashed = \roleQ$: by association \eqref{eq:safety-assoc-1}, we have 
$\stEnvApp{\qEnvi}{\roleP, \roleQ} = \stQEmpty$, a desired contradiction to \eqref{eq:safety-queue-not-empty}. 
\item $\roleQMaybeCrashed = \roleQCrashed$: we have $\roleQ \in \rolesC$. Hence, by association \eqref{eq:safety-assoc-1}, we have 
$\stEnvApp{\qEnvi}{\cdot, \roleQ} = \stQUnavail$, a desired contradiction to \eqref{eq:safety-queue-available}. 
\end{itemize}

It follows that $\unfoldOne{\gtGi} =  \gtCommTransit{\rolePMaybeCrashed}{\roleQ}{i \in
            I'}{\gtLab[i]}{\tyGroundi[i]}{\gtG[i]}{j}$. We are left to show that $j \in I$, $\tyGroundi[j] = \tyGround[j]$, and 
$ \stEnvApp{\qEnvi}{\roleP, \roleQ}  =  \stQCons{\stQMsg{\stLab[j]}{\tyGroundi[j]}}{\stQ}$,  and then by applying 
\inferrule{\iruleTCtxIn}, we can conclude with $\stEnvMoveAnnotP{\stEnvi; \qEnvi}{\stEnvInAnnot{\roleQ}{\roleP}{\stChoice{\stLab[j]}{\tyGround[j]}}}$. 
We consider two subcases: $\gtLab[j] = \gtCrashLab$ and $\gtLab[j] \neq \gtCrashLab$:
\begin{itemize}[leftmargin=*]
\item $\gtLab[j] = \gtCrashLab$: by association, 
$\stEnvApp{\qEnvi}{\roleP, \roleQ} = \stQEmpty$, a desired contradiction to \eqref{eq:safety-queue-not-empty}. 
\item $\gtLab[j] \neq \gtCrashLab$: by association, $\stEnvApp{\qEnvi}{\roleP, \roleQ} =
        \stQCons{\stQMsg{\gtLab[j]}{\tyGroundi[j]}}{\stQ}$. Moreover, with \eqref{eq:safety-subset} and \eqref{eq:stenvp-safety-receive-label-equvi}, 
        we obtain that $j \in I$, $\tyGroundi[j] = \tyGround[j]$, and $\stLab[j] = \gtLab[j]$, which follows that 
$ \stEnvApp{\qEnvi}{\roleP, \roleQ}  =  \stQCons{\stQMsg{\stLab[j]}{\tyGroundi[j]}}{\stQ}$, as desired. 
\end{itemize}

\item Case (2):  $\unfoldOne{\gtGi} =
            \gtComm{\roleS}{\roleTMaybeCrashed}{j \in J}{\gtLab[j]}{\tyGroundi[j]}{\gtG[j]}$,
          or
          $\unfoldOne{\gtGi} =
            \gtCommTransit{\roleSMaybeCrashed}{\roleT}{j \in
            J}{\gtLab[j]}{\tyGroundi[j]}{\gtG[j]}{k}$,
          where for all $j \in J$:
          $  \stEnvApp{\stEnvi}{%
        \roleQ%
      }  \stSub (\gtProj[\rolesR]{\gtG[j]}{\roleQ})$,
          with $\roleQ \neq \roleS$ and $\roleQ \neq \roleT$.
          
          We apply~\autoref{lem:eventual-global-form} to get that there exists a global type 
          $\gtWithCrashedRoles{\rolesCii}{\gtGii}$ and a queue environment $\qEnvii$ such that 
 $  \stEnvApp{\stEnvi}{%
        \roleQ%
      } \stSub 
 \gtProj[\rolesR]{\gtGii}{\roleQ}$, $\unfoldOne{\gtGii}$ is of the form 
 $ \gtCommTransit{\rolePMaybeCrashed}{\roleQ}{i \in
            I'}{\gtLab[i]}{\tyGroundi[i]}{\gtG[i]}{j}$ or
 $ \gtComm{\roleP}{\roleQMaybeCrashed}{i \in I'}{\gtLab[i]}{\tyGroundi[i]}{\gtG[i]}$, and $\qEnvii$ is 
 associated with $\gtWithCrashedRoles{\rolesCii}{\gtGii}$ with $\stEnvApp{\qEnvii}{\roleP, \roleQ} = \stEnvApp{\qEnvi}{\roleP, \roleQ}$.  
 The following proof argument is similar to that for Case (1). 
\end{itemize}

\item \inferrule{\iruleSafeCrash}:  from the premise, we have 
$ \stEnvApp{\stEnvi}{\roleQ} =
      \stExtSum{\roleP}{i \in I}{\stChoice{\stLab[i]}{\tyGround[i]} \stSeq \stT[i]}, 
       \stEnvApp{\stEnvi}{\roleP} = \stStop$, and $\stEnvApp{\qEnvi}{\roleP, \roleQ} = \stQEmpty$. 
       We just need to show that there exists $k \in I$ with $\stLab[k] = \stCrashLab$, and then by applying 
       \inferrule{\iruleTCtxCrashDetect}, we can conclude with 
       $\stEnvMoveAnnotP{\stEnvi; \qEnvi}{\ltsCrDe{\mpS}{\roleQ}{\roleP}}$. 
       By the association that $\stEnvAssoc{\gtWithCrashedRoles{\rolesCi}{\gtGi}}{\stEnvi; \qEnvi}{\rolesR}$, 
       we have $ \stEnvApp{\stEnvi}{\roleQ} =
      \stExtSum{\roleP}{i \in I}{\stChoice{\stLab[i]}{\tyGround[i]} \stSeq \stT[i]} 
      \stSub 
      \gtProj[\rolesR]{\gtGi}{\roleQ}$. Moreover, with $\roleP \notin \rolesR$, which follows directly from 
      $\stEnvApp{\stEnvi}{\roleP} = \stStop$, we can apply~\autoref{lem:crash-lab-exists} to 
      get $\exists k \in I : \stLab[k] = \stCrashLab$, as desired. 

\item \inferrule{\iruleSafeRec}: let $\stEnvii; \qEnvii$ be constructed from $\stEnvi; \qEnvi$ with
$\stEnvii = \stEnvUpd{\stEnvi}{%
        \roleP%
      }{%
        \stS\subst{\stRecVar}{\stRec{\stRecVar}{\stS}}%
      }$ and
$\qEnvii = \qEnvi$. 
We want to show that $\stEnvAssoc{\gtWithCrashedRoles{\rolesCi}{\gtGi}}{\stEnvii; \qEnvii}{\rolesR}$. 
From the premise, we have $\stEnvApp{\stEnvi}{{\roleP}}  =  \stRec{\stRecVar}{\stS}$.
Then by association that $\stEnvAssoc{\gtWithCrashedRoles{\rolesCi}{\gtGi}}{\stEnvi; \qEnvi}{\rolesR}$, 
it holds that $\stEnvApp{\stEnvi}{{\roleP}}   = \stRec{\stRecVar}{\stS} 
        \stSub
        \gtProj[\rolesR]{\gtGi}{\roleP}$. Hence, by inverting \inferrule{\iruleStSubRecL}, we have
        $\stEnvApp{\stEnvii}{{\roleP}} = \stS{}\subst{\stRecVar}{\stRec{\stRecVar}{\stS}}    
         \stSub
        \gtProj[\rolesR]{\gtGi}{\roleP}$. Therefore, by~\autoref{def:assoc}, we conclude with 
        $\stEnvAssoc{\gtWithCrashedRoles{\rolesCi}{\gtGi}}{\stEnvii; \qEnvii }{\rolesR}
    $, as desired.
\item \inferrule{\iruleSafeMove}:  from the premise, we have $\stEnvi; \qEnvi \stEnvMoveMaybeCrash[\rolesR] \stEnvii; \qEnvii$. 
Hence, by~\autoref{thm:gtype:proj-comp}, there exists $\gtWithCrashedRoles{\rolesCii}{\gtGii}$ with
 $\gtWithCrashedRoles{\rolesCi}{\gtGi} \gtMove[\stEnvAnnotGenericSym]{\rolesR}
    \gtWithCrashedRoles{\rolesCii}{\gtGii}$, where $\stEnvAnnotGenericSym \neq \ltsCrash{\mpS}{\roleP}$ for all $\roleP
  \in \rolesR$, and  $\stEnvAssoc{\gtWithCrashedRoles{\rolesCii}{\gtGii}}{\stEnvii; \qEnvii}{\rolesR}$.
By definition of $\predP$, the configuration after transition $\stEnvii; \qEnvii$ is in $\predP$, as desired. 
\qedhere
\end{itemize}
 \end{proof}

\subsection{Deadlock Freedom by Projection}
\label{sec:proof:deadlockfree}
\begin{lem}
\label{lem:proj:df}
\label{lem:ext-proj-deadlock-free}
  If \;$\stEnvAssoc{\gtWithCrashedRoles{\rolesC}{\gtG}}{\stEnv; \qEnv}{\rolesR}$,\;
  then\; $\stEnv; \qEnv$ \;is\; $\rolesR$-deadlock-free.%
\end{lem}

\begin{proof}
Since $\stEnvAssoc{\gtWithCrashedRoles{\rolesC}{\gtG}}{\stEnv; \qEnv}{\rolesR}$,
by \autoref{lem:ext-proj-safe}, we have that $\stEnv; \qEnv$ is $\rolesR$-safe.
By operational correspondence of global type $\gtWithCrashedRoles{\rolesC}{\gtG}$
and configuration $\stEnv; \qEnv$~(\autoref{thm:gtype:proj-comp} and~\autoref{thm:gtype:proj-sound}),
there exists
a global type $\gtWithCrashedRoles{\rolesCi}{\gtGi}$ such that
$\gtWithCrashedRoles{\rolesC}{\gtG} \!\gtMoveStar[\rolesR]\!
\gtWithCrashedRoles{\rolesCi}{\gtGi} \!\not\gtMove[]{\rolesR}$,
with associated configurations
 $\stEnv; \qEnv \!\stEnvMoveMaybeCrashStar[\rolesR]\! \stEnvi; \qEnvi
\!\not\stEnvMoveMaybeCrash[\rolesR]$. Since no further reductions are possible for the global type,
the global type $\gtWithCrashedRoles{\rolesCi}{\gtGi}$ must be of the form $\gtWithCrashedRoles{\rolesCi}{\stEnd}$.
By applying~\autoref{def:assoc} on $\stEnvAssoc{\gtWithCrashedRoles{\rolesCi}{\stEnd}}{\stEnvi; \qEnvi}{\rolesR}$,
we have
$\stEnvi =  \stEnvi[\stEnd] \stEnvComp \stEnvi[\stStopSym]$, where
$\dom{\stEnvi[\stEnd]} = \setcomp{\roleP}{\stEnvApp{\stEnvi}{\roleP} = \stEnd}$ and
$\dom{\stEnvi[\stStopSym]} = \rolesCi = \setcomp{\roleP}{\stEnvApp{\stEnvi}{\roleP} = \stStop}$, as required.
By association again, we have that,
for any $\roleP, \roleQ$, if $\roleQ \in \rolesCi (= \dom{\stEnvi[\stStopSym]}),
 \stEnvApp{\qEnvi}{\cdot, \roleQ} = \stQUnavail$, and otherwise, $\stEnvApp{\qEnvi}{\roleP, \roleQ} = \stQEmpty$, as required.
\end{proof}

\subsection{Liveness by Projection}
\label{sec:proof:liveness}
\begin{lem}
 \label{lem:ext-proj-live}
  If \;$\stEnvAssoc{\gtWithCrashedRoles{\rolesC}{\gtG}}{\stEnv; \qEnv}{\rolesR}$,\;
  then\; $\stEnv; \qEnv$ \;is\; $\rolesR$-live.%
\end{lem}
\begin{proof}
Since $\stEnvAssoc{\gtWithCrashedRoles{\rolesC}{\gtG}}{\stEnv; \qEnv}{\rolesR}$,
by~\autoref{lem:ext-proj-safe}, we have that $\stEnv; \qEnv$ is $\rolesR$-safe.

We are left to
show that if  $\stEnv; \qEnv \stEnvMoveMaybeCrashStar[\rolesR]
  \stEnvi; \qEnvi$, then any non-crashing path starting with $\stEnvi; \qEnvi$ which is fair is also live.
  By operational correspondence of global type $\gtWithCrashedRoles{\rolesC}{\gtG}$
and configuration $\stEnv; \qEnv$ (\autoref{thm:gtype:proj-comp} and~\autoref{thm:gtype:proj-sound}),
there exists
a global type $\gtWithCrashedRoles{\rolesCi}{\gtGi}$ such that
$\gtWithCrashedRoles{\rolesC}{\gtG} \!\gtMoveStar[\rolesR]\!
\gtWithCrashedRoles{\rolesCi}{\gtGi}$ and
$\stEnvAssoc{\gtWithCrashedRoles{\rolesCi}{\gtGi}}{\stEnvi; \qEnvi}{\rolesR}$.
We construct a non-crashing fair path $(\stEnvi[n]; \qEnvi[n])_{n \in N}$,
where $N = \setenum{0, 1, 2, \ldots}$, $\stEnvi[0] = \stEnvi$,
$\qEnvi[0] = \qEnvi$, and $\forall n \in N$, $\stEnvi[n]; \qEnvi[n] \!\stEnvMove\!
\stEnvi[n+1]; \qEnvi[n+1]$. Then we just need to show that $(\stEnvi[n]; \qEnvi[n])_{n \in N}$
is live.
  \begin{enumerate}[label=(L\arabic*)]
  \item Suppose that $\stEnvApp{\qEnvi[n]}{\roleP, \roleQ}
      =
      \stQCons{\stQMsg{\stLab}{\tyGround}}{\stQ}
      \neq
      \stQUnavail$ and $\stLab \neq \stCrashLab$. By operational correspondence
      of $\gtWithCrashedRoles{\rolesCi}{\gtGi}$ and configuration $\stEnvi[0]; \qEnvi[0]$,  and
      $\forall n \in N$,  $\stEnvi[n]; \qEnvi[n] \!\stEnvMove\!
\stEnvi[n+1]; \qEnvi[n+1]$, there
      exists $\gtWithCrashedRoles{\rolesCi[n]}{\gtGi[n]}$ such that
      $\stEnvAssoc{\gtWithCrashedRoles{\rolesCi[n]}{\gtGi[n]}}{\stEnvi[n]; \qEnvi[n]}{\rolesR}$.
      By~\autoref{lem:eventual-global-form} and~\autoref{def:assoc}, we only need to consider
      the case that
       $\gtGi[n] =
      \gtCommTransit{\rolePMaybeCrashed}{\roleQ}{i \in I}{\gtLab[i]}{\tyGround[i]}{\gtGii[i]}{j}$ with
      $\stLab = \gtLab[j] \neq \gtCrashLab$ and $\tyGround[j] = \tyGround$.
      By applying \inferrule{\iruleGtMoveIn},
      $\gtWithCrashedRoles{\rolesCi[n]}{\gtGi[n]}
      \gtMove[
      \stEnvInAnnotSmall{\roleQ}{\roleP}{\stChoice{\gtLab[j]}{\tyGround[j]}}
    ]{
      \rolesR
    }
     \gtWithCrashedRoles{\rolesCi[n]}{\gtGii[j]}$. Hence, by the soundness of association, it holds
     that $\stEnvi[n]; \qEnvi[n] \stEnvMoveInAnnot{\roleQ}{\roleP}{\stChoice{\stLab}{\tyGround}}$.
     Finally, combining the fact that $(\stEnvi[n]; \qEnvi[n])_{n \in N}$  is fair with $\stEnvi[n]; \qEnvi[n] \stEnvMoveInAnnot{\roleQ}{\roleP}{\stChoice{\stLab}{\tyGround}}$, we can conclude that there exists $k$ such that $n \leq k \in N$ and
     $\stEnvi[k]; \qEnvi[k] \stEnvMoveInAnnot{\roleQ}{\roleP}{\stChoice{\stLab}{\tyGround}} \stEnvi[k+1]; \qEnvi[k+1]$, as desired.

 \item Suppose that  $\stEnvApp{\stEnvi[n]}{%
        \roleP%
      } =
      \stExtSum{\roleQ}{i \in I}{\stChoice{\stLab[i]}{\tyGround[i]} \stSeq \stT[i]}$. By operational correspondence
      of $\gtWithCrashedRoles{\rolesCi}{\gtGi}$ and configuration $\stEnvi[0]; \qEnvi[0]$,  and
      $\forall n \in N$,  $\stEnvi[n]; \qEnvi[n] \!\stEnvMove\!
\stEnvi[n+1]; \qEnvi[n+1]$, there
      exists $\gtWithCrashedRoles{\rolesCi[n]}{\gtGi[n]}$ such that
      $\stEnvAssoc{\gtWithCrashedRoles{\rolesCi[n]}{\gtGi[n]}}{\stEnvi[n]; \qEnvi[n]}{\rolesR}$.
      By~\autoref{def:assoc}, \autoref{lem:inv-proj}, and~\autoref{lem:eventual-global-form},
      we only have to consider that  $\gtGi[n]$ is of the form
     $ \gtComm{\roleQ}{\rolePMaybeCrashed}{i \in I'}{\gtLab[i]}{\tyGroundi[i]}{\gtGii[i]}$
          or
          $
            \gtCommTransit{\roleQMaybeCrashed}{\roleP}{i \in
            I'}{\gtLab[i]}{\tyGroundi[i]}{\gtGii[i]}{j}$,
          where
          $I' \subseteq I$, and
          for all $i \in I'$:
          $\stLab[i] = \gtLab[i]$,
          $\tyGroundi[i] = \tyGround[i]$,
          and $\roleQ \notin \rolesR$ implies $\exists k \in I': \gtLab[k] =
          \gtCrashLab$.
          \begin{itemize}[leftmargin=*]
                \item %
      $\gtGi[n] =
      \gtComm{\roleQ}{\rolePMaybeCrashed}{i \in I'}{\gtLab[i]}{\tyGroundi[i]}{\gtGii[i]}$: we first show that
      $\rolePMaybeCrashed = \roleP$. Since $\stEnvApp{\stEnvi[n]}{%
        \roleP%
      } \neq \stStop$, by~\autoref{def:assoc}, we have that $\roleP \notin \rolesCi[n]$, and hence, $\rolePMaybeCrashed
      \neq \rolePCrashed$. Given that $\gtGi[n] =
      \gtComm{\roleQ}{\roleP}{i \in I'}{\gtLab[i]}{\tyGroundi[i]}{\gtGii[i]}$, by association, \autoref{def:assoc}, and inversion of subtyping,
       $\stEnvApp{\stEnvi[n]}{%
        \roleQ%
      }$ is of the form
     $\stIntSum{\roleP}{i \in I''}{\stChoice{\stLab[i]}{\tyGroundi[i]} \stSeq \stTi[i]}$ where
      $I'' \subseteq I'$. Then applying \inferrule{\iruleTCtxOut},
      $\stEnvi[n]; \qEnvi[n]
      \,\stEnvMoveOutAnnot{\roleQ}{\roleP}{\stChoice{\stLab[j]}{\tyGroundi[j]}}$ for some $j \in I''$.
      Then together with the fairness of $(\stEnvi[n]; \qEnvi[n])_{n \in N}$,
      we have that there exists $k, \stLabi, \tyGroundii$ such that $n \leq k \in N$ and
     $\stEnvi[k]; \qEnvi[k] \stEnvMoveOutAnnot{\roleQ}{\roleP}{\stChoice{\stLabi}{\tyGroundii}} \stEnvi[k+1]; \qEnvi[k+1]$,
     which follows that  $\stEnvApp{\qEnvi[k+1]}{\roleQ, \roleP} =
     \stEnvUpd{\qEnvi[k]}{\roleQ, \roleP}{
        \stQCons{
          \stEnvApp{\qEnvi[k]}{\roleQ, \roleP}
        }{
          \stQMsg{\stLabi}{\tyGroundii}
        }}$. Finally, by the previous case (L1), we can conclude that there exists $k', \stLabii, \tyGroundiii$ such that $k + 1 \leq k' \in N$ and
         $\stEnvi[k']; \qEnvi[k'] \stEnvMoveInAnnot{\roleP}{\roleQ}{\stChoice{\stLabii}{\tyGroundiii}} \stEnvi[k'+1]; \qEnvi[k'+1]$, as desired.

    \item $\gtGi[n] =
            \gtCommTransit{\roleQMaybeCrashed}{\roleP}{i \in
            I'}{\gtLab[i]}{\tyGroundi[i]}{\gtGii[i]}{j}$: we consider two subcases:
            \begin{itemize}[leftmargin=*]
            \item $\gtLab[j] \neq \gtCrashLab$: by applying \inferrule{\iruleGtMoveIn},
            $\gtWithCrashedRoles{\rolesCi[n]}{
     \gtGi[n]
    }
    \gtMove[
      \stEnvInAnnotSmall{\roleP}{\roleQ}{\stChoice{\gtLab[j]}{\tyGroundi[j]}}
    ]{
      \rolesR
    }
    \gtWithCrashedRoles{\rolesCi[n]}{\gtGii[j]}$.
    Hence, by the soundness association, we have that  $\stEnvi[n]; \qEnvi[n]
      \,\stEnvMoveInAnnot{\roleP}{\roleQ}{\stChoice{\stLab[j]}{\tyGroundi[j]}}$.
       Finally, combining the fact that $(\stEnvi[n]; \qEnvi[n])_{n \in N}$  is fair with
      $\stEnvi[n]; \qEnvi[n] \stEnvMoveInAnnot{\roleP}{\roleQ}{\stChoice{\stLab[j]}{\tyGroundi[j]}}$,
      we can conclude that there exists $k$ such that $n \leq k \in N$ and
     $\stEnvi[k]; \qEnvi[k] \stEnvMoveInAnnot{\roleP}{\roleQ}{\stChoice{\stLab[j]}{\tyGroundi[j]}} \stEnvi[k+1]; \qEnvi[k+1]$, as desired.
   \item $\gtLab[j] = \gtCrashLab$: consider the following two subcases:
   \begin{itemize}[leftmargin=*]
   \item $\roleQMaybeCrashed = \roleQCrashed$:
    $\stEnvApp{\stEnvi[n]}{%
        \roleQ%
      } = \stStop$ by $\roleQ \in \rolesCi[n]$. By~\autoref{def:assoc}, $\stEnvApp{\qEnvi[n]}{\roleQ, \roleP} =
      \stQEmpty$.  Now applying \inferrule{\iruleTCtxCrashDetect} on $\stEnvApp{\stEnvi[n]}{\roleP} =
      \stExtSum{\roleQ}{i \in I}{\stChoice{\stLab[i]}{\tyGround[i]} \stSeq \stT[i]},
          \stEnvApp{\stEnvi[n]}{\roleQ} = \stStop,
           \stEnvApp{\qEnvi[n]}{\roleQ, \roleP} = \stQEmpty$ and
       $\gtLab[j] = \stLab[j] = \stCrashLab$ with $j \in I' \subseteq I$, we have that
       $\stEnvi[n]; \qEnvi[n] \stEnvMoveAnnot{\ltsCrDe{\mpS}{\roleP}{\roleQ}}$.  Then together with the fairness of
       $(\stEnvi[n]; \qEnvi[n])_{n \in N}$ , we can conclude that there exists $k$ such that $n \leq k \in N$ and
       $\stEnvi[k]; \qEnvi[k] \stEnvMoveAnnot{\ltsCrDe{\mpS}{\roleP}{\roleQ}} \stEnvi[k+1]; \qEnvi[k+1]$.

   \item $\roleQMaybeCrashed \neq \roleQCrashed$: since $\gtLab[j] = \gtCrashLab$,  %
  together with $\gtWithCrashedRoles{\rolesC}{\gtG} \!\gtMoveStar[\rolesR]\!
\gtWithCrashedRoles{\rolesCi}{\gtGi}$,  we know that $\roleQMaybeCrashed \neq \roleQ$ holds, a desired contradiction.
\qedhere
   \end{itemize}
            \end{itemize}
          \end{itemize}
        \end{enumerate}
\end{proof}

\section{Proofs for Section \ref{sec:typing_system}}
\label{sec:proof:typesystem}

\begin{lem}[Typing Inversion]
\label{lem:typeinversion}
 Let $\Theta \vdash \mpP : \stTi$. Then, there exists $\stT \stSub \stTi$ such that:
 \begin{enumerate}[leftmargin=*]
\item
\label{proc_end}
$\mpP =  \mpNil$ implies $\stT = \stEnd$;
\item
\label{proc_stop}
$\mpP = \mpCrash$ implies $\stT = \stStop$;
\item
\label{proc_out}
$\mpP = \procout\roleQ{\mpLab}{\mpE}{\mpP[1]}$ implies
\begin{enumerate}[label={(a\arabic*)}, leftmargin=*, ref={(a\arabic*)}]
\item $\stT = \stOut{\roleQ}{\stLab}{\tyGround}\stSeq{\stT[1]}$, and
\item $\Theta \vdash \mpP[1] : \stT[1]$, and
\item   $\Theta \vdash \mpE:\tyGround$;
\end{enumerate}
\item
\label{proc_ext}
$\mpP = \sum_{i\in I}\procin{\roleQ}{\mpLab_i(\mpx_i)}{\mpP_i}$ implies
\begin{enumerate}[label={(a\arabic*)}, leftmargin=*, ref={(a\arabic*)}]
\item $\stT = \stExtSum{\roleQ}{i\in I}{\stChoice{\stLab[i]}{\tyGround[i]}\stSeq{\stT_i}}$, and
\item $\forall i\in I\;\;\; \Theta, x_i:\tyGround_i \vdash \mpP_i:\stT_i$;
\end{enumerate}
\item
\label{proc_cond}
$\mpP = \mpIf\mpE{\mpP_1}{\mpP_2}$ implies
\begin{enumerate}[label={(a\arabic*)}, leftmargin=*, ref={(a\arabic*)}]
\item $\Theta \vdash \mpE:\tyBool$, and
\item $\Theta  \vdash \mpP_1:\stT$, and
\item $\Theta  \vdash \mpP_2:\stT$;
\end{enumerate}
\item
\label{proc_rec}
$\mpP =  \mu X.\mpP[1]$ implies
$\Theta, X:\stT  \vdash \mpP[1]:\stT$;
\end{enumerate}
\noindent
Let $ \vdash \mpH :\qEnvPartial$. Then:
\begin{enumerate}[resume, align=left]
\item
\label{proc_empty_Q}
$\mpH = \mpQEmpty$ implies $\qEnvPartial =  \stQEmpty$;
\item
\label{proc_unavail_Q}
$\mpH = \mpQUnavail$ implies $\qEnvPartial = \stQUnavail$;
\item
\label{proc_message_Q}
$\mpH = (\roleQ , \mpLab(\mpV))$ implies  $\vdash \mpV:\tyGround$ and  %
 $\stEnvApp{\qEnvPartial}{\roleQ} = \stQMsg\mpLab\tyGround$;
 \item
 \label{proc_con_Q}
 $\mpH =  \mpH_1 \cdot \mpH_2$ implies $\vdash \mpH_1:\qEnvPartial[1]$, and
 $\vdash \mpH_2:\qEnvPartial[2]$, and $\qEnvPartial = \qEnvPartial[1] \cdot \qEnvPartial[2]$;
\end{enumerate}
\noindent
Let $ \gtWithCrashedRoles{\rolesC}{\gtG}
  \vdash \prod_{i\in I} (\mpPart{\roleP[i]}{\mpP[i]} \mpPar
  \mpPart{\roleP[i]}{\mpH[i]})$. Then:
  \begin{enumerate}[label={(a\arabic*)}, leftmargin=*, ref={(a\arabic*)}]
  \item $\exists \,\stEnv; \qEnv$ such that  $\stEnvAssoc{\gtWithCrashedRoles{\rolesC}{\gtG}}{\stEnv; \qEnv}{\rolesR}$, and
  \item $\dom{\stEnv} \subseteq
  \setcomp{\roleP[i]}{i \in I}$, and
  \item $\forall i \in I: \,\,\vdash \mpP_i:\stEnvApp{\stEnv}{\roleP[i]}$, and
  \item $\forall i \in I:   \,\,\vdash \mpH[i]: \stEnvApp{\qEnv}{-, \roleP[i]}$.
  \end{enumerate}
\end{lem}
\begin{proof}
By inverting the rules in~\Cref{fig:processes:typesystem}.
\qedhere
\end{proof}

\begin{lem}[Typing Precongruence]
\label{lem:typingcongruence}
\begin{enumerate}[leftmargin=*]
\item
\label{lem:proc_cong}
If $\Theta \vdash \mpP:\stT$ and $\mpP \Rrightarrow \mpQ$, then $\Theta \vdash \mpQ:\stT$.
\item
\label{lem:Q_cong}
If $\vdash \mpH[1]: \qEnvPartial[1]$ and $\mpH[1] \Rrightarrow \mpH[2]$, then there exists $\qEnvPartial[2]$ such that 
$\qEnvPartial[1] \,\stFmt{\Rrightarrow}\, \qEnvPartial[2]$ and $\vdash \mpH[2]:\qEnvPartial[2]$.
 \item
 \label{lem:cong}
 If $ \gtWithCrashedRoles{\rolesC}{\gtG}
  \vdash \mpM$ and $\mpM \Rrightarrow \mpMi$, then
 $ \gtWithCrashedRoles{\rolesC}{\gtG}
  \vdash \mpMi$.
\end{enumerate}
\end{lem}
\begin{proof}
\begin{enumerate}
\item By induction on the precongruence rules for $\mpP \Rrightarrow \mpQ$.
\item By induction on the precongruence rules for $\mpH[1] \Rrightarrow \mpH[2]$.
\item By induction on the precongruence rules for $\mpM \Rrightarrow \mpMi$.
\qedhere
\end{enumerate}
\end{proof}

\begin{lem}[Substitution]
\label{lem:substitution}
If $\Theta, x:\tyGround   \vdash \mpP :\stT$ and $\Theta \vdash \val:\tyGround$,
then
$\Theta \vdash \mpP\{\val/x\}:\stT$.
\end{lem}
\begin{proof}
By induction on the structure of $\mpP$.
\qedhere
\end{proof}

\begin{lem}%
\label{lem:val-eval}
If $\emptyset \vdash \mpE:\tyGround$ and $\eval{\mpE}\val$, then $\emptyset \vdash \mpV : \tyGround$.
\end{lem}
\begin{proof}
By induction on the derivation of $\emptyset \vdash \mpE : \tyGround$.
\qedhere
\end{proof}

\begin{lem}
\label{lem:recursion-unfolding-typing-session}
$\gtWithCrashedRoles{\rolesC}{\gtRec{\gtRecVar}{\gtG}} \vdash \mpM$ 
if and only if $\gtWithCrashedRoles{\rolesC}{\gtG\subst{\gtRecVar}{\gtRec{\gtRecVar}{\gtG}}} \vdash \mpM$. 
\end{lem}
\begin{proof}
The thesis follows directly by applying~\inferrule{t-sess} and~\autoref{prop:gt-lts-unfold-once}. 
\end{proof}

\lemSubjectReduction*
\begin{proof}
Let us recap the assumptions:
 \begin{align}
\gtWithCrashedRoles{\rolesC}{\gtG} \vdash \mpM
\label{eq:subred_ass_global_typing}
\\
\mpM \mpMove[\rolesR] \mpMi
\label{eq:subred_ass_session_red}
\end{align}
The proof proceeds by induction on the derivation of $\mpM \mpMove[\rolesR] \mpMi$.
Most cases hold by typing inversion~(\autoref{lem:typeinversion}), and by applying the induction
hypothesis.
\begin{itemize}[leftmargin=*]
\item Case \inferrule{r-send}: we have
\begin{gather}
\mpM = \mpPart\roleP{\procout{\roleQ}{\mpLab}{\mpE}{\mpP}}
\mpPar
\mpPart\roleP{\mpH[\roleP]}
\mpPar
\mpPart\roleQ{\mpQ}
\mpPar
\mpPart\roleQ{\mpH[\roleQ]}
\mpPar
\mpM[1]
\label{eq:r_send_M}
\\
\mpMi = \mpPart\roleP{\mpP}
\mpPar
\mpPart\roleP{\mpH[\roleP]}
\mpPar
\mpPart\roleQ{\mpQ}
\mpPar
\mpPart{\roleQ}{\mpH[\roleQ]}\cdot(\roleP,\mpLab(\val))
\mpPar
\mpM[1]
\label{eq:r_send_Mi}
\\
\eval{\mpE} \val
\label{eq:r_send_expression}
\\
\mpH[\roleQ] \neq \mpQUnavail
\label{eq:r_send_Q}
\\
\mpM[1] =  \prod_{i\in I} (\mpPart{\roleP[i]}{\mpP[i]} \mpPar
  \mpPart{\roleP[i]}{\mpH[i]})
 \label{eq:r_send_M1}
\end{gather}
By~\eqref{eq:subred_ass_global_typing} and~\autoref{lem:typeinversion}, we have that there
exists $\stEnv; \qEnv$ such that
\begin{gather}
\stEnvAssoc{\gtWithCrashedRoles{\rolesC}{\gtG}}{\stEnv; \qEnv}{\rolesR}
\label{eq:r_send_type_association}
\\
\vdash \procout\roleQ{\mpLab}{\mpE}{\mpP}:\stEnvApp{\stEnv}{\roleP}
\label{eq:r_send_type_out}
\\
 \vdash \mpH[\roleP]:  \stEnvApp{\qEnv}{-, \roleP}%
 \label{eq:r_send_type_out_Q}
 \\
 \vdash \mpQ: \stEnvApp{\stEnv}{\roleQ}
 \label{eq:r_send_type_out_another}
 \\
  \vdash \mpH[\roleQ]: \stEnvApp{\qEnv}{-, \roleQ}%
 \label{eq:r_send_type_out_Q_another}
 \\
 \forall i \in I: \,\,\vdash \mpP_i:\stEnvApp{\stEnv}{\roleP[i]}
 \label{eq:r_send_M1_type}
 \\
 \forall i \in I:  \,\,\vdash \mpH[i]: \stEnvApp{\qEnv}{-, \roleP[i]}%
  \label{eq:r_send_M1_type_Q}
\end{gather}
By~\eqref{eq:r_send_type_out} and 3 of~\autoref{lem:typeinversion}, we have that
\begin{gather}
\stEnvApp{\stEnv}{\roleP} = \stOut{\roleQ}{\mpLab}{\tyGround}\stSeq{\stT}
\label{eq:r_send_type_configuration}
\\
\stT[1] \stSub \stT
\label{eq:r_send_subtyping}
\\
 \vdash \mpP : \stT[1]
 \label{eq:r_send_type_configuration_2}
 \\
 \vdash \mpE:\tyGround
  \label{eq:r_send_type_configuration_3}
\end{gather}
We now let
\begin{gather}
\stEnvi = \stEnvUpd{\stEnv}{\roleP}{\stT}
\label{eq:r_send_new_context}
\\
\qEnvi = \stEnvUpd{\qEnv}{\roleP, \roleQ}{
        \stQCons{
          \stEnvApp{\qEnv}{\roleP, \roleQ}
        }{
          \stQMsg{\stLab}{\tyGround}
        }
        }
\label{eq:r_send_new_Q}
\end{gather}
Then by \inferrule{\iruleTCtxOut} in~\Cref{fig:gtype:tc-red-rules}, we have
\begin{align}
\stEnv; \qEnv \!\stEnvMoveMaybeCrash[\rolesR]\! \stEnvi; \qEnvi
\label{eq:r_send_reduction}
\end{align}
Hence, using~\autoref{thm:gtype:proj-comp}, we have that there exists
$\gtWithCrashedRoles{\rolesCi}{\gtGi}$ such that
\begin{gather}
\stEnvAssoc{\gtWithCrashedRoles{\rolesCi}{\gtGi}}{\stEnvi; \qEnvi}{\rolesR}
\label{eq:r_send_type_association_new}
\\
\gtWithCrashedRoles{\rolesC}{\gtG} \gtMove{\rolesR}
    \gtWithCrashedRoles{\rolesCi}{\gtGi}
\label{eq:r_send_type_global_reduction_new}
\end{gather}
Combine \eqref{eq:r_send_type_association_new}, \eqref{eq:r_send_type_global_reduction_new} with
 \begin{gather}
 \vdash \mpP : \stEnvApp{\stEnvi}{\roleP} \quad \quad (\text{by}~\eqref{eq:r_send_new_context}, \eqref{eq:r_send_subtyping}, \eqref{eq:r_send_type_configuration_2}\text{ and }\inferrule{{t-sub}})
 \\
  \vdash \mpH[\roleP] : \stEnvApp{\qEnvi}{-, \roleP} \quad  (\text{by}~\eqref{eq:r_send_new_Q}\text{ and }\eqref{eq:r_send_type_out_Q})
  \\
  \vdash \mpQ : \stEnvApp{\stEnvi}{\roleQ} \quad \quad (\text{by}~\eqref{eq:r_send_type_out_another}\text{ and }\eqref{eq:r_send_new_context})
  \\
 \vdash \mpH[\roleQ] \cdot(\roleP,\mpLab(\val)): \stEnvApp{\qEnvi}{-, \roleQ} \quad (\text{by}~\eqref{eq:r_send_new_Q}, \eqref{eq:r_send_type_out_Q_another}, \eqref{eq:r_send_expression}, \eqref{eq:r_send_type_configuration_3}, \text{\autoref{lem:val-eval}}, \inferrule{t-msg},~\text{and}~\inferrule{t-$\cdot$})
  \\
   \forall i \in I: \,\,\vdash \mpP_i:\stEnvApp{\stEnvi}{\roleP[i]} \quad (\text{by}~\eqref{eq:r_send_M1_type}\text{ and }\eqref{eq:r_send_new_context})
 \\
 \forall i \in I:  \,\,\vdash \mpH[i]: \stEnvApp{\qEnvi}{-, \roleP[i]} \quad (\text{by}~\eqref{eq:r_send_M1_type_Q}\text{ and }\eqref{eq:r_send_new_Q})
 \end{gather}
 We conclude that $\gtWithCrashedRoles{\rolesCi}{\gtGi} \vdash \mpMi$.

\item Case \inferrule{r-rcv}: similar to case \inferrule{r-send} above, except that we proceed by \inferrule{r-rcv}, inversion of \inferrule{{t-ext}} (4 of~\autoref{lem:typeinversion}), and~\autoref{lem:substitution}.

\item Case \inferrule{r-send-\mpCrash}: we have
\begin{gather}
\mpM = \mpPart\roleP{\procout{\roleQ}{\mpLab}{\mpE}{\mpP}}
\mpPar
\mpPart\roleP{\mpH[\roleP]}
\mpPar
\mpPart\roleQ{\mpCrash}
\mpPar
\mpPart\roleQ{\mpQUnavail}
\mpPar
\mpM[1]
\label{eq:r_send_crash_M}
\\
\mpMi = \mpPart\roleP{\mpP}
\mpPar
\mpPart\roleP{\mpH[\roleP]}
\mpPar
\mpPart\roleQ{\mpCrash}
\mpPar
\mpPart{\roleQ}{\mpQUnavail}
\mpPar
\mpM[1]
\label{eq:r_send_crash_Mi}
\\
\\
\mpM[1] =  \prod_{i\in I} (\mpPart{\roleP[i]}{\mpP[i]} \mpPar
  \mpPart{\roleP[i]}{\mpH[i]})
 \label{eq:r_send_crash_M1}
\end{gather}
By~\eqref{eq:subred_ass_global_typing} and~\autoref{lem:typeinversion}, we have that there
exists $\stEnv; \qEnv$ such that
\begin{gather}
\stEnvAssoc{\gtWithCrashedRoles{\rolesC}{\gtG}}{\stEnv; \qEnv}{\rolesR}
\label{eq:r_send_crash_type_association}
\\
\vdash \procout\roleQ{\mpLab}{\mpE}{\mpP}:\stEnvApp{\stEnv}{\roleP}
\label{eq:r_send_crash_type_out}
\\
 \vdash \mpH[\roleP]: \stEnvApp{\qEnv}{-, \roleP}
 \label{eq:r_send_crash_type_out_Q}
 \\
 \vdash \mpCrash: \stEnvApp{\stEnv}{\roleQ}
 \label{eq:r_send_crash_type_out_another}
 \\
  \vdash \mpQUnavail: \stEnvApp{\qEnv}{-, \roleQ}
 \label{eq:r_send_crash_type_out_Q_another}
 \\
 \forall i \in I: \,\,\vdash \mpP_i:\stEnvApp{\stEnv}{\roleP[i]}
 \label{eq:r_send_crash_M1_type}
 \\
 \forall i \in I:  \,\,\vdash \mpH[i]: \stEnvApp{\qEnv}{-, \roleP[i]}
  \label{eq:r_send_crash_M1_type_Q}
\end{gather}
It follows directly that
\begin{gather}
\stEnvApp{\stEnv}{\roleQ} = \stStop
\label{eq:r_send_crash_stop}
\\
 \stEnvApp{\qEnv}{-, \roleQ} = \stQUnavail
 \label{eq:r_send_crash_unavail}
\end{gather}
By~\eqref{eq:r_send_crash_type_out} and 3 of~\autoref{lem:typeinversion}, we have that
\begin{gather}
\stEnvApp{\stEnv}{\roleP} = \stOut{\roleQ}{\mpLab}{\tyGround}\stSeq{\stT}
\label{eq:r_send_crash_type_configuration}
\\
\stT[1] \stSub \stT
\label{eq:r_send_crash_subtyping}
\\
 \vdash \mpP : \stT[1]
 \label{eq:r_send_crash_type_configuration_2}
 \\
 \vdash \mpE:\tyGround
  \label{eq:r_send_crash_type_configuration_3}
\end{gather}
We now let
\begin{gather}
\stEnvi = \stEnvUpd{\stEnv}{\roleP}{\stT}
\label{eq:r_send_crash_new_context}
\\
\qEnvi = \stEnvUpd{\qEnv}{\roleP, \roleQ}{
        \stQCons{
          \stEnvApp{\qEnv}{\roleP, \roleQ}
        }{
          \stQMsg{\stLab}{\tyGround}
        }
        }
\label{eq:r_send_crash_new_Q}
\end{gather}
Then by \inferrule{\iruleTCtxOut} in~\Cref{fig:gtype:tc-red-rules}, we have
\begin{align}
\stEnv; \qEnv \!\stEnvMoveMaybeCrash[\rolesR]\! \stEnvi; \qEnvi
\label{eq:r_send_crash_reduction}
\end{align}
Hence, using~\autoref{thm:gtype:proj-comp}, we have that there exists
$\gtWithCrashedRoles{\rolesCi}{\gtGi}$ such that
\begin{gather}
\stEnvAssoc{\gtWithCrashedRoles{\rolesCi}{\gtGi}}{\stEnvi; \qEnvi}{\rolesR}
\label{eq:r_send_crash_type_association_new}
\\
\gtWithCrashedRoles{\rolesC}{\gtG} \gtMove{\rolesR}
    \gtWithCrashedRoles{\rolesCi}{\gtGi}
\label{eq:r_send_crash_type_global_reduction_new}
\end{gather}
 Combine \eqref{eq:r_send_crash_type_association_new}, \eqref{eq:r_send_crash_type_global_reduction_new} with
 \begin{gather}
 \vdash \mpP : \stEnvApp{\stEnvi}{\roleP} \quad \quad (\text{by}~\eqref{eq:r_send_crash_new_context}, \eqref{eq:r_send_crash_subtyping}, \eqref{eq:r_send_crash_type_configuration_2}\text{ and }\inferrule{{t-sub}})
 \\
  \vdash \mpH[\roleP] : \stEnvApp{\qEnvi}{-, \roleP} \quad  (\text{by}~\eqref{eq:r_send_crash_new_Q}\text{ and }\eqref{eq:r_send_crash_type_out_Q})
  \\
  \vdash \mpCrash : \stEnvApp{\stEnvi}{\roleQ} \quad \quad (\text{by}~\eqref{eq:r_send_crash_type_out_another},\eqref{eq:r_send_crash_stop}\text{ and }\eqref{eq:r_send_crash_new_context})
  \\
  \vdash \mpQUnavail: \stEnvApp{\qEnvi}{-, \roleQ} \quad  (\text{by}~\eqref{eq:r_send_crash_new_Q}, \eqref{eq:r_send_crash_type_out_Q_another}\text{ and }\eqref{eq:r_send_crash_unavail})
  \\
   \forall i \in I: \,\,\vdash \mpP_i:\stEnvApp{\stEnvi}{\roleP[i]} \quad (\text{by}~\eqref{eq:r_send_crash_M1_type}\text{ and }\eqref{eq:r_send_crash_new_context})
 \\
 \forall i \in I:  \,\,\vdash \mpH[i]: \stEnvApp{\qEnvi}{-, \roleP[i]} \quad (\text{by}~\eqref{eq:r_send_crash_M1_type_Q}\text{ and }\eqref{eq:r_send_crash_new_Q})
 \end{gather}
 We conclude that $\gtWithCrashedRoles{\rolesCi}{\gtGi} \vdash \mpMi$.

\item Case \inferrule{r-rcv-$\odot$}:  we have
\begin{gather}
\mpM = \mpPart\roleP{\sum_{i\in I} \procin\roleQ{\mpLab_i(\mpx_i)}\mpP_i}
\mpPar
\mpPart\roleP{\mpH[\roleP]}
\mpPar
\mpPart\roleQ\mpCrash
\mpPar
\mpPart\roleQ\mpQUnavail
\mpPar
\mpM[1]
\label{eq:r_rcv_det_M}
\\
\mpMi = \mpPart\roleP \mpP_k
\mpPar
\mpPart\roleP{\mpH_{\roleP}}
\mpPar
\mpPart\roleQ\mpCrash
\mpPar
\mpPart\roleQ\mpQUnavail
\mpPar
\mpM[1]
\label{eq:r_rcv_det_Mi}
\\
\mpM[1] =  \prod_{i\in I} (\mpPart{\roleP[i]}{\mpP[i]} \mpPar
  \mpPart{\roleP[i]}{\mpH[i]})
 \label{eq:r_rcv_det_M1}
 \\
 k \in I
 \label{eq:r_rcv_det_plus_con_1}
 \\
 \mpLab[k] = \mpCrashLab
 \label{eq:r_rcv_det_plus_con_2}
 \\
 \nexists \mpLab, \mpV: (\roleQ,
\mpLab(\mpV)) \in \mpH[\roleP]
 \label{eq:r_rcv_det_plus_con_3}
\end{gather}
By~\eqref{eq:subred_ass_global_typing} and~\autoref{lem:typeinversion}, we have that there
exists $\stEnv; \qEnv$ such that
\begin{gather}
\stEnvAssoc{\gtWithCrashedRoles{\rolesC}{\gtG}}{\stEnv; \qEnv}{\rolesR}
\label{eq:r_rcv_det_type_association}
\\
\vdash \sum_{i\in I}\procin{\roleQ}{\mpLab_i(\mpx_i)}{\mpP_i}:\stEnvApp{\stEnv}{\roleP}
\label{eq:r_rcv_det_type_out}
\\
 \vdash \mpH[\roleP]: \stEnvApp{\qEnv}{-, \roleP}
 \label{eq:r_rcv_det_type_out_Q}
 \\
 \vdash \mpCrash: \stEnvApp{\stEnv}{\roleQ}
 \label{eq:r_rcv_det_type_out_another}
 \\
  \vdash \mpQUnavail: \stEnvApp{\qEnv}{-, \roleQ}
 \label{eq:r_rcv_det_type_out_Q_another}
 \\
 \forall i \in I: \,\,\vdash \mpP_i:\stEnvApp{\stEnv}{\roleP[i]}
 \label{eq:r_rcv_det_M1_type}
 \\
 \forall i \in I:  \,\,\vdash \mpH[i]: \stEnvApp{\qEnv}{-, \roleP[i]}
  \label{eq:r_rcv_det_M1_type_Q}
\end{gather}
It follows directly that
\begin{gather}
\stEnvApp{\stEnv}{\roleQ} = \stStop
\label{eq:r_rcv_det_stop}
\\
 \stEnvApp{\qEnv}{-, \roleQ} = \stQUnavail
 \label{eq:r_rcv_det_unavail}
\end{gather}
By~\eqref{eq:r_rcv_det_type_out}, 4 of~\autoref{lem:typeinversion}, and \inferrule{\iruleStSubIn}, we have that
\begin{gather}
\stEnvApp{\stEnv}{\roleP} = \stExtSum{\roleQ}{i \in J}{\stChoice{\stLab[i]}{\tyGround[i]} \stSeq \stTi[i]}
\label{eq:r_rcv_det_type_configuration}
\\
J \subseteq I
\label{eq:r_rcv_det_subset}
\\
\forall i \in J:  \stT[i] \stSub \stTi[i]
\label{eq:r_rcv_det_subtyping}
\\
 \setcomp{\stLab[l]}{l \in I} \neq \setenum{\stCrashLab}
 \label{eq:r_rcv_det_not_crash}
\\
 \nexists j \in I \setminus J: \stLab[j] = \stCrashLab
 \label{eq:r_rcv_det_not_crash_2}
 \\
 \forall i \in I: x_i : \tyGround[i] \vdash \mpP[i] : \stT[i]
 \label{eq:r_rcv_det_type_configuration_2}
\end{gather}
From \eqref{eq:r_rcv_det_not_crash_2} and \eqref{eq:r_rcv_det_plus_con_2}, we get
\begin{align}
   k \in J
   \label{eq:r_rcv_det_k_in_J_1}
\end{align}
By \eqref{eq:r_rcv_det_type_out_Q} and \eqref{eq:r_rcv_det_plus_con_3}, we also know
\begin{align}
 \stEnvApp{\qEnv}{\roleQ, \roleP} = \stQEmpty
 \label{eq:r_rcv_det_empty_Q}
\end{align}
We now let
\begin{gather}
\stEnvi = \stEnvUpd{\stEnv}{\roleP}{\stTi[k]}
\label{eq:r_rcv_det_new_context}
\\
\qEnvi = \qEnv
\label{eq:r_rcv_det_new_Q}
\end{gather}
Then by \inferrule{\iruleTCtxCrashDetect} in~\Cref{fig:gtype:tc-red-rules}, we have
\begin{align}
\stEnv; \qEnv \!\stEnvMoveMaybeCrash[\rolesR]\! \stEnvi; \qEnvi
\label{eq:r_rcv_det_reduction}
\end{align}
Hence, using~\autoref{thm:gtype:proj-comp}, we have that there exists
$\gtWithCrashedRoles{\rolesCi}{\gtGi}$ such that
\begin{gather}
\stEnvAssoc{\gtWithCrashedRoles{\rolesCi}{\gtGi}}{\stEnvi; \qEnvi}{\rolesR}
\label{eq:r_rcv_det_type_association_new}
\\
\gtWithCrashedRoles{\rolesC}{\gtG} \gtMove{\rolesR}
    \gtWithCrashedRoles{\rolesCi}{\gtGi}
\label{eq:r_rcv_det_type_global_reduction_new}
\end{gather}
 Combine \eqref{eq:r_rcv_det_type_association_new}, \eqref{eq:r_rcv_det_type_global_reduction_new} with
 \begin{gather}
 \vdash \mpP[k] : \stEnvApp{\stEnvi}{\roleP} \quad \quad (\text{by}~\eqref{eq:r_rcv_det_new_context}, \eqref{eq:r_rcv_det_subtyping}, \eqref{eq:r_rcv_det_type_configuration_2}\text{ and }\inferrule{{t-sub}})
 \\
  \vdash \mpH[\roleP] : \stEnvApp{\qEnvi}{-, \roleP} \quad  (\text{by}~\eqref{eq:r_rcv_det_new_Q}\text{ and }\eqref{eq:r_rcv_det_type_out_Q})
  \\
  \vdash \mpCrash : \stEnvApp{\stEnvi}{\roleQ} \quad \quad (\text{by}~\eqref{eq:r_rcv_det_type_out_another},\eqref{eq:r_rcv_det_stop}\text{ and }\eqref{eq:r_rcv_det_new_context})
  \\
  \vdash \mpQUnavail: \stEnvApp{\qEnvi}{-, \roleQ} \quad  (\text{by}~\eqref{eq:r_rcv_det_new_Q}, \eqref{eq:r_rcv_det_type_out_Q_another}\text{ and }\eqref{eq:r_rcv_det_unavail})
  \\
   \forall i \in I: \,\,\vdash \mpP_i:\stEnvApp{\stEnvi}{\roleP[i]} \quad (\text{by}~\eqref{eq:r_rcv_det_M1_type}\text{ and }\eqref{eq:r_rcv_det_new_context})
 \\
 \forall i \in I:  \,\,\vdash \mpH[i]: \stEnvApp{\qEnvi}{-, \roleP[i]} \quad (\text{by}~\eqref{eq:r_rcv_det_M1_type_Q}\text{ and }\eqref{eq:r_rcv_det_new_Q})
 \end{gather}
 We conclude that $\gtWithCrashedRoles{\rolesCi}{\gtGi} \vdash \mpMi$.

\item Case \inferrule{r-$\lightning$}: we have
\begin{gather}
\mpM = \mpPart\roleP{\mpP}
\mpPar
\mpPart\roleP{\mpH[\roleP]}
\mpPar
\mpM[1]
\label{eq:r_lighting_M}
\\
\mpMi = \mpPart\roleP{\mpCrash}
\mpPar
\mpPart\roleP{\mpQUnavail}
\mpPar
\mpM
\label{eq:r_lighting_Mi}
\\
\mpP \neq \mpNil
\label{eq:r_lighting_M_Nil}
\\
\roleP \notin \rolesR
\label{eq:r_lighting_M_reliable}
\\
\mpM[1] =  \prod_{i\in I} (\mpPart{\roleP[i]}{\mpP[i]} \mpPar
  \mpPart{\roleP[i]}{\mpH[i]})
 \label{eq:r_lighting_M1}
\end{gather}
By~\eqref{eq:subred_ass_global_typing} and~\autoref{lem:typeinversion}, we have that there
exists $\stEnv; \qEnv$ such that
\begin{gather}
\stEnvAssoc{\gtWithCrashedRoles{\rolesC}{\gtG}}{\stEnv; \qEnv}{\rolesR}
\label{eq:r_lighting_type_association}
\\
\vdash \mpP:\stEnvApp{\stEnv}{\roleP}
\label{eq:r_lighting_type_out}
\\
 \vdash \mpH[\roleP]: \stEnvApp{\qEnv}{-, \roleP}
 \label{eq:r_lighting_type_out_Q}
 \\
 \forall i \in I: \,\,\vdash \mpP_i:\stEnvApp{\stEnv}{\roleP[i]}
 \label{eq:r_lighting_M1_type}
 \\
 \forall i \in I:  \,\,\vdash \mpH[i]: \stEnvApp{\qEnv}{-, \roleP[i]}
  \label{eq:r_lighting_M1_type_Q}
\end{gather}
By \eqref{eq:r_lighting_M_Nil}, \eqref{eq:r_lighting_M_reliable}, and \eqref{eq:r_lighting_type_out},
we have that
\begin{gather}
\stEnvApp{\stEnv}{\roleP} \neq \stStop
\label{eq:r_lighting_not_stop}
\\
 \stEnvApp{\stEnv}{\roleP} \neq \stEnd
 \label{eq:r_lighting_not_end}
\end{gather}
We now let
\begin{gather}
\stEnvi = \stEnvUpd{\stEnv}{\roleP}{\stStop}
\label{eq:r_lighting_new_context}
\\
\qEnvi = \stEnvUpd{\qEnv}{\cdot, \roleP}{
       \stQUnavail
        }
\label{eq:r_lighting_new_Q}
\end{gather}
Then by \inferrule{\iruleTCtxCrash} in~\Cref{fig:gtype:tc-red-rules}, we have
\begin{align}
\stEnv; \qEnv \!\stEnvMoveMaybeCrash[\rolesR]\! \stEnvi; \qEnvi
\label{eq:r_lighting_reduction}
\end{align}
Hence, using~\autoref{thm:gtype:proj-comp}, we have that there exists
$\gtWithCrashedRoles{\rolesCi}{\gtGi}$ such that
\begin{gather}
\stEnvAssoc{\gtWithCrashedRoles{\rolesCi}{\gtGi}}{\stEnvi; \qEnvi}{\rolesR}
\label{eq:r_lighting_type_association_new}
\\
\gtWithCrashedRoles{\rolesC}{\gtG} \gtMove{\rolesR}
    \gtWithCrashedRoles{\rolesCi}{\gtGi}
\label{eq:r_lighting_type_global_reduction_new}
\end{gather}
 Combine \eqref{eq:r_lighting_type_association_new}, \eqref{eq:r_lighting_type_global_reduction_new} with
 \begin{gather}
 \vdash \mpCrash : \stEnvApp{\stEnvi}{\roleP} \quad \quad (\text{by}~\eqref{eq:r_lighting_new_context})
  \\
  \vdash \mpQUnavail : \stEnvApp{\qEnvi}{-, \roleP} \quad  (\text{by}~\eqref{eq:r_lighting_new_Q})
    \\
    \forall i \in I: \,\,\vdash \mpP_i:\stEnvApp{\stEnvi}{\roleP[i]} \quad (\text{by}~\eqref{eq:r_lighting_M1_type}\text{ and }\eqref{eq:r_lighting_new_context})
 \\
 \forall i \in I:  \,\,\vdash \mpH[i]: \stEnvApp{\qEnvi}{-, \roleP[i]} \quad (\text{by}~\eqref{eq:r_lighting_M1_type_Q}\text{ and }\eqref{eq:r_lighting_new_Q})
 \end{gather}
  We conclude that $\gtWithCrashedRoles{\rolesCi}{\gtGi} \vdash \mpMi$.

\item Case \inferrule{r-cond-T}: we have
\begin{gather}
\mpM = \mpPart\roleP{\mpIf{\mpE}{\mpP}{\mpQ}}
\mpPar
\mpPart\roleP\mpH
\mpPar
\mpM[1]
\label{eq:r_cond_T_M}
\\
\mpMi = \mpPart\roleP\mpP
\mpPar
\mpPart\roleP\mpH
\mpPar
\mpM[1]
\label{eq:r_cond_T_Mi}
\\
\eval{\mpE}{\mpTrue}
\label{eq:r_cond_T_expression}
\\
\mpM[1] =  \prod_{i\in I} (\mpPart{\roleP[i]}{\mpP[i]} \mpPar
  \mpPart{\roleP[i]}{\mpH[i]})
 \label{eq:r_cond_T_M1}
\end{gather}
By~\eqref{eq:subred_ass_global_typing} and~\autoref{lem:typeinversion}, we have that there
exists $\stEnv; \qEnv$ such that
\begin{gather}
\stEnvAssoc{\gtWithCrashedRoles{\rolesC}{\gtG}}{\stEnv; \qEnv}{\rolesR}
\label{eq:r_cond_T_type_association}
\\
\vdash \mpIf{\mpE}{\mpP}{\mpQ}:\stEnvApp{\stEnv}{\roleP}
\label{eq:r_cond_T_type_out}
\\
 \vdash \mpH: \stEnvApp{\qEnv}{-, \roleP}
 \label{eq:r_cond_T_type_out_Q}
 \\
  \forall i \in I: \,\,\vdash \mpP_i:\stEnvApp{\stEnv}{\roleP[i]}
 \label{eq:r_cond_T_M1_type}
 \\
 \forall i \in I:  \,\,\vdash \mpH[i]: \stEnvApp{\qEnv}{-, \roleP[i]}
  \label{eq:r_cond_T_M1_type_Q}
\end{gather}
By~\eqref{eq:r_cond_T_expression},~\eqref{eq:r_cond_T_type_out}, 5 of~\autoref{lem:typeinversion} 
and \inferrule{t-sub}, we have that
\begin{gather}
\vdash \mpE : \tyBool
 \label{eq:r_cond_T_type_configuration_1}
 \\
 \vdash \mpP : \stEnvApp{\stEnv}{\roleP}
 \label{eq:r_cond_T_type_configuration_2}
 \\
 \vdash \mpQ: \stEnvApp{\stEnv}{\roleP}
  \label{eq:r_cond_T_type_configuration_3}
\end{gather}
Combine~\eqref{eq:r_cond_T_type_configuration_2} with \eqref{eq:r_cond_T_type_out_Q},
\eqref{eq:r_cond_T_M1_type},  and \eqref{eq:r_cond_T_M1_type_Q}, we can conclude
that $\gtWithCrashedRoles{\rolesC}{\gtG} \vdash \mpMi$, as desired.

\item Case \inferrule{r-cond-F}: similar to the case \inferrule{r-cond-T}.

\item Case \inferrule{r-struct}: assume that $\mpM \mpMove\mpMi$ is derived from
\begin{gather}
\mpM \Rrightarrow \mpM[1]
\label{eq:r_struct_eq_1}
\\
\mpM[1] \mpMove \mpMi[1]
\label{eq:r_struct_eq_2}
\\
\mpMi[1] \Rrightarrow \mpMi
\label{eq:r_struct_eq_3}
\end{gather}
From \eqref{eq:r_struct_eq_1}, \eqref{eq:subred_ass_global_typing}, by~\eqref{lem:cong} of~\autoref{lem:typingcongruence},
we have that $\gtWithCrashedRoles{\rolesC}{\gtG} \vdash \mpM[1]$. By induction hypothesis, either
$\gtWithCrashedRoles{\rolesC}{\gtG} \vdash \mpMi[1]$ or  there exists $\gtWithCrashedRoles{\rolesCi}{\gtGi}$ such that
  $\gtWithCrashedRoles{\rolesC}{\gtG} \gtMove{\rolesR}
  \gtWithCrashedRoles{\rolesCi}{\gtGi}$ and
  $\gtWithCrashedRoles{\rolesCi}{\gtGi} \vdash \mpMi[1]$. Now by \eqref{eq:r_struct_eq_3} and~\eqref{lem:cong} of \autoref{lem:typingcongruence}, we have that either $\gtWithCrashedRoles{\rolesC}{\gtG} \vdash \mpMi$ or
  $\gtWithCrashedRoles{\rolesCi}{\gtGi} \vdash \mpMi$, as desired.
  \qedhere
\end{itemize}
\end{proof}

\lemSessionFidelity*
\begin{proof}
Let us recap the assumptions:
 \begin{align}
\gtWithCrashedRoles{\rolesC}{\gtG} \vdash \mpM
\label{eq:sf_ass_global_typing}
\\
\gtWithCrashedRoles{\rolesC}{\gtG} \gtMove{\rolesR}
\label{eq:sf_ass_global_red}
\end{align}
The proof proceeds by induction on the derivation of $\gtWithCrashedRoles{\rolesC}{\gtG} \gtMove{\rolesR}$. 
\begin{itemize}[leftmargin=*]
\item Case \inferrule{\iruleGtMoveCrash}: by inversion of \inferrule{\iruleGtMoveCrash}, we have
\begin{gather}
 \roleP \notin \rolesR
 \label{eq:sf_global_crash_cond_1}
 \\
 \roleP \in \gtRoles{\gtG}
  \label{eq:sf_global_crash_cond_2}
  \\
   \gtWithCrashedRoles{\rolesC}{\gtG}
    \gtMove[\ltsCrashSmall{\mpS}{\roleP}]{\rolesR}
    \gtWithCrashedRoles{\rolesC \cup \setenum{\roleP}}{\gtCrashRole{\gtG}{\roleP}}
     \label{eq:sf_global_crash_cond_3}
  \end{gather}
 We can assume that $\mpM$ is of the form
 \begin{gather}
 \mpPart\roleP{\mpP}
\mpPar
\mpPart\roleP{\mpH[\roleP]}
\mpPar
\mpM[1]
\label{eq:sf_global_crash_M_form}
\\
\mpM[1] =  \prod_{i\in I} (\mpPart{\roleP[i]}{\mpP[i]} \mpPar
  \mpPart{\roleP[i]}{\mpH[i]})
\label{eq:sf_global_crash_M1_form}
\\
\mpP \neq \mpNil
\label{eq:sf_global_crash_M_form_not_nil}
\\
\roleP \notin \rolesR
\label{eq:sf_global_crash_M_form_reliable}
\end{gather}
Then by~\eqref{eq:sf_ass_global_typing} and~\autoref{lem:typeinversion}, there
exists $\stEnv; \qEnv$ such that
\begin{gather}
\stEnvAssoc{\gtWithCrashedRoles{\rolesC}{\gtG}}{\stEnv; \qEnv}{\rolesR}
\label{eq:sf_global_crash_type_association}
\\
\vdash \mpP:\stEnvApp{\stEnv}{\roleP}
\label{eq:sf_global__crash_type_out}
\\
 \vdash \mpH[\roleP]:  \stEnvApp{\qEnv}{-, \roleP}%
 \label{eq:sf_global_crash_type_out_Q}
 \\
 \forall i \in I: \,\,\vdash \mpP_i:\stEnvApp{\stEnv}{\roleP[i]}
 \label{eq:sf_global_crash_M1_type}
 \\
 \forall i \in I:  \,\,\vdash \mpH[i]: \stEnvApp{\qEnv}{-, \roleP[i]}%
  \label{eq:sf_global_crash_M1_type_Q}
\end{gather}
From~\eqref{eq:sf_global_crash_cond_3}, by~\autoref{thm:gtype:proj-sound}, there is $\stEnvi; \qEnvi$ such that
\begin{gather}
\stEnvAssoc{\gtWithCrashedRoles{\rolesC \cup \setenum{\roleP}}{\gtCrashRole{\gtG}{\roleP}}}{\stEnvi; \qEnvi}{\rolesR}
\label{eq:sf_global_crash_type_association_new}
\\
\stEnv; \qEnv
      \stEnvMoveAnnot{\ltsCrash{\mpS}{\roleP}}
      \stEnvi; \qEnvi
\label{eq:sf_global_crash_ctx_red}
\end{gather}
Using \eqref{eq:sf_global_crash_M_form_not_nil}, \eqref{eq:sf_global_crash_M_form_reliable}, \eqref{eq:sf_global__crash_type_out} and \inferrule{\iruleTCtxCrash} in~\Cref{fig:gtype:tc-red-rules}, we get
\begin{gather}
\stEnvi = \stEnvUpd{\stEnv}{\roleP}{\stStop}
\label{eq:sf_global_crash_ctx_new}
\\
\qEnvi = \stEnvUpd{\qEnv}{\cdot, \roleP}{\stQUnavail}
\label{eq:sf_global_crash_ctx_Q_new}
\end{gather}
It follows that
\begin{gather}
 \vdash \mpCrash : \stEnvApp{\stEnvi}{\roleP} \quad \quad (\text{by}~\eqref{eq:sf_global_crash_ctx_new})
  \\
  \vdash \mpQUnavail : \stEnvApp{\qEnvi}{-, \roleP} \quad  (\text{by}~\eqref{eq:sf_global_crash_ctx_Q_new})
    \\
    \forall i \in I: \,\,\vdash \mpP_i:\stEnvApp{\stEnvi}{\roleP[i]} \quad (\text{by}~\eqref{eq:sf_global_crash_M1_type}\text{ and }\eqref{eq:sf_global_crash_ctx_new})
 \\
 \forall i \in I:  \,\,\vdash \mpH[i]: \stEnvApp{\qEnvi}{-, \roleP[i]} \quad (\text{by}~\eqref{eq:sf_global_crash_M1_type_Q}\text{ and }\eqref{eq:sf_global_crash_ctx_Q_new})
\end{gather}
Therefore, by \inferrule{r-$\lightning$} and \inferrule{{t-sess}}, we can conclude that there exists $\mpMi = \mpPart\roleP{\mpCrash}
\mpPar
\mpPart\roleP{\mpQUnavail}
\mpPar
\mpM[1]$ and $\gtWithCrashedRoles{\rolesCi}{\gtGi}
   =
    \gtWithCrashedRoles{\rolesC \cup \setenum{\roleP}}{\gtCrashRole{\gtG}{\roleP}}$ such that
   $ \gtWithCrashedRoles{\rolesC}{\gtG}
    \gtMove[\ltsCrashSmall{\mpS}{\roleP}]{\rolesR}
   \gtWithCrashedRoles{\rolesCi}{\gtGi}$,
    $\mpM \;\redCrash{\roleP}{\rolesR}\;\mpMi$, and
   $ \gtWithCrashedRoles{\rolesCi}{\gtGi} \vdash \mpMi$.

\item Case \inferrule{\iruleGtMoveOut}: by inversion of \inferrule{\iruleGtMoveOut}, we have
\begin{gather}
j \in I
 \label{eq:sf_global_send_cond_1}
 \\
 \gtLab[j] \neq \gtCrashLab
  \label{eq:sf_global_send_cond_2}
  \\
  \gtWithCrashedRoles{\rolesC}{
      \gtCommSmall{\roleP}{\roleQ}{i \in I}{\gtLab[i]}{\tyGround[i]}{\gtGi[i]}
    }
    \gtMove[
      \stEnvOutAnnotSmall{\roleP}{\roleQ}{\stChoice{\gtLab[j]}{\tyGround[j]}}
    ]{
      \rolesR
    }
    \gtWithCrashedRoles{\rolesC}{
      \gtCommTransit{\roleP}{\roleQ}{i \in I}{\gtLab[i]}{\tyGround[i]}{\gtGi[i]}{j}
    }
     \label{eq:sf_global_send_cond_3}
  \end{gather}
 We can assume that $\mpM$ is of the form
 \begin{gather}
 \mpPart\roleP{\procout{\roleQ}{\mpLab}{\mpE}{\mpP}}
\mpPar
\mpPart\roleP{\mpH[\roleP]}
\mpPar
\mpPart\roleQ{\mpQ}
\mpPar
\mpPart\roleQ{\mpH[\roleQ]}
\mpPar
\mpM[1]
\label{eq:sf_global_send_M_form}
\\
\mpM[1] =  \prod_{i\in I} (\mpPart{\roleP[i]}{\mpP[i]} \mpPar
  \mpPart{\roleP[i]}{\mpH[i]})
\label{eq:sf_global_send_M1_form}
\\
\eval{\mpE} \val
\label{eq:sf_global_send_M_form_value}
\\
 \mpH[\roleQ] \neq \mpQUnavail
\label{eq:sf_global_csend_M_form_Q_availabe}
\end{gather}
Then by~\eqref{eq:sf_ass_global_typing} and~\autoref{lem:typeinversion}, there
exists $\stEnv; \qEnv$ such that
\begin{gather}
\stEnvAssoc{\gtWithCrashedRoles{\rolesC}{\gtG}}{\stEnv; \qEnv}{\rolesR}
\label{eq:sf_global_send_type_association}
\\
\vdash \procout\roleQ{\mpLab}{\mpE}{\mpP}:\stEnvApp{\stEnv}{\roleP}
\label{eq:sf_global_send_type_out}
\\
 \vdash \mpH[\roleP]:  \stEnvApp{\qEnv}{-, \roleP}%
 \label{eq:sf_global_send_type_out_Q}
 \\
 \vdash \mpQ: \stEnvApp{\stEnv}{\roleQ}
 \label{eq:sf_global_send_type_out_another}
 \\
  \vdash \mpH[\roleQ]: \stEnvApp{\qEnv}{-, \roleQ}%
 \label{eq:sf_global_send_type_out_Q_another}
 \\
 \forall i \in I: \,\,\vdash \mpP_i:\stEnvApp{\stEnv}{\roleP[i]}
 \label{eq:sf_global_send_M1_type}
 \\
 \forall i \in I:  \,\,\vdash \mpH[i]: \stEnvApp{\qEnv}{-, \roleP[i]}%
  \label{eq:sf_global_send_M1_type_Q}
\end{gather}
By~\eqref{eq:sf_global_send_type_out} and 3 of~\autoref{lem:typeinversion}, we have that
\begin{gather}
\stEnvApp{\stEnv}{\roleP} = \stOut{\roleQ}{\mpLab}{\tyGround}\stSeq{\stT}
\label{eq:sf_global_send_type_configuration}
\\
\stT[1] \stSub \stT
\label{eq:sf_global_send_subtyping}
\\
 \vdash \mpP : \stT[1]
 \label{eq:sf_global_send_type_configuration_2}
 \\
 \vdash \mpE:\tyGround
  \label{eq:sf_global_send_type_configuration_3}
\end{gather}
From~\eqref{eq:sf_global_send_cond_3}, by~\autoref{thm:gtype:proj-sound}, there are $\stEnvi; \qEnvi$ and
$\stEnvAnnotGenericSym =
\stEnvOutAnnot{\roleP}{\roleQ}{\stChoice{\stLab[k]}{\tyGround[k]}}$
 such that
\begin{gather}
  k \in I
  \label{eq:sf_global_send_M1_type_index}
  \\
\stEnvAssoc{ \gtWithCrashedRoles{\rolesC}{
      \gtCommTransit{\roleP}{\roleQ}{i \in I}{\gtLab[i]}{\tyGround[i]}{\gtGi[i]}{j}
}}{\stEnvi; \qEnvi}{\rolesR}
\label{eq:sf_global_send_type_association_new}
\\
\stEnv; \qEnv
      \stEnvMoveOutAnnot{\roleP}{\roleQ}{\stChoice{\stLab[k]}{\tyGround[k]}}
      \stEnvi; \qEnvi
\label{eq:sf_global_send_ctx_red}
\end{gather}
Using \eqref{eq:sf_global_send_type_configuration}, \eqref{eq:sf_global_send_M1_type_index}, and \inferrule{\iruleTCtxOut} in~\Cref{fig:gtype:tc-red-rules}, we get
\begin{gather}
\stEnvi = \stEnvUpd{\stEnv}{\roleP}{\stT}
\label{eq:sf_global_send_new_context}
\\
\qEnvi = \stEnvUpd{\qEnv}{\roleP, \roleQ}{
        \stQCons{
          \stEnvApp{\qEnv}{\roleP, \roleQ}
        }{
          \stQMsg{\stLab}{\tyGround}
        }
        }
\label{eq:sf_global_send_new_Q}
\end{gather}
It follows that
 \begin{gather}
 \vdash \mpP : \stEnvApp{\stEnvi}{\roleP} \quad \quad (\text{by}~\eqref{eq:sf_global_send_new_context}, \eqref{eq:sf_global_send_subtyping}, \eqref{eq:sf_global_send_type_configuration_2}\text{ and }\inferrule{{t-sub}})
 \\
  \vdash \mpH[\roleP] : \stEnvApp{\qEnvi}{-, \roleP} \quad  (\text{by}~\eqref{eq:sf_global_send_new_Q}\text{ and }\eqref{eq:sf_global_send_type_out_Q})
  \\
  \vdash \mpQ : \stEnvApp{\stEnvi}{\roleQ} \quad \quad (\text{by}~\eqref{eq:sf_global_send_type_out_another}\text{ and }\eqref{eq:sf_global_send_new_context})
  \\
  \vdash \mpH[\roleQ] : \stEnvApp{\qEnvi}{-, \roleQ} \quad  (\text{by}~\eqref{eq:sf_global_send_new_Q}, \eqref{eq:sf_global_send_type_out_Q_another}\text{ and } 9, 10\text{ of~\autoref{lem:typeinversion}})
  \\
   \forall i \in I: \,\,\vdash \mpP_i:\stEnvApp{\stEnvi}{\roleP[i]} \quad (\text{by}~\eqref{eq:sf_global_send_M1_type}\text{ and }\eqref{eq:sf_global_send_new_context})
 \\
 \forall i \in I:  \,\,\vdash \mpH[i]: \stEnvApp{\qEnvi}{-, \roleP[i]} \quad (\text{by}~\eqref{eq:sf_global_send_M1_type_Q}\text{ and }\eqref{eq:sf_global_send_new_Q})
 \end{gather}
Therefore, by \inferrule{r-send} and \inferrule{{t-sess}}, we can conclude that there exists $\mpMi = \mpPart\roleP{\mpP}
\mpPar
\mpPart\roleP{\mpH[\roleP]}
\mpPar
\mpPart\roleQ{\mpQ}
\mpPar
\mpPart{\roleQ}{\mpH[\roleQ]}\cdot(\roleP,\mpLab(\val))
\mpPar
\mpM[1]$ and $\gtWithCrashedRoles{\rolesCi}{\gtGi}
   =
    \gtWithCrashedRoles{\rolesC}{
      \gtCommTransit{\roleP}{\roleQ}{i \in I}{\gtLab[i]}{\tyGround[i]}{\gtGi[i]}{j}
    }$ such that
   $ \gtWithCrashedRoles{\rolesC}{\gtG}
    \gtMove[
      \stEnvOutAnnotSmall{\roleP}{\roleQ}{\stChoice{\gtLab[j]}{\tyGround[j]}}
    ]{
      \rolesR
    }
   \gtWithCrashedRoles{\rolesCi}{\gtGi}$,
   $\mpM \;\redSend{\roleP}{\roleQ}{\mpLab}\;\mpMi$, and
   $ \gtWithCrashedRoles{\rolesCi}{\gtGi} \vdash \mpMi$.

\item Case \inferrule{\iruleGtMoveIn}: similar to case \inferrule{\iruleGtMoveOut} above, except that we proceed by \inferrule{\iruleGtMoveIn}, \inferrule{r-rcv}, inversion of \inferrule{{t-ext}} (4 of~\autoref{lem:typeinversion}), and~\autoref{lem:substitution}.

\item Case \inferrule{\iruleGtMoveOrph}: by inversion of \inferrule{\iruleGtMoveOut}, we have
\begin{gather}
j \in I
 \label{eq:sf_global_orph_cond_1}
 \\
\gtLab[j] \neq \gtCrashLab
  \label{eq:sf_global_oprh_cond_2}
  \\
 \gtWithCrashedRoles{\rolesC}{\gtCommSmall{\roleP}{\roleQCrashed}{i \in
    I}{\gtLab[i]}{\tyGround[i]}{\gtGi[i]}}
    \gtMove[\stEnvOutAnnotSmall{\roleP}{\roleQ}{\stChoice{\gtLab[j]}{\tyGround[j]}}]{
      \rolesR
    }
    \gtWithCrashedRoles{\rolesC}{\gtGi[j]}
\label{eq:sf_global_orph_cond_3}
  \end{gather}
 We can assume that $\mpM$ is of the form
 \begin{gather}
 \mpPart\roleP{\procout{\roleQ}{\mpLab}{\mpE}{\mpP}}
\mpPar
\mpPart\roleP{\mpH[\roleP]}
\mpPar
\mpPart\roleQ{\mpCrash}
\mpPar
\mpPart\roleQ{\mpQUnavail}
\mpPar
\mpM[1]
\label{eq:sf_global_orph_M_form}
\\
\mpM[1] =  \prod_{i\in I} (\mpPart{\roleP[i]}{\mpP[i]} \mpPar
  \mpPart{\roleP[i]}{\mpH[i]})
\label{eq:sf_global_orph_M1_form}
\end{gather}
Then by~\eqref{eq:sf_ass_global_typing} and~\autoref{lem:typeinversion}, there
exists $\stEnv; \qEnv$ such that
\begin{gather}
\stEnvAssoc{\gtWithCrashedRoles{\rolesC}{\gtG}}{\stEnv; \qEnv}{\rolesR}
\label{eq:sf_global_orph_type_association}
\\
\vdash \procout\roleQ{\mpLab}{\mpE}{\mpP}:\stEnvApp{\stEnv}{\roleP}
\label{eq:sf_global_orph_type_out}
\\
 \vdash \mpH[\roleP]: \stEnvApp{\qEnv}{-, \roleP}
 \label{eq:sf_global_orph_type_out_Q}
 \\
 \vdash \mpCrash: \stEnvApp{\stEnv}{\roleQ}
 \label{eq:sf_global_orph_type_out_another}
 \\
  \vdash \mpQUnavail: \stEnvApp{\qEnv}{-, \roleQ}
 \label{eq:sf_global_orph_type_out_Q_another}
 \\
 \forall i \in I: \,\,\vdash \mpP_i:\stEnvApp{\stEnv}{\roleP[i]}
 \label{eq:sf_global_orph_M1_type}
 \\
 \forall i \in I:  \,\,\vdash \mpH[i]: \stEnvApp{\qEnv}{-, \roleP[i]}
  \label{eq:sf_global_orph_M1_type_Q}
\end{gather}
By~\eqref{eq:sf_global_orph_type_out} and 3 of~\autoref{lem:typeinversion}, we have that
\begin{gather}
\stEnvApp{\stEnv}{\roleP} = \stOut{\roleQ}{\mpLab}{\tyGround}\stSeq{\stT}
\label{eq:sf_global_orph_type_configuration}
\\
\stT[1] \stSub \stT
\label{eq:sf_global_orph_subtyping}
\\
 \vdash \mpP : \stT[1]
 \label{eq:sf_global_orph_type_configuration_2}
\end{gather}
From~\eqref{eq:sf_global_orph_cond_3}, by~\autoref{thm:gtype:proj-sound}, there are $\stEnvi; \qEnvi$ and
$\stEnvAnnotGenericSym =
\stEnvOutAnnot{\roleP}{\roleQ}{\stChoice{\stLab[k]}{\tyGround[k]}}$
 such that
\begin{gather}
  k \in I
  \label{eq:sf_global_orph_M1_type_index}
  \\
\stEnvAssoc{ \gtWithCrashedRoles{\rolesC}{
     \gtGi[j]
}}{\stEnvi; \qEnvi}{\rolesR}
\label{eq:sf_global_orph_type_association_new}
\\
\stEnv; \qEnv
      \stEnvMoveOutAnnot{\roleP}{\roleQ}{\stChoice{\stLab[k]}{\tyGround[k]}}
      \stEnvi; \qEnvi
\label{eq:sf_global_orph_ctx_red}
\end{gather}
Using \eqref{eq:sf_global_orph_type_configuration}, \eqref{eq:sf_global_orph_M1_type_index}, and \inferrule{\iruleTCtxOut} in~\Cref{fig:gtype:tc-red-rules}, we get
\begin{gather}
\stEnvi = \stEnvUpd{\stEnv}{\roleP}{\stT}
\label{eq:sf_global_orph_new_context}
\\
\qEnvi = \stEnvUpd{\qEnv}{\roleP, \roleQ}{
        \stQCons{
          \stEnvApp{\qEnv}{\roleP, \roleQ}
        }{
          \stQMsg{\stLab}{\tyGround}
        }
        }
\label{eq:sf_global_orph_new_Q}
\end{gather}
It follows that
 \begin{gather}
 \vdash \mpP : \stEnvApp{\stEnvi}{\roleP} \quad \quad (\text{by}~\eqref{eq:sf_global_orph_new_context}, \eqref{eq:sf_global_orph_subtyping}, \eqref{eq:sf_global_orph_type_configuration_2}\text{ and }\inferrule{{t-sub}})
 \\
  \vdash \mpH[\roleP] : \stEnvApp{\qEnvi}{-, \roleP} \quad  (\text{by}~\eqref{eq:sf_global_orph_new_Q}\text{ and }\eqref{eq:sf_global_orph_type_out_Q})
  \\
  \vdash \mpCrash : \stEnvApp{\stEnvi}{\roleQ} \quad \quad (\text{by}~\eqref{eq:sf_global_orph_type_out_another}\text{ and }\eqref{eq:sf_global_orph_new_context})
  \\
  \vdash \mpQUnavail : \stEnvApp{\qEnvi}{-, \roleQ} \quad  (\text{by}~\eqref{eq:sf_global_orph_new_Q}\text{ and }\eqref{eq:sf_global_orph_type_out_Q_another})
  \\
   \forall i \in I: \,\,\vdash \mpP_i:\stEnvApp{\stEnvi}{\roleP[i]} \quad (\text{by}~\eqref{eq:sf_global_orph_M1_type}\text{ and }\eqref{eq:sf_global_orph_new_context})
 \\
 \forall i \in I:  \,\,\vdash \mpH[i]: \stEnvApp{\qEnvi}{-, \roleP[i]} \quad (\text{by}~\eqref{eq:sf_global_orph_M1_type_Q}\text{ and }\eqref{eq:sf_global_orph_new_Q})
 \end{gather}
Therefore, by \inferrule{r-send-\mpCrash} and \inferrule{{t-sess}}, we can conclude that there exists $\mpMi = \mpPart\roleP{\mpP}
\mpPar
\mpPart\roleP{\mpH[\roleP]}
\mpPar
\mpPart\roleQ{\mpCrash}
\mpPar
\mpPart{\roleQ}{\mpQUnavail}
\mpPar
\mpM[1]$ and $\gtWithCrashedRoles{\rolesCi}{\gtGi}
   =
    \gtWithCrashedRoles{\rolesC}{
    \gtGi[j]
    }$ such that
   $ \gtWithCrashedRoles{\rolesC}{\gtG}
    \gtMove[
      \stEnvOutAnnotSmall{\roleP}{\roleQ}{\stChoice{\gtLab[j]}{\tyGround[j]}}
    ]{
      \rolesR
    }
   \gtWithCrashedRoles{\rolesCi}{\gtGi}$,
   $\mpM \;\redSend{\roleP}{\roleQ}{\mpLab}\;\mpMi$, and
   $ \gtWithCrashedRoles{\rolesCi}{\gtGi} \vdash \mpMi$.

\item Case \inferrule{\iruleGtMoveCrDe}: by inversion of \inferrule{\iruleGtMoveCrDe}, we have
\begin{gather}
j \in I
 \label{eq:sf_global_det_cond_1}
 \\
\gtLab[j] = \gtCrashLab
  \label{eq:sf_global_det_cond_2}
  \\
 \gtWithCrashedRoles{\rolesC}{
      \gtCommTransit{\roleQCrashed}{\roleP}{i \in I}{\gtLab[i]}{\tyGround[i]}{\gtGi[i]}{j}
    }
    \gtMove[\ltsCrDe{\mpS}{\roleP}{\roleQ}]{\rolesR}
    \gtWithCrashedRoles{\rolesC}{\gtGi[j]}
\label{eq:sf_global_det_cond_3}
  \end{gather}
 We can assume that $\mpM$ is of the form
 \begin{gather}
 \mpPart\roleP{\sum_{i\in I} \procin\roleQ{\mpLab_i(\mpx_i)}\mpP_i}
\mpPar
\mpPart\roleP{\mpH[\roleP]}
\mpPar
\mpPart\roleQ\mpCrash
\mpPar
\mpPart\roleQ\mpQUnavail
\mpPar
\mpM[1]
\label{eq:sf_global_det_M_form}
\\
\mpM[1] =  \prod_{i\in I} (\mpPart{\roleP[i]}{\mpP[i]} \mpPar
  \mpPart{\roleP[i]}{\mpH[i]})
\label{eq:sf_global_det_M1_form}
\\
k \in I
\label{eq:sf_global_det_M_form_index}
\\
 \mpLab[k] = \mpCrashLab
 \label{eq:sf_global_det_M_form_crash}
 \\
 \nexists \mpLab, \mpV: (\roleQ,
\mpLab(\mpV)) \in \mpH[\roleP]
\label{eq:sf_global_det_M_form_empty_Q}
\end{gather}
Then by~\eqref{eq:sf_ass_global_typing} and~\autoref{lem:typeinversion}, there
exists $\stEnv; \qEnv$ such that
\begin{gather}
\stEnvAssoc{\gtWithCrashedRoles{\rolesC}{\gtG}}{\stEnv; \qEnv}{\rolesR}
\label{eq:sf_global_det_type_association}
\\
\vdash \sum_{i\in I}\procin{\roleQ}{\mpLab_i(\mpx_i)}{\mpP_i}:\stEnvApp{\stEnv}{\roleP}
\label{eq:sf_global_det_type_out}
\\
 \vdash \mpH[\roleP]: \stEnvApp{\qEnv}{-, \roleP}
 \label{eq:sf_global_det_type_out_Q}
 \\
 \vdash \mpCrash: \stEnvApp{\stEnv}{\roleQ}
 \label{eq:sf_global_det_type_out_another}
 \\
  \vdash \mpQUnavail: \stEnvApp{\qEnv}{-, \roleQ}
 \label{eq:sf_global_det_type_out_Q_another}
 \\
 \forall i \in I: \,\,\vdash \mpP_i:\stEnvApp{\stEnv}{\roleP[i]}
 \label{eq:sf_global_det_M1_type}
 \\
 \forall i \in I:  \,\,\vdash \mpH[i]: \stEnvApp{\qEnv}{-, \roleP[i]}
  \label{eq:sf_global_det_M1_type_Q}
\end{gather}
It follows directly that
\begin{gather}
\stEnvApp{\stEnv}{\roleQ} = \stStop
\label{eq:sf_global_det_stop}
\\
 \stEnvApp{\qEnv}{-, \roleQ} = \stQUnavail
 \label{eq:sf_global_det_unavail}
\end{gather}

By~\eqref{eq:sf_global_det_type_out}, 4 of~\autoref{lem:typeinversion}, and \inferrule{\iruleStSubIn}, we have that
\begin{gather}
\stEnvApp{\stEnv}{\roleP} = \stExtSum{\roleQ}{i \in J}{\stChoice{\stLab[i]}{\tyGround[i]} \stSeq \stTi[i]}
\label{eq:sf_global_det_type_configuration}
\\
J \subseteq I
\label{eq:sf_global_det_subset}
\\
\forall i \in J:  \stT[i] \stSub \stTi[i]
\label{eq:sf_global_det_subtyping}
\\
 \setcomp{\stLab[l]}{l \in I} \neq \setenum{\stCrashLab}
 \label{eq:sf_global_det_not_crash}
\\
 \nexists j \in I \setminus J: \stLab[j] = \stCrashLab
 \label{eq:sf_global_det_not_crash_2}
 \\
 \forall i \in I: x_i : \tyGround[i] \vdash \mpP[i] : \stT[i]
 \label{eq:sf_global_det_type_configuration_2}
\end{gather}
From \eqref{eq:sf_global_det_not_crash_2} and \eqref{eq:sf_global_det_M_form_crash}, we get
\begin{align}
   k \in J
   \label{eq:sf_global_det_k_in_J_1}
\end{align}
By \eqref{eq:sf_global_det_type_out_Q} and \eqref{eq:sf_global_det_M_form_empty_Q}, we also know
\begin{align}
 \stEnvApp{\qEnv}{\roleQ, \roleP} = \stQEmpty
 \label{eq:sf_global_det_empty_Q}
\end{align}

From~\eqref{eq:sf_global_det_cond_3}, by~\autoref{thm:gtype:proj-sound}, there is $\stEnvi; \qEnvi$ such that
\begin{gather}
\stEnvAssoc{ \gtWithCrashedRoles{\rolesC}{
     \gtGi[j]
}}{\stEnvi; \qEnvi}{\rolesR}
\label{eq:sf_global_det_type_association_new}
\\
\stEnv; \qEnv
      \stEnvMoveAnnot{\ltsCrDe{\mpS}{\roleP}{\roleQ}}
      \stEnvi; \qEnvi
\label{eq:sf_global_det_ctx_red}
\end{gather}

Using \eqref{eq:sf_global_det_type_configuration}, \eqref{eq:sf_global_det_stop}, \eqref{eq:sf_global_det_k_in_J_1},
\eqref{eq:sf_global_det_M_form_crash}, \eqref{eq:sf_global_det_empty_Q}, and \inferrule{\iruleTCtxCrashDetect} in \Cref{fig:gtype:tc-red-rules}, we get
\begin{gather}
\stEnvi =  \stEnvUpd{\stEnv}{\roleP}{\stTi[k]}
\label{eq:sf_global_det_new_context}
\\
\qEnvi = \qEnv
\label{eq:sf_global_det_new_Q}
\end{gather}
It follows that
 \begin{gather}
  \vdash \mpP[k] : \stEnvApp{\stEnvi}{\roleP} \quad \quad (\text{by}~\eqref{eq:sf_global_det_new_context}, \eqref{eq:sf_global_det_subtyping}, \eqref{eq:sf_global_det_type_configuration_2}\text{ and }\inferrule{{t-sub}})
 \\
  \vdash \mpH[\roleP] : \stEnvApp{\qEnvi}{-, \roleP} \quad  (\text{by}~\eqref{eq:sf_global_det_new_Q}\text{ and }\eqref{eq:sf_global_det_type_out_Q})
  \\
  \vdash \mpCrash : \stEnvApp{\stEnvi}{\roleQ} \quad \quad (\text{by}~\eqref{eq:sf_global_det_type_out_another},\eqref{eq:sf_global_det_stop}\text{ and }\eqref{eq:sf_global_det_new_context})
  \\
  \vdash \mpQUnavail: \stEnvApp{\qEnvi}{-, \roleQ} \quad  (\text{by}~\eqref{eq:sf_global_det_new_Q}, \eqref{eq:sf_global_det_type_out_Q_another}\text{ and }\eqref{eq:sf_global_det_unavail})
  \\
   \forall i \in I: \,\,\vdash \mpP_i:\stEnvApp{\stEnvi}{\roleP[i]} \quad (\text{by}~\eqref{eq:sf_global_det_M1_type}\text{ and }\eqref{eq:sf_global_det_new_context})
 \\
 \forall i \in I:  \,\,\vdash \mpH[i]: \stEnvApp{\qEnvi}{-, \roleP[i]} \quad (\text{by}~\eqref{eq:sf_global_det_M1_type_Q}\text{ and }\eqref{eq:sf_global_det_new_Q})
 \end{gather}
Therefore, by \inferrule{r-rcv-$\odot$} and \inferrule{{t-sess}}, we can conclude that there exists $\mpMi = \mpPart\roleP \mpP_k
\mpPar
\mpPart\roleP{\mpH_{\roleP}}
\mpPar
\mpPart\roleQ\mpCrash
\mpPar
\mpPart\roleQ\mpQUnavail
\mpPar
\mpM[1]$ and $\gtWithCrashedRoles{\rolesCi}{\gtGi}
   =
    \gtWithCrashedRoles{\rolesC}{
    \gtGi[j]
    }$ such that
   $ \gtWithCrashedRoles{\rolesC}{\gtG}
   \gtMove[\ltsCrDe{\mpS}{\roleP}{\roleQ}]{\rolesR}
      \gtWithCrashedRoles{\rolesCi}{\gtGi}$,
   $\mpM \;\redRecv{\roleP}{\roleQ}{\mpLab_k}\;\mpMi$, and
   $ \gtWithCrashedRoles{\rolesCi}{\gtGi} \vdash \mpMi$.
   
   \item Case \inferrule{\iruleGtMoveRec}:  follows by~\autoref{lem:recursion-unfolding-typing-session} and inductive hypothesis. 
   
   \item Cases~\inferrule{\iruleGtMoveCtx} and~\inferrule{\iruleGtMoveCtxi}: these two cases do not need to be considered, 
   as the reduction $\gtWithCrashedRoles{\rolesC}{\gtG} \gtMove{\rolesR}$ can always be triggered by the  preceding cases.    
   \qedhere
   \end{itemize}
\end{proof}

\lemSessionDF*
\begin{proof}

Assume that $\gtWithCrashedRoles{\rolesC}{\gtG} \vdash \mpM$ and $\mpM \mpMoveStar[\rolesR] \mpMi \mpNotMove[\rolesR]$, we need to prove that either 
$\mpMi \Rrightarrow \mpPart\roleP\mpNil \mpPar \mpPart\roleP \mpQEmpty$ for some $\roleP$, or $\mpMi \Rrightarrow \prod_{i\in I} (\mpPart{\roleP[i]}{\mpCrash} \mpPar \mpPart{\roleP[i]}{\mpQUnavail})$ with $I \neq \emptyset$.

By applying~\autoref{lem:sr} (subject reduction) repeatedly as needed, there exists a global type 
$\gtWithCrashedRoles{\rolesCi}{\gtGi}$ such that 
$\gtWithCrashedRoles{\rolesC}{\gtG} \!\gtMoveStar[\rolesR]\!
\gtWithCrashedRoles{\rolesCi}{\gtGi}$ and $\gtWithCrashedRoles{\rolesCi}{\gtGi} \vdash \mpMi$. 
Since no further reductions are possible for 
$\mpMi$, by~\autoref{lem:sf} (session fidelity),  
the global type 
$\gtWithCrashedRoles{\rolesCi}{\gtGi}$ cannot be reduced further either, implying that 
$\gtWithCrashedRoles{\rolesCi}{\gtGi}$ must be of the form 
$\gtWithCrashedRoles{\rolesCi}{\stEnd}$. 

Furthermore, using~\inferrule{t-sess} and the proof for~\autoref{lem:proj:df}, we can conclude that 
$\mpMi$ is of the form  $\prod_{i\in I} (\mpPart{\roleP[i]}{\mpPi[i]} \mpPar
  \mpPart{\roleP[i]}{\mpHi[i]})$ 
  such that $\forall i\in I: \, 
\vdash \mpPi[i]:\stEnvApp{\stEnvi}{\roleP[i]}$, 
$\vdash \mpHi[i]: \stEnvApp{\qEnvi}{-, \roleP[i]}$, %
  where $\stEnvi =  \stEnvi[\stEnd] \stEnvComp \stEnvi[\stStopSym]$, with 
$\forall \roleP \in \dom{\stEnvi[\stEnd]}: \stEnvApp{\stEnvi}{\roleP} = \stEnd$, 
$\forall \roleP \in \dom{\stEnvi[\stStopSym]}: \stEnvApp{\stEnvi}{\roleP} = \stStop$, and 
for any $\roleP, \roleQ$, if $\roleQ \in \dom{\stEnvi[\stStopSym]},
 \stEnvApp{\qEnvi}{\cdot, \roleQ} = \stQUnavail$, and otherwise, $\stEnvApp{\qEnvi}{\roleP, \roleQ} = \stQEmpty$. 

Then, by applying \inferrule{t-$\mpQEmpty$}, \inferrule{t-$\mpQUnavail$}, \inferrule{t-$\mpCrash$}, \inferrule{t-$\mpNil$}, 
and the precongruence rules 
$\mpPart\roleP\mpNil
\mpPar
\mpPart\roleP\mpQEmpty
\mpPar
\mpM
\Rrightarrow
\mpM$ and $\mpM \Rrightarrow \mpMi
\text{ and }
\mpMi \Rrightarrow \mpMii
\implies
\mpM
\Rrightarrow
\mpMii$, the thesis follows. 
\qedhere 
\end{proof}

\lemSessionLive*
\begin{proof}

The proof follows the same structure as~\autoref{lem:session_deadlock_free}, but instead uses the definition of live sessions~(\autoref{def:session_live}), the proof for~\autoref{lem:ext-proj-live}, and the relevant typing rules in~\Cref{fig:processes:typesystem}. 
\qedhere 
\end{proof}

\end{document}